\documentclass[master,intern,english,e-technik,12pt,itpaper]{ituuthesis}

\usepackage{nomencl}
\let\abbrev\nomenclature
\renewcommand{\nomname}{Abbreviations}
\makenomenclature
\newcommand{\Listofabbrev}{\printnomenclature}
\usepackage[justification=centering]{caption} 
\graphicspath{{PDF_figures/}} 
\usepackage{subfig} 
\usepackage{amsmath} 
\usepackage{amssymb} 
\newcommand{\argmax}{\operatornamewithlimits{argmax}} 
\newcommand{\argmin}{\operatornamewithlimits{argmin}} 
\newcommand{\mini}{\operatornamewithlimits{min}}
\newcommand{\rank}{\text{rank}}
\newtheorem{theorem}{Theorem}[section]	
\newenvironment{definition}[1][Definition]{\begin{trivlist}
\item[\hskip \labelsep {\bfseries #1}]}{\end{trivlist}}

\newenvironment{proof}[1][Proof]{\begin{trivlist}
\item[\hskip \labelsep {\bfseries #1}]}{\end{trivlist}}
\usepackage{colortbl, pstricks, colortab} 
\providecommand{\shadeRow}{\rowcolor[rgb]{0.9, 0.9, 0.9}}
\DeclareMathOperator{\tr}{tr} 
\DeclareMathOperator{\diag}{diag}
\usepackage{nicefrac}
\usepackage{units} 
\usepackage{accents} 

\usepackage{array} 
\newcommand{\hr}{\sp\dagger} 
\newcommand{\eq}{\!=\!}
\newcommand\ST{\rule{0pt}{2.5em}} 
\newcommand{\V}[1]{\boldsymbol{#1}} 
\newcommand{\definedas}{\overset{\underset{\Delta}{}}{=}}
\newcommand{\req}{\overset{\underset{!}{}}{=}}
\newcommand{\E}{\text{E}}
\newcommand{\tn}{\tabularnewline} 
\usepackage{arydshln} 
\usepackage[plainpages=false,pdfborder={0 0 0}, bookmarksnumbered=true]{hyperref}
\makeatletter
\usepackage{etoolbox}
\pretocmd{\tableofcontents}{%
  \if@openright\cleardoublepage\else\clearpage\fi
  \pdfbookmark[0]{\contentsname}{toc}%
}{}{}%
\makeatother

\begin{document}


\premake{all}{ %
  \addtolength{\oddsidemargin}{5mm}
  \addtolength{\evensidemargin}{-5mm}}
\make[blank]{titlepage}
\make[blank]{description}
\make[blank]{authentication}
\chapter*{Acknowledgement}
First and above all, I thank God, the almighty for all His given blessings in my life and for granting me the capability to proceed successfully in my studies. I am heartily thankful to all those who supported me in any respect toward the completion of my Master thesis. I owe my deepest gratitude to my supervisor Alexander Linduska who supported me with his great guidance and rich knowledge while allowing me the room to work in my own way. His well-explained illustrations enabled me to develop a solid understanding of the subject. I sincerely appreciate his proficient review on the thesis and his valuable comments. Without his kind help in both the practical work as well as the documentation, this thesis would not have been possible. \\

It is an honor for me to acknowledge the German academic exchange service DAAD which offered me a scholarship  during the whole period of my master studies in the international master program "Communications Technology". I would also like to thank Ulm University and in particular the institute of Information Technology for providing all necessary facilities required by this project.\\

Last, but by no means least, I would like to express my deep thanks to my lovely parents Magda and Moustafa and my dear sister Mai for their personal support, continuous prayers and great patience at all times. Thanks to them for always being my main source of motivation for success. It is through the encouragement of my family that I travelled to Germany for my post graduate studies. Such a trip, which enriched my experience on the academic, cultural and personal level. 
%

\pagenumbering{roman} 
\tableofcontents

\chapter*{Notations}
\markboth{Notations}{Notations}
\addcontentsline{toc}{chapter}{Notations}
\label{chap:notations}
\renewcommand{\arraystretch}{1} 
\begin{table}[htp]
\begin{center}  
\begin{tabular}{p{5.2cm}p{10cm}}  
$\thicksim$ & is distributed as \\
$j=\sqrt{-1}$ & imaginary unit\\
$\Re\{x\}$ or $x^R$ & the real part of $x$\\
$\Im\{x\}$ or $x^I$ & the imaginary part of $x$\\
$|x|=\sqrt{(x^R)^2+(x^I)^2}$ & the norm of the complex scalar $x$\\
$x^*$ & the complex conjugate of $x$\\
$\lceil x \rceil$ & the smallest integer not less than x\\
$d_{Eu}(x,y)=(x-y)^*(x-y)$ &is the Euclidean distance between $x$ and $y$\\
$\V{A}$ & uppercase bold letters denote matrices\\
$\V{x}$ &lowercase bold letters denote vectors \\
$\V{A}_{M\!\times\! N}=\{a_{mn}\}$ & matrix $\V{A}$ with dimensions $M\!\times\! N$ and entries $a_{mn}$ on the $m$th row and the $n$th column\\
$\V{I}_{M}$ &an $M\!\times\! M$ identity matrix\\
$\V{0}_{M\!\times\!N}$ & all zero matrix of dimension $M\!\times\! N$\\
$\V{1}_{M\!\times\!N}$ & all ones matrix of dimension $M\!\times\! N$\\
$\det(\V{A})$ & the determinant of a square matrix $\V{A}$\\
$\rank(\V{A})$ & the rank of matrix $\V{A}$\\
tr$(\V{A}_{N\!\times\! N})=\sum \limits_{i=1}^{N}{a_{ii}}$ &the trace of a square matrix $\V{A}$ \\
$\V{A}^T$ & transpose of matrix $\V{A}$\\
$\V{A}\hr$ &complex conjugate transpose (Hermitian) of matrix $\V{A}$ \\
$\sigma_{m}(\V{A})$ &the $m^{\text{th}}$ singular value of matrix $\V{A}$\\
$\lambda_{m}(\V{A})$ &the $m^{\text{th}}$ eigenvalue of matrix $\V{A}$\\
$\|\V{x}\|^2=\V{x}\hr\V{x}=\sum\limits_{i}{|x_i|^2}$ & the square of the norm of vector $\V{x}$\\
$|\V{A}|^2$ & a square matrix defined by $\V{A}\hr\V{A}$ (not to be confused with the scalar $\|\V{A}\|_F^2$ defined next)\\
$ \left\|\V{A}\right\|_{F}^{2}\!=$tr$(\V{A}\hr \V{A})\!=$tr$(\V{A}\V{A}\hr)$&the squared Frobenius norm of a matrix $\V{A}$ \\
$\mathcal{CN}(0,1)$ & the complex normal distribution with zero mean and unit variance, where the real and imaginary components are independent and each has a variance $1/2$\\
$\mathbb{C}^{M\!\times\! N}$ &the set of complex matrices of dimension $M\!\times\!N$\\
$\mathbb{R}^{M\!\times\! N}$ &the set of real matrices of dimension $M\!\times\!N$\\
$\E[x]$ &statistical average of the random variable $x$
\end{tabular}
\end{center}
\end{table}
\newpage
\begin{table}[htp]
\begin{center}  
\begin{tabular}{p{5.2cm}p{10cm}}  
$\V{\mu_{x}}=\E[\V{x}]$ &the mean vector of the random vector $\V{x}$\\
$\V{\Lambda_{x}}=\E[(\V{x}-\V{\mu_{x}})(\V{x}-\V{\mu_{x}})\hr]$ & the covariance matrix of the random vector $\V{x}$\\
$\mathcal{CN}(\V{\mu_{x}}, \V{\Lambda_{x}})$ &circularly symmetric complex Gaussian distributed random vector $\V{x}$, with a mean vector $\V{\mu_{x}}$ and a covariance matrix $\V{\Lambda_{x}}$ \\
p$(A|B)$& denotes the probability that event A occurs given that event B happened.\\
$\text{p}_{X}(x)$ & denotes the probability that random variable $X$ takes value $x$.\\
\end{tabular}
\end{center}
\end{table}


\cleardoublepage
\phantomsection
\addcontentsline{toc}{chapter}{\nomname}
\Listofabbrev
\markboth{\nomname}{\nomname}
\clearpage
\pagenumbering{arabic}

\chapter{Introduction}
 \label{chap:Introduction}


Since wireless communications was born, several challenges face its progress including the limited available radio spectrum as well as the complicated nature
of the wireless environment being time-variant and continuously fading. Thus, successfully receiving the transmitted information is a non-trivial task. One way to restore the information at the receiver is to estimate the channel and then use it for subsequent detection. This is possible by transmitting informationless symbols that are known to the receiver. However, if the channel changes fast relative to the symbol rate, the channel estimation block will fail to track these changes which in turn spoils the information detection. Moreover, in multiple antenna systems, the channel between every transmit and receive antenna pair is to be estimated which increases the receiver complexity and the transmission overhead. All the addressed impracticalities necessitate the need of transmission schemes that do not require channel knowledge at the receiver, the so-called non-coherent schemes. Differential modulation is one such scheme.\\

This thesis starts by addressing differential modulation schemes in single carrier single antenna systems. Differential Phase Shift Keying (DPSK) is a well-known scheme that encodes the information on the phase difference between two successive transmit symbols. However, the constellation symbols of DPSK do not best utilize the complex space, leading to performance degradation as more bits are to be transmitted per symbol. At the same time, the limited available radio spectrum demands the need of more spectrally efficient schemes. To this end, we investigate the use of amplitude modulation in the differential domain through a scheme known as Differential Amplitude Phase Shift Keying (DAPSK). For further performance enhancement of differential systems, multiple symbols can be jointly decoded through a technique known as Multiple-Symbol Differential Detection (MSDD). This technique bridges the performance gap between coherent and non-coherent systems and removes the error floor associated with fast-fading channels.\\

The use of multiple antennas at the transmitter and/or the receiver is a promising solution to combat fading and improve the reliability of transmission. Such antennas create several links between the transmitter and the receiver reducing the probability of a simultaneous fade to all links. The resulting performance improvement (known as diversity) does not require additional bandwidth nor power, the two most precious resources in wireless communications. Achieving diversity using multiple antennas at the receiver (receive diversity) is proved to be a much simpler task than using multiple antennas at the transmitter (transmit diversity). However in downlink mobile communications, mounting multiple antennas on the receiving mobile handsets results in an increase in the size and cost of mobiles. This opposes the ongoing desire in making mobiles as low profile and cheap as possible. Furthermore, since there are thousands more mobile stations than base stations, it is economically better to place the complexity in the transmitting base stations and make the mobile terminals as simple as possible. This motivates the need of transmit diversity.\\

Achieving transmit diversity through multiple transmit antennas is known as Space-Time Coding. Such a technique attracted the research interest in the last decade. It is the goal of this thesis to investigate STC schemes in differential non-coherent systems. We start by investigating a class of ST codes known as Orthogonal Space-Time Block Codes (OSTBC). This class achieves good performance with linear decoding complexity. As an attempt to increase the data rate and improve the transmission quality, the orthogonality requirement of OSTBCs is relaxed leading to the so-called Quasi-Orthogonal STBCs (QOSTBC). For most QOSTBCs proposed in the literature, the rate and performance advantage comes at the expense of increased decoding complexity. Recently a new class of QOSTBCs known as Minimum Decoding Complexity QOSTBCs (MDC-QOSTBC) has proved to achieve such advantage with the same complexity requirement of OSTBCs. These codes have been used in coherent systems that assume perfect channel knowledge at the receiver. In this thesis, we propose the use of MDC-QOSTBC in differential non-coherent systems. The performance of the proposed scheme proved to be superior to differential OSTBC schemes especially for high spectral efficiencies.\\

The thesis is organized as follows. Chapter \ref{chap:DPSK_DAPSK} introduces the DAPSK scheme and derives the MSDD receiver metrics in differential single antenna systems. Then Chapter \ref{chap:Diversity_MIMO} starts the investigations on multiple antenna systems by explaining how receive diversity is achieved in coherent and non-coherent systems. It then addresses the difficulty of achieving transmit diversity. Chapter \ref{chap:STC} follows with the introduction of space-time codes with a thorough literature review. A detailed proof on the design criteria of ST codes in the non-coherent domain is provided together with a derivation of the Maximum Likelihood (ML) metric. In Chapter \ref{chap:DSTC_schemes}, several differential OSTBCs are investigated and their error performance is analyzed. Afterwards the theory of MDC-QOSTBCs is studied in Chapter \ref{chap:QOSTBCs}. Their use in the differential domain is evaluated and compared to the previously addressed schemes in terms of complexity and performance. The last chapter concludes the thesis and suggests open problems for future research.\\

Note that due to the different concepts addressed in single antenna and multiple antenna systems, in few cases the same symbols are used in both systems with different meanings. For the used notations, a common list is provided on page \pageref{chap:notations} and is valid for the entire thesis. For the purpose of aiding the understanding of the thesis content, some fundamentals of linear Algebra are summarized in Appendix \ref{App:math_formulas}. 
%
 
\chapter{Differential Modulation Techniques in Single Antenna Systems}
\label{chap:DPSK_DAPSK}

On addressing the problem of data detection at the receiver in case of wireless transmission over an unknown channel, one solution is to estimate the channel and use the estimate for subsequent detection. Channel estimation is a technique that sends training symbols that are known to the receiver, and uses the corresponding received symbols to estimate the channel coefficients. Such a technique assumes the channel is approximately constant over a relatively long time period. In this case the receiver will have full channel knowledge, and the detection scheme is denoted as \textit{coherent} detection. Having full Channel State Information (CSI) at the receiver is however not reasonable when the channel changes fast relative to the symbol rate. In this case, the validity of the estimates spans only over a few symbol periods, thus forcing the system to repeat the estimation more often. Moreover, channel estimation adds to the complexity of the receiver and the transmission overhead. Here comes the advantage of using non-coherent systems like differential modulation schemes.\\

This chapter addresses two types of differential modulation techniques in single antenna systems, namely Differential Phase Shift Keying (DPSK) and Differential Amplitude Phase Shift Keying (DAPSK). DPSK is a technique that modulates the information over the phase difference between two successive transmit symbols. Each transmit symbol is the product of the previous transmit symbol and the new information symbol. In DAPSK, part of the symbol bits is carried on the phase difference between two successive transmit symbols and the other part is carried on the amplitude ratio between both symbols. If the channel is considered constant over at least two successive symbols, then both symbols will suffer the same amplitude and/or phase distortion, hence the constellation translation and/or rotation caused by the channel can be removed by considering both symbols in the detection. This is the conventional non-coherent detection of DAPSK/DPSK.\\

One disadvantage of conventional non-coherent detection is that it results in a performance degradation of about $\unit[3]{dB}$ compared to coherent detection. One way to partially compensate for such loss is to increase the receiver memory, where instead of deciding on a symbol by symbol basis, we decide on a symbol sequence basis \cite{Divsalar94}. Such a technique is known as Multiple-Symbol Differential Detection (MSDD). This chapter starts by introducing the DAPSK modulation scheme in Section \ref{s:DAPSK}. Then Section \ref{s:MSDD} derives the MSDD ML decision metric  and another sub-optimum less complex metric for different channel conditions. The last section shows the error performance of the MSDD receiver using both DPSK and DAPSK modulation techniques. 	 
\section{DAPSK}
\label{s:DAPSK}
In DPSK, the transmit signal $s_{t,p}$ at time slot $t$ is constructed as
\begin{equation}
\label{eq:DPSK}
 s_{t,p}=v_{z_{t,p}}s_{t-1,p}\,
\end{equation}
where $z_{t,p}$ is the information integer generated by the  transmitter at time slot $t$. The subscript $p$ is used to signify that the equation describes phase modulation. $v_{z_{t,p}}$ is the corresponding information symbol drawn from a DPSK alphabet (same as PSK alphabet) defined as
\begin{equation}
\label{eq:PSK_alphabet}
\mathcal{A}_\text{DPSK}=\{v_d=e^{\frac{j2\pi d}{q_p}}\, ,\:\: d\in\{0,..,q_p-1\}\},
\end{equation}
where $q_p$ is the DPSK alphabet size (order) and $j$ is the imaginary unit defined as $\sqrt{-1}$. From (\ref{eq:DPSK}), the transmit signal at any time slot $t$ can be equivalently written as
\begin{equation}
 \label{eq:DPSK_2}
s_{t,p}=\Big(\prod\limits_{\kappa=1}^{t}{v_{z_\kappa,p}}\Big)s_{0,p},\:\:t=1,2,...,
\end{equation}
where $s_{0,p}$ is the very first transmit signal used to initiate the differential transmission. In fact this is the only informationless signal in the whole transmission and is chosen to be any arbitrary symbol in $\mathcal{A}_\text{DPSK}$ (usually $s_{0,p}\eq1$ is used). \\

Using only the phase of the signal for information transmission will cause  performance degradation as more bits are to be transmitted per second. Carrying some part of the information on the amplitude will result in a better distribution of the constellation points in the complex space, potentially leading to a better error performance. In the context of differential modulation, one technique that differentially carries information on both the amplitude and the phase is the DAPSK technique. \\

There are several constructions of DAPSK constellations. Here we consider a DAPSK constellation which can modulate the phase bits independently from the amplitude bits. In other words, an independent DAPSK constellation satisfies that for every amplitude level, one sees all possible phase levels and for every phase level, one sees all possible amplitude levels. This potentially may make it possible to perform independent demodulation for the information embedded in the amplitude and the phase simplifying the demodulation process. \\

For the purpose of differential encoding of the amplitude, one may define the DASK alphabet in a way that imitates the DPSK alphabet by defining the amplitude base $a$ to be analogous to the phasor base $e^{\frac{j2\pi}{q_p}}$. In DPSK the phasor base is the incremental step (or ratio) between the $q_p$ phase levels. In DASK, $a$ is also defined as the amplitude ratio between the amplitude levels. That is why it is commonly known as the ring ratio. In DPSK, the $d^{\text{th}}$ constellation point is the base phasor raised to power $d$. Similarly, in DASK the $d^{\text{th}}$ amplitude level is the amplitude base $a$ raised to power $d$. Thus, the DASK alphabet is defined as 
\begin{equation}
 \label{eq:DASK_alphabet}
\mathcal{A}_\text{DASK}=\{v_d=a^d\, ,\:\: d\in\{0,..,q_a-1\}\},
\end{equation}
where $q_a$ is the DASK alphabet size. In the independent constellation considered, the final DAPSK alphabet is having the form of concentric DPSK circles (all of the same order) with radii set by the DASK amplitude levels. The DAPSK alphabet size is then $q_{a,p}\eq q_aq_p$. Figure \ref{fig:64_DAPSK} shows the DAPSK constellation for a 64-DAPSK alphabet constructed as 16-DPSK and 4-DASK. \\
\begin{figure}[htp]
\centering
\scalebox{0.7}{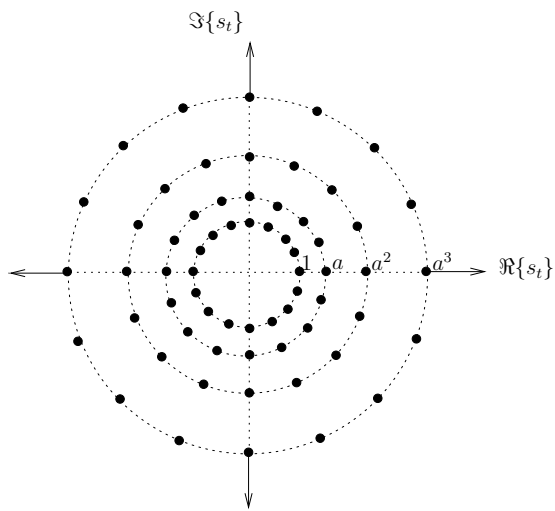}
\caption{Constellation of the 64-DAPSK alphabet}
 \label{fig:64_DAPSK}
\end{figure}

The differential encoding of DAPSK is described in Figure \ref{fig:DAPSK_mod}. The user bit stream is divided into blocks of $n$ bits. Each block is divided into phase bits and amplitude bits. The phase bits select the phase information symbol using the DPSK alphabet defined in (\ref{eq:PSK_alphabet}). Then the chosen information symbol is differentially encoded using (\ref{eq:DPSK}) resulting in a DPSK signal. Due to the inherent \emph{group} nature of the DPSK alphabet, the product of any two elements in $\mathcal{A}_\text{DPSK}$ is also an element in $\mathcal{A}_\text{DPSK}$. In other words, the product of any two points in the DPSK constellation circle falls to a point on the same circle. This is because the product of two phasors is a phasor whose angle is the modulo addition of the angles of the two phasors with respect to $2\pi$. Therefore, the transmit DPSK signal $s_{t,p}$ is also element in $\mathcal{A}_{\text{DPSK}}$ $\forall \,t$.\\

\begin{figure}[htp]
\centering
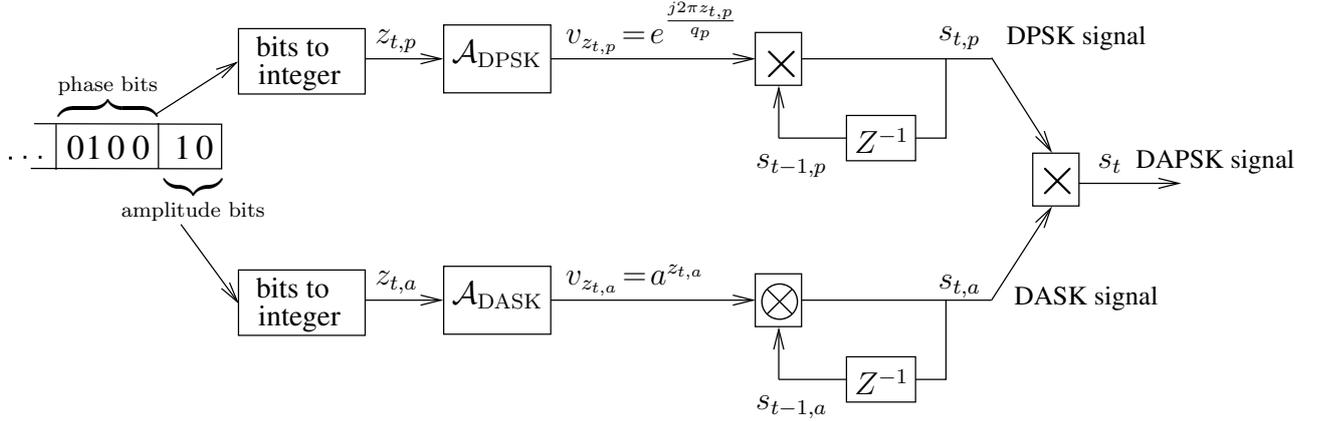
\caption{A description of DAPSK modulation using an independent constellation}
 \label{fig:DAPSK_mod}
\end{figure}

Unlike the DPSK alphabet, the DASK alphabet does not inherently possess a group nature. This means if differential encoding of the amplitude just performs a multiplication operation of the information amplitude with the last transmitted amplitude, the transmit amplitude may increase or decrease indefinitely. To avoid this, we force the DASK alphabet to form a group under multiplication by defining the operator $\otimes$ as
\begin{equation}
 a^{d_1}\otimes a^{d_2}=a^{(d_1+d_2)\,\text{mod}\,q_a}.
\end{equation}
The modulo operation on the power of $a$ ensures that when  the operator $\otimes$ applies on members of $\mathcal{A}_{\text{DASK}}$, the outcome is also a member in the same alphabet. The power of $a$ belongs to a finite cyclic group $G=\{0,...,q_a-1\}$ of order $q_a$ and generator '1'. The group is defined by the operation $d_1\oplus d_2=(d_1+d_2)\text{mod}\,q_a$, where $d_1$ and $d_2$ $\in G$. Through group isomorphism, the group $G$ can be mapped to another group $a^G$ which forms the DASK alphabet. In this case, the differential encoding of the amplitude signal is done  similar to the phase signal as
\begin{equation}
\label{eq:DASK}
 s_{t,a}=v_{z_{t,a}}\otimes s_{t-1,a},
\end{equation}
where the subscript $a$ is used to represent the amplitude modulation. $z_{t,a}$ is the information integer of the amplitude bits, $v_{z_{t,a}}$ is drawn from the DASK alphabet in (\ref{eq:DASK_alphabet}) and $s_{t,a}$ is the transmitted amplitude at time slot $t$, and is also $\in$ $\mathcal{A}_{\text{DASK}}$. Finally, the total DAPSK signal is just the product of the amplitude and phase signals, i.e.
\begin{equation}
\label{eq:DAPSK}
 s_t=s_{t,a}\,s_{t,p}.
\end{equation}

\subsection{Optimal ring ratio and bits distribution}
\label{ss:Ring_ratio_bits_distribution}
In order to optimize the DAPSK constellation, we address the following two questions: 1) what is the optimum ring ratio $a$? 2) what is the optimum assignment of the $n$ bits to the amplitude and phase parts? Based on investigations, we provide answers to the addressed questions for several DAPSK modulation orders.\\

Consider the DAPSK alphabet as the one shown in Figure \ref{fig:64_DAPSK}. Increasing the ring ratio $a$ will set the PSK circles further apart. This at the beginning will have the effect of separating the symbols of the different circles, and therefore providing more robustness for the amplitude detection. However, --for a constant average energy of the constellation symbols-- the higher the ring ratio $a$ gets, the smaller will be the innermost circle. Hence the probability of mistaking the phase symbols on this circle increases and the phase error performance deteriorates. This indicates the existence of an optimal value of $a$ which results in the optimal \emph{overall} performance of the DAPSK modulation. \\

To see the effect of $a$ on the error performance, simulations have been made over an Additive White Gaussian Noise (AWGN) channel. To start with, consider a 16-DAPSK alphabet constructed from 8-DPSK and 2-DASK. The rate of phase and amplitude bit errors are both measured and are referred to as DPSK and DASK Bit Error Rate (BER). Figure \ref{fig:16_DAPSK_a_influence} shows the DPSK BER, the DASK BER and the total DAPSK BER at $E_b/N_0=\unit[18]{dB}$ \footnote{$E_b/N_0$ is the average energy per information bit to noise power spectral density ratio.} as a function of the ring ratio $a$. Clearly, the DPSK performance degrades as $a$ increases, whereas the DASK performance improves with increasing $a$ but then it deteriorates for large values of $a$. The overall DAPSK performance is dominated by the worse performance of both DPSK and DASK at any $a$. Namely, for small values of $a$, the DASK performance dominates and for large values of $a$, the DPSK performance dominates. In the considered modulation order, $a=2.1$ has shown to provide the best overall performance.\\
\begin{figure}
 \centering
\includegraphics[width=0.5\textwidth]{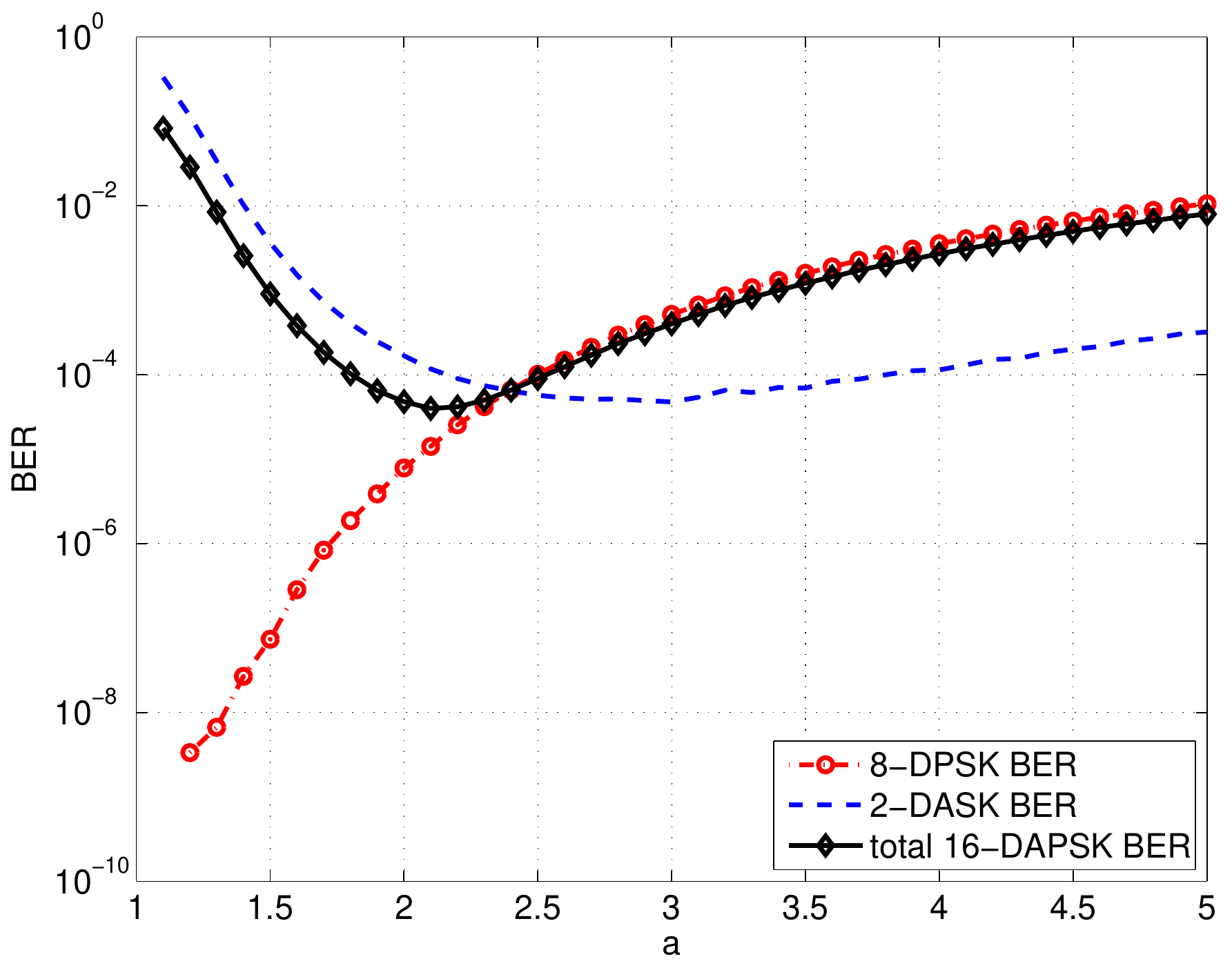}
\caption{Effect of ring ratio $a$ on the amplitude and phase errors}
\label{fig:16_DAPSK_a_influence}
\end{figure}

Another point of interest is how many bits to give to the phase alphabet and how many to the amplitude. As an initial guess, more bits are to be given to the phase than to the amplitude since the phase is more robust to channel influences. Simulations have been made for different DAPSK alphabet orders, and for each order all possible combinations of $q_a$ and $q_p$ are tested. One combination is given the notation $q_p$-DPSK/$q_a$-DASK. The simulations have been carried out at $E_b/N_0=\unit[18]{dB}$ over an AWGN channel and for different values of $a$ so as to search in parallel for the optimal $a$ of every modulation order. Figure \ref{fig:opt_bit_distribution_opt_a} shows the overall DAPSK BER performance results. As shown, for up to a modulation order of 32-DAPSK, a constellation with two amplitude levels has shown the best BER performance. However, for 64-DAPSK, a constellation with 4 amplitude levels and 16-DPSK provides better performance than a constellation with only 2 amplitude levels and 32-DPSK. Additionally the optimal ring ratio $a$ has shown to be dependent on the constellation order. Table \ref{tab:opt_DAPSK} summarizes the optimal DPSK/DASK combination as well as the optimal ring ratio for modulation orders of 8, 16, 32, and 64.
 \begin{table}[htp]
\begin{center}
\caption{Optimal ring ratio $a$ and optimal DPSK/DASK combination for different DAPSK modulation orders.}
\label{tab:opt_DAPSK}
\begin{tabular}{|m{2.5cm}|m{2.5cm}|m{2.5cm}|m{2.5cm}|m{2.5cm}|}
    \hline 
& 8-DAPSK & 16-DAPSK& 32-DAPSK & 64-DAPSK \tn \hline
Optimal DPSK/DASK combination & 4-DPSK/2-DASK & 8-DPSK/2-DASK & 16-DPSK/2-DASK & 16-DPSK/4-DASK \tn \hline
Optimal ring ratio $a$ & 2.7 & 2.1 & 1.5 & 1.4 \tn \hline
 \end{tabular}	
\end{center}
\end{table}

\begin{figure}[!htp]
 \centering
  \subfloat[8-DAPSK]{\label{subfig:opt_a_bit_distribution_8DAPSK}\includegraphics[width=0.49\textwidth]{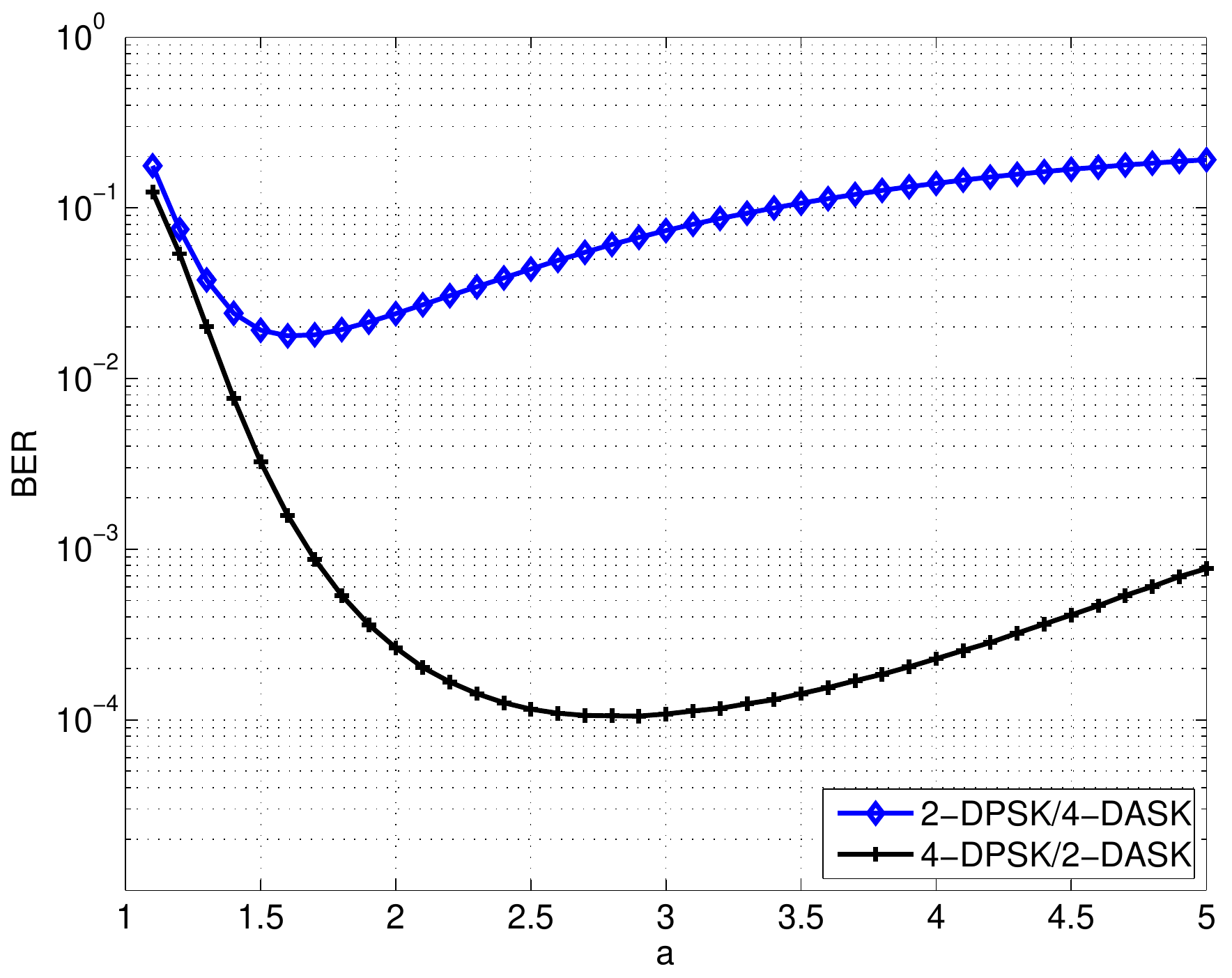}}
  \subfloat[16-DAPSK]{\label{subfig:opt_a_bit_distribution_16DAPSK}\includegraphics[width=0.49\textwidth]{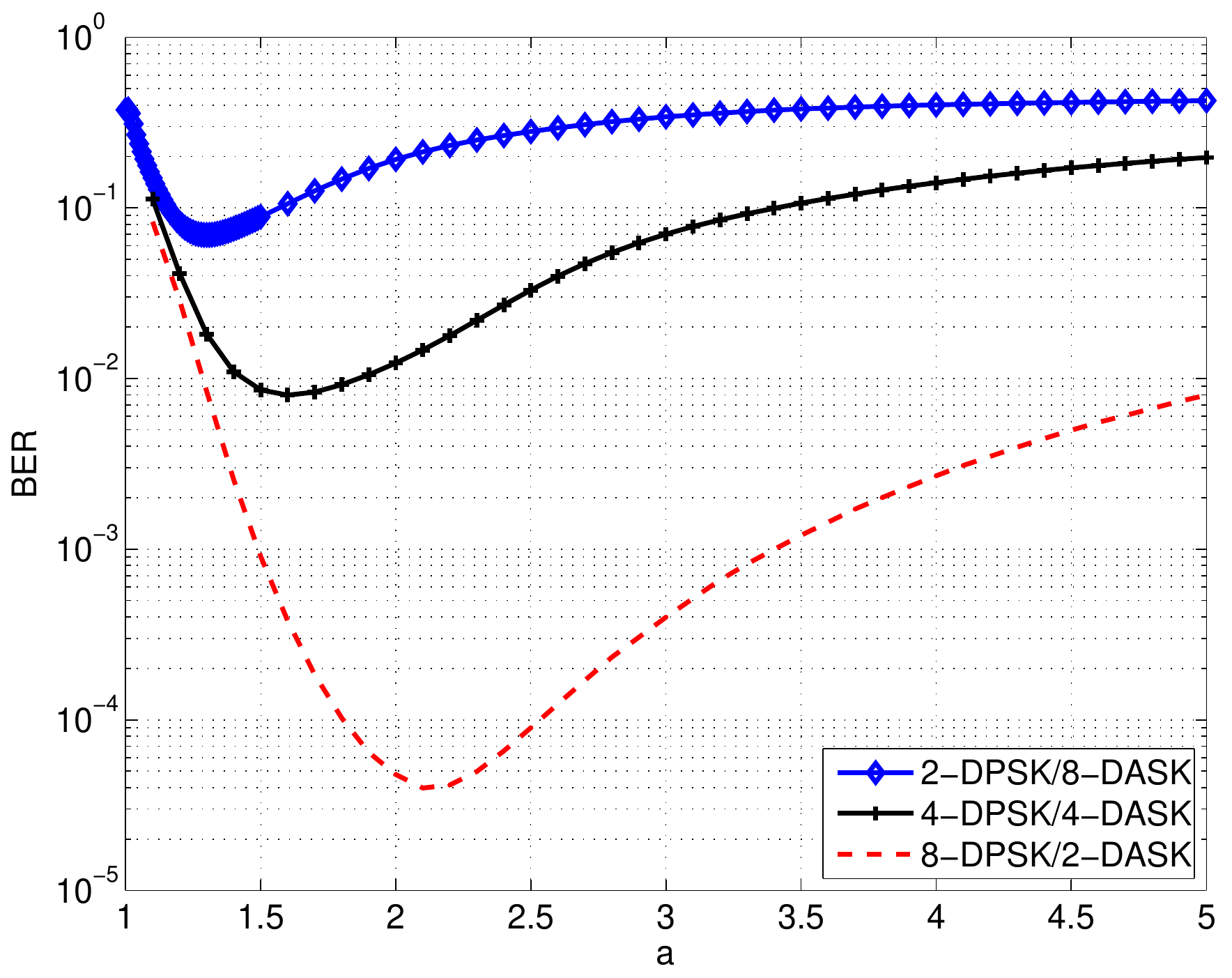}}\\
  \subfloat[32-DAPSK]{\label{subfig:opt_a_bit_distribution_32DAPSK}\includegraphics[width=0.49\textwidth]{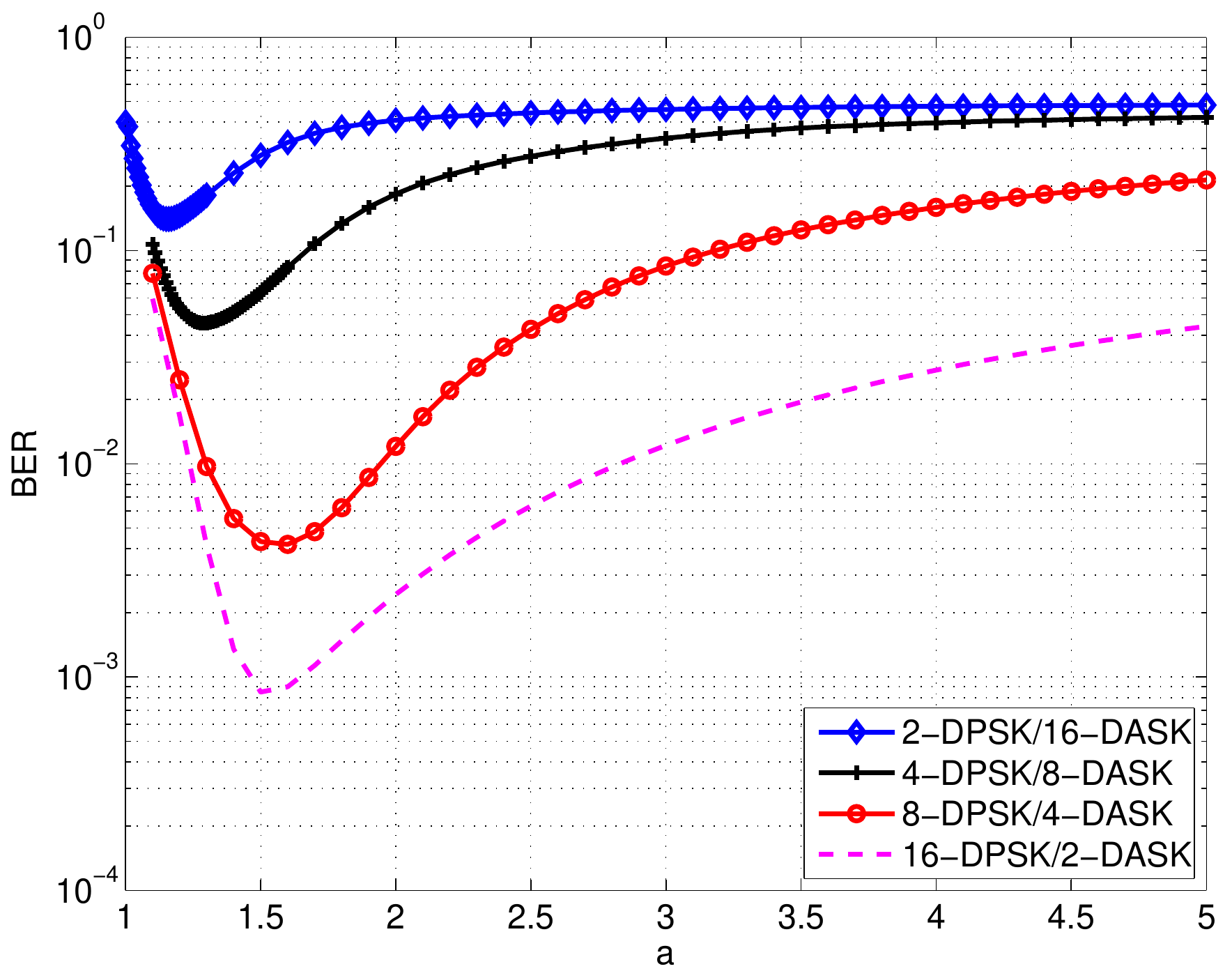}}
  \subfloat[64-DAPSK]{\label{subfig:opt_a_bit_distribution_64DAPSK}\includegraphics[width=0.49\textwidth]{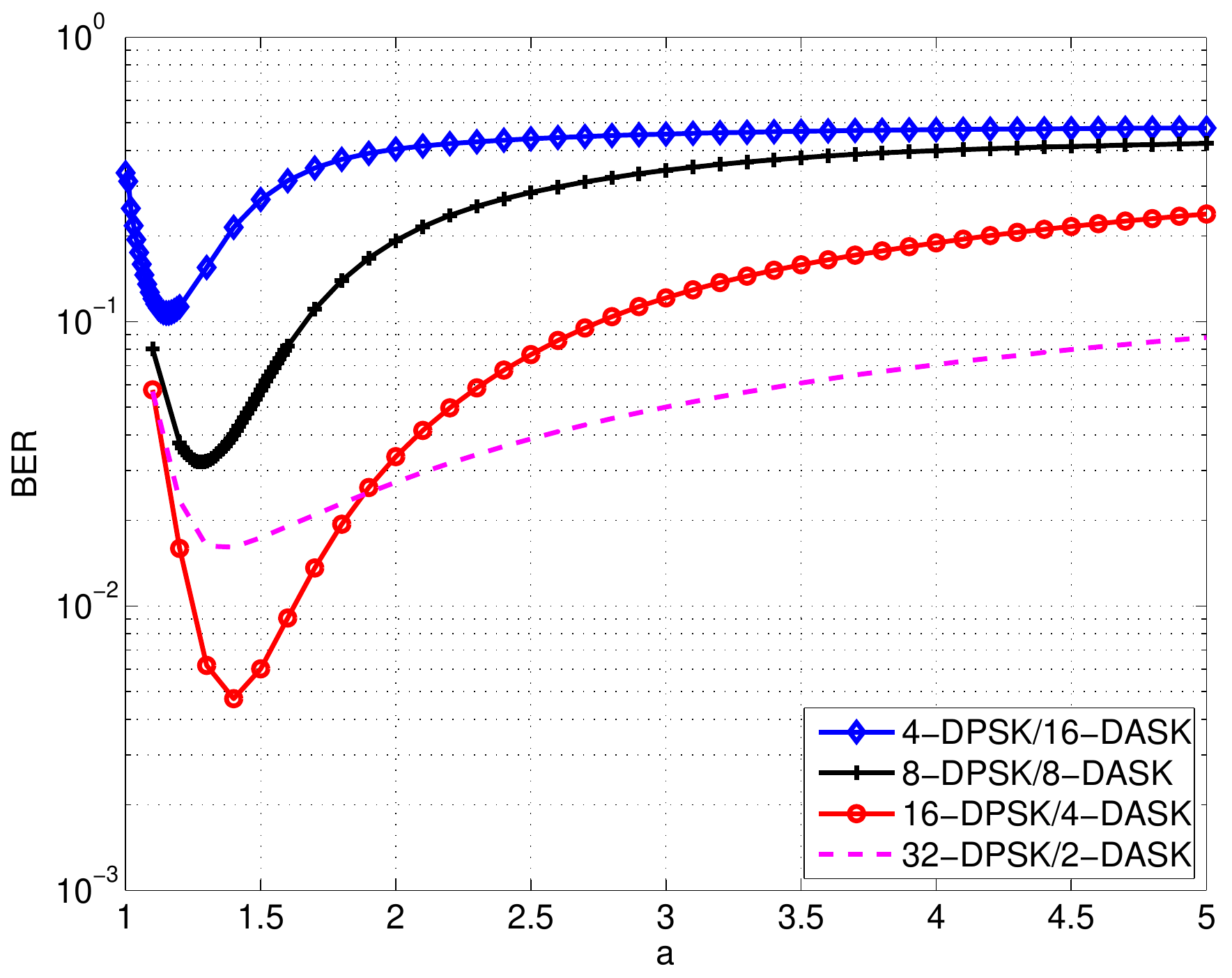}}
  \caption{BER curves at $E_b/N_0\eq\unit[18]{dB}$ showing the optimal ring ratio and optimal amplitude-phase bits distribution for different DAPSK alphabet orders}
  \label{fig:opt_bit_distribution_opt_a}
\end{figure}
\section{Multiple-Symbol Differential Detection}
\label{s:MSDD}
In general, differential modulation can be described in the form of (\ref{eq:DPSK_2}), where the transmit signal at time slot $t$ actually carries information of \emph{all} previous information symbols. Conventionally, only two successive received symbols are used to differentially decode one information symbol. In this case the modulation memory is not fully accounted for in the detection. This causes differential schemes to be inferior in performance to their coherent counterparts. To bridge the performance gap between coherent and differential non-coherent systems, more symbols are to be used for the detection resulting in the so-called Multiple-Symbol Differential Detection (MSDD).\\

In this technique, the receiver divides the entire received symbols into blocks of length $T$ with one overlapping symbol between every two successive blocks as shown in Figure \ref{fig:MSDD-explanation}. Such symbol overlap is necessary since the first symbol in each block is used as a reference for the whole block, and the last symbol is a reference for the next block and so on \cite{Fung92}. The decoder then works on every block and decides on $T-1$ information symbols. We will refer to $T$ as the observation window length. $T=2$ is the case of conventional differential detection.\\

\begin{figure}[htp]
\centering
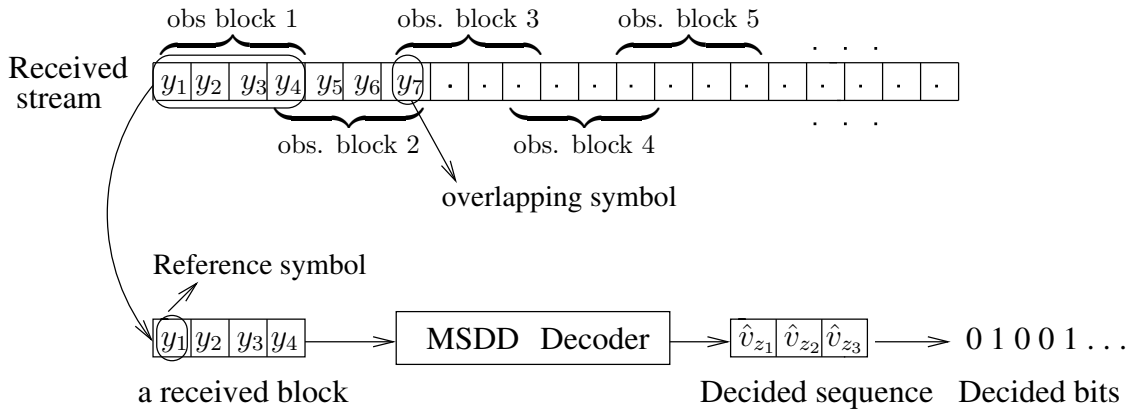
\caption{Multiple-symbol differential detection at $T=4$}
 \label{fig:MSDD-explanation}
\end{figure}
In general the higher the value of $T$, the better the performance but also the higher the complexity. In the limiting case when $T$ tends to $\infty$, the error performance of an MSDD in an AWGN channel approaches that of a coherent detector \cite{Fung92}.\\

In this section a general framework is used to analytically derive receiver metrics for the MSDD technique. Based on the channel statistics and the differential technique used, different metrics are obtained. Three channel types are considered: AWGN, Rayleigh block fading and Rayleigh fast fading channel. In the Rayleigh block fading case, the channel is assumed to be constant over at least $T$ symbols, whereas Rayleigh fast fading channel is defined for the case when the channel varies within $T$ symbol transmissions with known statistical variation. Using the derived metrics, Matlab simulations are performed and error performance in the different cases are compared. Two types of receiver metrics are considered, the Maximum Likelihood (ML) metric  and the Generalized Likelihood Ratio Test (GLRT) metric. \\
\subsection{Derivation of ML metric for MSDD}
\label{ss:ML_metric_derivation}

Consider a wireless transmission medium where the data is modulated differentially and the baseband transmit complex signal $s_t$ is generated. The transmitted signal will be distorted by a multiplicative complex random variable $h_t$. The received signal $y_t$ is modeled as in Figure \ref{fig:Rayleigh_fading} by

\begin{equation}
\label{eq:Tx_Rx}
y_t=\sqrt{\rho}\,s_t\,h_t\,e^{j\Theta}+w_t,
\end{equation}
where $\rho$ is the mean signal power, and $h_t$ is the complex fading coefficient at time slot $t$ following some Gaussian distribution. The phase $\Theta$ is used to have a general channel model that includes the AWGN channel, in which case the channel coefficient $h_t$ is deterministic and the channel is modeled only by an unknown carrier phase $\Theta$ that is assumed to be uniformally distributed in the range $]-\pi,\pi]$. $w_t$ is a complex AWGN noise sample at time slot $t$ which is $\thicksim\mathcal{CN}(0,1)$, with the $w_t$ samples being uncorrelated in time.
\begin{figure}[htp]
\centering
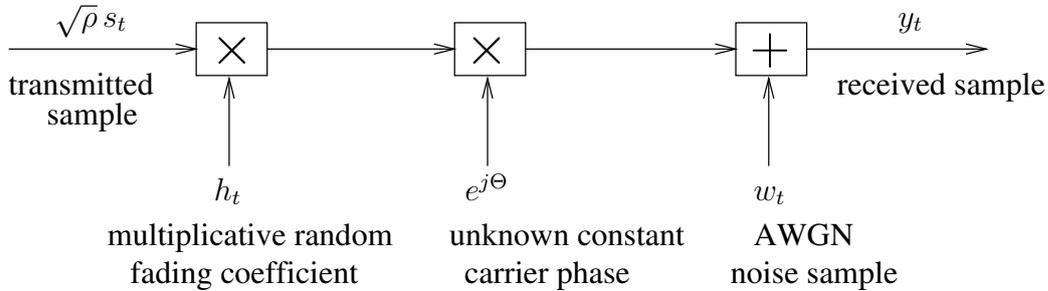
\caption{A general model for AWGN and flat fading channels.}
\label{fig:Rayleigh_fading}
\end{figure}

In MSDD, the receiver buffers $T$ symbols in one observation block $\V{y}$ and works on the whole block in the decoding process. Hence (\ref{eq:Tx_Rx}) can be written in a matrix-vector form as 

\begin{equation}
\label{eq:Tx-Rx-expanded}
 \begin{bmatrix}
  y_0  \\
  \vdots \\
  \vdots \\
  y_{T-1}
 \end{bmatrix} = \sqrt{\rho} 
\begin{bmatrix}
s_0 & & & \\
& \ddots & & \text{\huge{0}} \\
\:\text{\huge{0}} & & \ddots & \\
& & & s_{T-1}
\end{bmatrix}
 \begin{bmatrix}
  h_0  \\
  \vdots \\
  \vdots \\
  h_{T-1}
 \end{bmatrix}\,e^{j\Theta} +
  \begin{bmatrix}
  w_0  \\
  \vdots \\
  \vdots \\
  w_{T-1}
 \end{bmatrix},
\end{equation}
and in a compact form as
\begin{equation}
\label{eq:Tx_Rx_vector}
\V{y}=\sqrt{\rho}\,\V{D_{s}}\,\V{h}\,e^{j\Theta}+\V{w},
\end{equation}
where $\V{D_{s}}$ is a diagonal matrix with the signal vector $\V{s}$ on the main diagonal. $\V{y}, \V{h}$ and $\V{w} \in \mathbb{C}^{T\!\times\! 1}$ and represent the received vector, the channel vector and the AWGN noise vector, respectively. $\V{h} \thicksim \mathcal{CN}(\V{\mu_{h}}, \V{\Lambda_{h}})$, where $\V{\mu_{h}} \in \mathbb{C}^{T\!\times\! 1}$ is the channel mean vector and $\V{\Lambda_{h}} \in \mathbb{R}^{T \!\times\! T}$ is the channel covariance matrix. At any time slot $t$, the channel power $\E[\left|h_t\right|^2]$ as well as the symbol power $\E[\left|s_t\right|^2]$ are normalized to 1. The noise vector $\V{w} \thicksim \mathcal{CN}(\V{0}_{T\!\times\! 1},\V{I}_{T})$. With all powers normalized to 1, $\rho$ represents the Signal to Noise Ratio (SNR) at the input of the receiver.\\

Our target is to derive the ML metric for the MSDD receiver, hence we need to get an expression for the conditional Probability Density Function (PDF) of the received vector $\V{y}$ under the condition that vector $\V{s}$ is sent and the carrier phase $\Theta$ is given, namely we need to derive $\text{p}(\V{y}|\V{s},\Theta)$. From (\ref{eq:Tx_Rx_vector}), since $\V{h}$ and $\V{w}$ are both complex Gaussian distributed, then $\V{y}$ follows the same distribution, and its conditional PDF is given by \cite{Steven_Kay93}

\begin{equation}
\label{eq:pr_s_theta}
\text{p}(\V{y}|\V{s},\Theta)= \frac{1}{\pi^{T} \det(\V{\Lambda}_{\V{y}|\V{s},\Theta})}\,e^{-(\V{y}-\V{\mu}_{\V{y}|\V{s},\Theta})\hr\,\V{\Lambda}_{\V{y}|\V{s},\Theta}^{-1}\,(\V{y}-\V{\mu}_{\V{y}|\V{s},\Theta})},
\end{equation}
where $\V{\mu}_{\V{y}|\V{s},\Theta}$ is the conditional mean vector of $\V{y}$ given $\V{s}$ and $\Theta$. Using (\ref{eq:Tx_Rx_vector}), $\V{\mu}_{\V{y}|\V{s},\Theta}$ becomes
\begin{equation}
\label{eq:mean_r}
\V{\mu}_{\V{y}|\V{s},\Theta} = \sqrt{\rho}\,\V{D_{s}}\,\V{\mu_{h}}\,e^{j\Theta},
 \end{equation}
and $\V{\Lambda}_{\V{y}|\V{s},\Theta}$ is the conditional covariance matrix of $\V{y}$ given $\V{s}$ and $\Theta$ and it can be calculated as
\begin{eqnarray}
\label{eq:cov_mat_r}
\V{\Lambda}_{\V{y}|\V{s},\Theta} &=& \E[(\V{y}-\V{\mu_{y}})(\V{y}-\V{\mu_{y}})\hr |\V{s},\Theta] \nonumber \\
& = & E\Big[\Big(\sqrt{\rho}\V{D_{s}}(\V{h}-\V{\mu_{h}})\,e^{j\Theta}+\V{w}\Big) \Big( \sqrt{\rho}\V{D_{s}}(\V{h}-\V{\mu_{h}})\,e^{j\Theta}+\V{w} \Big)\hr|\V{s},\Theta\Big] \nonumber \\
& = & E\Big[\Big(\sqrt{\rho}\V{D_{s}}(\V{h}-\V{\mu_{h}})\,e^{j\Theta}+\V{w}\Big) \Big( \sqrt{\rho}\,e^{-j\Theta}(\V{h}-\V{\mu_{h}})\hr\V{D_{s}}\hr+\V{w}\hr \Big)|\V{s},\Theta\Big] \nonumber \\
& = & \rho\,\V{D_{s}}\,\E[(\V{h}-\V{\mu_{h}})(\V{h}-\V{\mu_{h}})\hr]\,\V{D_{s}}\hr+ \E[\V{ww}\hr] \nonumber \\
&=& \rho\,\V{D_{s}}\,\V{\Lambda_{h}}\,\V{D_{s}}\hr+\V{I}_{T} = \V{\Lambda_{y|s}},
\end{eqnarray}
where we have used the fact that the channel $\V{h}$ and the noise $\V{w}$ are statistically independent. The conditional covariance matrix was found to be independent of $\Theta$. Now (\ref{eq:pr_s_theta}) can be written as 
\begin{equation*}
\text{p}(\V{y}|\V{s},\Theta)= \frac{1}{\pi^{T} \det(\V{\Lambda}_{\V{y}|\V{s}})}\,e^{-\V{y}\hr\,\V{\Lambda}_{\V{y|s}}^{-1}\,\V{y}\,+\,\V{y}\hr\,\V{\Lambda}_{\V{y|s}}^{-1}\,\V{\mu}_{\V{y|s},\Theta}\,+\,\V{\mu}_{\V{y|s},\Theta}\hr\,\V{\Lambda}_{\V{y|s}}^{-1}\,\V{y}\,-\,\V{\mu}_{\V{y|s},\Theta}\hr\,\V{\Lambda}_{\V{y|s}}^{-1}\,\V{\mu}_{\V{y|s},\Theta}}.
\end{equation*}
Using the identity $\V{y}\hr\,\V{\Lambda}_{\V{y|s}}^{-1}\,\V{\mu}_{\V{y|s},\Theta}\,+\,\V{\mu}_{\V{y|s},\Theta}\hr\,\V{\Lambda}_{\V{y|s}}^{-1}\,\V{y} = 2\,\Re\{\V{y}\hr\,\V{\Lambda}_{\V{y|s}}^{-1}\,\V{\mu}_{\V{y|s},\Theta}\}$ together with (\ref{eq:mean_r}) we get
\begin{equation*}
\label{eq:pr_s_derivation}
\text{p}(\V{y}|\V{s},\Theta) =\frac{1}{\pi^{T} \det(\V{\Lambda}_{\V{y}|\V{s}})}\,e^{-\V{y}\hr\,\V{\Lambda}_{\V{y|s}}^{-1}\,\V{y}\,-\rho\,\V{\mu_{h}}\hr\,\V{D_{s}}\hr\,\V{\Lambda}_{\V{y|s}}^{-1}\,\V{D_{s}}\,\V{\mu_{h}}\,+\,2\,\Re\{\sqrt{\rho}\,\V{y}\hr\,\V{\Lambda}_{\V{y|s}}^{-1}\,\V{D_{s}}\,\V{\mu_{h}}\,e^{j\Theta}\}}.
\end{equation*}
Using $\Re\{x\}=|x|\cos{\alpha}$, with $\alpha=\tan^{-1}{\frac{\Im\{x\}}{\Re\{x\}}}$, we can replace $\Re\{\,\sqrt{\rho}\,\V{y}\hr\,\V{\Lambda}_{\V{y|s}}^{-1}\,\V{D_{s}}\,\V{\mu_{h}}\,e^{j\Theta}\}$ by $|\,\sqrt{\rho}\,\V{y}\hr\,\V{\Lambda}_{\V{y|s}}^{-1}\,\V{D_{s}}\,\V{\mu_{h}}|\cos(\alpha +\Theta)$, where $\alpha=\tan^{-1}{\frac{\Im\{\V{y}\hr\,\V{\Lambda}_{\V{y|s}}^{-1}\,\V{D_{s}}\,\V{\mu_{h}}\}}{\Re\{\V{y}\hr\,\V{\Lambda}_{\V{y|s}}^{-1}\,\V{D_{s}}\,\V{\mu_{h}}\}}}$ resulting in
\begin{equation*}
\text{p}(\V{y}|\V{s},\Theta) =\frac{1}{\pi^{T} \det(\V{\Lambda}_{\V{y}|\V{s}})}\,e^{-\V{y}\hr\,\V{\Lambda}_{\V{y|s}}^{-1}\,\V{y}\,-\rho\,\V{\mu_{h}}\hr\,\V{D_{s}}\hr\,\V{\Lambda}_{\V{y|s}}^{-1}\,\V{D_{s}}\,\V{\mu_{h}}\,+\,2\,\sqrt{\rho}\,|\,\V{y}\hr\,\V{\Lambda}_{\V{y|s}}^{-1}\,\V{D_{s}}\,\V{\mu_{h}}|\cos(\alpha +\Theta)}.
\end{equation*}\\

Since the carrier phase $\Theta$ is assumed to be uniformally distributed in the range $]-\pi,\pi]$, then by averaging over $\Theta$ we get
\begin{eqnarray*}
\text{p}(\V{y|s})\! &=&\! \int_{\!-\infty}^{\infty}\!\text{p}(\V{y|s},\eta)\,\text{p}_{\Theta}(\eta)\,d\eta \\
\!&=&\! \frac{1}{\pi^{T} \det(\V{\Lambda}_{\V{y}|\V{s}})}\,e^{-\V{y}\hr\,\V{\Lambda}_{\V{y|s}}^{-1}\,\V{y}\,-\rho\,\V{\mu_{h}}\hr\,\V{D_{s}}\hr\,\V{\Lambda}_{\V{y|s}}^{-1}\,\V{D_{s}}\,\V{\mu_{h}}}.\,\frac{1}{2\pi}\,\int_{-\pi}^{\pi}\!e^{2\,\sqrt{\rho}\,|\,\V{y}\hr\,\V{\Lambda}_{\V{y|s}}^{-1}\,\V{D_{s}}\,\V{\mu_{h}}|\cos(\alpha +\eta)}\,d\eta.
\end{eqnarray*}
If we then use the definition of the zeroth order modified Bessel function of the first kind $I_0(z)=\frac{1}{\pi}\int_0^{\pi}\!e^{z\cos(\eta)}d\eta$, and the identity $\int_{-\pi}^{\pi}\!e^{z\cos(\alpha+\eta)}d\eta = \int_{-\pi}^{\pi}\!e^{z\cos(\eta)}d\eta$, we finally have

\begin{equation}
\label{eq:pr_s}
\boxed{\text{p}(\V{y|s})=\frac{1}{\pi^{T}\det(\V{\Lambda}_{\V{y}|\V{s}})}\,e^{-\V{y}\hr\,\V{\Lambda}_{\V{y|s}}^{-1}\,\V{y}\,-\rho\,\V{\mu_{h}}\hr\,\V{D_{s}}\hr\,\V{\Lambda}_{\V{y|s}}^{-1}\,\V{D_{s}}\,\V{\mu_{h}}}.I_0(2\,\sqrt{\rho}\,|\,\V{y}\hr\,\V{\Lambda}_{\V{y|s}}^{-1}\,\V{D_{s}}\,\V{\mu_{h}}|).}
\end{equation}

The above equation is the general form of the conditional probability density function of the MSDD detection which is aimed to be maximized in order to result in the most likely decision on the transmit symbol vector $\V{s}$. In the next three subsections, we will use (\ref{eq:pr_s}) and simplify it for the different channels considered and the different modulation techniques used.\\

\subsubsection{AWGN Channel}
In the AWGN case, the channel vector and its mean are simply the all ones vector $\V{h}=\V{1}_{T\!\times\! 1}$ and $\V{\mu_h}=\V{1}_{T\!\times\! 1}$. The covariance matrix $\V{\Lambda_{h}}$ is an all zero matrix $\V{0}_{T\!\times\! T}$ since $\V{h}$ is deterministic. From (\ref{eq:cov_mat_r}), $\V{\Lambda_{y|s}}=\V{I}_{T}$, and the conditional PDF in (\ref{eq:pr_s}) is simplified to
\begin{equation}
\label{eq:pr_s_AWGN}
\text{p}(\V{y|s})=\frac{1}{\pi^{T}}\,e^{-\V{y}\hr\,\V{y}\,-\,\rho\,\V{s}\hr\,\V{s}}\,I_0(2\,\sqrt{\rho}\,|\,\V{y}\hr\,\V{s}\,|),
\end{equation}
where $\V{D_s}\,\V{\mu_h}\!=\!\V{D_s}\,\V{1}_{T\!\times\!1}\!=\!\V{s}$ has been used. We note that $\V{y}\hr\,\V{y}$ is independent of $\V{s}$, and that $\V{s}\hr\,\V{s}=\|\V{s}\|^2$. Furthermore, since the natural logarithm function is a monotonically increasing function of its argument, hence maximizing $\text{p}(\V{y|s})$ is equivalent to maximizing $\ln(\text{p}(\V{y|s}))$. We then get the MSDD ML decision in the AWGN channel case as
\begin{equation}
\label{eq:ML_metric_AWGN}
\V{\hat{s}}_\text{ML}=\argmax\limits_{\V{s}^i}\:\ln\Big(I_0(2\,\sqrt{\rho}\,|\,\V{y}\hr\,\V{s}^i|)\Big)\,-\,\rho\|\V{s}^i\|^2,
\end{equation}
where --with a small abuse of the notation-- $\V{s}^i$ is the $i^{th}$ possible transmit vector, which will be termed as a candidate sequence.\\

For the special case of constant envelope modulation like in DPSK, $\|\V{s}^i\|^2$ is constant $\forall i$, so the second term can be removed. Moreover, the $\ln$ and the $I_0$ functions are both monotonically increasing functions of their arguments, hence they can be omitted in getting the maximum of the expression and the ML decision metric is further simplified to
\begin{equation}
\label{eq:ML_metric_AWGN_DPSK}
\V{\hat{s}}_\text{ML}=\argmax\limits_{\V{s}^{i}}\:|\V{y}\hr\V{s}^i|,
\end{equation}
which is the known correlation detection. In other words, the most likely transmit symbol vector when DPSK is used over an AWGN channel is the one that has the highest correlation with the received vector (the one that is as parallel as possible to $\V{y}$).\\

Note that the metric does not depend on the Euclidean distance between $\V{y}$ and $\V{s}$ as is the case in coherent detection, but rather it depends on the magnitude of their correlation. It is worth mentioning that gray coding which is known to be the optimal bit-to-symbol mapping in coherent detection might not be the optimal one here. This is due to the fact that in differential detection, the distance profile between $\V{y}$ and $\V{s}$ is not the Euclidean distance but rather a different one defined by (\ref{eq:ML_metric_AWGN}). It is however not an easy task to get the optimal bit-to-symbol mapping for differential non-coherent systems. To surmount this problem, one may use several different bit-to-symbol mappings during a transmission, so as to achieve an averaged performance of all mappings. However assigning several bit-to-symbol mappings will complicate the transmission process, so here only gray coding is used unless otherwise specified.
\subsubsection{Rayleigh Block Fading Channel}
Rayleigh fading model is described by (\ref{eq:Tx_Rx}) and in Figure \ref{fig:Rayleigh_fading} if we take $h_t\thicksim\mathcal{CN}(0,1)$. The model assumes no dominant Line of Sight (LOS) component, therefore the mean of the channel coefficient $h_t$ is zero at any time slot $t$. This leads to an all zero mean vector $\V{\mu_{h}}=\V{0}_{T\!\times\! 1}$. The word \textit{block} indicates that the channel coefficient remains constant over a long transmission period, hence the channel does not vary within the observation duration $T$ and the covariance matrix is an all ones matrix $\V{\Lambda_{h}}=\V{1}_{T\!\times\! T}$ since the channel coefficients are 100\% correlated. In such a special case, (\ref{eq:cov_mat_r}) is simplified to

\begin{equation}
\label{eq:lambda_RBF}
 \V{\Lambda_{y|s}}=\rho\,\V{D_{s}}\,\V{1}_{T\!\times\! T}\,\V{D_{s}}\hr+\V{I}_{T},
\end{equation}
where,
\begin{eqnarray*}
\V{D_{s}}\,\V{1}_{T\!\times\! T}\,\V{D_{s}}\hr &=&
\begin{bmatrix} 
s_0 & & & \\
& \ddots & & \text{\huge{0}} \\
\:\text{\huge{0}} & & \ddots & \\
& & & s_{T-1}
\end{bmatrix}
\begin{bmatrix} 
1 & 1& \cdots & 1\\
1& \ddots & & 1 \\
\vdots & & \ddots &\vdots \\
1& 1& \cdots & 1
\end{bmatrix}
\begin{bmatrix} 
s_0^* & & & \\
& \ddots & & \text{\huge{0}} \\
\:\text{\huge{0}} & & \ddots & \\
& & & s_{T-1}^*
\end{bmatrix}\\
&=&
\begin{bmatrix}
s_0s_0^* & s_0s_1^*& \cdots & s_0s_{T-1}^*\\
s_1s_0^*& \ddots & & s_1s_{T-1}^* \\
\vdots & & \ddots &\vdots \\
s_{T-1}s_0^*& s_{T-1}s_1^*& \cdots & s_{T-1}s_{T-1}^*
\end{bmatrix}
=\V{s}\,\V{s}\hr.
\end{eqnarray*}
Hence (\ref{eq:lambda_RBF}) is simplified to $\V{\Lambda_{y|s}}=\rho\,\V{s}\V{s}\hr+\V{I}_{T}$ and the general $\text{p}(\V{y|s})$ in (\ref{eq:pr_s}) can be written as
\begin{equation}
\label{eq:pr_s_RBF1}
\text{p}(\V{y|s})=\frac{1}{\pi^{T}\det(\V{I}_{T}+\rho\,\V{s}\V{s}\hr)}\,e^{-\V{y}\hr\,(\V{I}_{T}+\rho\,\V{s}\V{s}\hr)^{-1}\,\V{y}}.
\end{equation}
In the following, we will use two identities, namely Sylvester's determinant theorem 
\begin{equation}
\label{eq:det_identity}
\det(\V{I}_p+\V{A}_{p\!\times\! q}\,\V{B}_{q\!\times\! p})=\det(\V{I}_{q}+\V{B}_{q\!\times\! p}\,\V{A}_{p\!\times\! q}),
\end{equation}
and the matrix inversion lemma \cite{Steven_Kay93}
\begin{equation}
\label{eq:matrix_inversion_lemma}
(\V{A}+\V{B}\,\V{C}\,\V{D})^{-1}=\V{A}^{-1}-\V{A}^{-1}\,\V{B}(\V{C}^{-1}+\V{D}\V{A}^{-1}\V{B})^{-1}\,\V{D}\V{A}^{-1},
\end{equation}
with $\V{A} \in \mathbb{C}^{p\!\times\! p}$, $\V{B} \in \mathbb{C}^{p\!\times\! m}$, $\V{C} \in \mathbb{C}^{m\!\times\! m}$, and $\V{D} \in \mathbb{C}^{m\!\times\! p}$.
Using the above two identities, we can do the following replacements in (\ref{eq:pr_s_RBF1}),
\begin{equation*}
\det(\V{I}_{T}+\rho\,\V{s}\V{s}\hr)=\det(1+\rho\,\|\V{s}\|^2),
\end{equation*}
and
\begin{eqnarray*}
(\V{I}_{T}+\rho\V{s}\V{s}\hr)^{-1} &=& \V{I}_{T}-\rho\,\V{s}\,(1+\rho\,\V{s}\hr\V{s})^{-1}\V{s}\hr\\
&=& \V{I}_{T}-\frac{\rho\,\V{s}\V{s}\hr}{1+\rho\,\|\V{s}\|^2}.
\end{eqnarray*}
In getting the inverse we have used the following substitutions, $\V{A}=\V{I}_{T}$, $\V{B}=\rho\,\V{s}$, $\V{C}=1$, and $\V{D}=\V{s}\hr$.
Finally, $\text{p}(\V{y|s})$ is simplified to
\begin{equation}
\label{eq:pr_s_RBF2}
\text{p}(\V{y|s})=\frac{1}{\pi^{T}(1+\rho\,\|\V{s}\|^2)}\,e^{-\V{y}\hr\,(\V{I}_{T}-\frac{\rho}{1+\rho\,|\V{s}|^2}\,\V{s}\,\V{s}\hr)\,\V{y}},
\end{equation}
and the ML decision metric for the Rayleigh block fading case becomes 
\begin{equation}
\label{ML_metric_RBF}
\V{\hat{s}}_\text{ML}=\argmax\limits_{\V{s}^{i}}\:\frac{\rho}{1+\rho\,\|\V{s}^i\|^2}\,|\V{y}\hr\V{s}^i|^2\,-\,\ln(1+\rho\,\|\V{s}^i\|^2).
\end{equation}
For the special case of constant envelope modulation, the ML decision is further simplified to 
\begin{equation}
\label{eq:Ml_metric_RBF_DPSK}
\V{\hat{s}}_\text{ML}=\argmax\limits_{\V{s}^{i}}\:|\V{y}\hr\V{s}^i|^2 \equiv \argmax\limits_{\V{s}^{i}}\:|\V{y}\hr\V{s}^i|,
\end{equation}
which is the same as the ML metric for the AWGN channel in the constant envelope modulation case as shown in (\ref{eq:ML_metric_AWGN_DPSK}).

\subsubsection{Rayleigh Fast Fading Channel}
In the Rayleigh fast fading channel, the channel coefficients change within the observation interval with a variation described by the channel covariance matrix $\V{\Lambda_h}$. The assumed statistical behaviour of the variation follows the Jake's model, with the autocorrelation function defined by \cite{Proakis2000}
\begin{equation}
\label{eq:Jakes_ACF}
\phi(m)=J_0(2\pi mf_DT_s),
\end{equation}
where $J_0(\V{.})$ is the zeroth order Bessel function of the first kind, $f_D$ is the maximum Doppler frequency defined by $f_D=v/\lambda$, where $v$ is the velocity of a moving vehicle, and $\lambda$ is the center wavelength of the bandpass signal. $T_s$ is the symbol duration, hence $f_DT_S$ is the normalized Doppler frequency which is limited between $0$ and $1$. $f_DT_s=0$ describes the Rayleigh block fading case. The power density spectrum is the Jakes spectrum. As in Rayleigh block fading channel, the mean of the channel coefficient is zero and consequently $\V{\mu_{h}}=\V{0}_{T\!\times\! 1}$. The conditional probability density function $\text{p}(\V{y|s})$ of (\ref{eq:pr_s}) is simplified to
\begin{equation}
\label{eq:pr_s_RFF1}
\text{p}(\V{y|s})=\frac{1}{\pi^{T} \det(\V{\Lambda_{y|s}})}\,e^{-\V{y}\hr\V{\Lambda_{y|s}}^{-1}\V{y}}.
\end{equation}
To get the determinant and the inverse of $\V{\Lambda_{y|s}}$, we again use the definition of $\V{\Lambda_{y|s}}$ in (\ref{eq:cov_mat_r}) together with identity (\ref{eq:det_identity}) resulting in
\begin{equation*}
\det(\V{I}_{T}+\rho\V{D_s}\V{\Lambda_h}\V{D_s}\hr)=\det(\V{I}_{T}+\rho\V{D_s}\hr\V{D_s}\V{\Lambda_h})=\det(\V{I}_{T}+\rho|\V{D_s|^2\V{\Lambda_h}}),
\end{equation*}
and using the matrix inversion lemma of (\ref{eq:matrix_inversion_lemma})
\begin{equation*}
(\V{I}_{T}+\rho\V{D_s}\V{\Lambda_h}\V{D_s}\hr)^{-1}=\V{I}_{T}-\rho\V{D_s}(\V{\Lambda_h}^{-1}+\rho\V{D_s}\hr\V{D_s})^{-1}\V{D_s}\hr,
\end{equation*}
where we have used the following substitutions $\V{A}=\V{I}_{T}$, $\V{B}=\rho\,\V{D_s}$, $\V{C}=\V{\Lambda_h}$, and $\V{D}=\V{D_s}\hr$.
Substituting back into (\ref{eq:pr_s_RFF1}), we get 
\begin{eqnarray}
\label{pr_s_RFF2}
\text{p}(\V{y|s})&=&\frac{1}{\pi^{T} \det(\V{I}_{T}+\rho|\V{D_s|^2\V{\Lambda_h}})}\,e^{-\V{y}\hr\Big(\V{I}_{T}-\rho\V{D_s}(\V{\Lambda_h}^{-1}+\rho\V{D_s}\hr\V{D_s})^{-1}\V{D_s}\hr\Big)\V{y}}	\nonumber\\
&=& \frac{1}{\pi^{T} \det(\V{I}_{T}+\rho|\V{D_s}|^2\V{\Lambda_h})}\,e^{-\V{y}\hr\V{y}+\rho\,\V{y}\hr\V{D_s}(\V{\Lambda_h}^{-1}+\rho|\V{D_s}|^2)^{-1}\V{D_s}\hr\V{y}}.
\end{eqnarray}\\

The ML decision metric follows as 
\begin{equation}
\V{\hat{s}}_\text{ML}=\argmax\limits_{\V{s}^i}\:\,\rho\,\V{y}\hr\V{D}_{\V{s}^i}(\V{\Lambda_h}^{-1}+\rho|\V{D}_{\V{s}^i}|^2)^{-1}\V{D}_{\V{s}^i}\hr\V{y}-\ln(\det(\V{I}_{T}+\rho|\V{D}_{\V{s}^i}|^2\V{\Lambda_h})).
\label{eq:ML_metric_RFF}
\end{equation}
In the special case of constant envelope modulation like in DPSK, $\:|\V{D}_{\V{s}^i}|^2\!=\!\V{I}_{T} \: \forall{i}$ and the ML metric can be simplified to 
\begin{equation}
\label{eq:ML_metric_RFF_DPSK}
\V{\hat{s}}_\text{ML}=\argmax\limits_{\V{s}^i}\:\,\V{y}\hr\V{D}_{\V{s}^i}(\V{\Lambda_h}^{-1}+\rho\V{I}_{T})^{-1}\V{D}_{\V{s}^i}\hr\V{y}.
\end{equation}
Again using the matrix inversion lemma in (\ref{eq:matrix_inversion_lemma})
\begin{equation*}
(\V{\Lambda_h}^{-1}+\rho\V{I}_{T})^{-1}=\frac{1}{\rho}\V{I}_{T}-\frac{1}{\rho^2}(\V{\Lambda_h}+\frac{1}{\rho}\V{I}_{T})^{-1},
\end{equation*}
where we have used the following substitutions $\V{A}=\rho\V{I}_{T}$, $\V{B}=\V{I}_{T}$, $\V{C}=\V{\Lambda_h}^{-1}$, and $\V{D}=\V{I}_{T}$.
Substituting back in (\ref{eq:ML_metric_RFF_DPSK})
\begin{eqnarray}
\label{eq:ML_metric_RFF_DPSK_final}
\V{\hat{s}}_\text{ML}&=&\argmax\limits_{\V{s}^i}\:\,\V{y}\hr\V{D}_{\V{s}^i}\Big(\frac{1}{\rho}\V{I}_{T}-\frac{1}{\rho^2}(\V{\Lambda_h}+\frac{1}{\rho}\V{I}_{T})^{-1}\Big)\V{D}_{\V{s}^i}\hr\V{y} \nonumber \\
&=&\argmax\limits_{\V{s}^i}\:\,\frac{1}{\rho}\V{y}\hr\V{D}_{\V{s}^i}\V{D}_{\V{s}^i}\hr\V{y}-\frac{1}{\rho^2}\V{y}\hr\V{D}_{\V{s}^i}(\V{\Lambda_h}+\frac{1}{\rho}\V{I}_{T})^{-1}\V{D}_{\V{s}^i}\hr\V{y} \nonumber \\
&=&\argmin\limits_{\V{s}^i}\:\,\V{y}\hr\V{D}_{\V{s}^i}(\V{\Lambda_h}+\frac{1}{\rho}\V{I}_{T})^{-1}\V{D}_{\V{s}^i}\hr\V{y},
\end{eqnarray}
where the term with $\V{D}_{\V{s}^i}\V{D}_{\V{s}^i}\hr=|\V{D}_{\V{s}^i}|^2$ has been removed since it is constant $\forall i$.
A summary of the ML decision metrics in the different channels and in the non-constant and constant envelope modulation schemes is given  in Table \ref{tab:ML_metrics}. For space limitations, we used the abbreviations RBF for Rayleigh Block Fading channel and RFF for Rayleigh Fast Fading channel.

\renewcommand{\arraystretch}{1.5} 
\begin{table}[h]
\begin{center}
\caption{ML MSDD metrics in a single antenna system}
\label{tab:ML_metrics}
  \begin{tabular}{|m{1.4cm}|m{8cm}|m{6cm}|}
    \hline    
      \textbf{ML MSDD metric} $\V{\hat{s}}_\text{ML}$	 &  \textbf{non-constant envelope modulation e.~\!g DAPSK} &  \textbf{constant envelope modulation e~\!.g DPSK} \tn[5pt] \hline 
     
	  AWGN &  $\argmax\limits_{\V{s}^{i}}\:\ln\Big(I_0(2\,\sqrt{\rho}\,|\,\V{y}\hr\,\V{s}^i|)\Big)\,-\,\rho\|\V{s}^i\|^2$	&	$\argmax\limits_{\V{s}^{i}}\:|\V{y}\hr\V{s}^i|$	 \tn \hline
	  RBF  & $\argmax\limits_{\V{s}^{i}}\:\frac{\rho}{1+\rho\,\|\V{s}^i\|^2}\,|\V{y}\hr\V{s}^i|^2\,-\,\ln(1+\rho\,\|\V{s}^i\|^2)$				& $\argmax\limits_{\V{s}^{i}}\:|\V{y}\hr\V{s}^i|$			 \tn \hline
	  RFF& $\argmax\limits_{\V{s}^i}\:\,\rho\,\V{y}\hr\V{D}_{\V{s}^i}(\V{\Lambda_h}^{-1}+\rho|\V{D}_{\V{s}^i}|^2)^{-1}\V{D}_{\V{s}^i}\hr\V{y}-\ln(\det(\V{I}_{T}+\rho|\V{D}_{\V{s}^i}|^2\V{\Lambda_h}))$			& $\argmin\limits_{\V{s}^i}\:\,\V{y}\hr\V{D}_{\V{s}^i}(\V{\Lambda_h}+\frac{1}{\rho}\V{I}_{T})^{-1}\V{D}_{\V{s}^i}\hr\V{y}$	\tn \hline
  \end{tabular}	
\end{center}
\end{table}

 \subsection{Derivation of GLRT metric for MSDD}
\label{ss:GLRT}
The MSDD ML decision metrics derived for the non-constant envelope modulation like DAPSK requires the knowledge of the SNR ($\rho$) at the receiver. Additionally, the metric calculation itself looks complicated with the need of using the modified Bessel function in the AWGN channel case, and the $\ln$ and $\det$ functions in the Rayleigh fading case. This motivates the use of some sub-optimum metric which reduces the complexity of the receiver. To this end, we use the so-called Generalized Likelihood Ratio Test (GLRT) decision metric. The GLRT metric basically uses a coherent receiver metric and replaces the unknown channel with its ML estimate. In this section, we will provide a derivation of the GLRT metric for an MSDD receiver in the different channel conditions.\\
\subsubsection{GLRT metric for Rayleigh Fast Fading channel}
Starting from the channel model defined in (\ref{eq:Tx_Rx_vector}) and omitting $\Theta$ when considering a Rayleigh fading channel --since the channel coefficients are complex--, we have
\begin{equation}
\label{eq:Tx_Rx_vector_Rayleigh}
\V{y}=\sqrt{\rho}\,\V{D_{s}}\,\V{h}+\V{w}.
\end{equation}
An ML coherent receiver aims at maximizing the conditional PDF of receiving $\V{y}$ given that both the transmit signal $\V{s}$ and the channel vector $\V{h}$ are known. Namely,
\begin{equation}
 \V{\hat{s}}_\text{ML}=\argmax\limits_{\V{s}^{i}}\:\text{p}(\V{y}|\V{s}^i,\V{h}),
\end{equation}
where,
\begin{equation}
\label{eq:p_y_s_h_GLRT}
 \text{p}(\V{y}|\V{s},\V{h})=\frac{1}{\pi^{T}}e^{-(\V{y}-\sqrt{\rho}\V{D}_{\V{s}}\V{h})\hr(\V{y}-\sqrt{\rho}\V{D}_{\V{s}}\V{h})}.
\end{equation}
Define $\V{\hat{h}}$ as the ML Estimate (MLE) of the channel coefficients under the hypothesis that $\V{s}^i$ was sent. The GLRT metric is then defined as 
\begin{equation}
 \V{\hat{s}}_\text{GLRT}=\argmax\limits_{\V{s}^{i}}\:\text{p}(\V{y}|\V{s}^i,\hat{\V{h}}).
\end{equation}
Based on (\ref{eq:p_y_s_h_GLRT}), the optimal choice of $\V{\hat{h}}$ is the one that minimizes the exponent
\begin{equation*}
 J=(\V{y}-\sqrt{\rho}\V{D}_{\V{s}}\V{\hat{h}})\hr(\V{y}-\sqrt{\rho}\V{D}_{\V{s}}\V{\hat{h}}).
\end{equation*}
Optimizing $\V{\hat{h}}$ can be done by optimizing its entries $\hat{h}_t\:\forall\,t\in\{0,...,T-1\}$ element by element. To obtain the optimal $\hat{h}_t$, we differentiate $J$ with respect to $\hat{h}_t$ and set the derivative to zero. Namely,
\begin{equation*}
 \frac{\partial J}{\partial \hat{h}_t}=-\sqrt{\rho}\,s_t(y_t^*-\sqrt{\rho}s_t^*\hat{h}_t^*)\req 0,
\end{equation*}
making the optimal $\hat{h}_t$
\begin{equation*}
 \hat{h}_{t,\text{opt}}=\frac{y_t}{\sqrt{\rho}\,s_t}.
\end{equation*}
Therefore, the optimal MLE of the channel vector is 
\begin{equation}
\label{eq:opt_h_RFF}
 \V{\hat{h}}_{t,\text{opt}}=\frac{1}{\sqrt{\rho}}\V{D}_{\V{s}^i}^{-1}\V{y}.
\end{equation}
Substituting (\ref{eq:opt_h_RFF}) in the GLRT metric reduces it to
\begin{equation}
 \V{\hat{s}}_\text{GLRT}=\argmin\limits_{\V{s}^{i}}\:\|\V{y}-\sqrt{\rho}\V{D}_{\V{s}^i}\V{\hat{h}}_{t,\text{opt}}\|_F^2=\argmin\limits_{\V{s}^{i}}\:0.
\end{equation}
Therefore for a Rayleigh fast fading channel, the GLRT metric can not be used. One interpretation to this is that it is not possible to jointly decide on $T$ channel coefficients and $(T-1)$ information symbols having only $T$ observation variables $y_t$, $t=0,...,T-1$.
\subsubsection{GLRT metric for Rayleigh Block Fading and AWGN channels}
In the Rayleigh block fading channel case, the channel coefficient $h$ remains constant over the observation window of the MSDD. Therefore $\V{D}_{\V{s}}\V{h}$ can be replaced by $h\V{s}$, and the transmission equation can be written as
\begin{equation*}
 \V{y}=\sqrt{\rho}\,h\V{s}+\V{w}.
\end{equation*}
The above transmission equation can also describe an AWGN channel if $h$ was just a phasor. Therefore, the following derivation is valid for both the AWGN and the Rayleigh block fading channel. Similar to the previous section, the optimal choice of the MLE $\hat{h}$ is the one that minimizes
\begin{equation*}
 J=(\V{y}-\sqrt{\rho}\hat{h}\V{s}^i)\hr(\V{y}-\sqrt{\rho}\hat{h}\V{s}^i).
\end{equation*}
By differentiating $J$ and setting its derivative to zero, we get
\begin{equation*}
 \frac{\partial J}{\partial \hat{h}}=-\sqrt{\rho}\,(\V{y}\hr-\sqrt{\rho}\,\hat{h}^*\V{s}\sp{i\,\hr})\V{s}^i\req 0,
\end{equation*}
and the optimal MLE of $\hat{h}$ is
\begin{equation}
 \hat{h}_{\text{opt}}=\frac{\V{s}\sp{i\,\hr}\V{y}}{\sqrt{\rho}\,\|\V{s}^i\|^2}.
\end{equation}
Substituting $\hat{h}_{\text{opt}}$ in the GLRT metric reduces it to
\begin{eqnarray}
 \V{\hat{s}}_\text{GLRT}&=&\argmin\limits_{\V{s}^{i}}\:\|\V{y}-\sqrt{\rho}\frac{\V{s}\sp{i\,\hr}\V{y}}{\sqrt{\rho}\|\V{s}^i\|^2}\V{s}^i\|_F^2\nonumber\\
&=&\argmin\limits_{\V{s}^{i}}\:\|\V{y}\|^2-2\Re\{\frac{\V{y}\hr\V{s}\sp{i\,\hr}\V{y}\V{s}^i}{\|\V{s}^i\|^2}\}+\frac{\V{s}\sp{i\,\hr}\overbrace{\V{y}\hr\V{s}^i\V{s}\sp{i\,\hr}\V{y}}^{\text{scalar}}\V{s}^i}{\|\V{s}^i\|^4}\nonumber\\
&=&\argmax\limits_{\V{s}^{i}}\:2\frac{|\V{s}\sp{i\,\hr}\V{y}|^2}{\|\V{s}^i\|^2}-\frac{|\V{s}\sp{i\,\hr}\V{y}|^2}{\|\V{s}^i\|^2}\nonumber\\
&=&\argmax\limits_{\V{s}^{i}}\:\frac{|\V{s}\sp{i\,\hr}\V{y}|^2}{\|\V{s}^i\|^2}=\argmax\limits_{\V{s}^{i}}\:\frac{|\V{y}\hr\V{s}^i|^2}{\|\V{s}^i\|^2}.
\label{eq:GLRT_RBF_AWGN}
\end{eqnarray}\\

Interestingly, the GLRT MSDD metric is the same as the ML MSDD metric for a constant envelope modulation like DPSK in a Rayleigh block fading or an AWGN channel as shown in Table \ref{tab:ML_metrics}. For DAPSK, the GLRT metric in (\ref{eq:GLRT_RBF_AWGN}) is much simpler than the AWGN and Rayleigh block fading ML metrics shown in Table \ref{tab:ML_metrics} and does not need the knowledge of SNR ($\rho$).\\

Table \ref{tab:GLRT_metrics} summarizes the GLRT metrics for the different channel conditions using constant and non-constant envelope modulation techniques. RBF signifies the Rayleigh block fading channel and RFF signifies the Rayleigh fast fading channel. A general MSDD decoder which is valid for both the ML metric and the GLRT metric can be visualized as shown in Figure \ref{fig:metric_decoder}. We define $\xi_i$ as the argument of $\argmax$ or $\argmin$ for the $i^{\text{th}}$ candidate sequence $\V{s}^i$ in the MSDD metric. The cardinality of the candidate sequences set is defined as L.\\
\renewcommand{\arraystretch}{1.8} 
\begin{table}[!htp]
\begin{center}
\caption{GLRT MSDD metrics in a single antenna system}
\label{tab:GLRT_metrics}
  \begin{tabular}{|p{3cm}|p{5cm}|p{5cm}|}
    \hline    
      \textbf{GLRT MSDD metric} $\V{\hat{s}}_\text{GLRT}$&  \textbf{non-constant envelope modulation e.~\!g DAPSK} &  \textbf{constant envelope modulation e~\!.g DPSK} \tn[5pt] \hline 
       AWGN \& RBF & \centering $\argmax\limits_{\V{s}^{i}}\:\frac{|\V{y}\hr\V{s}^i|^2}{\|\V{s}^i\|^2}$	& \centering $\argmax\limits_{\V{s}^{i}}\:|\V{y}\hr\V{s}^i|$ \tn \hline
       RFF& \multicolumn{2}{c|}{Not possible} \tn \hline
  \end{tabular}
\end{center}
\end{table}

\begin{figure}[!ht]
  \centering
\scalebox{0.8}{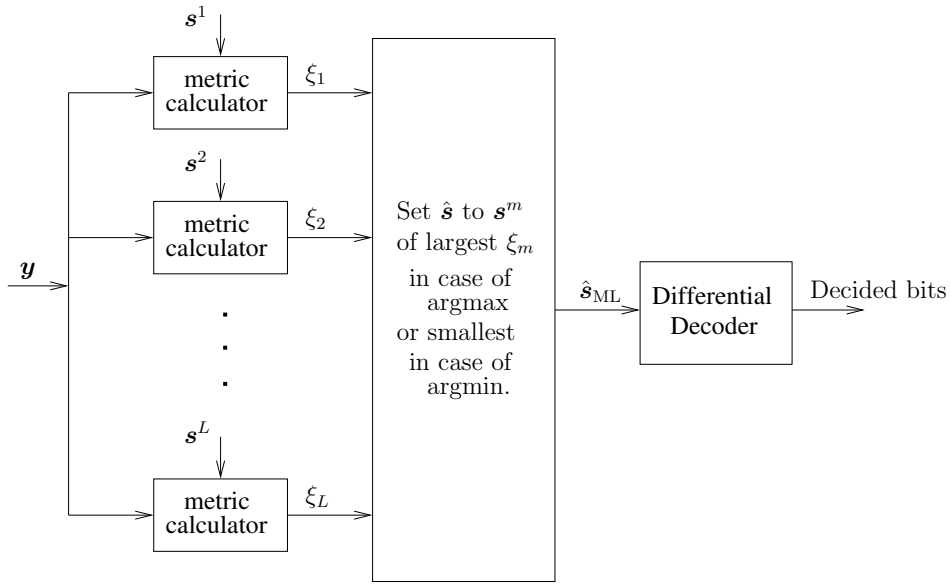}
\caption{MSDD decoder}
\label{fig:metric_decoder}
\end{figure}
\section{Performance Analysis}
\label{s:Performance_MSDD}
This section shows the simulation results of the non-coherent MSDD receiver using both DPSK and DAPSK modulation techniques in different channel conditions. 
For DPSK, the ML and GLRT metrics for both AWGN and Rayleigh block fading channels are just the same, and in the fast fading case the GLRT metric is not usable. Therefore, in the following, DPSK uses by default the ML metric.\\

The decoding complexity of the receiver is mainly governed by the search space of the metric. For DPSK, the search space is $L_p=q_p^{T-1}$, whereas for DAPSK the search space is $L=(q_pq_a)^{T-1}$. Thus, the complexity increases exponentially with the window length. Therefore, the following simulations were done for up to a window length of $T=4$. Modulation orders of 8, 16, 32, and 64 are used. For DAPSK, the constellations use the optimal parameters shown in Table \ref{tab:opt_DAPSK}.\\

As an attempt to reduce the search space of the DAPSK scheme, we additionally propose a less complex detection approach that tries to exploit the independency of the amplitude and phase in the used DAPSK constellation. As shown in Figure \ref{fig:metric_indep_detection}, this technique first applies the received vector $\V{y}$ on the DPSK metric and decides on the phase sequence $\hat{\V{s}}_p$. Then the decided phase sequence is multiplied by all amplitude candidates $\V{s}_a^i$ getting $\V{s}^i=\hat{\V{s}}_p\,\V{s}_a^i$. Finally, $\V{s}^i$ is applied on the DAPSK metric to decide on the amplitude sequence $\hat{\V{s}}_a$. This will result in reducing the search space from $L\eq L_pL_a\eq(q_pq_a)^{T-1}$ to  $L\eq L_p+L_a\eq q_p^{T-1}+q_a^{T-1}$. \\

Using the simplified scheme the decision on the phase sequence is totally independent of the amplitude, whereas the decision on the amplitude sequence depends on the decided phase sequence. Thus, we call such a scheme quasi-independent detection. Also used is a detection scheme which directly applies the received vector $\V{y}$ on the DAPSK metric and tests all $(q_pq_a)^{T-1}$ candidates constructed as the product of all possible amplitude and phase sequence combinations. Such a scheme decides on the amplitude and phase sequences simultaneously and is referred to as combined detection.\\

\begin{figure}[htp]
  \centering
\scalebox{0.8}{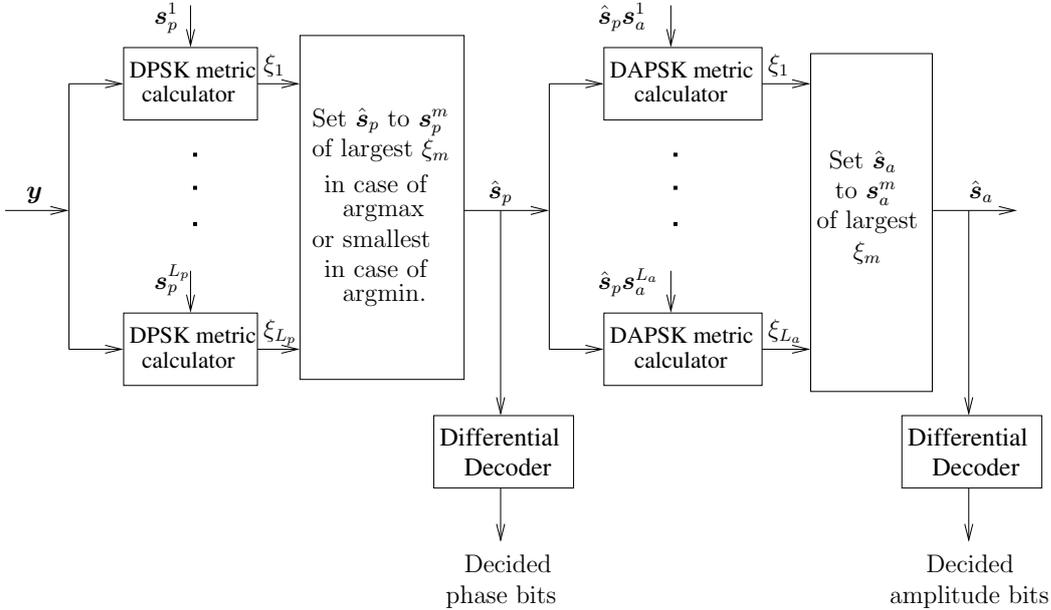}
\caption{MSDD with quasi-independent detection for DAPSK modulation scheme.}
\label{fig:metric_indep_detection}
\end{figure}


For an AWGN channel, the modified Bessel function in the ML metric for DAPSK results in overflow of the metric values. Therefore, we approximate the Bessel function using
\begin{equation}
 I_{0}(z)\approxeq \frac{e^z}{\sqrt{2\pi z}},\hspace{0.5cm} z\gg\frac{1}{4} \hspace{0.7cm}\implies \ln(I_0(z))\approxeq z-\ln(\sqrt{2\pi z}),
\end{equation}
and the decision scheme is termed as approx. ML decision. Figure \ref{fig:DAPSK_simulation_summary} summarizes the different cases whose results will be shown in the rest of this section.
\begin{figure}[htp]
  \centering
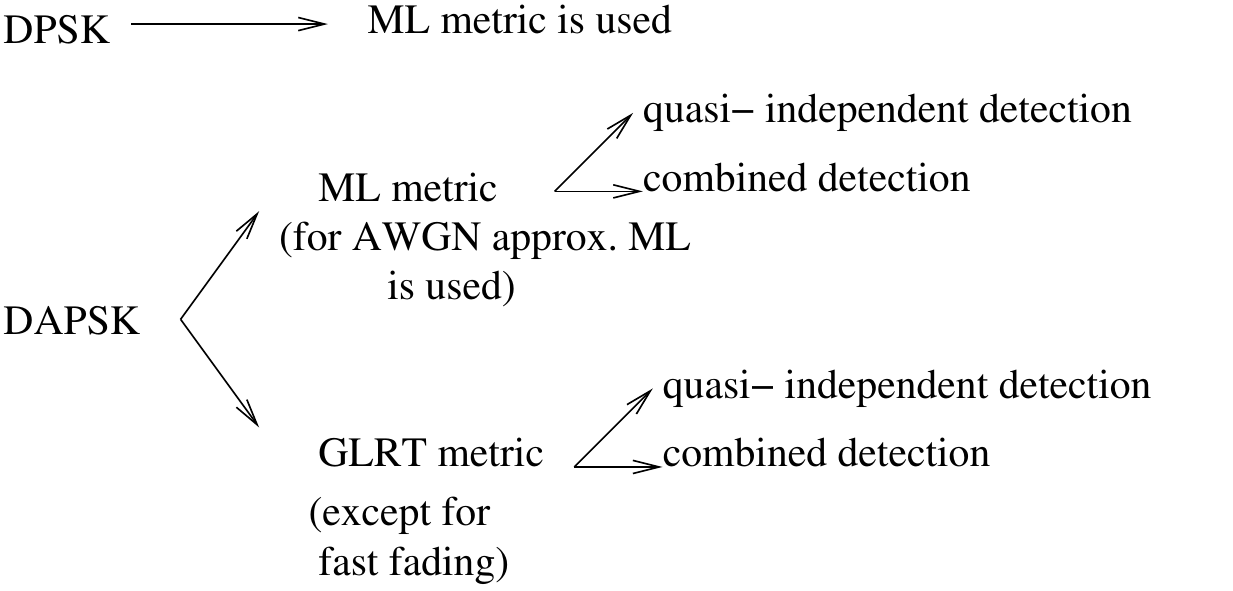
\caption{Summary of the different cases considered for simulation results of DPSK and DAPSK with an MSDD receiver.}
\label{fig:DAPSK_simulation_summary}
\end{figure}

\subsection{AWGN}
\label{ss:results_AWGN}
\subsubsection{Quasi-independent vs. combined detection}
We first compare the quasi-independent scheme shown in Figure \ref{fig:metric_indep_detection} with the combined detection for both the GLRT and the ML metrics in an AWGN channel. Figure \ref{fig:AWGN_indep_vs_combined} shows that the two methods are identical. This is reasonable since the modulation of the amplitude and the phase bits is done independently, so there is no gain in decoding them jointly. Therefore, in the AWGN case, only results of the quasi-independent detection scheme will be shown. In the figure, arbitrary modulation orders are chosen since the effect is the same for all modulation orders. The figure also shows the effect of multiple-symbol detection. For example in the case of 32-DAPSK using the ML metric with a window length of 3 achieves about $\unit[1]{dB}$ gain compared to the conventional detection. Additionally, the \emph{increase} in the gain decreases with increasing $T$. Increasing $T$ above $4$ leads to only a marginal improvement which is not much worth the added complexity. 
\begin{figure}[!htp]
 \centering
  \subfloat[ML quasi-independent vs. combined for 32-DAPSK]{\label{subfig:32_DAPSK_AWGN_ML_indep_vs_combined}\includegraphics[width=0.49\textwidth]{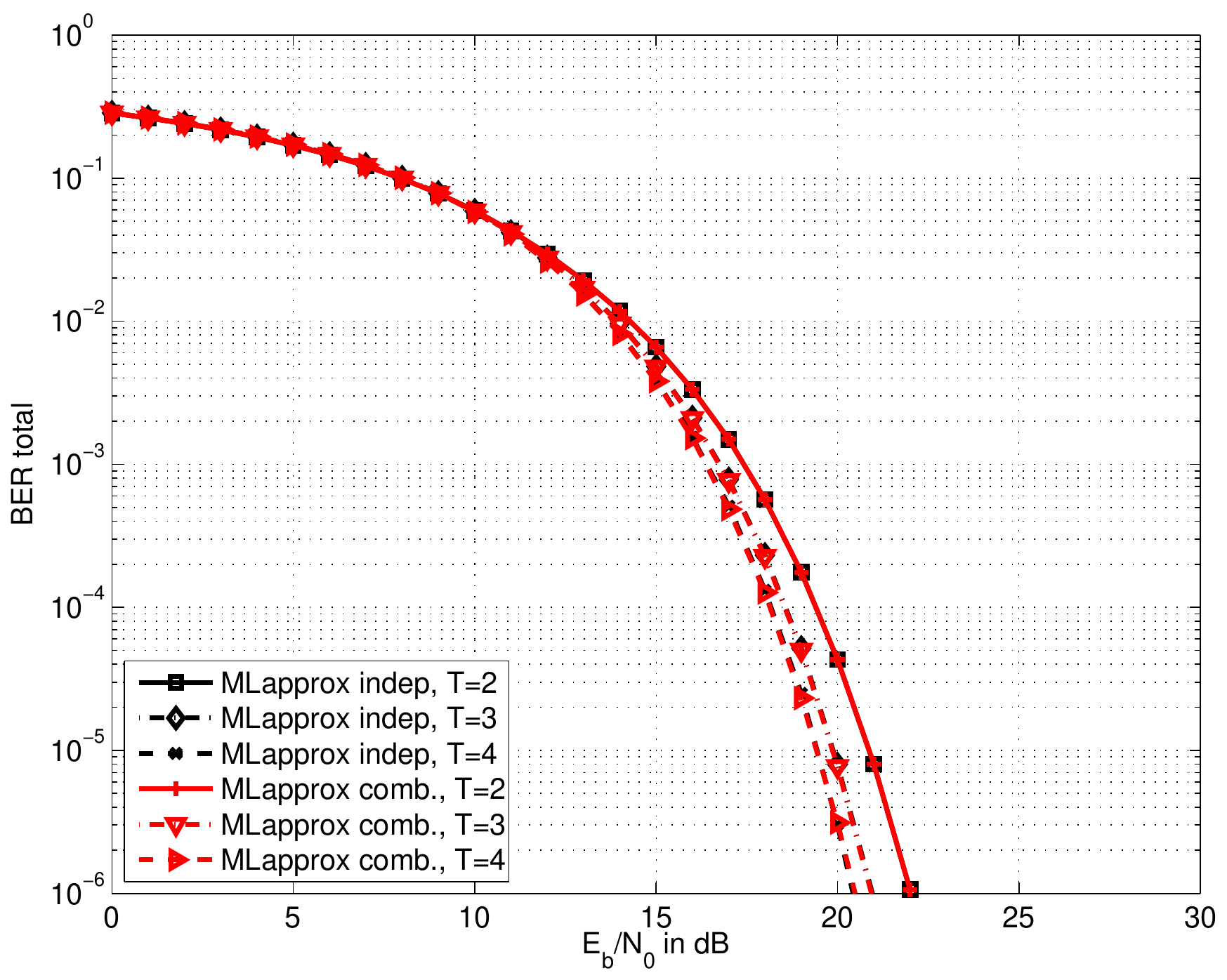}}
  \subfloat[GLRT quasi-independent vs. combined for 16-DAPSK]{\label{subfig:16_DAPSK_AWGN_GLRT_indep_vs_combined}\includegraphics[width=0.49\textwidth]{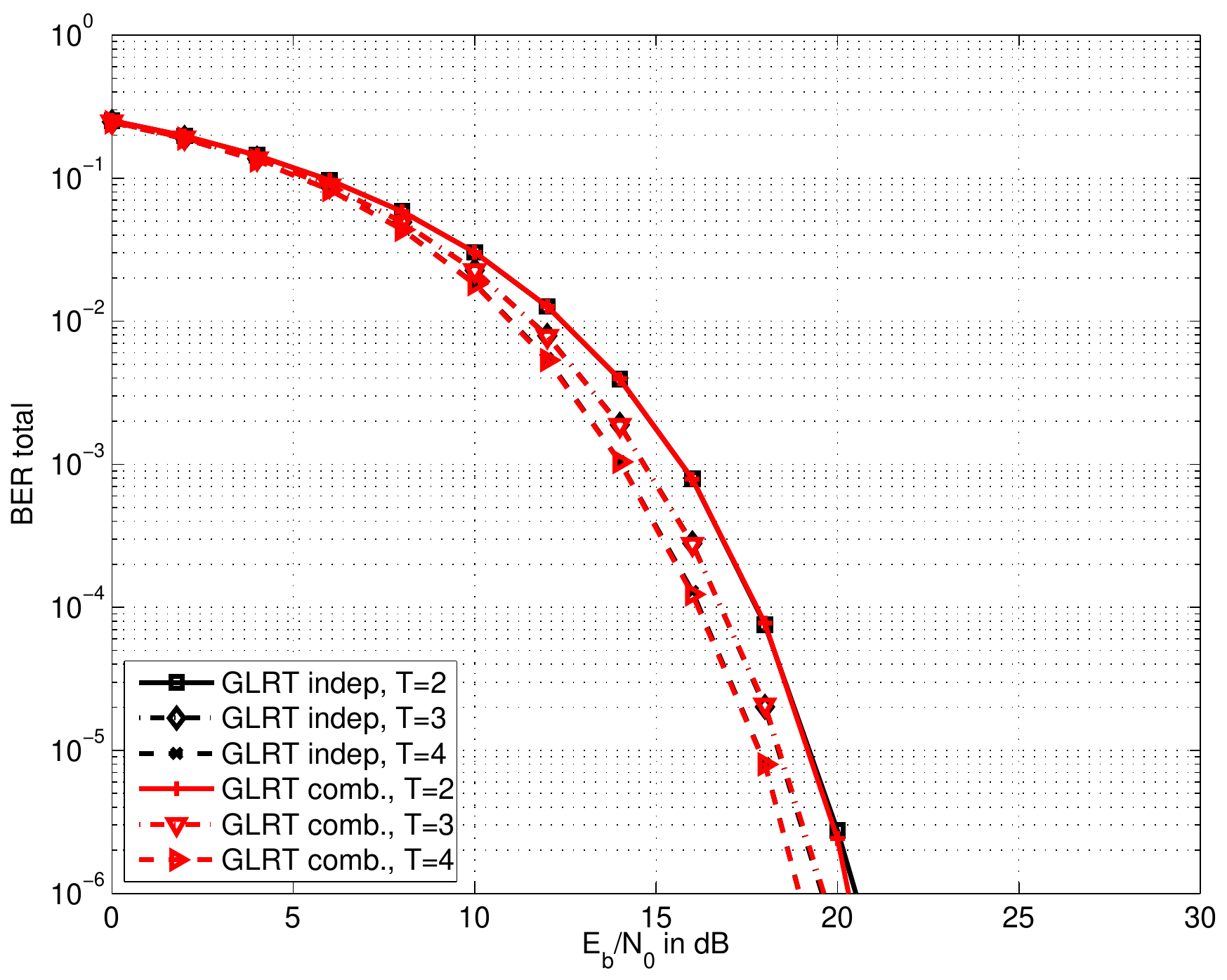}}
  \caption{BER performance of both quasi-independent detection and combined detection for DAPSK using ML and GLRT metrics over an AWGN channel.}
  \label{fig:AWGN_indep_vs_combined}
\end{figure}

\begin{figure}[!htp]
 \centering
  \subfloat[8-DPSK and 8-DAPSK]{\label{subfig:8_DAPSK_AWGN}\includegraphics[width=0.49\textwidth]{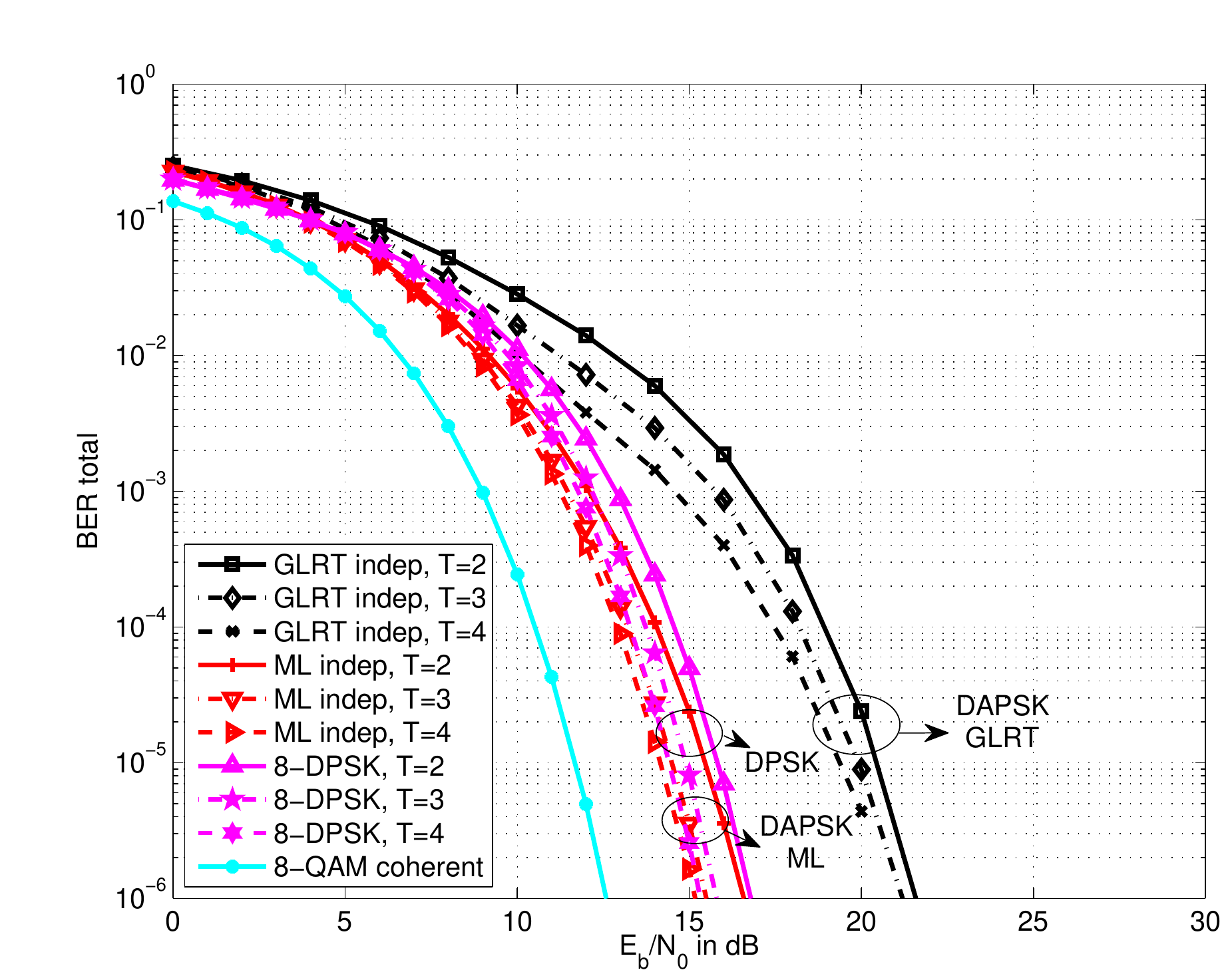}}
  \subfloat[16-DPSK and 16-DAPSK]{\label{subfig:16_DAPSK_AWGN}\includegraphics[width=0.49\textwidth]{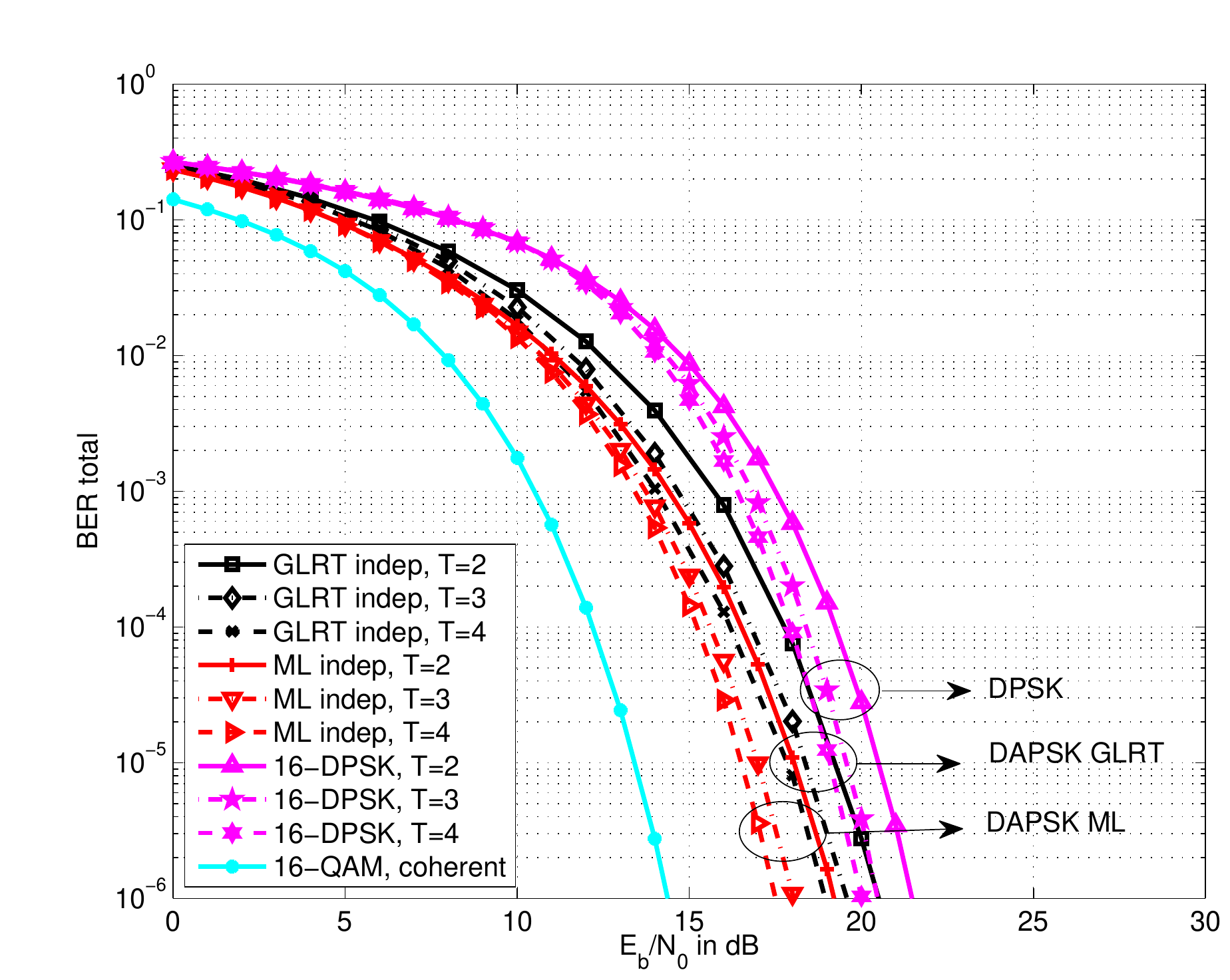}}\\
  \subfloat[32-DPSK and 32-DAPSK]{\label{subfig:32_DAPSK_AWGN}\includegraphics[width=0.49\textwidth]{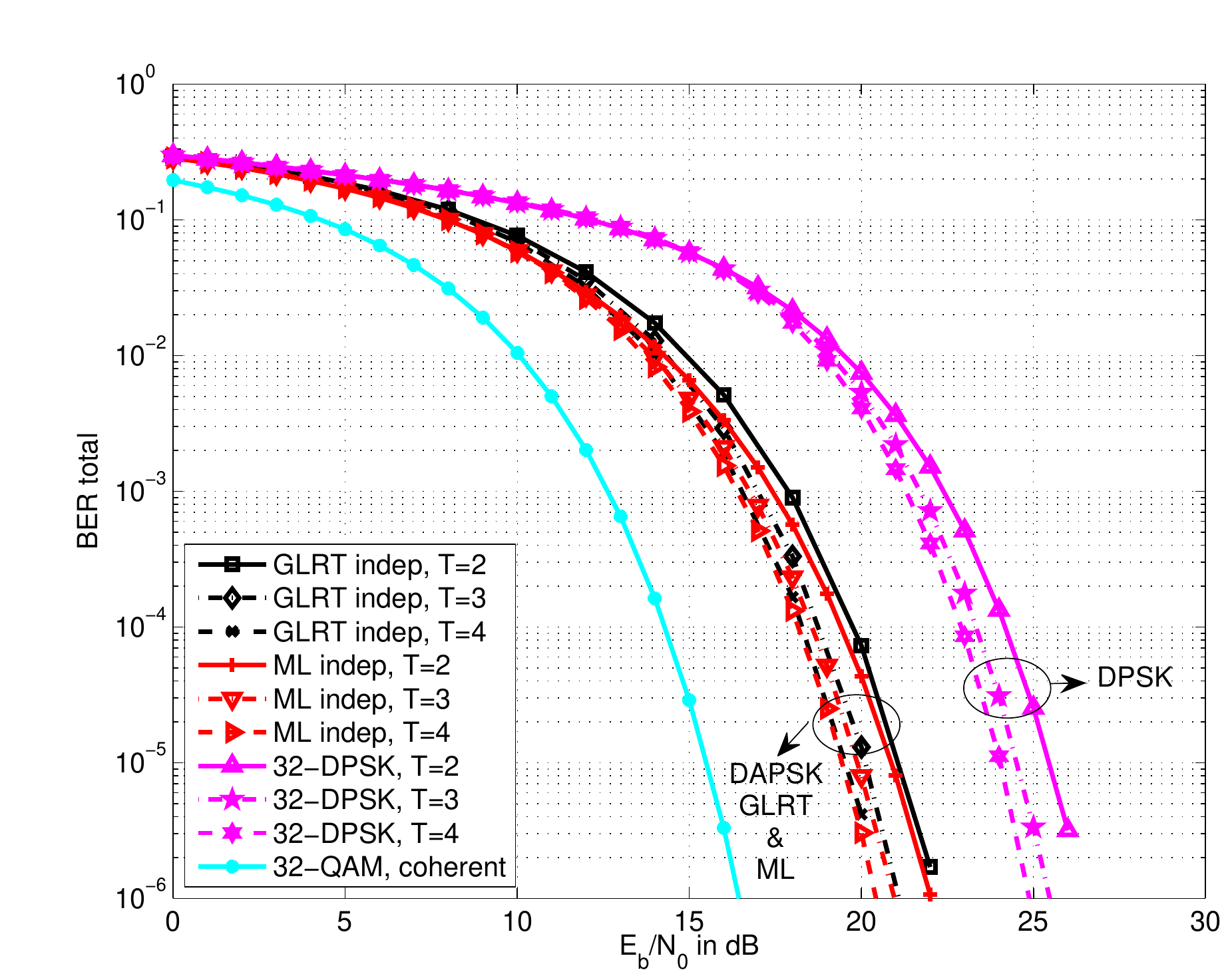}}
  \subfloat[64-DPSK and 64-DAPSK]{\label{subfig:64_DAPSK_AWGN}\includegraphics[width=0.49\textwidth]{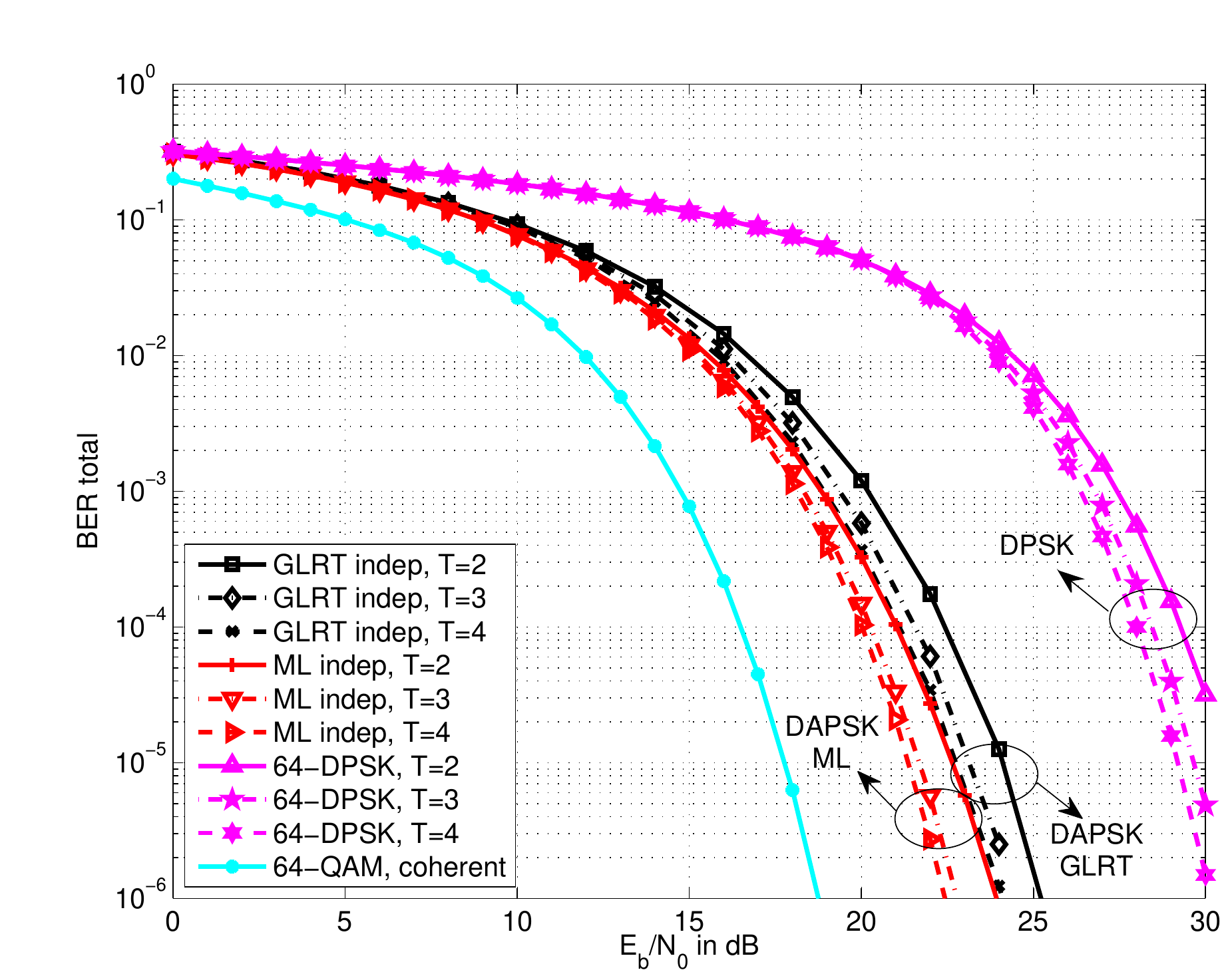}}
  \caption{BER curves of DPSK using ML metric and DAPSK using both GLRT and ML metrics over an AWGN channel.}
  \label{fig:DAPSK_AWGN}
\end{figure}

\subsubsection{ML vs. GLRT}
Figure \ref{fig:DAPSK_AWGN} compares the BER curves of different-order DAPSK schemes using the ML and the GLRT decision metrics over an AWGN channel. Also shown is the DPSK BER curves and the coherent QAM counterpart. For an alphabet order of 8, the GLRT performance is significantly worse than the ML performance. In fact the performance of DPSK with ML (same as GLRT) is better than that of DAPSK with a GLRT receiver. The ML DAPSK curves are a little bit better  than the ML DPSK ones. Also shown is the effect of increasing the window length $T$. For example a performance gain of about $\unit[1]{dB}$ is incurred going from $T\eq 2$ to $T\eq 4$ in an ML receiver. For a modulation order of 16, the gain of using DAPSK over DPSK with an ML receiver is about $\unit[3]{dB}$. For 16-DAPSK, there is still some SNR loss in using the sub-optimum GLRT metric compared to the ML decision, but the gap between both is much less than in the 8-DAPSK case.\\

For a modulation order of 32, the performance gap between the DPSK and the DAPSK increases further. Whereas the gap between the GLRT and the ML decision metrics is insignificant. When increasing the modulation order to 64, DPSK shows about $\unit[8]{dB}$ performance loss compared to ML DAPSK. The performance gap between the ML and the GLRT metrics increases again although it was gradually decreasing going from 8-DAPSK to 16-DAPSK to 32-DAPSK. The difference with 64-DAPSK is the use of four amplitude levels instead of two. Thus, the performance of the sub-optimal detector GLRT degrades as more amplitude levels are introduced.
\subsection{Rayleigh Block Fading}
\label{ss:results_RBF}
Just as in the AWGN channel, for a Rayleigh block fading channel, the combined detection has shown identically the same performance as the quasi-independent one. Therefore we only show the results of the quasi-independent detection. Figure \ref{fig:RBF_all} shows the results of the ML and the GLRT receivers for the DAPSK scheme. The DPSK  as well as the coherent QAM curves are also shown. Clearly, there is hardly any effect of multiple-symbol detection for a Rayleigh block fading channel. Such an effect has been incurred in all modulation orders, and is shown here for order 8 and 32. A possible explanation for this is that in the Rayleigh block fading channel, every possible channel realization can be viewed as an AWGN channel with SNR being the channel power of such a realization. That is the Rayleigh block fading performance can be viewed as averaging over the BER of different SNR values in an AWGN channel and the low SNR values dominate the performance.\\

Similar to the AWGN case, DAPSK shows better error rate compared to DPSK for a modulation order above 8.  For example, a gain of about $\unit[4]{dB}$ is achieved using 32-DAPSK over 32-DPSK.
\begin{figure}[!htp]
 \centering
  \subfloat[8-DPSK and 8-DAPSK]{\label{subfig:8_DAPSK_RBF_all}\includegraphics[width=0.49\textwidth]{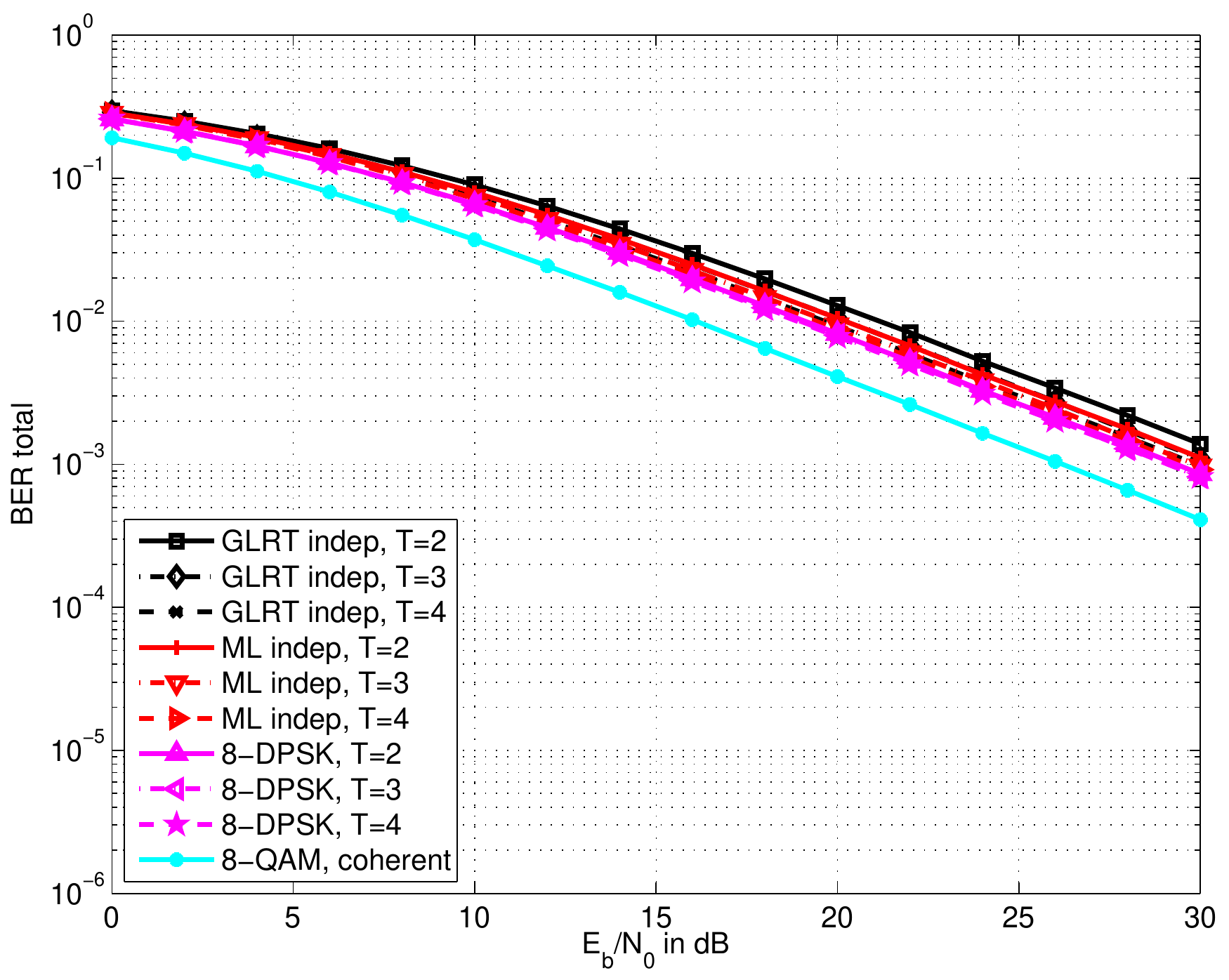}}
  \subfloat[32-DPSK and 32-DAPSK]{\label{subfig:32_DAPSK_RBF_all}\includegraphics[width=0.49\textwidth]{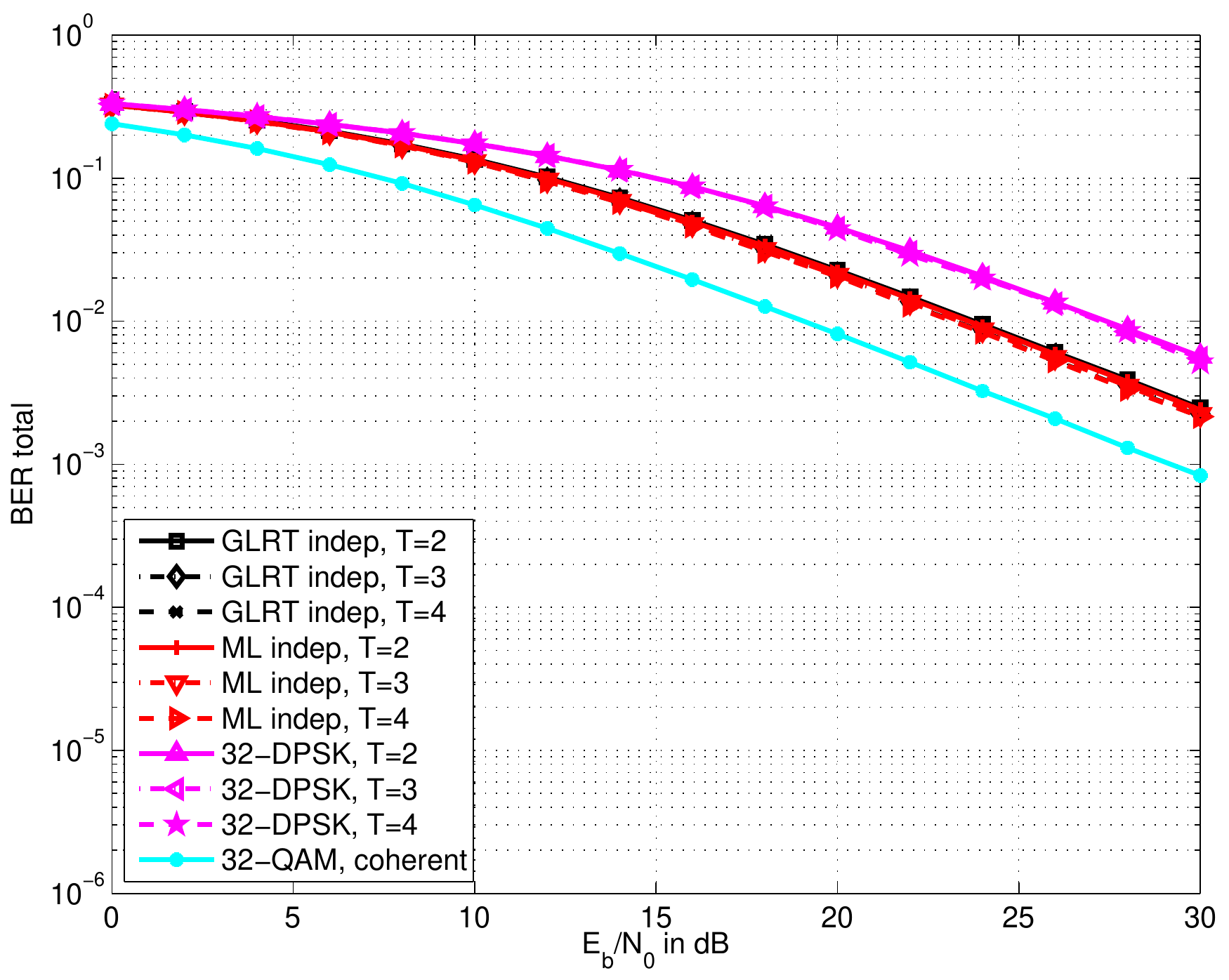}}
  \caption{BER performance of DPSK and DAPSK using quasi-independent ML and GLRT metrics over a Rayleigh block fading channel.}
  \label{fig:RBF_all}
\end{figure}

\subsection{Rayleigh Fast Fading}
\label{ss:results_RFF}
In a fast fading channel, we first investigate the performance of the ML metric with the quasi-independent detection. Figure \ref{fig:RFF_indep} shows the DPSK BER, the DASK BER and the overall DAPSK BER. As shown, the BER curves look unreasonable for a window length larger than 2. The DPSK BER curve indicates that the decision on the phase sequence is wrong which in turn affects the amplitude decision resulting in wrong DASK BER curve, and thus wrong overall performance. One interpretation for the failure of the DPSK metric in the fast fading case is that amplitude modulation destroys the assumed channel correlation for the phase modulation. In other words, the phase modulation encounters a different statistical channel covariance matrix which includes the statistics of the amplitude variation. It is not yet known how to modify the equivalent covariance matrix to circumvent this effect. So we reside to the combined detection in the fast fading case.\\
\begin{figure}[!htp]
 \centering
  \subfloat[DPSK performance]{\label{subfig:8DAPSK_RFF_fdTs002_MS_MLindep_vs_combined_BER_PSK}\includegraphics[width=0.49\textwidth]{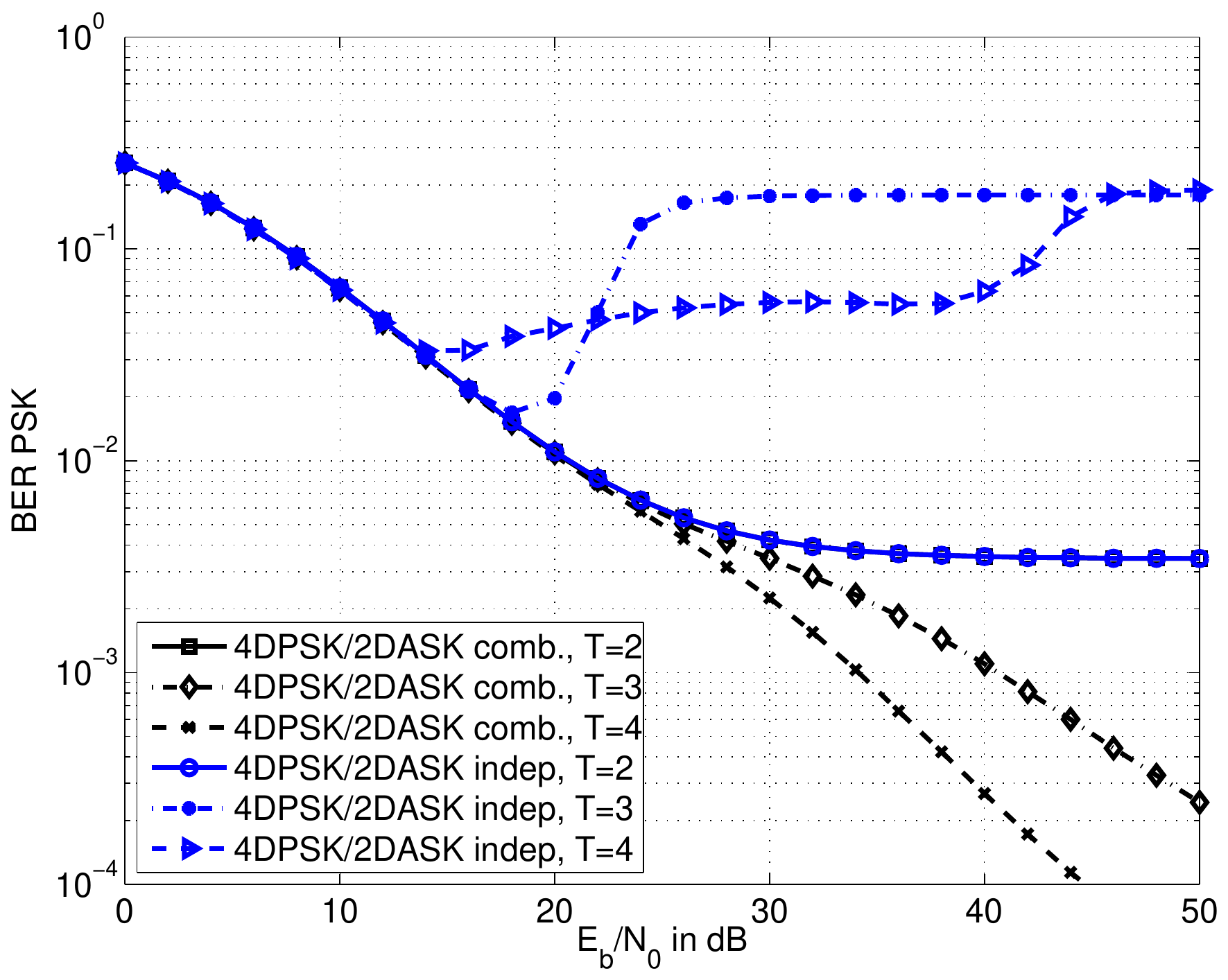}}
  \subfloat[DASK performance]{\label{subfig:8DAPSK_RFF_fdTs002_MS_MLindep_vs_combined_BER_ASK}\includegraphics[width=0.49\textwidth]{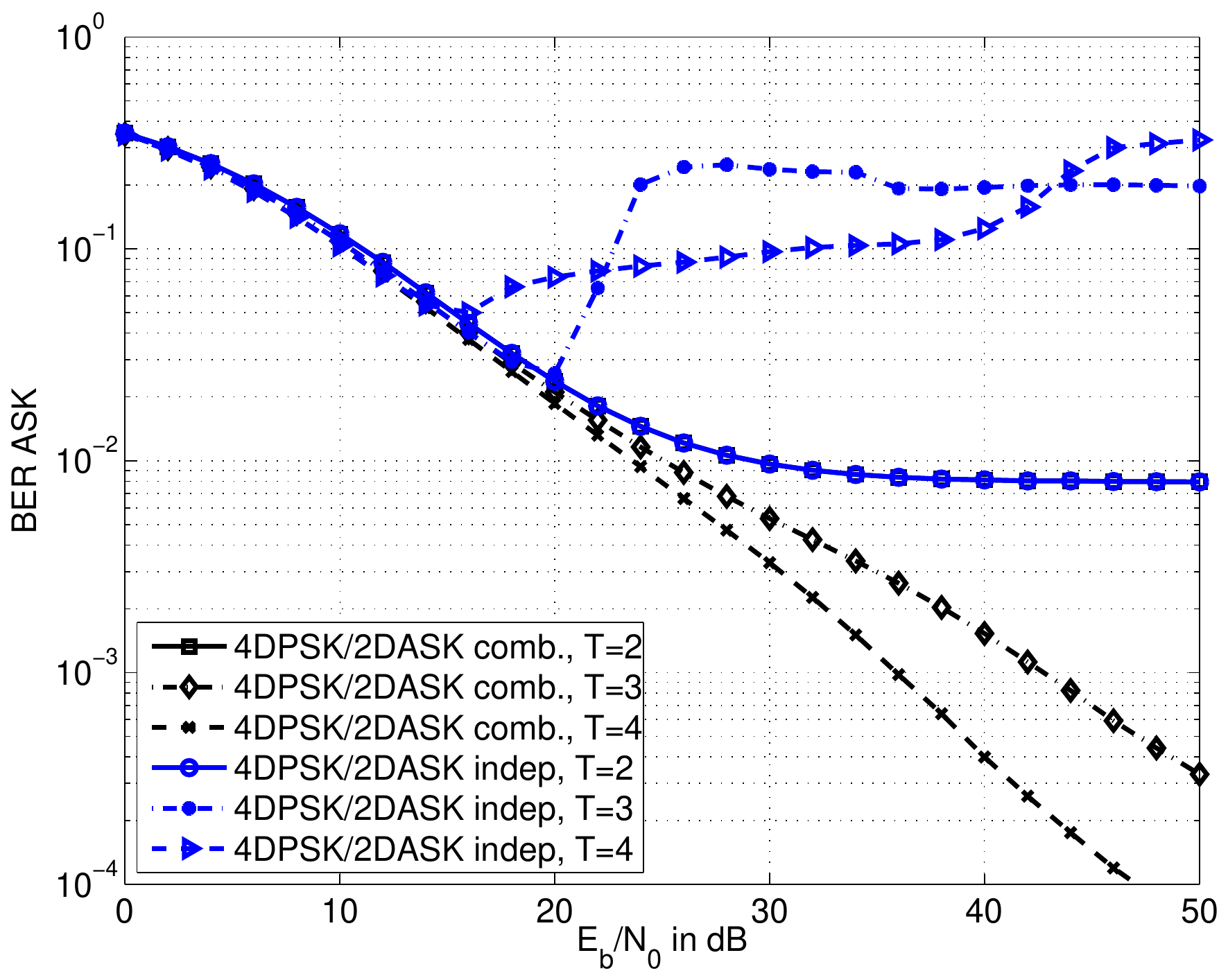}}\\
\subfloat[Total DAPSK performance]{\label{subfig:8DAPSK_RFF_fdTs002_MS_MLindep_vs_combined_BERtotal}\includegraphics[width=0.49\textwidth]{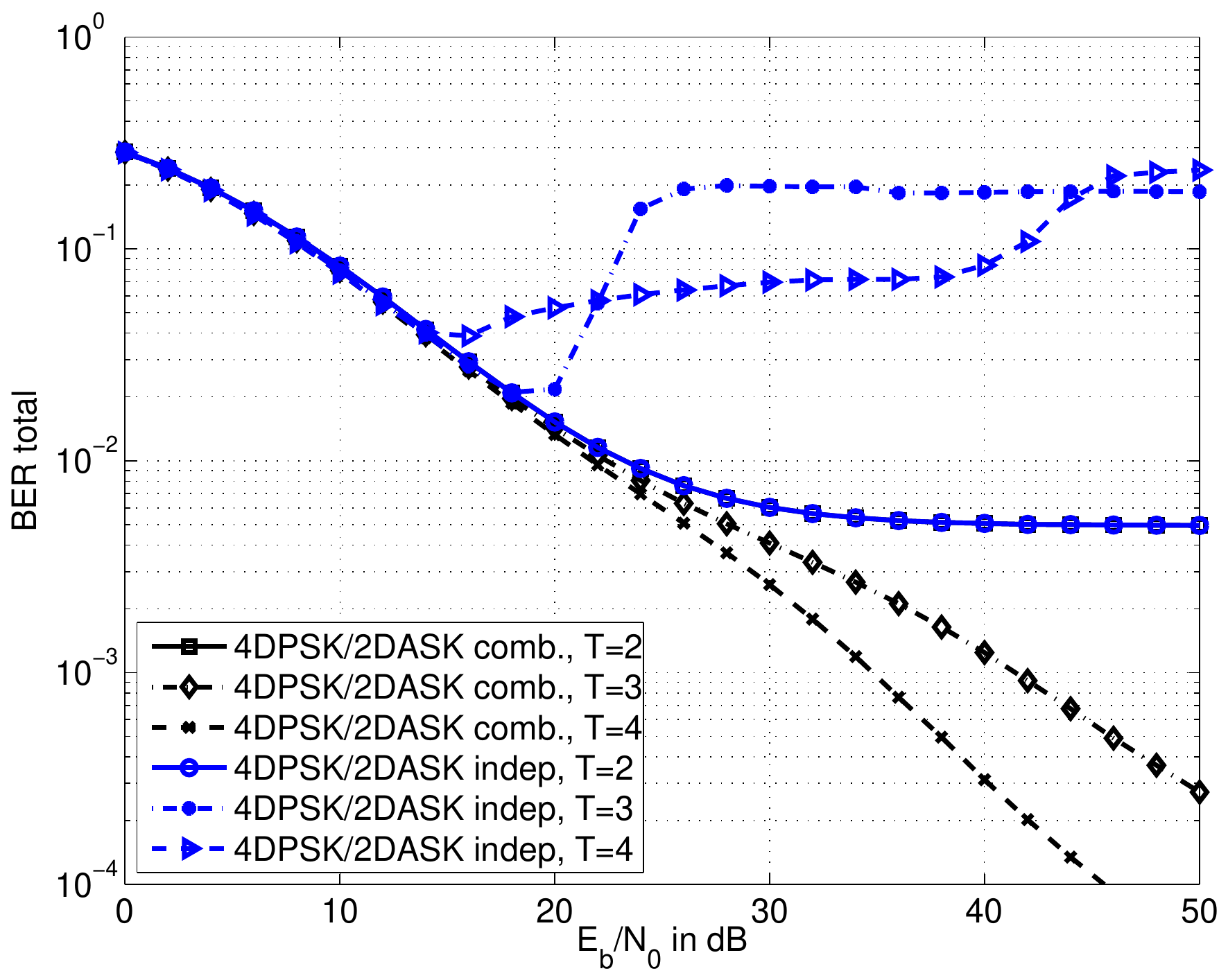}}
  \caption{BER performance of DAPSK using independent ML metric over a Rayleigh fast fading channel at $f_DT_s\eq0.02$.}
  \label{fig:RFF_indep}
\end{figure}

Figure  \ref{fig:RFF_combined} shows the BER curves of DPSK and DAPSK using an ML receiver over a fast varying channel with a normalized Doppler shift $f_DT_s$ of 0.02. Clearly MSDD removes the error floor caused by the channel variation. A main reason for the performance improvement of  MSDD compared to conventional differential detection in the fast fading case is that MSDD exploits the correlation between the channel distortions suffered by the successive symbols within one observation block \cite{Fung92}. The gain is still significant going from $T\eq3$ to $T\eq4$. This effect is incurred for all modulation orders and is shown here for two orders. Higher Doppler shifts for up to $f_DT_s\eq0.05$ have shown the same MSDD performance effect. \\

\begin{figure}[!htp]
 \centering
  \subfloat[16-DPSK and 16-DAPSK]{\label{subfig:16_DAPSK_RFF_fdTs_0_02}\includegraphics[width=0.49\textwidth]{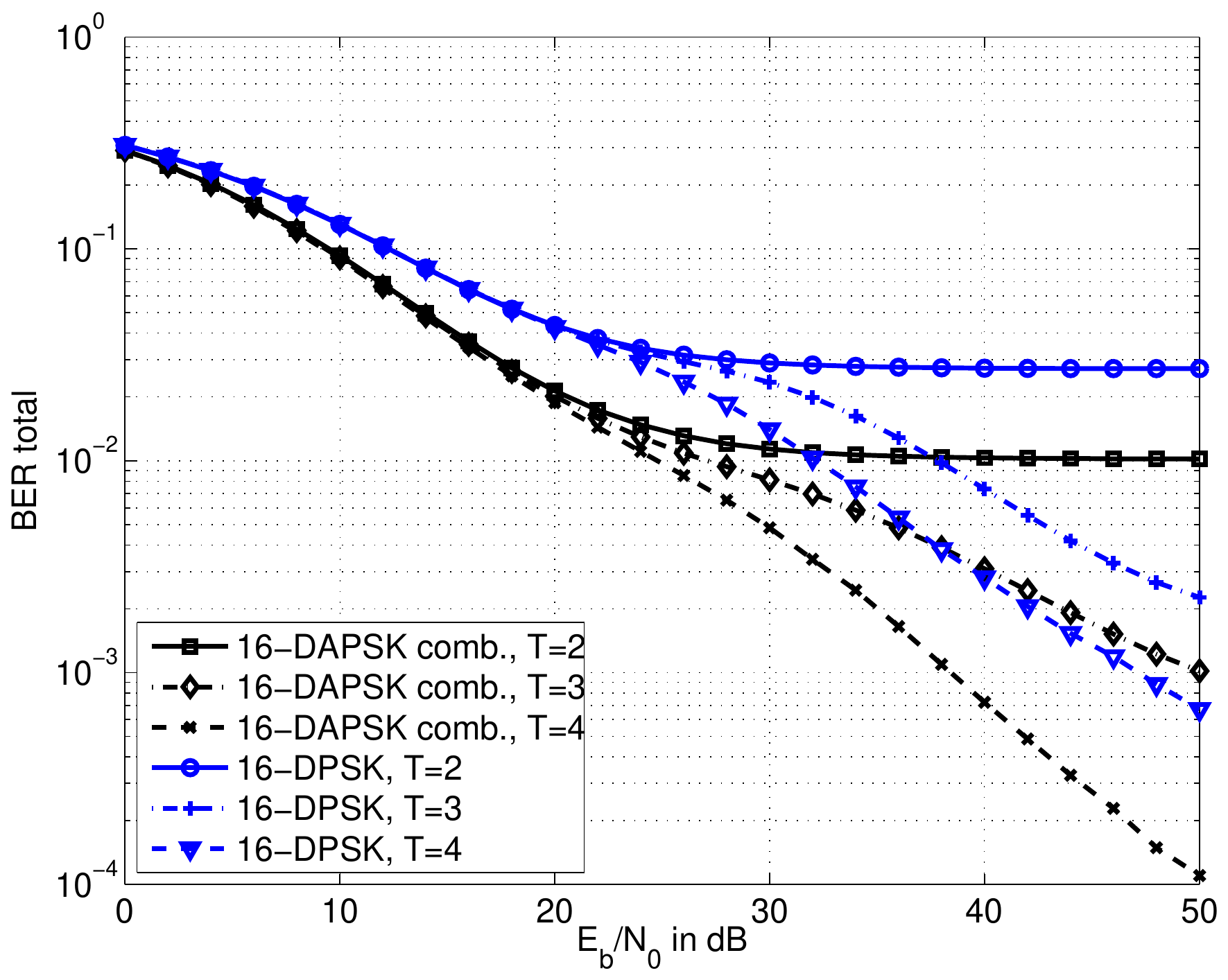}}
  \subfloat[32-DPSK and 32-DAPSK]{\label{subfig:32_DAPSK_RFF_fdTs_0_02}\includegraphics[width=0.49\textwidth]{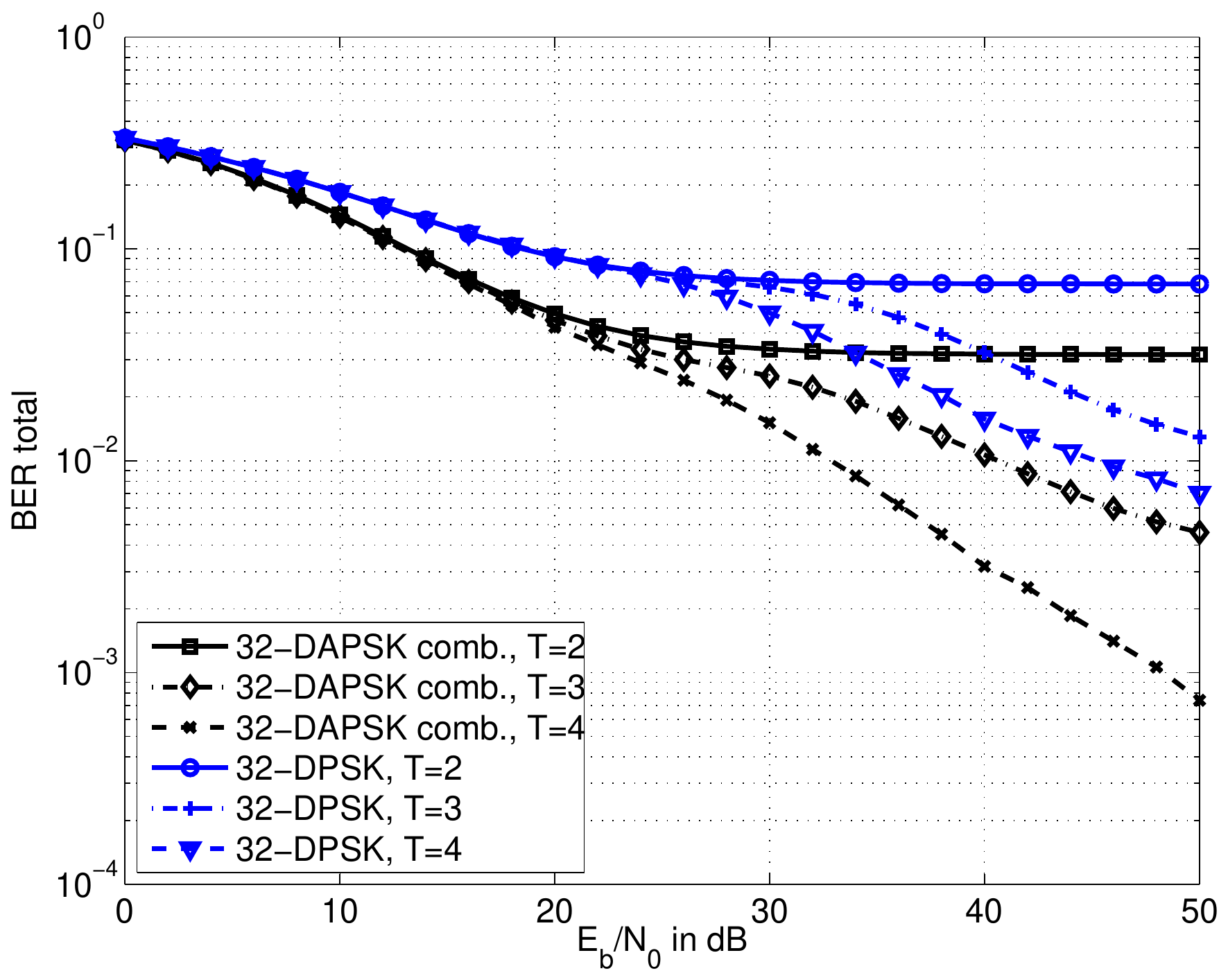}}
  \caption{BER performance of DPSK and DAPSK using combined ML metric over a Rayleigh fast fading channel at $f_DT_s\eq0.02$.}
  \label{fig:RFF_combined}
\end{figure}

In conclusion, this chapter has shown that the use of amplitude modulation in differential non-coherent systems through DAPSK improves the error rate performance compared to DPSK for alphabet size above 8. MSDD proved to remove the error floor associated with fast fading channels, and has also shown slight improvement in the AWGN channel, but almost no improvement in the Rayleigh block fading channel. The rest of the thesis will show the extension of differential modulation to systems employing more than one antenna at the transmitter and/or the receiver side to allow a higher degree of freedom for a potential performance improvement.

\abbrev{BER}{Bit Error Rate}
\abbrev{DPSK}{Differential Phase Shift Keying}
\abbrev{DASK}{Differential Amplitude Shift Keying}
\abbrev{DAPSK}{Differential Amplitude Phase Shift Keying}
\abbrev{MSDD}{Multiple-Symbol Differential Detection}
\abbrev{ML}{Maximum Likelihood}
\abbrev{GLRT}{Generalized Likelihood Ratio Test}
\abbrev{AWGN}{Additive White Gaussian Noise}
\abbrev{PDF}{Probability Density Function}
\abbrev{SNR}{Signal to Noise Ratio}
\abbrev{LOS}{Line of Sight}
\abbrev{RBF}{Rayleigh Block Fading}
\abbrev{RFF}{Rayleigh Fast Fading}
\abbrev{BER}{Bit Error Rate}
\abbrev{MLE}{Maximum Likelihood Estimate}
	


\chapter{Diversity Techniques in MIMO Systems}
\label{chap:Diversity_MIMO}

Several challenges face the progress of wireless communications
including the limited available radio spectrum as well as the complicated nature
of the wireless environment. On the
other hand, --as a natural progress of the technology-- several demands are
expected to be fulfilled in the near future, including higher data rates and
better Quality of Service (QoS). One way to face such challenges and meet the
current demands is to exploit more Degrees of Freedom (DoF) in the available
resources. The use of multiple antennas is a promising solution to the addressed
problem. Such antennas create several links between the transmitter (Tx) and the
receiver (Rx), resulting in performance improvement without the need of extra
bandwidth nor power  \cite{PNG2003}.\\

In Chapter \ref{chap:DPSK_DAPSK} we introduced differential modulation schemes
in single antenna systems, or so-called Single Input Single Output (SISO)
systems. In the rest of the thesis, the differential framework is extended to
multiple antenna systems. A system with multiple antennas at the Tx and the Rx
is referred to as a MIMO (Multiple Input Multiple Output) system. There exists
two types of gains that a MIMO system can provide, namely multiplexing gain and
diversity gain \cite{Zheng2003}. By multiplexing gain, we mean the increase in
the data rate achieved by utilizing the different spatial links to transmit
independent information streams in parallel. Whereas diversity gain is the gain achieved when the
receiver is provided with multiple independent copies of the same transmitted signal resulting in a
more reliable communication \cite{Bolcskei_Paulraj_2000}. With multiple antennas, this is possible by carrying the
same information over different independent antenna links. Studies in \cite{Zheng2003, Zheng2002}
by Zheng and Tse have shown that, there exists a fundamental trade off between both gains. In other
words, achieving higher diversity to combat fading comes at the price of
achieving lower spatial multiplexing, or equivalently lower data rate, and vice
versa.\\

When the goal is to achieve only multiplexing gain, then the main question is
how fast data can be transmitted, or equivalently what are the theoretical
capacity limits offered by a MIMO system. In \cite{Foschini98onlimits}, Foschini
proved the capacity limits in a Rayleigh fading channel that is known to the
receiver. However, when the goal is to combat fading in order to improve the transmission quality, then diversity gain is the advantage of interest.\\

There exist different types of diversity based on the communication resource
used for the repeated transmission. The same symbol can be transmitted over
several uncorrelated time slots, uncorrelated frequency bands, or independent
spatial paths, leading to time diversity, frequency diversity and spatial
diversity, respectively. The different portions of the used resource over which
the same symbol is transmitted can be defined as the diversity elements. In all
diversity types, the condition for achieving diversity is to ensure that the
diversity elements suffer independent fading, this reduces the chances of all
copies suffering a deep fade \cite{Janaswamy2001}. In MIMO systems, spatial
diversity is achieved with the diversity elements being multiple independent
links between the Tx and the Rx. Such links are known as diversity branches.
Unlike other diversity types, spatial diversity does not sacrifice bandwidth nor
power, the two most precious resources in wireless communications
\cite{PNG2003}. This makes MIMO systems attractive for next generation wireless
communications.\\

In this thesis, the main concern is on utilizing the spatial links of a MIMO
system to achieve diversity gain and hence improve the reliability of the
transmission. In this chapter, we first define the
channel model of a MIMO system in Section \ref{s:MIMO-channel_model}. Then
Section \ref{s:Rx-diversity} proceeds by explaining the diversity achieved by a
system that employs a single antenna at the transmitter and multiple antennas at
the receiver, the so-called SIMO (Single Input Multiple Output) system. The last
section introduces the possibility of achieving diversity in a MISO (Multiple
Input Single Output) system, where only the transmitter has multiple antennas.\\
 \section{MIMO Channel Model}
\label{s:MIMO-channel_model}
Consider a wireless communication system comprising $M$ antennas at the Tx and
$N$ antennas at the Rx (referred to as an $M\!\times\!N$ system) as shown in
Figure \ref{fig:MIMO_model}. Each Tx-Rx antenna pair is connected via a single non-line-of-sight
Rayleigh flat fading path. The different links are assumed to
fade independently from one another, and in general each link fades in a
time-varying manner. At each time slot, (also known as channel use) $M$ symbols
are emitted from the $M$ transmit antennas, transmitted over the different links
to all $N$ receive antennas, and finally corrupted by independent AWGN noise
added at each receive antenna. Such a system can be mathematically modeled by
\begin{equation}
 \left. y_{tn}=\sqrt{\rho\,}\sum \limits_{m=1}^{M}{s_{tm}h_{mn}^t+w_{tn}}
 \right. \hspace{1cm}\begin{aligned} &t=0,1,...\\
  &n=1,...,N
\end{aligned}	
\label{eq:MIMO_model_scalar}   
\end{equation}
where at time slot $t$, $y_{tn}$ is the received signal at receive antenna $n$,
$s_{tm}$ denotes the normalized symbol emitted by antenna $m$, and $h_{mn}^t$
represents the complex fading coefficient (or impulse response) between the
$m^{\text{th}}$ transmit antenna and the $n^{\text{th}}$ receive antenna. All
channel coefficients are $\mathcal{CN}(0,1)$ distributed and are assumed to be
statistically independent with respect to $m$ and $n$, $\forall$  $m=1,...,M$ and
$ n=1,...,N$. They are however in general \emph{dependent} with respect to $t$.
That means each channel impulse response in general varies with time in a
correlated manner. The received signal at receive antenna $n$ is then corrupted
by AWGN noise sample $w_{tn}$. All noise samples are drawn from a stochastic
Gaussian process with zero mean and unit power spectral density, i.e. $w_{tn}
\thicksim\mathcal{CN}(0,1)$ and are statistically independent among all receive
antennas and with time, i.e. with respect to (w.r.t.) $n$ and $t$. \\
\begin{figure}[!h]
\centering
\scalebox{0.8}{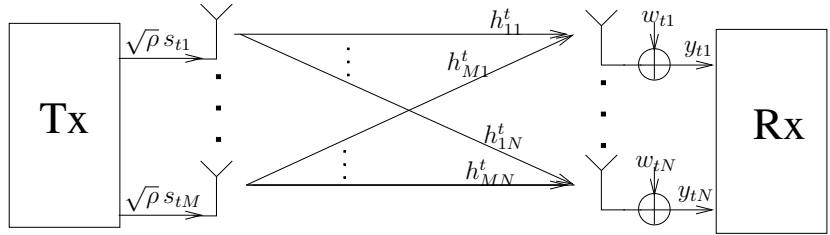}
\caption{MIMO system model}
\label{fig:MIMO_model}
\end{figure}

The transmitted symbols $s_{tm}$ are normalized in a way to make the average
total symbol power transmitted from all $M$ antennas to be unity, i.e.
\begin{equation}
\E[\sum\limits_{m=1}^{M}{|s_{tm}|^2}]=1\:\:\forall \:t
\label{eq:power_st}
\end{equation}
this ensures that the total transmitted power per time slot is constant
regardless the number of transmit antennas $M$. This is necessary for fair
comparison between systems with different $M$. Furthermore the average channel
power is normalized to 1, i.e. $\E[|h_{mn}^t|^2]=1$ $\forall\, m,n,t$.
With the above normalizations of the symbol power, the channel power, and the
noise power spectral density, $\rho\,$ represents the average SNR at each
receive antenna per time slot \cite{Hochwald2000}.\\

Unless mentioned otherwise, the channel model considered in this work assumes
the channel coefficients $h_{mn}^t$ to be constant over one time frame, and varies
independently from one frame to the next. The frame length is on the order of
100 to 300 symbols. That is in the MIMO case, we assume the channel to be Rayleigh block fading, which is also known as quasi-static Rayleigh fading or
piece-wise constant fading.
\section{Receive Diversity}
\label{s:Rx-diversity}

In uplink\footnote{In mobile communications, uplink refers to the transmission
from the mobile station to the base station (mobile is Tx), while downlink
refers to the transmission from the base station to the mobile station (mobile
is Rx).} mobile communications, it is important to reduce the power consumption
of the mobile handsets in order to increase the battery lifetime. At the same
time, the transmission quality should not deteriorate. Using multiple receive
antennas in the base stations of mobile communications systems is a simple and efficient solution to compensate
for the low power transmission of the mobile stations by providing spatial
receive diversity \cite{Tarokh_STTC_98}. This section explains receive diversity
in both coherent and non-coherent systems.\\

Consider the channel model defined in Section \ref{s:MIMO-channel_model}, with
one transmit antenna ($M=1$) and N receive antennas. Assuming a quasi-static Rayleigh fading, (\ref{eq:MIMO_model_scalar}) reduces in the SIMO case to
\begin{equation}
 \left. y_{tn}=\sqrt{\rho\,}s_{t}h_{n}+w_{tn}
 \right. \hspace{1cm}\begin{aligned} &t=0,1,...\\
  &n=1,...,N
\end{aligned}	
\label{eq:SIMO_model_scalar}
\end{equation}
which can be written in a vector form as
\begin{equation}
 \left.\V{y}_t=\sqrt{\rho\,}s_{t}\V{h}+\V{w}_t
 \right. \hspace{1cm}t=0,1,...
\label{eq:SIMO_model_vector}
\end{equation}
where $\V{y}_t$, $\V{h}$, and $\V{w}_t$ are all $\in\mathbb{C}^{1\!\times\!
N}$ representing the vector of the received signals, the channel coefficients
and the noise samples across the $N$ receive antennas, respectively.\\

Consider first the case of coherent detection, where perfect channel knowledge
is assumed to be available at the receiver. At any time slot $t$, the $N$
different received signals $y_{tn}, \,n=1,...,N$, can be combined to make a more
reliable decision on the transmit symbol $s_t$. The optimal way to combine those
signals is by using a technique known as Maximum Ratio Combining (MRC) (also
known as Channel Matched Filter (CMF)). This technique was proved to maximize
the SNR at the MRC output. It performs a weighted sum of all received signals,
with the weights being the complex conjugate of the estimated channel
coefficients. Assuming perfect channel estimation, the combined signal
$\tilde{s_t}$ is 
\begin{eqnarray}
\tilde{s_t}&=&\V{y}_t\V{h}\hr=\sum \limits_{n=1}^{N}{h_{n}^*y_{tn}}=\sum
\limits_{n=1}^{N}{h_{n}^*(\sqrt{\rho\,}s_th_{n}+w_{tn})} \nonumber \\
&=&\sum\limits_{n=1}^{N}{\sqrt{\rho\,}s_t|h_{n}|^2+h_{n}^*w_{tn}}=\sqrt{\rho\,
}s_t\|\V{h}\|^2+\V{w}_t\V{h}\hr.
\label{eq:SIMO_coherent_estimate}
\end{eqnarray}
where $\tilde{s_t}$ serves as an estimate value of $s_t$. In the simple case of
PSK transmission with modulation alphabet $\mathcal{A_{\text{PSK}}}$, the ML
decision metric is based on the Euclidean distance between $\tilde{s_t}$ and the
symbol candidate $s_t^i$ \cite{Alamouti98}
\begin{equation}
\hat{s_t}=\argmin\limits_{s_t^i\in\mathcal{A_{\text{PSK}}}}
\:\,d_{Eu}^2(s_t^i,\tilde{s_t})
\label{eq:ML_metric_SIMO_coherent}
\end{equation}
As seen in (\ref{eq:SIMO_coherent_estimate}), the estimate $\tilde{s_t}$
contains the transmit symbol $s_t$ weighted by the channel power $\|\V{h}\|^2$
which is always positive. Hence the signal contributions from the different
paths superimpose constructively at the receiver. Whereas the noise samples in
$\V{w}_t$ are weighted by the corresponding channel coefficients in $\V{h}$.
Since the channel and the noise processes are independent, they in general add
destructively causing the noise to average out. Consequently, the signal power
is boosted and the noise power is reduced, resulting in a high SNR at the MRC
output.\\

Using similar arguments, receive diversity can be achieved with non-coherent
detection assuming no Channel State Information (CSI) at the Tx nor at the Rx.
To start with, consider the simple case of DPSK transmission, where the transmit
symbol $s_t$ is the product of the new information symbol $v_t$ and the previous
transmit symbol $s_{t-1}$, i.e.
\begin{equation}
s_t=v_ts_{t-1}.
\label{eq:fundamental_Diff_Transmission}	
\end{equation}
Assuming a quasi-static Rayleigh fading channel and using
(\ref{eq:SIMO_model_vector}), the received signals at time slots $t-1$ and $t$
over all receive antennas can be written as
\begin{eqnarray}
\V{y}_{t-1}=\sqrt{\rho\,}s_{t-1}\V{h}+\V{w}_{t-1}
\label{eq:y_t_1}\\
\V{y}_{t}=\sqrt{\rho\,}s_{t}\V{h}+\V{w}_{t-1},
\label{eq:y_t}
\end{eqnarray}
where all vectors are $\in\mathbb{C}^{1\!\times\! N}$. Using
(\ref{eq:fundamental_Diff_Transmission}) in (\ref{eq:y_t}),  and multiplying
(\ref{eq:y_t_1}) by $v_t$, we get
\begin{eqnarray*}
v_t\V{y}_{t-1}&=&\sqrt{\rho\,}v_ts_{t-1}\V{h}+v_t\V{w}_{t-1}\\
\V{y}_{t}&=&\sqrt{\rho\,}v_ts_{t-1}\V{h}+\V{w}_{t}.
\end{eqnarray*}
By subtracting the above two equations, the unknown channel vector $\V{h}$
vanishes and we get
\begin{equation}
v_t\V{y}_{t-1}-\V{y}_{t}= v_t\V{w}_{t-1}-\V{w}_t=-\sqrt{2}\V{w}'
\label{eq:PSK_Diff_Rx}
\end{equation}
which can be rewritten as
\begin{equation}
\V{y}_{t}= v_t\V{y}_{t-1} +\sqrt{2}\V{w}'.
\label{eq:PSK_Diff_Rx_2}
\end{equation}\\
Useful insights can be extracted from the last equation. Comparing
(\ref{eq:PSK_Diff_Rx_2}) with (\ref{eq:SIMO_model_vector}), we see that a
differential non-coherent system can be viewed as a coherent system where the information
symbol $v_t$ is transmitted over the known channel vector $\V{y}_{t-1}$ and
corrupted by an AWGN noise vector $\sqrt{2}\V{w}'$, where
$\V{w}'\V{\sim}\mathcal{CN}(\V{0},\V{I}_N)$. In other words, in differential detection,
the previously received signal vector $\V{y}_{t-1}$ is used as an estimate of
the scaled channel vector $\sqrt{\rho\,}\V{h}$. The factor $\sqrt{2}$ in the
noise term is used to emphasize that the resulting noise has twice as much
variance compared to the noise vectors $\V{w}_t$ or $\V{w}_{t-1}$. One may see
that with non-coherent detection, the noise from two time slots contribute to
the effective noise suffered by the information symbol. This corresponds to the
well-known $\unit[3]{dB}$ performance degradation in SNR as the price payed by
non-coherent detection compared to coherent detection \cite{Hochwald2000}. \\

In order to see how receive diversity can be achieved with non-coherent
detection, the ML decision metric is proved in the following. Using the above
described analogy with the coherent case, an ML decision on the transmit symbol
$v_t$ is the one that has the minimum power of the noise term in
(\ref{eq:PSK_Diff_Rx_2}) resulting in the following metric\footnote{Note that in
general, the ML decision maximizes the conditional PDF of the received signal
given a candidate symbol was transmitted. However the above simpler approach
also results in ML decision in the case of PSK modulated symbols. The
equivalence of both approaches in the constant-envelope case is proved in the discussion of (\ref{eq:x_hat_non_unitary}).}
\begin{equation*}
\hat{v_t}=\argmin\limits_{v^i\in\mathcal{A_{\text{PSK}}}}
\:\,\|v^i\V{y}_{t-1}-\V{y}_t\|_F^2=\argmin\limits_{v^i\in\mathcal{A_{\text{PSK}}}} \:\,
\tr\{(v^i\V{y}_{t-1}-\V{y}_t)(v^i\V{y}_{t-1}-\V{y}_t)\hr\}
\end{equation*}
\begin{eqnarray}
\hat{v_t}&=& \argmin\limits_{v^i\in\mathcal{A_{\text{PSK}}}} \:\,
\tr\{\underbrace{|v^i|^2\V{y}_{t-1}\V{y}_{t-1}\hr}_{\text{constant }\forall\,
v^i}-2\Re\{\tr\{v^i\underbrace{\V{y}_{t-1}\V{y}_t\hr}_{\tilde{v_t}}\}\}+\underbrace{\V
{y}_t\V{y}_t\hr}_{\text{independent of }v^i}\}\nonumber \\
&=& \argmax\limits_{v^i\in\mathcal{A_{\text{PSK}}}} \:\,
\Re\{v^i\tilde{v_t}\}
\label{eq:SIMO_PSK_ML_metric}
\end{eqnarray}
where we have used the fact that all symbols have the same amplitude, making
$|v^i|$ constant $\forall\,i$ and therefore can be omitted from the metric.
Some properties of the trace operator have been used. They are summarized in
Appendix \ref{App:math_formulas}, Section \ref{s:trace_properties}.
$\tilde{v_t}$ serves as an estimate of $v_t$. Using (\ref{eq:y_t_1}) and
(\ref{eq:y_t}), the estimate $\tilde{v_t}$ can be expanded to see the parameters
upon which it is based;
\begin{eqnarray}
\tilde{v_t}&=&\V{y}_{t-1}\V{y}_t\hr=(\sqrt{\rho\,}s_{t-1}\V{h}+\V{w}_{t-1})(\sqrt{\rho\,}v_ts_{t-1}\V{h}+\V{w}_{t
})\hr \nonumber \\
&=& \rho\, v_t^* |s_{t-1}|^2 \|\V{h}\|^2
+\underbrace{\sqrt{\rho\,}s_{t-1}\V{h}\V{w}_t\hr+\sqrt{\rho\,}v_t^*s_{t-1}^*\V{w
}_{t-1}\V{h}\hr+\V{w}_{t-1}\V{w}_t\hr}_{\text{noise terms with uncorrelated
products}} 
\label{eq:v_t_estimate}
\end{eqnarray}
Analogous to (\ref{eq:SIMO_coherent_estimate}), the estimate $\tilde{v_t}$
comprises a constructive addition of weighted copies of the information symbol $v_t$
with the weights being a scaled version of the channel coefficients' power of
the $N$ diversity branches. The rest of the terms in $\tilde{v_t}$ are all noise
terms. Each term is a product of uncorrelated quantities (the noise
$\V{w}_{t-1}$ and $\V{w}_t$ are uncorrelated with the channel $\V{h}$, and the
noise samples are also uncorrelated in the time dimension, i.e. $\V{w}_{t-1}$
and $\V{w}_t$ are uncorrelated). This is to say that the noise averages out.
Accordingly, the SNR at the receiver output is increased and receive diversity
is achieved.\\

Figure \ref{fig:SIMO} show the BER curves of the addressed SIMO diversity scheme with
coherent and non-coherent detection. The channel is assumed to be
quasi-static Rayleigh flat fading and the modulation alphabet considered is QPSK
(Quadrature PSK). In the coherent case, perfect channel knowledge is assumed at
the receiver. As shown in the figures, there is a clear gain achieved with
receive diversity. Note that the \emph{increase} in the gain is reduced as the
number of receive antennas increase. For example using non-coherent detection
and at a BER of $10^{-2}$, a system with two receive antennas achieves about
$\unit[8]{dB}$ gain in SNR compared to the SISO system, whereas the gain obtained
going from a $1\!\times\!2$ system to a $1\!\times\!3$ system is only $\unit[3.5]{dB}$.
 Figure \ref{fig:SIMO_coherent_vs_noncoherent} compares the error performance of coherent versus non-coherent reception
for the $1\!\times\!2$  and the $1\!\times\!3$ systems. As expected a $\unit[3]{dB}$ performance
degradation is incurred in the non-coherent case.

\begin{figure}[!ht]
 \centering
  \subfloat[SIMO with coherent detection]{\label{subfig:SIMO_coherent}\includegraphics[width=0.48\textwidth]{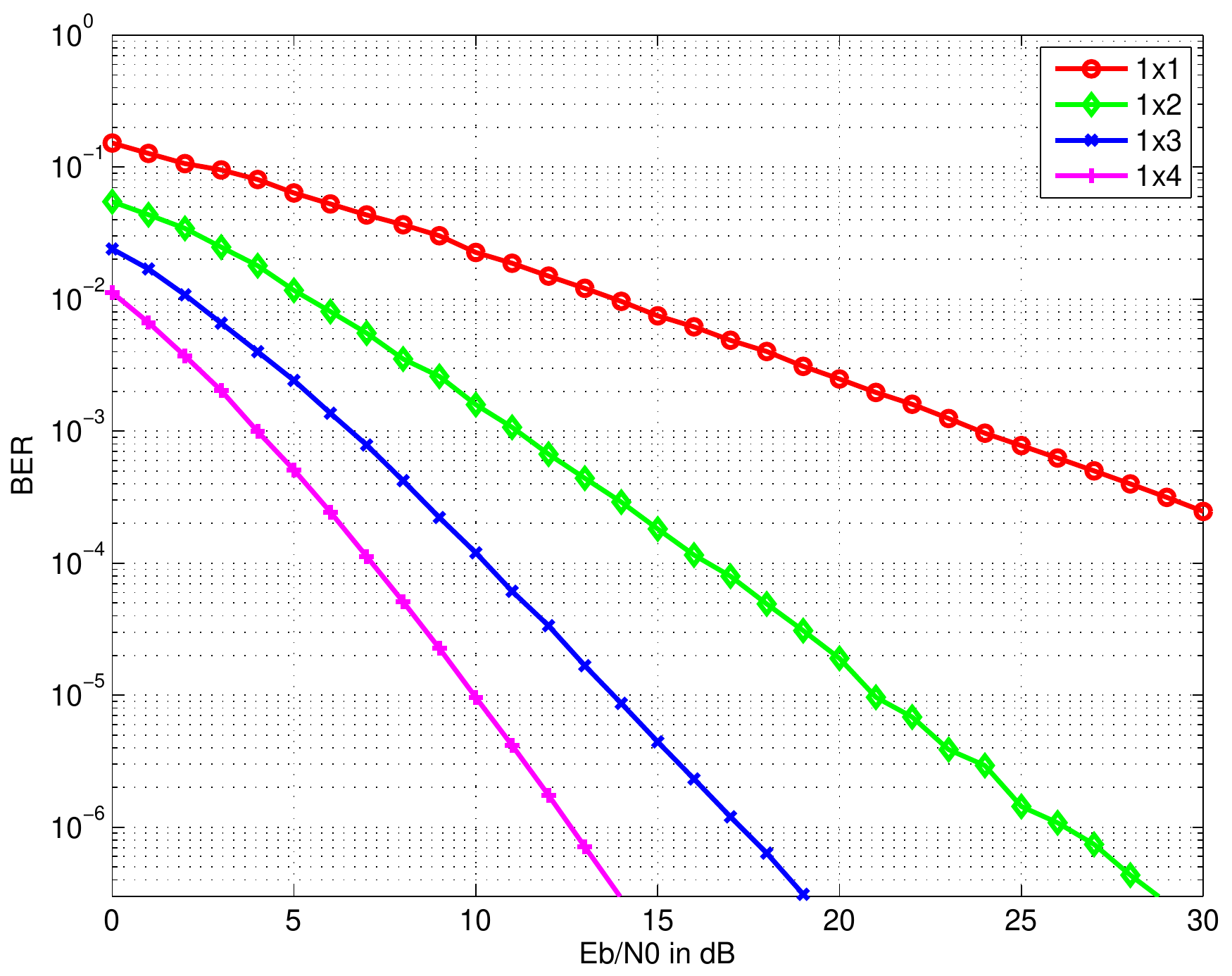}}
  \subfloat[SIMO with non-coherent detection]{\label{subfig:SIMO_non_coherent}\includegraphics[width=0.48\textwidth]{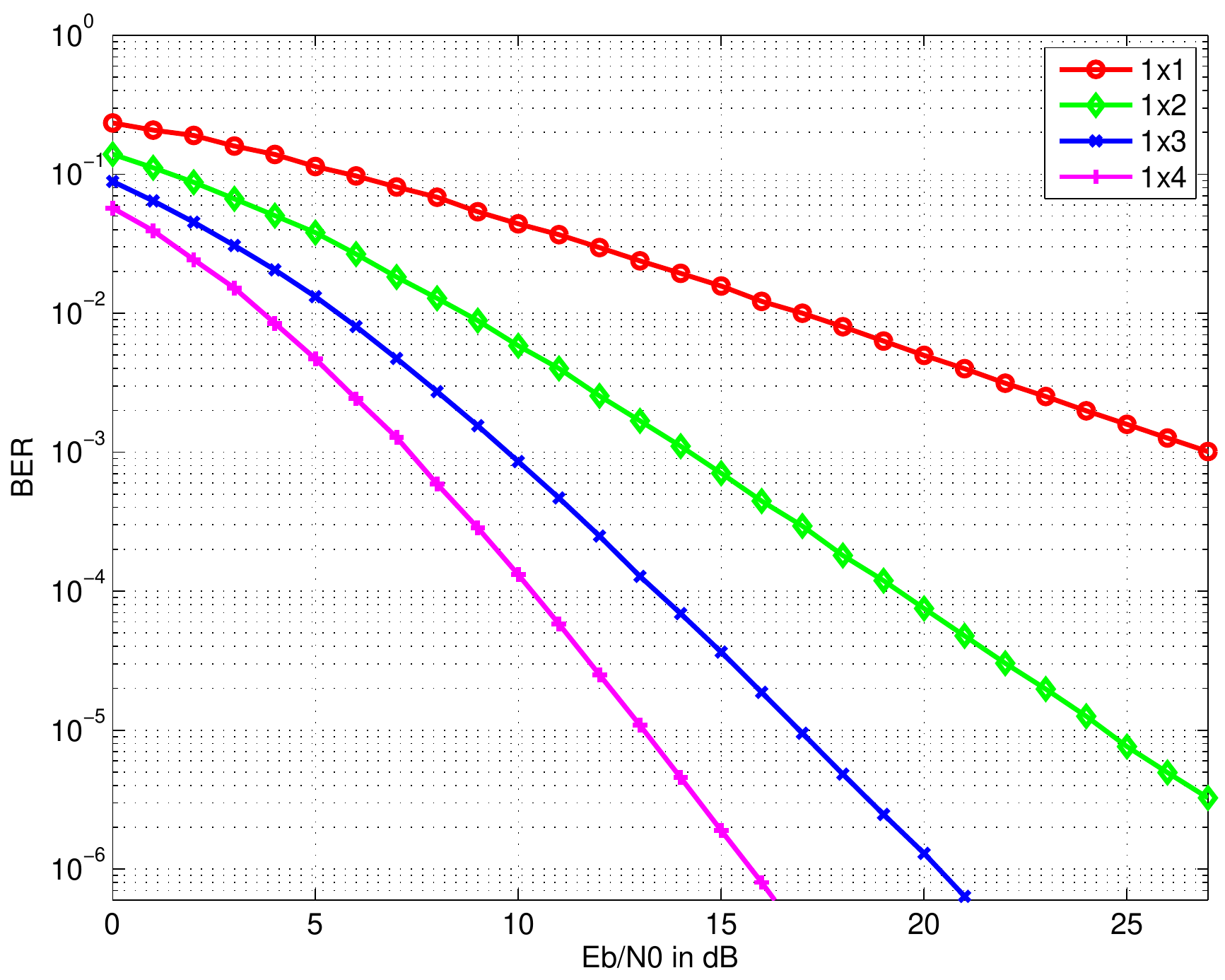}}\\
  \subfloat[Comparing coherent reception vs. differential non-coherent in SIMO systems]{\label{fig:SIMO_coherent_vs_noncoherent}\includegraphics[width=0.48\textwidth]{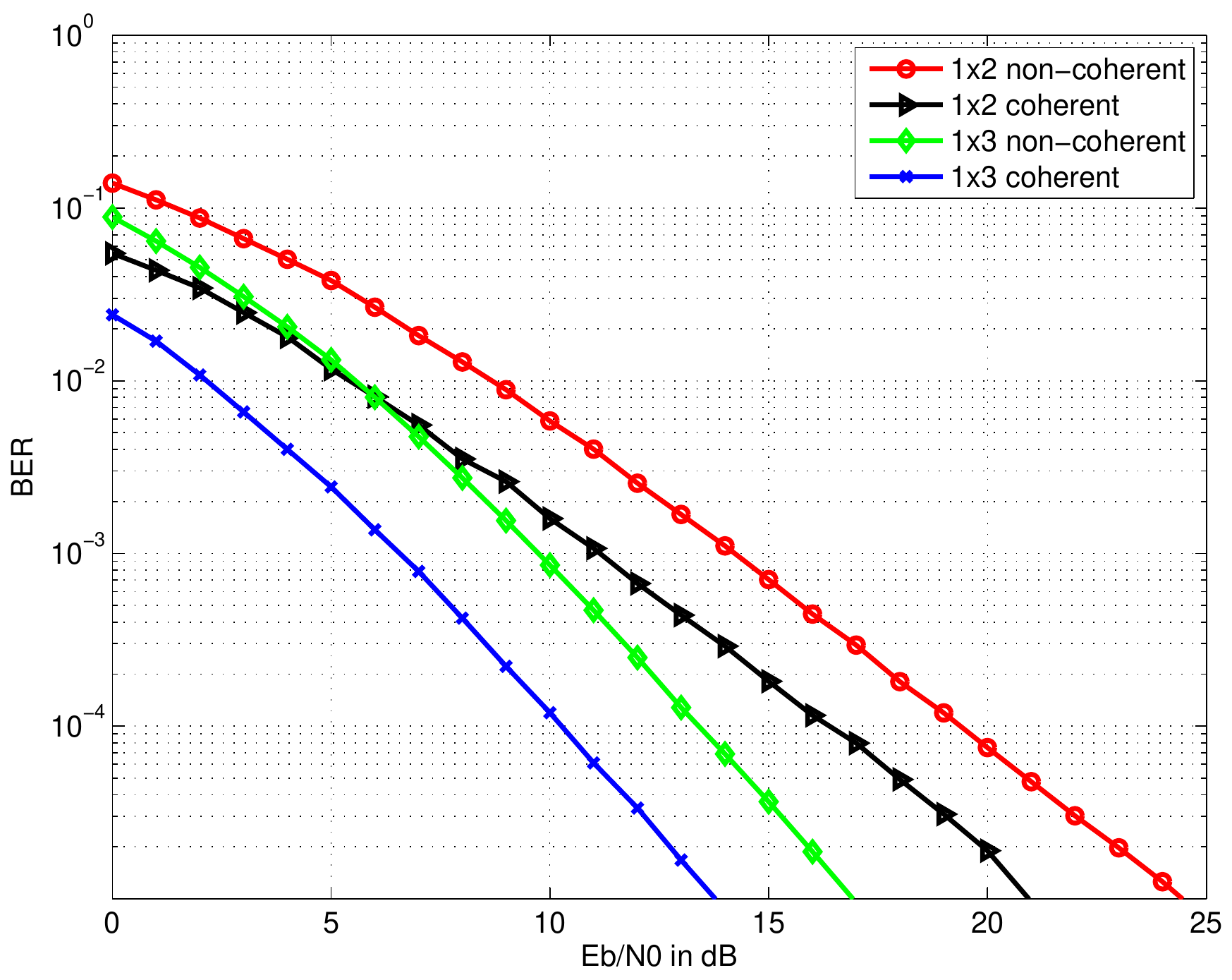}}
  \caption{Effect of receive diversity on the error performance in fading channels using 4-PSK symbols}
  \label{fig:SIMO}
\end{figure}
\section{Transmit Diversity}
\label{s:Tx-diversity}
Although receive diversity looks simple and with low complexity receivers, yet
it can be practically not suitable when the receiver is a mobile station. In
downlink mobile communications, mounting multiple antennas on the mobile handset
results in an increase in size and cost of the mobiles. This opposes the
ongoing desire in making mobiles as low profile and cheap as possible. Even
when using small-size antennas, the spacing between the different antennas
should be sufficiently large to ensure independently fading spatial branches and
therefore ensure diversity. As a rule of thumb, the minimum antenna spacing to
achieve this condition is $\lambda/2$, with $\lambda$ being the center
wavelength of the transmission. In a system operating at a \unit[1]{GHz} center
frequency, $\lambda/2$ translates to \unit[15]{cm} which is about double the
size of the current average-sized mobile sets. Furthermore, since there are
thousands more mobile stations than base stations, it is economically better to
place the complexity in the base stations and make the mobile terminals as
simple as possible \cite{Alamouti98}. All the addressed impracticalities in
using receive diversity in mobile downlink communications motivates the need of
transmit diversity.\\

Consider a MISO system with $M$ transmit antennas following the model described
in Section \ref{s:MIMO-channel_model} with one receive antenna. As an initial
guess of achieving transmit diversity, one might think of transmitting the same
symbol over all transmit antennas simultaneously. In this case the received
signal can be written as 
\begin{equation}
y_t=s_t\sum\limits_{m=1}^{M}{h_{tm}+w_t}
\label{eq:MISO_y_t}
\end{equation}
where $y_t$ is the only signal the receiver can use for detection. Unlike in
receive diversity, --where the receiver posses $N$ independently received
versions of the transmitted symbol $s_t$ and combines them to achieve
diversity-- in the transmit approach in question, it is the role of the channel
to combine the $M$ transmitted replicas of the symbol $s_t$ resulting in $y_t$.
By examining $y_t$ in (\ref{eq:MISO_y_t}), it is seen that the $M$ symbol
replicas are weighted by the corresponding $M$ channel coefficients and not by
the channel power as is the case with the estimate in Rx diversity. This indicates that the
different copies of the same symbol may add destructively, hence diversity is
not ensured. This can be viewed as an intentionally created multipath channel
where the symbol $s_t$ traverses $M$ paths to the receiver, leading to the
possibility of destructive interference.\\

In conclusion, achieving transmit diversity is not as easy as repeating the
transmission of the same symbol over multiple antennas, it however requires some
more complicated structure of the transmitted signal that guarantees diversity.
More formally, the transmitter should preprocess or precode the signal prior
to transmission to provide the receiver with \emph{independent} replicas of the
same transmit symbol \cite{PNG2003}. The previously investigated naive technique
failed to achieve transmit diversity using \emph{only} the spatial dimension.
This indicates that the spatial dimension alone is not sufficient, and another
dimension is needed for transmit diversity to be achieved. Either the frequency
dimension or the time dimension, together with the spatial dimension, can be
used for this purpose leading to schemes known as Space-Frequency Coding (SFC)
and Space-Time Coding (STC), respectively. \\

The choice of which space coding technique to use depends primarily on the
channel conditions. In multipath environments where the channel is frequency
selective, the delay spread of the channel can be exploited to achieve frequency
diversity \cite{PNG2003}. In such channels a common and simple way to counteract
impairments like Inter-Symbol Interference (ISI) is the use of multi-carrier
modulation such as Orthogonal Frequency Division Multiplexing (OFDM) techniques.
If the subcarriers (tones) are separated by more than the channel coherence
bandwidth, then the different bands experience independent fading and therefore
can be used together with multiple antennas to achieve frequency and spatial
diversity using space-frequency codes \cite{PNG2003,Bolcskei_Paulraj_2000}.\\

In the case of narrowband flat fading channels, single carrier systems are used.
In this case, there is no diversity in the frequency domain so the time domain
can be used instead. When multiple antennas are employed, spatial and temporal
diversity can be achieved using space-time codes. This thesis is concerned with
single carrier systems in flat fading channels, hence space-time coding is the
transmit diversity of interest. It is the role of the next chapter to lay the
foundation of space-time codes and present the literature review covered in this
regard.
\abbrev{Tx}{Transmitter}
\abbrev{Rx}{Receiver}
\abbrev{MIMO}{Multiple Input Multiple Output}
\abbrev{SIMO}{Single Input Multiple Output}
\abbrev{MISO}{Multiple Input Single Output}
\abbrev{QoS}{Quality of Service}
\abbrev{DoF}{Degrees of Freedom}
\abbrev{MRC}{Maximum Ratio Combining}
\abbrev{CMF}{Channel Matched Filer}
\abbrev{QPSK}{Quadrature Phase Shift Keying}
\abbrev{w.r.t.}{with respect to}
\abbrev{STC}{Space-Time Coding}
\abbrev{ISI}{Intersymbol Interference}
\abbrev{OFDM}{Orthogonal Frequency Division Multiplexing}

\chapter{Space-Time Coding}
\label{chap:STC}
In Chapter \ref{chap:Diversity_MIMO}, the possibility of achieving transmit diversity has been introduced. Space-time coding is considered as one of the most popular transmit diversity techniques in single carrier systems. This chapter starts with a literature review on space time codes in Section \ref{s:STC_Literature_Review}. The structure of ST codes is then defined in Section \ref{s:Code_Structure}. Next, the general ML decision metric for non-coherent detection of ST codes is derived in Section \ref{s:ML_STC_metric}. The last section derives the design criteria for ST codes in differential non-coherent systems and defines the notion of diversity order and coding gain.
\section{STC Literature Review}
\label{s:STC_Literature_Review}
It all started in 1998 when Tarokh et al. introduced in \cite{Tarokh_STTC_98} a space-time coding technique known as Space-Time Trellis Coding (STTC). Such a technique combines channel coding and transmit diversity to provide significant performance improvement. A major drawback of STTC is that the decoding complexity increases exponentially with the transmission rate \cite{OSTBC_Tarokh99}. In the same year, Alamouti presented a remarkable transmit diversity scheme using two transmit antennas \cite{Alamouti98}. Even though the scheme does not achieve as much gain as that achieved by STTC, its complexity is considerably lower since it only needs simple \emph{linear} processing at the receiver. Alamouti proved that the proposed scheme achieves the same diversity order\footnote{For the definition of diversity order, see Section \ref{s:Design_Criteria}.} as that achieved by MRC reception with one transmit and two receive antennas.\\

Alamouti's scheme was then considered as the foundation of a new class of space time codes named Space-Time Block Codes (STBC). The word \emph{block} is used to indicate that the data stream to be transmitted is encoded in blocks, which contain symbols distributed across space and time. One year later, Tarokh et al. extended Alamouti's scheme to arbitrary number of transmit antennas \cite{OSTBC_Tarokh99}. They defined the class of Orthogonal Space-Time Block Codes (OSTBC) where the data streams transmitted over the different transmit antennas within a block are mutually orthogonal. The work is based on the theory of orthogonal designs, first presented by Hurwitz and independently by Radon in 1922 \cite{Radon1922}. The scheme is remarkable in that it provides the maximum achievable diversity order (full diversity), and allows ML decoding algorithm with only linear processing at the receiver.\\

To eliminate the need of channel estimation and hence reduce the cost and complexity of the receivers, STC has been extended to include non-coherent reception. In the year 2000, Tarokh and Jafarkhani applied differential encoding and non-coherent detection to Alamouti's scheme with equal energy constellation \cite{Tarokh_DOSTBC_2000}. The scheme maintains the full diversity property and the low decoding complexity of Alamouti's scheme without the need of channel knowledge neither at the Tx nor at the Rx. In the same year, Hochwald and Marzetta proposed in \cite{Hochwald_Marzetta_2000} a differential space time technique that is based on unitary matrices. They coined it the term Differential Unitary Space-Time Modulation (DUSTM). They derived performance criteria, error bounds and ML decoder for DUSTM under no CSI, and additionally they make use of channel estimates when they are available. Independently, Hughes proposed in \cite{DSTM_Hughes2000} a similar non-coherent transmission scheme. The approach is valid for any number of transmit antennas. In \cite{Hochwald2000} Hochwald and Sweldens defined an example of DUSTM where any transmit block is a member of a group. This technique simplifies the transmission process and it eventually led to constellations of diagonal signals, where only one transmit antenna is operating at a time. The disadvantage of this group-based signals is the exponential increase in decoding complexity with the transmission rate.\\

As more bits are to be transmitted per channel use, the constellation points of $q$-ary PSK get closer resulting in performance degradation. In this case introducing amplitude modulation is expected to enhance the performance \cite{Proakis2000}. For the purpose of achieving higher spectral efficiency while maintaining good transmission quality, Tao introduced in \cite{Tao_2001} the use of QAM (Quadrature Amplitude Modulation) constellation with Differential OSTBCs (DOSTBCs). He derived the optimal and a near-optimal differential decoder with linear complexity. The approach outperforms DOSTBCs with PSK constellation, since with QAM modulation the constellation symbols are more efficiently separated. A similar approach has been adopted in \cite{Tarokh_QAM_2003} and \cite{Chen_QAM_2003}. In 2005, Bauch and Mengi have attempted in \cite{Bauch_SOD_QAM_2005} to extend the use of QAM symbols in DOSTBCs to include soft-output decoder which allows the scheme to be combined with outer error control coding.\\

To achieve higher transmission rates, several attempts have been made by relaxing the condition of orthogonality in STBCs leading to the so-called Quasi-Orthogonal STBCs (QOSTBC). This new class of STBCs has been proposed by Jafarkhani in \cite{QOSTBC_Jafarkhani_2000}. For most QOSTBCs used in the literature, achieving higher rate comes at the expense of higher decoding complexity. As an attempt to reduce the complexity of QOSTBCs, Yuen et al. proposed a new scheme in \cite{MDC_QOSTBC_2004} which they named Minimum Decoding Complexity QOSTBC (MDC-QOSTBC). 

The literature review presented in this section serves to give an idea of the research covered in STCs thus far. In the next two chapters a detailed description of most of the aforementioned techniques in STCs is provided, and a fair comparison between them will be concluded. 
\section{Code Structure}
\label{s:Code_Structure}
In this section we define the structure of space-time codes as well as the notion of code rate and spectral efficiency. The space-time codes considered in this thesis are of block type. Each block consists of signals transmitted over space ($M$ antennas) and time ($T$ time slots). A transmit block $\V{S}_{\tau}$ can be described in a matrix form as
\begin{figure}[!h]
\centering
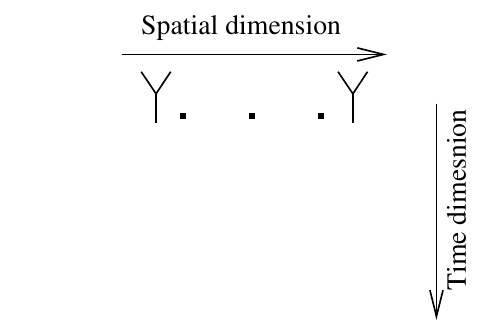
\caption{ST transmit matrix structure}
\label{fig:STC_Tx_matrix}
\end{figure}

\hspace{-0.65cm}where $\tau$ is the transmit block index constituting $T$ time slots, and $s_{tm}$ is the signal transmitted in time slot $t$ over antenna $m$. This indicates that the horizontal dimension of matrix $\V{S}_\tau$ is the spatial dimension and the vertical dimension is the temporal (time) dimension. The $m^{\text{th}}$ column of matrix $\V{S}_\tau$ represents the signal transmitted over antenna $m$ as a function of time, while the $t^{\text{th}}$ row represents the transmitted signals over all $M$ antennas at time slot $t$.\\

In differential non-coherent systems, matrix $\V{S}_\tau$ is generated from a differential encoder whose inputs are in general the previously transmitted matrix $\V{S}_{\tau-1}$ or some function of it and the new information matrix $\V{V}_{z_\tau}$. Information matrices have the same \emph{structure} as transmit matrices in Figure \ref{fig:STC_Tx_matrix}. Any information matrix $\V{V}_{z_\tau}$ is an element in a space-time codebook $\Omega$ and is referred to as space-time codeword. The codebook $\Omega$ has cardinality $L$ and is defined as 

\begin{equation}
\label{eq:STC_book}
\Omega=\{\V{V}_l, l=0,...,L-1 ; \,\V{V}_l\in\mathbb{C}^{T\!\times\!M}\}.
\end{equation}

Figure \ref{fig:STC} shows two possibilities for constructing differential space-time codes. In Figure \ref{subfig:DSTM}, at block time index $\tau$, $\log_2L$ bits are buffered from the user's data, converted to their decimal equivalent $z_\tau$ and mapped by the space-time modulator to the information matrix $\V{V}_{z_\tau}$, where $\V{V}_{z_\tau}$ is a codeword in the space-time codebook $\Omega$, i.e. $z_\tau\in\{0,...,L-1\}$. $\V{V}_{z_\tau}$ together with the previously transmitted matrix $\V{S}_{\tau-1}$ are then fed to the differential encoder to generate the new transmit matrix $\V{S}_\tau$. Finally, $\V{S}_\tau$ is fed to the transmitter which transmits its inner symbols across the $M$ transmit antennas over $T$ time slots. This process repeats every $T$ time slots. The receiver may optionally employ $N$ receive antennas to additionally achieve receive diversity. Codes constructed from such an architecture will be referred to as Space-Time Modulation (STM) as they were first defined in \cite{Hochwald_Marzetta_2000}.
\begin{figure}[htp]
 \centering
  \subfloat[Differential Space-Time Modulation architecture]{\label{subfig:DSTM}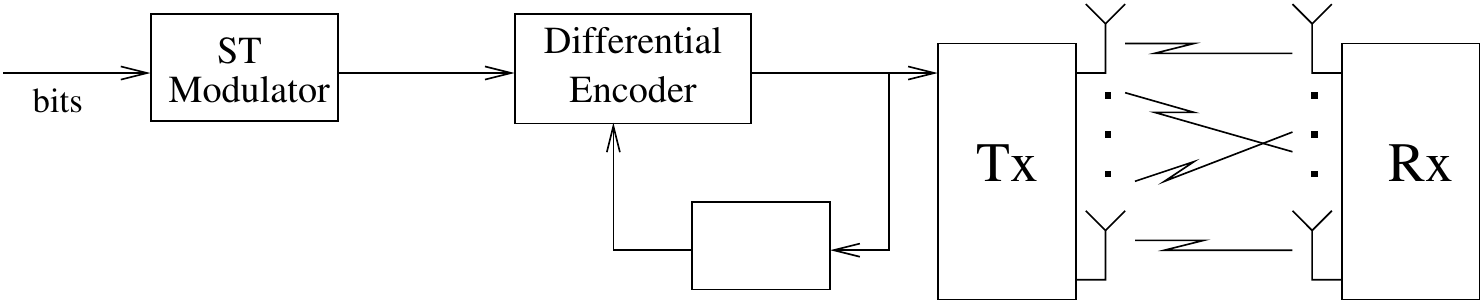}\\\vspace{1cm}
  \subfloat[Differential Space-Time Block Codes architecture]{\label{subfig:DSTBC}\hspace{-1cm}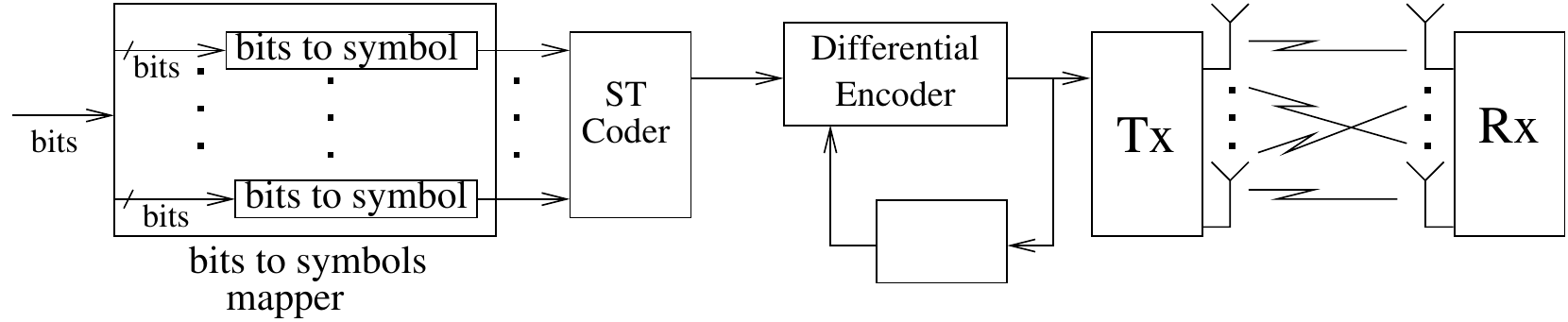}
  \caption{Space-Time Coding architectures}
  \label{fig:STC}
\end{figure}\\

The second architecture of space-time codes is shown in Figure \ref{subfig:DSTBC}. Unlike STM, the construction of the information codeword $\V{V}_{z_\tau}$ is performed in two steps. First $\log_2L$ bits are buffered from the user's data and mapped to $K$ symbols  $x_1,...,x_K$. The symbols are in general drawn from different constellation alphabets of different size. Namely, $x_i$ is element in the constellation alphabet $\mathcal{A}_i$ whose size is $q_i$ $\forall i\in {1,...,K}$. The $K$ symbols are then fed to the space-time encoder which constructs the information matrix $\V{V}_{z_\tau}$, where $\V{V}_{z_\tau}$ is element from a ST codebook $\Omega$, i.e. $z_\tau\in\{0,...,L-1\}$. Hence it follows that the total number of bits per information matrix is 
\begin{equation}
\label{eq:bits_per_block_general}
\log_2L=\sum\limits_{i=1}^{K}{\log_2q_i},
\end{equation}
where $\log_2q_i$ is the number of bits carried by symbol $x_i$. In the special case when all symbols are drawn from the same alphabet $\mathcal{A}$ of size $q$, the total number of bits per information matrix is 
\begin{equation}
\label{eq:bits_per_block_special}
\log_2L=K\log_2q.
\end{equation}

After $\V{V}_{z_\tau}$ is generated, both architectures are just the same. They both generate the differentially modulated matrix $\V{S}_\tau$ through the differential encoder, and then the entries of matrix $\V{S}_\tau$ are transmitted by the $M$ antennas over $T$ time slots. Codes constructed form the second architecture will be referred to as Space-Time Block Codes\footnote{Note that both architectures defined here are block codes since the transmission in both cases is done block-wise. However the different terms STM and STBC are used for shortly referring to each architecture.} (STBC) as they are first defined in \cite{OSTBC_Tarokh99}. In short, if the bits are mapped directly to transmit matrices, then the scheme used is STM shown in Figure \ref{subfig:DSTM}, but if the bits are mapped first to symbols then to transmit matrices, then the scheme is considered as STBC shown in Figure \ref{subfig:DSTBC}.\\

For STBC codes, we define the code rate $R$ as the number of symbols transmitted per time slot\footnote{The code rate defined here should not be confused with the code rate defined in the context of channel coding. In this thesis, channel coding is not used and all transmission is done uncoded.}, i.e.
\begin{equation}
\label{eq:code_rate}
R=\frac{K}{T} \:\text{  symbols/time slot},
\end{equation} 
and for both STM and STBC codes, we define spectral efficiency $\eta$ as the number of bits transmitted per time slot, i.e.
\begin{equation}
\label{eq:spectral_efficiency_STM}
\eta=\frac{\log_2L}{T} \:\text{  bits/time slot},
\end{equation}
where $\log_2L$ is the number of bits transmitted per information matrix. Spectral efficiency is the term used to refer to the information rate that can be transmitted over a given bandwidth. It is therefore a measure of how efficiently a limited frequency band is utilized \cite{Digital_Wireless_comm1999}. Since the term time slot is essentially the same as channel use, the unit of spectral efficiency $\eta$ can be bits/channel use. Furthermore, since the spectral efficiency is the rate per unit bandwidth, another common unit is bits/s/Hz. All three units are used interchangeably. Now substituting (\ref{eq:bits_per_block_general}) in (\ref{eq:spectral_efficiency_STM}), the spectral efficiency in the case of STBCs is 
\begin{equation}
\label{eq:spectral_efficiency_STBC_general}
\eta\big|_{\text{STBC}}=\frac{\sum\limits_{i=1}^{K}{\log_2q_i}}{T} \text{  bits/time slot}.
\end{equation}
In the special case when all $K$ symbols are drawn from the same alphabet $\mathcal{A}$, the spectral efficiency of STBCs reduces to
\begin{equation}
\label{eq:spectral_efficiency_STBC_special}
\eta\big|_{\text{STBC}}=\frac{K\log_2q}{T} =R\log_2q\:\text{  bits/time slot}.
\end{equation}
\section{ML Decision Metric for ST Codes}
\label{s:ML_STC_metric}
After defining the code structure in the previous section, it is now possible to proceed to the non-coherent detection of space-time codes. This section provides a derivation of the ML decoding metric which is valid for STCs of any of the two architectures defined in the previous section in the piece-wise constant Rayleigh flat fading channel.\\

Using the $M\!\times\!N$ MIMO model defined in Section \ref{s:MIMO-channel_model}, the signal received at time slot $t$ by antenna $n$ is -as defined in (\ref{eq:MIMO_model_scalar})-
\begin{equation}
 \left. y_{tn}=\sqrt{\rho\,}\sum \limits_{m=1}^{M}{s_{tm}h_{mn}^t+w_{tn}},
 \right. \hspace{1cm}\begin{aligned} &t=0,1,...\\
 										   &n=1,...,N
 				\end{aligned}	
\label{eq:MIMO_model_scalar2}					   
\end{equation}
Now after defining the structure of the ST transmit signals in Figure \ref{fig:STC_Tx_matrix}, the received signals over the $N$ receive antennas during $T$ time slots can be written in a matrix form as
\begin{equation}
\label{eq:Y_MIMO}
\begin{bmatrix} 
y_{11} & \cdots & y_{1N}\\
\vdots & \ddots & \vdots \\
y_{T1} & \cdots & y_{TN}
\end{bmatrix}
=\sqrt{\rho}\begin{bmatrix} 
s_{11} & \cdots & s_{1M}\\
\vdots & \ddots & \vdots \\
s_{T1} & \cdots & s_{TM}
\end{bmatrix}
\begin{bmatrix} 
h_{11} & \cdots & h_{1N}\\
\vdots & \ddots & \vdots \\
h_{M1} & \cdots & h_{MN}
\end{bmatrix} +
\begin{bmatrix} 
w_{11} & \cdots & w_{1N}\\
\vdots & \ddots & \vdots \\
w_{T1} & \cdots & w_{TN}
\end{bmatrix}
\end{equation}
which can be written in a compact form as
\begin{equation}
\label{eq:Y_MIMO_STC}
\V{Y}_\tau=\sqrt{\rho}\V{S}_\tau\V{H}+\V{W}_\tau
\end{equation}
where $\tau$ is the time index of the transmit block constituting $T$ time slots. $\V{Y}_\tau$ is the $T\!\times\!N$ received matrix. $\V{S}_\tau$ is the $T\!\times\!M$ differentially encoded transmit matrix whose total power at any time slot $t$ is normalized to be one as mentioned in (\ref{eq:power_st}). $\V{H}$ is the $M\!\times\!N$ channel matrix whose entries $h_{mn}$ are all independent and identically distributed (i.i.d.) complex Gaussian random variables with zero mean and unit variance. All channel coefficients $h_{mn}$ are assumed to be constant within the $T$ time slots in a transmission block. Furthermore, assuming piece-wise constant Rayleigh flat fading channel, the coefficients $h_{mn}$ are constant within several successive blocks like for example 30$\sim$60 blocks. $\V{W}_\tau$ is the $T\!\times\!N$ noise matrix whose entries are all i.i.d. complex Gaussian samples with zero mean and unit variance. With the power normalizations of $\V{S}_\tau$, $\V{H}_\tau$ and $\V{W}_\tau$, $\rho$ is the expected SNR at each receive antenna. \\

Since the transmission is done differentially, this means that the information matrix $\V{V}_{z_\tau}$ is embedded in both transmit matrices $\V{S}_{\tau-1}$ and $\V{S}_\tau$ through some differential encoding function defined based on the ST coding scheme used. Hence both received signal matrices $\V{Y}_{\tau-1}$ and $\V{Y}_\tau$ are needed for the differential detection. Similar to (\ref{eq:Y_MIMO_STC}), the previously received matrix $\V{Y}_{\tau-1}$ is
\begin{equation}
\label{eq:Y_1_MIMO_STC}
\V{Y}_{\tau-1}=\sqrt{\rho}\V{S}_{\tau-1}\V{H}+\V{W}_{\tau-1}.
\end{equation}
Combining both (\ref{eq:Y_1_MIMO_STC}) and (\ref{eq:Y_MIMO_STC}) in one matrix, we define matrix $\bar{\V{Y}}_\tau$ as
\begin{eqnarray}
\label{eq:Y_bar}
\V{\bar{Y}}_\tau&=&
\begin{bmatrix} 
\V{Y}_{\tau-1}\\
\V{Y}_\tau\\
\end{bmatrix} =\sqrt{\rho}
\begin{bmatrix} 
\V{S}_{\tau-1}\\
\V{S}_\tau\\
\end{bmatrix}\V{H}+
\begin{bmatrix} 
\V{W}_{\tau-1}\\
\V{W}_\tau\\
\end{bmatrix}\nonumber \\
\V{\bar{Y}}_\tau&=&\sqrt{\rho}\V{\bar{S}}_\tau\V{H}+\V{\bar{W}}_\tau,
\end{eqnarray}
The bar sign is used to indicate the appending of two successive blocks making both $\V{\bar{Y}}_\tau$, $\V{\bar{W}}_\tau \in\mathbb{C}^{2T\!\times\!N}$ and $\V{\bar{S}}_\tau\,\in\mathbb{C}^{2T\!\times\!M}$.\\

To obtain a Maximum A Posteriori (MAP) decision on the information matrix $\V{V}_{z_\tau}$, one needs to maximize the conditional PDF of transmitting $\V{\bar{S}}^{(l)}$ given that $\V{\bar{Y}}_\tau$ is observed\footnote{We omit the subscript $\tau$ from $\bar{\V{S}}$ when considering a candidate transmit block since it doesn't dependent on time. For $\bar{\V{S}}$, the subscript $\tau$ is only meaningful in a transmission equation as in (\ref{eq:Y_bar}).}, namely $\text{p}(\V{\bar{S}}^{(l)}|\V{\bar{Y}}_\tau)$, where  $\V{\bar{S}}^{(l)}$ is function of  $\V{V}_l$ which is the $l^{\text{th}}$ possible information matrix drawn from the codebook $\Omega$ defined in (\ref{eq:STC_book}). Hence the MAP decision rule is
\begin{equation}
\label{eq:MAP_STC}
\hat{\V{V}}_{z_\tau}\Big|_{\text{MAP}}=\argmax\limits_{l\in\{0,..,L-1\}}\:\,\text{p}(\V{\bar{S}}^{(l)}|\V{\bar{Y}}_\tau).
\end{equation}
Using Bayes' rule, (\ref{eq:MAP_STC}) can be written as 
\begin{equation}
\label{eq:STC_Bayes}
\hat{\V{V}}_{z_\tau}\Big|_{\text{MAP}}=\argmax\limits_{l\in\{0,..,L-1\}}\:\,\frac{\text{p}(\V{\bar{Y}}_\tau|\V{\bar{S}}^{(l)})\text{p}(\V{\bar{S}}^{(l)})}{\text{p}(\V{\bar{Y}}_\tau)}.
\end{equation}\\

If all $L$ possible information matrices $\V{V}_l$ are equiprobable, then all transmit matrices $\V{\bar{S}}^{(l)}$ are also equiprobable. In this case $\text{p}(\V{\bar{S}}^{(l)})$ is constant $\forall\, l$. Furthermore, $\text{p}(\V{\bar{Y}}_\tau)$ is independent of $l$. Hence, maximizing $\text{p}(\V{\bar{S}}^{(l)}|\V{\bar{Y}}_\tau)$ is in this case equivalent to maximizing $\text{p}(\V{\bar{Y}}_\tau|\V{\bar{S}}^{(l)})$ leading to the ML decision metric
\begin{equation}
\label{eq:ML_STC}
\hat{\V{V}}_{z_\tau}\Big|_{\text{ML}}=\argmax\limits_{l\in\{0,..,L-1\}}\:\,\text{p}(\V{\bar{Y}}_\tau|\V{\bar{S}}^{(l)}).
\end{equation}
Since the natural logarithm is a monotonically increasing function of its argument, then the ML decision metric can be equivalently written as
\begin{equation}
\label{eq:ML_STC2}
\hat{\V{V}}_{z_\tau}\Big|_{\text{ML}}=\argmax\limits_{l\in\{0,..,L-1\}}\:\,\ln\{\text{p}(\V{\bar{Y}}_\tau|\V{\bar{S}}^{(l)})\}.
\end{equation}

For a given transmit matrix, the received signals over the $N$ receive antennas are all independent from each other due to the independency of the fading of all links, so it follows\footnote{Outside the decision metrics, $(l)$ is omitted from $\V{\bar{S}}$ for simplicity}
\begin{equation}
\label{eq:P_Y_p_y}
\text{p}(\V{\bar{Y}}_\tau|\V{\bar{S}})=\prod_{n=1}^{N}{\text{p}(\V{\bar{y}}_{\tau n}|\V{\bar{S}})}
\end{equation}
where $\V{\bar{y}}_{\tau n}$ is the signal vector that carries information about $\V{V}_{z_\tau}$ and is received by the $n^{\text{th}}$ receive antenna in $2T$ time slots. Similar to (\ref{eq:Y_bar}), $\V{\bar{y}}_{\tau n}$ can be written as
\begin{equation}
\label{eq:y_n_bar}
\V{\bar{y}}_{\tau n}=\sqrt{\rho}\V{\bar{S}}_\tau\V{h}_n+\V{\bar{w}}_{\tau n}
\end{equation}
where $\V{\bar{y}}_{\tau n}$ is the $n^{\text{th}}$ column of $\V{\bar{Y}}_\tau$ and therefore $\in\mathbb{C}^{2T\!\times\!1}$, the transmit matrix $\V{\bar{S}}_\tau\,\!\in\!\mathbb{C}^{2T\!\times\!M}$, the channel vector $\V{h}_n\!\in\!\mathbb{C}^{M\!\times\!1}$ and represents the $n^{\text{th}}$ column of $\V{H}$, and the noise vector $\V{\bar{w}}_{\tau n} \!\in\!\mathbb{C}^{2T\!\times\!1}$ and represents the $n^{\text{th}}$ column of $\V{\bar{W}}_\tau$.\\

Since both $\V{h}_n$ and $\V{\bar{w}}_{\tau n}$ are complex Gaussian distributed, then for a given transmit matrix $\V{\bar{S}}$, the received vector $\V{\bar{y}}_{\tau n}$ follows the same distribution, namely
\begin{equation}
\label{eq:y_S}
\V{\bar{y}}_{\tau n}|\V{\bar{S}}\thicksim\mathcal{CN}(\V{\mu}_{\V{\bar{y}}_{\tau n}|\V{\bar{S}}},\V{\Lambda}_{\V{\bar{y}}_{\tau n}|\V{\bar{S}}}),
\end{equation}
where $\V{\mu}_{\V{\bar{y}}_{\tau n}|\V{\bar{S}}}$ and $\V{\Lambda}_{\V{\bar{y}}_{\tau n}|\V{\bar{S}}}$ represent the mean vector and the covariance matrix of $\V{\bar{y}}_{\tau n}|\V{\bar{S}}$, respectively. The multivariate conditional PDF of receiving vector $\V{\bar{y}}_{\tau n}$ given that $\V{\bar{S}}$ has been transmitted is given by \cite{Steven_Kay93}
\begin{equation}
\label{eq:pdf_y_S}
\text{p}(\V{\bar{y}}_{\tau n}|\V{\bar{S}})=\frac{1}{\pi^{2T}\det(\V{\Lambda}_{\V{\bar{y}}_{\tau n}|\V{\bar{S}}})}e^{-(\V{\bar{y}}_{\tau n}-\V{\mu}_{\V{\bar{y}}_{\tau n}|\V{\bar{S}}})\hr\V{\Lambda}_{\V{\bar{y}}_{\tau n}|\V{\bar{S}}}^{-1}(\V{\bar{y}}_{\tau n}-\V{\mu}_{\V{\bar{y}}_{\tau n}|\V{\bar{S}}})}.
\end{equation}\\

To obtain an expression for $\V{\mu}_{\V{\bar{y}}_{\tau n}|\V{\bar{S}}}$ and $\V{\Lambda}_{\V{\bar{y}}_{\tau n}|\V{\bar{S}}}$ in terms of the transmit matrix $\V{\bar{S}}$, consider first the statistics of the channel vector $\V{h}_n$. The channel coefficients are all independent complex Gaussian distributed with zero mean and unit variance, hence
\begin{eqnarray}
\label{eq:h_statistics}
&\E[h_{mn}]=0 & \hspace{1cm}\text{(zero mean)}\nonumber\\
&\E[|h_{mn}|^2]=1 &  \hspace{1cm} \text{(unit variance)}\\
&\E[h_{m_1n}h_{m_2n}^*]=0&  \hspace{1cm}
\begin{aligned} &m_1\neq m_2\in\{1,...,M\}\\
 										   &n\in\{1,...,N\}\nonumber
 				\end{aligned}	
\end{eqnarray}

Consequently, the channel vector $\V{h}_n$ has a mean vector $\V{\mu}_{\V{h}_n}=\V{0}_{M\!\times\!1}$, and a covariance matrix
\begin{eqnarray}
\label{eq:cov_mat_h}
\V{\Lambda}_{\V{h}_n}&=&\E[(\V{h}_n-\V{\mu}_{\V{h}_n})(\V{h}_n-\V{\mu}_{\V{h}_n})\hr]=\E[\V{h}_n\V{h}_n\hr]\nonumber\\
&=&\E[\begin{bmatrix}
h_{1n}\\
\vdots\\
h_{Mn}
\end{bmatrix}
\begin{bmatrix}
h_{1n}^*\cdots h_{Mn}^*
\end{bmatrix}]=
\begin{bmatrix}
\E[|h_{1n}|^2] & & \text{\huge{0}} \\
& \ddots &  \\
\:\text{\huge{0}}& &\E[|h_{Mn}|^2]
\end{bmatrix}=\V{I}_{M}
\end{eqnarray}
Consequently,
\begin{equation}
\label{eq:h_n_distribution}
\V{h}_n\thicksim\mathcal{CN}(\V{0}_{M\!\times\!1},\V{I}_{M}).
\end{equation}\\

Accordingly, using (\ref{eq:y_n_bar}), $\V{\bar{y}}_{\tau n}|\V{\bar{S}}$ has the following statistical parameters
\begin{eqnarray}
\label{eq:mean_y_S}
\V{\mu}_{\V{\bar{y}}_{\tau n}|\V{\bar{S}}}&=&\sqrt{\rho}\V{\bar{S}}\underbrace{\E[\V{h}_n]}_{\V{0}_{M\!\times\!1}}+\underbrace{\E[\V{\bar{w}}_{\tau n}]}_{\V{0}_{2T\!\times\!1}}=\V{0}_{2T\!\times\!1}\\
\label{eq:cov_y_S}
\V{\Lambda}_{\V{\bar{y}}_{\tau n}|\V{\bar{S}}}&=&\E[(\V{\bar{y}}_{\tau n}-\V{\mu}_{\V{\bar{y}}_{\tau n}|\V{\bar{S}}})(\V{\bar{y}}_{\tau n}-\V{\mu}_{\V{\bar{y}}_{\tau n}|\V{\bar{S}}})\hr]=\E[\V{\bar{y}}_{\tau n}\V{\bar{y}}_{\tau n}\hr]\nonumber\\
&=&\E[(\sqrt{\rho}\V{\bar{S}}\V{h}_n+\V{\bar{w}}_{\tau n})(\sqrt{\rho}\V{\bar{S}}\V{h}_n+\V{\bar{w}}_{\tau n})\hr]\nonumber \\
&=&\rho\V{\bar{S}} \cdot \underbrace{\E[\V{h}_n\V{h}_n\hr]}_{=\V{I}_M\text{ using (\ref{eq:cov_mat_h})}}\cdot\V{\bar{S}}\hr + \sqrt{\rho}\V{\bar{S}} \E[\underbrace{\V{h}_n\V{\bar{w}}_{\tau n}\hr}_{\text{independent}}]+\sqrt{\rho}\E[\underbrace{\V{\bar{w}}_{\tau n}\V{h}_n\hr}_{\text{independent}}]\V{\bar{S}}\hr+\underbrace{\E[\V{\bar{w}}_{\tau n}\V{\bar{w}}_{\tau n}\hr]}_{=\V{I}_{2T}}\nonumber\\
&=&\rho\V{\bar{S}} \V{\bar{S}}\hr +\V{I}_{2T}=\V{\Lambda}_{\V{\bar{y}}_{\tau}|\V{\bar{S}}},
\end{eqnarray}
which are the same $\forall\,n$, therefore the received signals at all receive antennas expectedly follow the same distribution, namely
\begin{equation}
\label{eq:y_S_statistics}
\V{\bar{y}}_{\tau n}|\V{\bar{S}}\thicksim\mathcal{CN}(\V{0}_{2T\!\times\!1},\,\rho\V{\bar{S}} \V{\bar{S}}\hr +\V{I}_{2T}).
\end{equation}
Substituting (\ref{eq:mean_y_S}) in (\ref{eq:pdf_y_S}), then $\text{p}(\V{\bar{Y}}_\tau|\V{\bar{S}})$ in (\ref{eq:P_Y_p_y}) reduces to
\begin{equation}
\label{eq:P_Y_S}
\text{p}(\V{\bar{Y}}_\tau|\V{\bar{S}})=\frac{1}{\pi^{2TN}\det^N(\V{\Lambda}_{\V{\bar{y}}_{\tau}|\V{\bar{S}}})}e^{-\sum\limits_{n=1}^{N}\V{\bar{y}}_{\tau n}\hr\V{\Lambda}_{\V{\bar{y}}_{\tau}|\V{\bar{S}}}^{-1}\V{\bar{y}}_{\tau n}},
\end{equation}
where the magnitude of the exponent can be rewritten as 
\begin{eqnarray}
\label{eq:exponent_in P_Y_S}
\sum\limits_{n=1}^{N}\V{\bar{y}}_{\tau n}\hr\V{\Lambda}_{\V{\bar{y}}_{\tau}|\V{\bar{S}}}^{-1}\V{\bar{y}}_{\tau n}&=&\sum\limits_{n=1}^{N}
{\left[
y_{1,n}^* \cdots y_{2T,n}^*
\right]
\left[
\V{\Lambda}_{\V{\bar{y}}_{\tau}|\V{\bar{S}}}
\right]^{-1}
\begin{bmatrix}
y_{1,n}\\ \vdots \\y_{2T,n}
\end{bmatrix}}\nonumber\\
&=&\tr\{\begin{bmatrix}
y_{1,1}^*&\cdots&y_{2T,1}^*\\
\vdots&\ddots&\vdots\\
y_{1,N}^*&\cdots&y_{2T,N}^*
\end{bmatrix}
\left[
\V{\Lambda}_{\V{\bar{y}}_{\tau }|\V{\bar{S}}}
\right]^{-1}
\begin{bmatrix}
y_{1,1}&\cdots&y_{1,N}\\
\vdots&\ddots&\vdots\\
y_{2T,1}&\cdots&y_{2T,N}
\end{bmatrix}
\}\nonumber\\
&=&\tr\left\{\V{\bar{Y}}_\tau\hr\V{\Lambda}_{\V{\bar{y}}_{\tau}|\V{\bar{S}}}^{-1}\V{\bar{Y}}_\tau\right\}
\end{eqnarray}
Substituting (\ref{eq:exponent_in P_Y_S}) into (\ref{eq:P_Y_S}), we get
\begin{equation}
\label{eq:P_Y_S_general}\text{p}(\V{\bar{Y}}_\tau|\V{\bar{S}})=\frac{1}{\pi^{2TN}\det^N(\V{\Lambda}_{\V{\bar{y}}_{\tau}|\V{\bar{S}}} )}e^{-\tr\{\V{\bar{Y}}_\tau\hr\V{\Lambda}_{\V{\bar{y}}_{\tau}|\V{\bar{S}}}^{-1}\V{\bar{Y}}_\tau\}}
\end{equation}
If we then replace $\V{\Lambda}_{\V{\bar{y}}_{\tau}|\V{\bar{S}}}$ by (\ref{eq:cov_y_S}), $\text{p}(\V{\bar{Y}}_\tau|\V{\bar{S}})$ becomes
\begin{equation}
 \label{eq:P_Y_S_general_useful}\text{p}(\V{\bar{Y}}_\tau|\V{\bar{S}})=\frac{1}{\pi^{2TN}\det^N(\V{I}_{2T}+\rho\V{\bar{S}} \V{\bar{S}}\hr)}e^{-\tr\{\V{\bar{Y}}_\tau\hr(\V{I}_{2T}+\rho\V{\bar{S}} \V{\bar{S}}\hr)^{-1}\V{\bar{Y}}_\tau\}}
\end{equation}
A very interesting observation in (\ref{eq:P_Y_S_general_useful}) is that $\text{p}(\V{\bar{Y}}_\tau|\V{\bar{S}})$ does not change if $\V{\bar{S}}$ is right multiplied by any arbitrary $M\!\times\!M$ unitary matrix $\V{\Phi}$. This is because 
\begin{equation*}
(\V{\bar{S}}\V{\Phi})(\V{\bar{S}}\V{\Phi})\hr\eq \V{\bar{S}}\V{\Phi}\V{\Phi}\hr\V{\bar{S}}\hr\eq \V{\bar{S}}\V{\bar{S}}\hr.
\end{equation*}
In other words, $\V{\bar{S}}$ and $\V{\bar{S}}\V{\Phi}$ are indistinguishable to a non-coherent receiver.\\

Using (\ref{eq:cov_y_S}), $\det(\V{\Lambda}_{\V{\bar{y}}_{\tau}|\V{\bar{S}}})$ and $\V{\Lambda}_{\V{\bar{y}}_{\tau}|\V{\bar{S}}}^{-1}$ can be simplified using Sylvester's determinant theorem in (\ref{eq:det_identity}) and the matrix inversion lemma in (\ref{eq:matrix_inversion_lemma}), respectively to
\begin{eqnarray}
\label{eq:det_Lambda_STC}
\det(\V{\Lambda}_{\V{\bar{y}}_{\tau}|\V{\bar{S}}})&=&\det(\V{I}_{2T}+ \rho\V{\bar{S}} \V{\bar{S}}\hr)\nonumber\\
&=&\det(\V{I}_M+ \rho\V{\bar{S}}\hr \V{\bar{S}})\\
\label{eq:inv_Lambda_STC}
\V{\Lambda}_{\V{\bar{y}}_{\tau}|\V{\bar{S}}}^{-1}&=&(\V{I}_{2T}+ \rho\V{\bar{S}} \V{\bar{S}}\hr)^{-1}\nonumber \\
&=&\V{I}_{2T}-\rho\V{\bar{S}}(\V{I}_{M}+ \rho\V{\bar{S}}\hr \V{\bar{S}})^{-1}\V{\bar{S}}\hr
\end{eqnarray}
where in the matrix inversion lemma, the substitutions $\V{A}=\V{I}_{2T}$, $\V{B}=\rho\,\V{\bar{S}}$, $\V{C}=\V{I}_M$, and $\V{D}=\V{\bar{S}}\hr$ have been made. Substituting (\ref{eq:det_Lambda_STC}) and (\ref{eq:inv_Lambda_STC}) back in (\ref{eq:P_Y_S_general}), we get
\begin{equation}
\label{eq:P_Y_S_general2}
\text{p}(\V{\bar{Y}}_\tau|\V{\bar{S}})=\frac{1}{\pi^{2TN}\det^N(\V{I}_M+ \rho\V{\bar{S}}\hr \V{\bar{S}} )}e^{-\tr\{\V{\bar{Y}}_\tau\hr(\V{I}_{2T}-\rho\V{\bar{S}}(\V{I}_{M}+ \rho\V{\bar{S}}\hr \V{\bar{S}})^{-1}\V{\bar{S}}\hr)\V{\bar{Y}}_\tau\}}
\end{equation}
Hence the ML decision metric in (\ref{eq:ML_STC}) reduces to 
\begin{eqnarray}
\label{eq:ML_STC_general}
\hat{\V{V}}_{z_\tau}\Big|_{\text{ML}}&=&\argmax\limits_{l\in\{0,..,L-1\}}\:\,\frac{e^{-\tr\{\V{\bar{Y}}_\tau\hr\V{\bar{Y}}_\tau-\rho\V{\bar{Y}}_\tau\hr\V{\bar{S}}^{(l)}(\V{I}_{M}+ \rho\V{\bar{S}}\sp{(l)\dagger} \V{\bar{S}}^{(l)})^{-1}\V{\bar{S}}\sp{(l)\dagger}\V{\bar{Y}}_\tau\}}}{\pi^{2TN}\det^N(\V{I}_M+ \rho\V{\bar{S}}\sp{(l)\dagger} \V{\bar{S}}^{(l)})}\nonumber\\\ST
&=&\argmax\limits_{l\in\{0,..,L-1\}}\:\,\frac{e^{\tr\{\rho\V{\bar{Y}}_\tau\hr\V{\bar{S}}^{(l)}(\V{I}_{M}+ \rho\V{\bar{S}}\sp{(l)\dagger} \V{\bar{S}}^{(l)})^{-1}\V{\bar{S}}\sp{(l)\dagger}\V{\bar{Y}}_\tau\}}}{\det^N(\V{I}_M+ \rho\V{\bar{S}}\sp{(l)\dagger} \V{\bar{S}}^{(l)})}
\end{eqnarray}
Taking the natural logarithm of (\ref{eq:ML_STC_general}) and substituting it back in (\ref{eq:ML_STC2}), the ML decision metric reduces finally to
\begin{equation}
\label{eq:ML_STC_general2}
\boxed{\hat{\V{V}}_{z_\tau}\Big|_{\text{ML}}=\argmax\limits_{l\in\{0,..,L-1\}}\:\,\tr\{\rho\V{\bar{Y}}_\tau\hr\V{\bar{S}}^{(l)}(\V{I}_{M}+ \rho\V{\bar{S}}\sp{(l)\dagger} \V{\bar{S}}^{(l)})^{-1}\V{\bar{S}}\sp{(l)\dagger}\V{\bar{Y}}_\tau\}-N\ln\{\det(\V{I}_M+ \rho\V{\bar{S}}\sp{(l)\dagger} \V{\bar{S}}^{(l)})\}.}
\end{equation}
Since no restrictions on the codebook $\Omega$  have been assumed, (\ref{eq:ML_STC_general2}) is the general ML decision metric for non-coherent detection which applies in the case of piece-wise constant Rayleigh flat fading channels with independent coefficients for all differential ST codes considered in this thesis. \\

Consider next the special case of ST codes with unitary information matrices, i.e. $\V{V}_{z_\tau}\hr\V{V}_{z_\tau}=\V{I}_{M}$. This means that all columns of $\V{V}_{z_\tau}$ are orthogonal, and every column is of unit norm. In the unitary transmission, the differential encoding function is defined as 
\begin{equation}
\label{eq:Diff_encoding_unitary}
\V{S}_\tau=\V{V}_{z_\tau}\V{S}_{\tau-1},
\end{equation}
which is similar to the differential encoding of PSK signals in the SISO system. The above relation assumes $T=M$ for the dimensions in the matrix multiplication above to fit. This means $\V{V}_{z_\tau}$ and $\V{S}_\tau$ are square matrices $\in\mathbb{C}^{M\!\times\!M}\:\forall\,\tau$. The possibility of transmitting rectangular matrices with differential encoding remains an open question. The problem is that differential encoding in most cases includes a multiplication operation of the previously transmitted matrix and the new information matrix, which is not easily realizable with rectangular matrices.  Hence all the transmit matrices considered all over this thesis will be square matrices.\\

Since the product of two unitary matrices yields another unitary matrix, then if any information matrix $\V{V}_{z_\tau}$ is unitary, and the initial transmit matrix $\V{S}_0$ is also unitary, then it follows from (\ref{eq:Diff_encoding_unitary}) that all transmit matrices $\V{S}_\tau$ are also unitary $\forall\,\tau$, i.e. $\V{S}_\tau\hr\V{S}_\tau=\V{I}_M$. In words, this means that the data streams transmitted
over the different transmit antennas within a block are mutually orthogonal. Consequently, 
\begin{equation}
\label{eq:Unitary_S_bar}
\V{\bar{S}}\hr\V{\bar{S}}=[\V{S}_{\tau-1}\hr\, \V{S}_{\tau}\hr]\begin{bmatrix}\V{S}_{\tau-1}\\\V{S}_\tau\end{bmatrix}=\V{S}_{\tau-1}\hr\V{S}_{\tau-1}+\V{S}_\tau\hr\V{S}_\tau=2\V{I}_M
\end{equation}
Using (\ref{eq:Unitary_S_bar}), the determinant and the inverse terms in (\ref{eq:ML_STC_general2}) reduce to 
\begin{eqnarray}
\label{eq:det_inv_unitary}
\det(\V{I}_M+ \rho\V{\bar{S}}\sp{(l)\dagger} \V{\bar{S}}^{(l)})&=&\det(\V{I}_M+ \rho\!\times\!2\V{I}_M)\nonumber\\
&=&\det((1+2\rho)\V{I}_M)=(1+2\rho)^M\\\ST
(\V{I}_M+ \rho\V{\bar{S}}\sp{(l)\dagger} \V{\bar{S}}^{(l)})^{-1}&=&((1+2\rho)\V{I}_M)^{-1}=\frac{1}{1+2\rho}\V{I}_M
\end{eqnarray}
This makes the second term in (\ref{eq:ML_STC_general2}) constant for all $l$, and therefore irrelevant for the decision. The ML metric reduces to 
\begin{eqnarray}
\label{eq:ML_Unitary_STC_two_block}
\hat{\V{V}}_{z_\tau}\Big|_{\text{ML}}&=&\argmax\limits_{l\in\{0,..,L-1\}}\:\,\tr\{\frac{\rho}{1+2\rho}\V{\bar{Y}}_\tau\hr\V{\bar{S}}^{(l)}\V{\bar{S}}\sp{(l)\dagger}\V{\bar{Y}}_\tau\}\nonumber\\
&=&\argmax\limits_{l\in\{0,..,L-1\}}\:\,\tr\{\V{\bar{Y}}_\tau\hr\V{\bar{S}}^{(l)}\V{\bar{S}}\sp{(l)\dagger}\V{\bar{Y}}_\tau\}\nonumber\\
&=&\argmax\limits_{l\in\{0,..,L-1\}}\:\,\|\V{\bar{Y}}_\tau\hr\V{\bar{S}}^{(l)}\|_F^2.
\end{eqnarray}
where $\V{\bar{S}}$ is
\begin{equation}
\label{eq:S_bar_unitary}
\V{\bar{S}}= \begin{bmatrix}\V{S}_{\tau-1}\\\V{S}_\tau\end{bmatrix}=\begin{bmatrix}\V{S}_{\tau-1}\\\V{V}_{z_\tau}\V{S}_{\tau-1}\end{bmatrix}=\begin{bmatrix}\V{I}_M\\\V{V}_{z_\tau}\end{bmatrix}\V{S}_{\tau-1}.
\end{equation}
Recall from the discussion on (\ref{eq:P_Y_S_general_useful}) that multiplying $\V{\bar{S}}$ by a unitary matrix from the right does not change the ML receiver metric. Therefore the transmit two-block matrices $\textstyle\begin{bmatrix}\V{I}_M\\\V{V}_l\end{bmatrix}$ and $\textstyle\begin{bmatrix}\V{I}_M\\\V{V}_l\end{bmatrix}\V{S}_{\tau-1}$ are indistinguishable to the receiver.
As a result, the transmit candidate matrix can be written in a canonical form as
\begin{equation}
 \V{S}^{(l)}\equiv\begin{bmatrix}\V{I}_M\\\V{V}_l\end{bmatrix}
\end{equation}
and the ML metric in (\ref{eq:ML_Unitary_STC_two_block}) reduces to,
\begin{eqnarray}
\label{eq:ML_unitray}
\hat{\V{V}}_{z_\tau}\Big|_{\text{ML}}&=&\argmax\limits_{l\in\{0,..,L-1\}}\:\,\|[\V{Y}_{\tau-1}\hr \,\V{Y}_\tau\hr]\begin{bmatrix}\V{I}_M\\\V{V}_l\end{bmatrix}\|_F^2\nonumber\\
&=&\argmax\limits_{l\in\{0,..,L-1\}}\:\,\|\V{Y}_{\tau-1}\hr +\V{Y}_\tau\hr\V{V}_l\|_F^2\nonumber\\
&=&\argmax\limits_{l\in\{0,..,L-1\}}\:\,\|\V{Y}_{\tau-1}+\V{V}_l\hr\V{Y}_\tau\|_F^2
\end{eqnarray}
which can be alternatively written as 
\begin{eqnarray*}
\hat{\V{V}}_{z_\tau}\Big|_{\text{ML}}&=&\argmax\limits_{l\in\{0,..,L-1\}}\:\,\tr\{(\V{Y}_{\tau-1}+\V{V}_l\hr\V{Y}_\tau)\hr(\V{Y}_{\tau-1}+\V{V}_l\hr\V{Y}_\tau)\}\\
&=&\argmax\limits_{l\in\{0,..,L-1\}}\:\,\tr\{\underbrace{\V{Y}_{\tau-1}\hr\V{Y}_{\tau-1}}_{\text{indep. of }l}+\V{Y}_{\tau-1}\hr\V{V}_l\hr\V{Y}_{\tau}+\V{Y}_{\tau}\hr\V{V}_l\V{Y}_{\tau-1}+\underbrace{\V{Y}_{\tau}\hr\overbrace{\V{V}_l\V{V}_l\hr}^{\eq\V{I}_M}\V{Y}_{\tau}}_{\text{indep of }l}\\
&=&\argmax\limits_{l\in\{0,..,L-1\}}\:\,\tr(\V{Y}_{\tau-1}\hr\V{V}_l\hr\V{Y}_{\tau}+\V{Y}_{\tau}\hr\V{V}_l\V{Y}_{\tau-1})
\end{eqnarray*}
and the ML metric for unitary ST codes in the case of piece-wise constant Rayleigh flat fading channels with independent coefficients reduces finally to
\begin{equation}
\label{eq:ML_unitary2}
 \boxed{\hat{\V{V}}_{z_\tau}\Big|_{\text{ML}}=\argmax\limits_{l\in\{0,..,L-1\}}\:\,\Re\{\tr(\V{V}_l\V{Y}_{\tau-1}\V{Y}_{\tau}\hr)\}}
\end{equation}


\section{Design Criteria}
\label{s:Design_Criteria}
In this section we derive the design criteria for space-time codes from which we extract the notion of diversity order and coding gain. The criteria used to design a ST codebook $\Omega$ is based on minimizing the worst Pair-wise Error Probability (PEP) between two codewords in the codebook. In \cite{Hochwald_Marzetta_2000}, Hochwald and Marzetta proved the Chernoff upper bound of the PEP between two different transmit matrices $\V{\bar{S}}^{(l)}$ and $\V{\bar{S}}^{(l')}$ in the case of unitary ST codes. The PEP between $\V{\bar{S}}^{(l)}$ and $\V{\bar{S}}^{(l')}$or equivalently between codewords $\V{V}_l$ and $\V{V}_{l'}$ is the probability of mistaking $\V{V}_l$ for $\V{V}_{l'}$ or vice versa\footnote{In \cite{Hochwald_Marzetta_2000}, it was proved that $\text{PEP}(l,l')\eq \text{PEP}(l',l)$} and is defined in the case of unitary transmission as
\begin{eqnarray}
\label{eq:PEP_l_l_dash}
\text{PEP}(l,l')&=&\text{p}(\hat{\V{V}}=\V{V}_{l'}|\:\V{V}_l \text{ is transmitted})\nonumber\\
&=&\text{p}(\text{ML metric for }\V{V}_{l'}>\text{ML metric for }\V{V}_l|\:\V{V}_l)\nonumber\\
&\overset{\underset{\text{from  (\ref{eq:ML_unitray})}}{}}{=}&\text{p}(\|\V{Y}_{\tau-1}+\V{V}_{l'}\hr\V{Y}_\tau\|_F^2>\|\V{Y}_{\tau-1}+\V{V}_{l}\hr\V{Y}_\tau\|_F^2\:|\:\V{V}_l).
\end{eqnarray}
The Chernoff upper bound of PEP$(l,l')$ is derived in \cite[eq. (B.11)]{Hochwald_Marzetta_2000} and modified in \cite[eq. (7)]{Liang_Xia_2002} with the same notations used in this thesis as
\begin{equation}
\label{eq:PEP_l_ldash2}
\text{PEP}(l,l')\leq\frac{1}{2}\prod\limits_{m=1}^{M}[1+\frac{\rho^2}{4(1+2\rho)}\sigma_m^2(\V{V}_l-\V{V}_{l'})]^{-N}
\end{equation}
where $\sigma_m(\V{A})$ is the $m^{\text{th}}$ singular value of matrix $\V{A}$. For further understanding of singular values and singular value decomposition (SVD), refer to Section \ref{ss:SVD}. Singular values and eigenvalues can be related using
\begin{equation}
\label{eq:SV_EV}
\sigma_m^2(\V{A})=\lambda_m(\V{A}\hr\V{A}),
\end{equation}
where $\lambda_m(\V{B})$ is the $m^{\text{th}}$ eigenvalue of matrix $\V{B}$. For further understanding of eigenvalues and eigenvalue decomposition, refer to Section \ref{ss:EVD}. The above relation is shown in (\ref{eq:EV_SV2}), theorem \ref{theo:EV_SV2}. Hence $\sigma_m^2(\V{V}_l-\V{V}_{l'})$ can be replaced by
\begin{eqnarray}
\label{eq:SV_EV2}
\sigma_m^2(\V{V}_l-\V{V}_{l'})&=&\lambda_m((\V{V}_l-\V{V}_{l'})\hr(\V{V}_l-\V{V}_{l'}))\nonumber\\ &=&\lambda_m(\V{D}_{ll'}\hr\V{D}_{ll'})\overset{\underset{\Delta}{}}{=}\lambda_{ll'm},
\end{eqnarray}
where $\V{D}_{ll'}$ is defined as the distance matrix between the two codewords $\V{V}_l$ and $\V{V}_{l'}$, and $\lambda_{ll'm}$ is the $m^{\text{th}}$ eigenvalue of the squared distance matrix $\V{D}_{ll'}\hr\V{D}_{ll'}$. Consequently PEP$(l,l')$ in (\ref{eq:PEP_l_ldash2}) can be rewritten as
\begin{eqnarray}
\label{eq:PEP_l_l_dash3}
\text{PEP}(l,l')&\leq&\underbrace{\frac{1}{2}}_{\text{constant}}\prod\limits_{m=1}^{M}[1+\underbrace{\frac{\rho^2}{4(1+2\rho)}}_{f}\lambda_{ll'm}]^{-N}\nonumber\\
&\leq&\prod\limits_{m=1}^{M}[1+f\lambda_{ll'm}]^{-N}
\end{eqnarray}
where the common factor $\nicefrac{1}{2}$ is constant w.r.t. the codewords $\V{V}_l$ and $\V{V}_{l'}$, so it does not affect the design criterion and can be omitted and since it is less than $1$ the inequality still holds. The Right-Hand Side (RHS) of (\ref{eq:PEP_l_l_dash3}) can be decomposed to
\begin{equation}
\label{eq:PEP_l_l_dash4}
\text{PEP}(l,l')\leq\Big[\frac{1}{(1+f\lambda_{ll'1})......(1+f\lambda_{ll'M})}\Big]^{N}.
\end{equation}
If $\V{D}_{ll'}\hr\V{D}_{ll'}$ has rank $r_{ll'}\leq M$, i.e.\footnote{rank($\V{A}\hr\V{A}$)=rank($\V{A}$), as shown in (\ref{eq:rank_AA_A}), theorem \ref{theo:rank_AA_A}.}
\begin{equation}
\label{eq:rank_dist_matrix}
r_{ll'}=\rank(\V{D}_{ll'}\hr\V{D}_{ll'})=\rank(\V{D}_{ll'}),
\end{equation}
then only $r_{ll'}$ eigenvalues are non-zero (refer to theorem \ref{theo:rank_EV}), making (\ref{eq:PEP_l_l_dash4}) reduce to
\begin{eqnarray}
\label{eq:PEP_l_l_dash5}
\text{PEP}(l,l')&\leq&\Big[\frac{1}{(1+f\lambda_{ll'1})......(1+f\lambda_{ll'r_{ll'}})}\Big]^{N}\nonumber\\
&\leq&\Big[\frac{1}{1+f\sum\limits_{i=1}^{r_{ll'}}{\lambda_{ll'i}}+f^2\sum\limits_{i=1, j=1}^{r_{ll'}}{\lambda_{ll'i}\lambda_{ll'j}}+......+f^{r_{ll'}}\prod\limits_{i=1}^{r_{ll'}}{\lambda_{ll'i}}}\Big]^{N},
\end{eqnarray}
where $f=\frac{\rho^2}{4(1+2\rho)}$. In the low SNR $(\rho)$ range, $f^{\xi}$ decreases with increasing $\xi$, so the higher order terms in the denominator of (\ref{eq:PEP_l_l_dash5}) are insignificant and only the first two terms are dominant making PEP$(l,l')$ reduce to
\begin{equation}
\label{eq:PEP_l_l_dash_low_SNR}
\text{PEP}(l,l')\Big|_{\text{Low SNR}}\leq[1+f\sum\limits_{i=1}^{r_{ll'}}{\lambda_{ll'i}}]^{-N}
\end{equation}
Note that the inequality still holds since ignoring the insignificant terms in the RHS enlarges its value.\\

The design criteria is based on the worst PEP among all codeword pairs $\V{V}_l$ and $\V{V}_{l'}$ $\forall\,l\neq l'\in\{0,...,L-1\}$. The worst pair has the maximum PEP and therefore the minimum $\sum\limits_{i=1}^{r_{ll'}}{\lambda_{ll'i}}$, which is defined as the diversity sum $\delta$
\begin{equation}
\label{eq:diversity_sum}
\delta=\mini\limits_{l\neq l'\in\{0,...,L-1\}}\:\,\sum\limits_{i=1}^{r_{ll'}}{\lambda_{ll'i}}.
\end{equation}
In conclusion, the design criterion for ST codebooks in the low SNR range is to maximize the diversity sum defined in (\ref{eq:diversity_sum}).\\

On the contrary, in the high SNR range, $f\eq \frac{\rho^2}{4(1+2\rho)}\!\simeq\!\frac{\rho^2}{8\rho}\eq \frac{\rho}{8}$. Therefore $f^{\xi}$ increases with increasing $\xi$, so the lower order terms in the denominator of (\ref{eq:PEP_l_l_dash5}) are insignificant and only the last term is the dominant one making the upper bound of PEP$(l,l')$ reduce to
\begin{eqnarray}
\label{eq:PEP_l_l_dash_high_SNR}
\text{PEP}(l,l')\Big|_{\text{High SNR}}&\leq &[f^{r_{ll'}}\prod\limits_{i=1}^{r_{ll'}}\lambda_{ll'i}]^{-N}\nonumber\\
&\leq&\bigg[\Big(\frac{\rho}{8}\Big)^{r_{ll'}}\prod\limits_{i=1}^{r_{ll'}}{\lambda_{ll'i}}\bigg]^{-N}\nonumber\\
&\leq&\bigg[\frac{\rho}{8}\Big(\prod\limits_{i=1}^{r_{ll'}}{\lambda_{ll'i}}\Big)^{\frac{1}{r_{ll'}}}\bigg]^{-r_{ll'}N}
\end{eqnarray}
Taking $10\log_{10}$ of (\ref{eq:PEP_l_l_dash_high_SNR}), the PEP in the logarithmic scale (in dB) is
\begin{eqnarray}
\label{eq:PEP_l_l_dash_high_SNR_log}
10\log_{10}\text{PEP}(l,l')\Big|_{\text{High SNR}}\!&\!\leq\! &\!-r_{ll'}N\Big(10\log_{10}\frac{\rho}{8}+10\log_{10}\Big(\prod\limits_{i=1}^{r_{ll'}}{\lambda_{ll'i}}\Big)^{\frac{1}{r_{ll'}}}\Big)\nonumber\\
\!&\!\leq\!&\!\underbrace{-r_{ll'}N}_{a}\Big(\underbrace{10\log_{10}\frac{\rho}{4}}_{x} +\underbrace{10\log_{10}\frac{1}{2}}_{\unit[-3]{dB}} +\underbrace{10\log_{10}\Big(\prod\limits_{i=1}{\lambda_{ll'i}}\Big)^{\frac{1}{r_{ll'}}}}_{b}\Big)\nonumber\\
\!&\!\leq\!&\!a(x\unit[-3]{dB}+b)
\end{eqnarray}
Useful insights can be extracted from the above relation. First, $\nicefrac{\rho}{8}$ has been split into $\nicefrac{\rho}{4}$ and $\nicefrac{1}{2}$ because when (\ref{eq:PEP_l_l_dash_high_SNR}) is compared to the corresponding PEP$(l,l')$ in the high SNR range for \emph{coherent} systems as in \cite[p.132]{Jankiraman2004}, \cite[p.116, eq.(6.11)]{PNG2003}, and \cite[p.31, eq.(3.4)]{Lusina2003}, the only difference is an extra $\nicefrac{1}{2}$ factor in our (non-coherent) case. Such a factor translates to the well-known $\unit[3]{dB}$ loss in SNR experienced by non-coherent systems compared to coherent ones. Other than the half factor, PEP$(l,l')$ for unknown-channel systems has the same form as that for known-channel systems. This leads to the important conclusion that a good design for a known-channel system is also good for an unknown channel system in the case of differential unitary transmission. Consequently the design criterion for both systems in the unitary case is expected to be the same.\\

\begin{figure}%
\centering
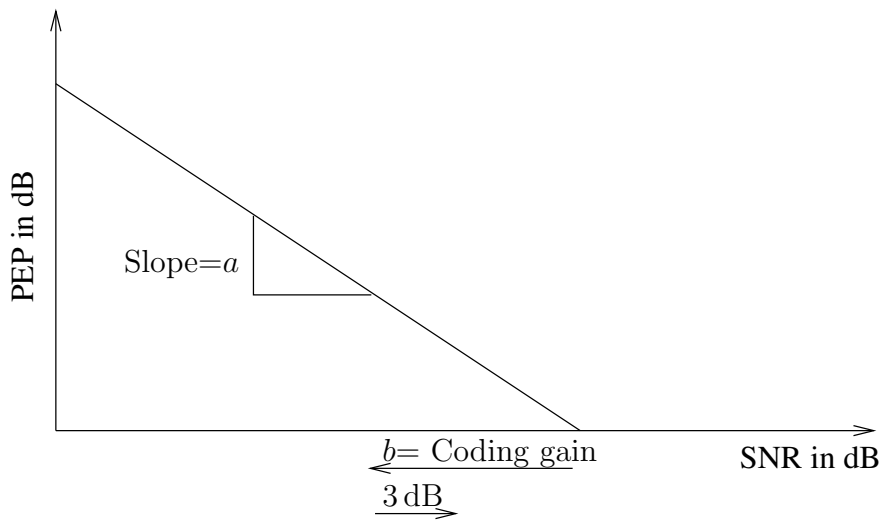
\caption{PEP$(l,l')$ in dB vs SNR in $M\!\times\!N$ system with Differential ST coding}%
\label{fig:SER_SNR_STBC}%
\end{figure}

Investigating (\ref{eq:PEP_l_l_dash_high_SNR_log}) further, the upper bound of the PEP$(l,l')$ in the logarithmic scale takes the form of $y\eq a(x+b-\unit[3]{dB})$, where $y$ is PEP$(l,l')$ in dB and $x$ is a scaled version of SNR in dB. This  relation is illustrated in Figure \ref{fig:SER_SNR_STBC}. $a$ is the slope of the curve, specifically the magnitude of the slope of the PEP$(l,l')$ vs SNR curve in the logarithmic scale is $r_{ll'}N$. The slope of the error rate curve is a measure of how fast the error performance improves with SNR. For the design criterion, consider the worst codeword pair which leads to the least steepness of the PEP curve, and therefore to the lowest absolute slope. Such a codeword pair has the minimum rank of the distance matrix $\V{D}_{ll'}$ or equivalently of $\V{D}_{ll'}\hr\V{D}_{ll'}$ among all codeword pairs (i.e. $\min r_{ll'}\:\forall\:l\neq l'\in\{0,...,L-1\}$). Based on this we define the Diversity Order (DO) as the smallest absolute slope of the PEP vs SNR curve in the logarithmic scale. Namely, 
\begin{equation}
\label{eq:DO}
\text{DO}=\mini\limits_{l\neq l'\in\{0,..,L-1\}}\:\,r_{ll'}N.
\end{equation}
Two codebooks having the same DO will have parallel error rate curves. The gain achieved due to the increase of the slope of the error rate curve is known as the diversity gain. One design criterion for ST codes in the high SNR regime is to maximize the DO in order to optimize the spatial diversity. Such a design criterion is known as the \emph{rank criterion} since it is based on the minimum rank of the difference matrices. A codebook whose \emph{all} codeword pairs have distance matrices of maximum rank (i.e. $(r_{ll'}\eq M)\:\ \forall\:\: l\neq l'\in\{0,...,L-1\}$) achieves the maximum diversity order of $MN$ which is the total number of diversity branches of an $M\!\times\!N$ MIMO system. Such a code is said to have achieved full diversity. \\

Consider again the relation $y\eq a(x+b-\unit[3]{dB})$, one can see that $y\eq 0$ when $x\eq -b+\unit[3]{dB}$. This means that due to ST coding, the logarithmic PEP curve is shifted horizontally to the left by $b$ which indicates performance improvement, and due to non-coherent detection, the curve is shifted $\unit[3]{dB}$ to the right. The worst codeword pair can alternatively be defined as the one that has the smallest value of $b$, which in the linear scale is defined as the coding gain (CG), namely
\begin{equation}	
\label{eq:CG}
\text{CG}=\mini\limits_{l\neq l'\in\{0,..,L-1\}}\:\,\Big(\prod\limits_{i=1}^{r_{ll'}}{\lambda_{ll'i}}\Big)^{\frac{1}{r_{ll'}}}.
\end{equation}
Such a gain defines the second design criterion for space-time codes in the high SNR range. Namely, a ST code can be designed to maximize the coding gain. It is termed as coding gain since it is analogous to the coding gain achieved with channel coding as it is a \emph{horizontal} shift to the left of the error rate curve in the logarithmic scale. As seen in (\ref{eq:CG}), CG is proportional to the smallest product of the non-zero eigenvalues of matrix $\V{D}_{ll'}\hr\V{D}_{ll'}$ over all codeword pairs. Such a product is the reason why CG is sometimes referred to as diversity product as in \cite{Hochwald2000}. The term diversity product is also analogous to diversity sum defined in (\ref{eq:diversity_sum}), which is the measure used in the design criterion for ST codes in the low SNR regime.\\

Having defined the diversity gain and the coding gain, the two gains have different effects on the error rate curve. The diversity gain is the result of the increase in the slope of the error rate curve defined by the diversity order, and therefore the SNR improvement due to diversity gain increases with SNR. Whereas the coding gain is the horizontal shift of the error rate curve, and therefore the SNR improvement due to coding gain remains constant with increasing SNR. Figure \ref{fig:DG_vs_CG} illustrates the difference between diversity gain and coding gain.\\
\begin{figure}
\centering
 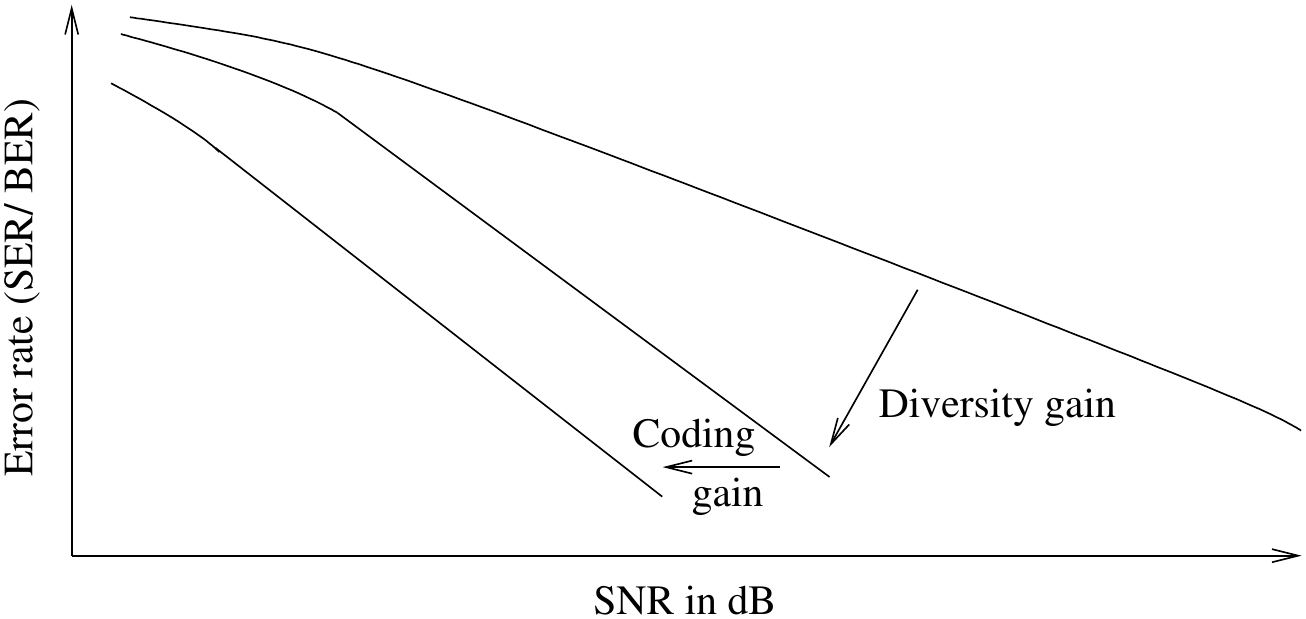
\caption{A schematic illustrating the difference between diversity gain and coding gain.}
\label{fig:DG_vs_CG}
\end{figure}

In the case of codes that achieve full diversity (i.e. $\!\text{DO}\eq MN$), all distance matrices of all codeword pairs are of full rank. In this case the only design criterion is maximizing the coding gain which then includes the minimum product of \emph{all} eigenvalues of matrix $\V{D}_{ll'}\hr\V{D}_{ll'}$. This is the same as the minimum determinant of the same matrix (see theorem \ref{theo:det_EV}), i.e.
\begin{eqnarray}
\label{eq:CG_full_diversity}
\text{CG}\Big|_{\text{DO=}MN}\!&=&\!\mini\limits_{l\in\{0,..,L-1\}}\:\,\Big(\prod\limits_{i=1}^{M}{\lambda_{ll'i}}\Big)^{\frac{1}{M}}\nonumber\\
&=&\mini\limits_{l\in\{0,..,L-1\}}\:\,\Big[\det(\V{D}_{ll'}\hr\V{D}_{ll'})\Big]^{\frac{1}{M}}
\end{eqnarray}
For this reason, the second design criterion is known as the \emph{determinant criterion}, which can be restated as follows; design a ST code that maximizes the minimum determinant of the squared distance matrix $\V{D}_{ll'}\hr\V{D}_{ll'}$. Note also that if an arbitrary codebook achieves a non-zero determinant for all squared distance matrices, it means that all distance matrices are of full rank and therefore the code achieves full diversity. Therefore, evaluating the minimum determinant in (\ref{eq:CG_full_diversity}) allows us to see whether the full diversity condition is satisfied, and also to know how much coding gain can be achieved.\\

In conclusion, this section has derived the design criteria for ST codes in both the low and the high SNR range based on minimizing the worst pair-wise error probability. To improve the performance in the low SNR range, the ST code should be designed to maximize the diversity sum. If however the high SNR range is the operating range, then two design criteria have been defined. The first is the rank criterion which aims at maximizing the minimum absolute slope of the error rate curve by optimizing the diversity order of the code. The second criterion is the determinant criterion whose goal is to maximize the minimum determinant of the squared distance matrix in order to increase the horizontal left shift of the error rate curve, and therefore optimize the coding gain.


\abbrev{STTC}{Space-Time Trellis Coding}
\abbrev{STBC}{Space-Time Block Code}
\abbrev{OSTBC}{Orthogonal Space-Time Block Code}
\abbrev{DOSTBC}{Differential Orthogonal Space-Time Block Code}
\abbrev{i.i.d.}{independent and identically distributed}
\abbrev{PEP}{Pair-wise Error Probability}
\abbrev{RHS}{Right-Hand Side}
\abbrev{DO}{Diversity Order}

\chapter{Orthogonal Space-Time Schemes}
\label{chap:DSTC_schemes}
In Chapter \ref{chap:STC}, we laid the foundation of space-time coding in differential systems. A general form for the code structure, the design criteria and the ML decision metric have been defined. In this chapter, we show several realizations of some STC schemes. Specifically, the \emph{orthogonal} STC schemes described in the literature review in Section \ref{s:STC_Literature_Review} will be unfolded here.  The chapter describes two classes of ST schemes, namely Unitary Space-Time Modulation (USTM), and Orthogonal Space-Time Block Codes (OSTBCs). Furthermore, the different schemes are compared in terms of complexity and error performance.

\section{Differential Unitary Space-Time Modulation}
\label{s:DSTM}
Hochwald and Marzetta proposed in \cite{Hochwald_Marzetta_2000} a modulation scheme for multiple transmit antenna systems which they named Differential Unitary Space-Time Modulation (DUSTM). The scheme is well-suited for Rayleigh flat fading environments when neither the transmitter nor the receiver knows the channel coefficients. In \cite{Hochwald_USTM_design_2000}, Hochwald et al. showed a systematic approach to design unitary space-time signals. Then in \cite{Hochwald2000}, Hochwald and Sweldens proposed one simple design for DUSTM that led eventually to constellations of the so-called diagonal signals, where only one transmit antenna is active at a time. This section will explain such a scheme in detail.\\

Recall the MIMO channel model described in Section \ref{s:MIMO-channel_model}, where the transmitter is equipped with $M$ transmit antennas and the receiver with $N$ receive antennas and each antenna pair is connected through a piece-wise constant Rayleigh flat fading channel. Such a system can be described as in (\ref{eq:Y_MIMO_STC}) by (restated here for convenience)
\begin{equation}
\label{eq:Y_MIMO_STC2}
\V{Y}_\tau=\sqrt{\rho}\V{S}_\tau\V{H}+\V{W}_\tau
\end{equation}
where at block index $\tau$, $\V{Y}_\tau$ is the $T\!\times\! N$ received matrix, $\V{S}_\tau$ is the $T\!\times\! M$ differentially encoded transmit matrix, $\V{W}_\tau$ is the $T\!\times\! N$ noise matrix and $\V{H}$ is the $M\!\times\! N$ channel matrix in some transmission time frame (a frame spans multiple transmit blocks during which the channel matrix is assumed constant). Recall also that the transmit matrices considered are square matrices, i.e. ($T\eq  M$). The transmission scheme of DUSTM follows the STC architecture shown in Figure \ref{subfig:DSTM}, where the information bits are directly mapped to information matrices.\\

As the name DUSTM suggests, the transmit signal matrices considered in this scheme are unitary. Explicitly stated, all information matrices $\V{V}_{l}$ are unitary matrices, i.e. $\V{V}_l\hr\V{V}_l\eq\V{I}_M\:\forall\, l\in\{0,..,L-1\}$, and using the differential transmission equation
\begin{equation}
\label{eq:Diff_encoding_unitary2}
\V{S}_\tau=\V{V}_{z_\tau}\V{S}_{\tau-1},
\end{equation}
if the initial transmit matrix $\V{S}_0$ is any arbitrary unitary matrix, then it follows that all transmit matrices $\V{S}_\tau$ are also unitary, i.e. $\V{S}_\tau\hr\V{S}_\tau\eq\V{I}_M\:\forall\,\tau$.  The unitary condition ensures that the transmit symbol stream over the $M$ antennas are mutually orthogonal, which has the advantage of simplifying the demodulation metric as shown in (\ref{eq:ML_unitray}) to
\begin{equation}
\label{eq:ML_Unitary_STC2}
 \hat{\V{V}}_{z_\tau}\Big|_{\text{ML}}=\argmax\limits_{l\in\{0,..,L-1\}}\:\,\|\V{Y}_{\tau-1}+\V{V}_l\hr\V{Y}_\tau\|_F^2.
\end{equation}
\subsection{DUSTM vs. DPSK}
\label{ss:DUSTM_DPSK}
Up to this point, a clear similarity between the DUSTM scheme and the DPSK single antenna scheme can be inferred. Table \ref{tab:DPSK_vs_DUSTM} summarizes such a similarity.
\renewcommand{\arraystretch}{1.5} 
\begin{table}[!ht]
\begin{center}
\caption{Comparison between DPSK modulation scheme for single antenna system and DUSTM modulation scheme for MIMO systems}
\label{tab:DPSK_vs_DUSTM}
  \begin{tabular}{|m{3cm}|m{5.2cm}|m{6cm}|}
    \hline    
      \textbf{Parameter} 	 &  \centering\textbf{DPSK} &  \centering\textbf{DUSTM} \tn[5pt] \hline 
	  Information integer sequence & \multicolumn{2}{m{10cm}|}{\centering $z_1, z_2, ..., z_t, ...$} \tn \hline
	  Information symbol sequence &  \vspace{-1cm}\begin{eqnarray*} & v_{z_1}, v_{z_2}, ...,v_{z_t},... \\ & v_{l}=e^{\frac{j2\pi l}{L}},\\ & |v_l|=1\:\forall\, l=\{0,..,L-1\}\end{eqnarray*}\vspace{-1cm} & \vspace{-1cm}\begin{eqnarray*} & \V{V}_{z_1}, \V{V}_{z_2}, ..., \V{V}_{z_\tau},... \\ & \hspace{-0.5cm} \V{V}_l \text{ is an }M\!\times\!M \text{ unitary matrix,} \\ & \hspace{-0.3cm}\text{i.e. } \V{V}_l\hr\V{V}_l=\V{I}_M\:\forall\, l=\{0,..,L-1\}\end{eqnarray*} \vspace{-1cm} \tn \hline
	  Transmit symbol sequence  & \vspace{-1cm}\begin{eqnarray*} & s_1, s_2, ..., s_t \\  & s_t=v_{z_t}s_{t-1},\:\: t=1,2,.... \\ & s_0=1 \end{eqnarray*}\vspace{-1cm} & \vspace{-1cm}\begin{eqnarray*} & \V{S}_1, \V{S}_2, ..., \V{S}_\tau \\ & \V{S}_\tau=\V{V}_{z_\tau}\V{S}_{\tau-1},\:\: \tau=1,2,.... \\ & \V{S}_0=\V{I}_M \end{eqnarray*}\vspace{-1cm} \tn \hline
	  $l^{\text{th}}$ candidate transmit matrix in a canonical form & \vspace{-1cm}\begin{equation*}\bar{\V{s}}^{(l)}=\begin{bmatrix} 1\\e^{\frac{j2\pi l}{L}}\end{bmatrix}\end{equation*}\vspace{-1cm} & \vspace{-1cm}\begin{equation*}\bar{\V{S}}^{(l)}=\begin{bmatrix}\V{I}_M\\\V{V}_l\end{bmatrix}\end{equation*}\vspace{-1cm}\tn\hline
	  Received symbols for differential detection	 &\vspace{-1cm}\begin{equation*}\bar{y}_t =\begin{bmatrix} y_{t -1} \\ \ y_t \end{bmatrix}\end{equation*}\vspace{-1cm} &  \vspace{-1cm}\begin{equation*}\V{\bar{Y}}_\tau =\begin{bmatrix} \V{Y}_{\tau -1} \\ \V{Y}_{\tau} \end{bmatrix}\end{equation*}\vspace{-1cm} \tn \hline
	  ML metric in Rayleigh piece-wise constant fading channel & \vspace{-0.95cm}\begin{equation*}\begin{aligned}\hspace{-0.15cm}&\hat{v}_t\Big|_{\text{ML}}=\argmax\limits_{l\in\{0,..,L-1\}}\:\,\|\bar{\V{y}}_t\hr \bar{\V{s}}^{(l)}\|_F \\\!&\eq\!\argmax\limits_{l\in\{0,..,L-1\}}\:\,\|[y_{t-1}^*\, y_t^*]\begin{bmatrix} 1\\ v_l \end{bmatrix}\|_F \\ \!&\eq\!\argmax\limits_{l\in\{0,..,L-1\}}\:\,|y_{t-1}^*+ y_t^*v_l|\end{aligned}\end{equation*}\vspace{-0.7cm} & \vspace{-0.9cm}\begin{equation*}\begin{aligned}\hspace{-0.2cm}&\hat{\V{V}}_\tau\Big|_{\text{ML}}\!=\argmax\limits_{l\in\{0,..,L-1\}}\:\,\|\bar{\V{Y}}_\tau\hr \bar{\V{S}}^{(l)}\|_F \\\!&\eq \!\argmax\limits_{l\in\{0,..,L-1\}}\:\,\|[\V{Y}_{\tau-1}\hr \V{Y}_\tau\hr]\begin{bmatrix}\V{I}_M \\\V{V}_l\end{bmatrix}\|_F \\ \!&\eq \!\argmax\limits_{l\in\{0,..,L-1\}}\:\,\|\V{Y}_{\tau-1}\hr+ \V{Y}_\tau\hr\,\V{V}_l\|_F\end{aligned}\end{equation*} \vspace{-0.7cm}\tn \hline
\end{tabular}	
\end{center}
\end{table}
The table shows that DUSTM is a direct extension to DPSK, by extending the scalar elements to $M\!\times\!M$ matrices.\\

In DPSK, any transmit symbol can be written as
\begin{equation}
\label{eq:s_t_DPSK}
 s_t=\Big(\prod\limits_{\kappa=1}^{t}{v_{z_\kappa}}\Big)s_0,\:\:t=1,2,...
\end{equation}
with all $v_{z_\kappa}$ drawn from the alphabet
\begin{equation}
\label{eq:PSK_Alphabet}
\mathcal{A}_\text{PSK}=\{v_l=e^{\frac{j2\pi l}{L}}\:\forall\, l\in\{0,..,L-1\}\}.
\end{equation}
Due to the inherent \emph{group} nature of the PSK alphabet, the product of any two elements in $\mathcal{A}_\text{PSK}$ is also an element in $\mathcal{A}_\text{PSK}$. In other words, the product of any two points in the PSK constellation circle falls to a point on the same circle. This is because the product of two phasors is a phasor whose angle is the modulo addition of the angels of the two phasors with respect to $2\pi$. If additionally the initial transmit symbol $s_0$ is also $\in\mathcal{A}_\text{PSK}$ (e.g. $s_0=1$), then it follows from (\ref{eq:s_t_DPSK}) that any transmit symbol $s_t$ is element in $\mathcal{A}_\text{PSK}$. Consequently, the construction of $s_t$ from $s_{t-1}$ doesn't really require multiplying $v_{z_t}$ by $s_{t-1}$, but it rather suffices to add the angles of their phasors.
\subsection{DUSTM Codebook Design}
\label{ss:DUSTM_design}
The group structure of DPSK is however not inherent in the information matrices $\V{V}_l$ of the DUSTM scheme. It should rather be imposed on the design of the codebook. Consider the codebook $\Omega$ defined as
\begin{equation}
 \Omega=\{\V{V}_0,...,\V{V}_{L-1}\}.
\end{equation}
In order for the codebook $\Omega$ to form a finite group under multiplication, four conditions must be satisfied. Namely, enclosure, associativity, existence of an identity element and existence of an inverse element for every element in the set. Enclosure is the condition that the product of any two elements in the set is also an element in the same set. That is, for any $l,\,l'\in\{0,..,L-1\}$, it is required that 
\begin{equation}
 \V{V}_l\V{V}_{l'}=\V{V}_{l''}
\end{equation}
for some $l''\in\{0,..,L-1\}$. Similar to the inherent modulo operation in DPSK transmission, the operator $\oplus$ can be defined for DUSTM scheme to operate on the indices of the group members as
\begin{equation}
\label{eq:modulo_operator}
 l''=l\oplus l'.
\end{equation}
And $\V{V}_{l''}$ can be constructed as 
\begin{equation}
 \V{V}_{l''}=\V{V}_{l\oplus l'}.
\end{equation}\\

The existence of an identity element can be satisfied by making the identity matrix $\V{I}_M$ a member in the group. For example, let $\V{V}_0=\V{I}_M$. By imposing the conditions of enclosure and the existence of an identity element, every element will automatically have an inverse in the group. To see this, consider some element $\V{V}_l$ in $\Omega$, then from enclosure, there exists a group member $\V{V}_{l'}$ such that
\begin{equation*}
 \V{V}_l\V{V}_{l'}=\V{V}_0=\V{I}_M
\end{equation*}
therefore $\V{V}_l^{-1}\eq \V{V}_{l'}$. Finally, the associativity condition follows directly from the associativity of matrix multiplication (i.e. $(\V{V}_l\V{V}_{l''})\V{V}_{l'''}=\V{V}_l(\V{V}_{l''}\V{V}_{l'''})$). If additionally the initial transmit matrix $\V{S}_0$ is chosen to be a group member (e.g. $\V{S}_0=\V{I}_M$), then it follows from the differential transmission equation in (\ref{eq:Diff_encoding_unitary2}) that all transmit matrices $\V{S}_\tau$ are element in the group $\forall\:\tau$. In this case the differential transmission can be described as follows;
if at block index $\tau-1$, the transmit matrix $\V{S}_{\tau-1}$ is the group member with index $x_{\tau-1}$, i.e. $\V{S}_{\tau-1}\eq \V{V}_{x_{\tau-1}}$, and if the new information matrix to be encoded is $\V{V}_{z_\tau}$, then the new transmit matrix $\V{S}_\tau$ is
\begin{equation}
\V{S}_\tau=\V{V}_{z_\tau}\V{S}_{\tau-1}=\V{V}_{z_\tau}\V{V}_{x_{\tau-1}}=\V{V}_{x_\tau}
\end{equation}
where,
\begin{equation}
\label{eq:modulo_encoding}
 x_\tau=z_\tau\oplus x_{\tau-1}.
\end{equation}
This shows the major advantage of finite group constellations, where the transmitter never needs to explicitly multiply matrices for differential encoding, it rather requires only the \emph{index} of the previously transmit matrix ($x_{\tau-1}$) and that of the new information matrix ($z_{\tau}$) to compute the new index $x_\tau$ using a lookup table. Therefore the group structure simplifies the transmitter's role significantly.\\

Requiring an additional constraint on the group $\Omega$ to satisfy commutativity will further simplify the transmission process. Commutativity is the condition that all elements in the group commute, namely
\begin{equation}
\V{V}_l\V{V}_{l'}=\V{V}_{l'}\V{V}_l\:\:\:\forall\:l,l'\in\{0,..,L-1\}.
\end{equation}
If commutativity is satisfied then the group is called an abelian or a commutative group. The advantage of imposing  the commutative property can be described as follows. Since the group members $\V{V}_0,...,\V{V}_{L-1}$ are unitary, they can be eigendecomposed as
\begin{equation*}
 \V{V}_l=\V{X}_l\V{\Lambda}_l\V{X}_l^{-1},
\end{equation*}
where $\V{X}_l$ is the matrix of eigenvectors of $\V{V}_l$ which is also unitary and $\V{\Lambda}_l$ is the matrix of eigenvalues. Now since $\V{V}_0,...,\V{V}_{L-1}$ commute, they share the same set of eigenvectors, namely
 \begin{equation}
  \V{X}_0=\V{X}_1=...=\V{X}_{L-1}\overset{\underset{\Delta}{}}{=}\V{X},
 \end{equation}
and therefore $\V{V}_l=\V{X}\V{\Lambda}_l\V{X}^{-1}$ $\forall\:l\in\{0,..,L-1\}$. Furthermore, premultiplying or postmultiplying all constellation members by unitary matrices doesn't change the properties of the codebook in terms of error performance. So if all group members are premultiplied by the unitary matrix $\V{X}^{-1}$ and postmultiplied by the unitary matrix $\V{X}$, then such a transform results in an equivalent group with the group members
\begin{equation}
 \V{V}_l\longrightarrow\V{X}^{-1}\V{V}_l\V{X}=\V{X}^{-1}(\V{X}\V{\Lambda}_l\V{X}^{-1})\V{X}=\V{\Lambda}_l,
\end{equation}
which are diagonal matrices $\forall\:l$. In other words, imposing the commutative property on the group means that we can restrict ourselves to codebooks with diagonal matrices having the form
\begin{equation}
\V{V}_l=
 \begin{bmatrix}
\lambda_{l1}& & \text{\huge{0}} \\
& \ddots &  \\
\:\text{\huge{0}}& &\lambda_{lM}
\end{bmatrix}.
\end{equation}
Since the transmit matrices $\V{S}_\tau$ are also elements in the group, then they also have a diagonal structure implying that only one antenna is active at a time. For this reason, the transmit signals are termed as \emph{diagonal} signals.\\

One simple way of constructing an abelian group is to make it cyclic, meaning that $\V{V}_l$ is constructed as
\begin{equation}
\label{eq:v_l_G_l}
 \V{V}_l=\V{G}^{l}\:\:\:\forall\:l\in\{0,..,L-1\},
\end{equation}
where $\V{G}$ is defined as the generator matrix of the cyclic codebook $\Omega$ since it can be used to generate all codewords. In this case, 
\begin{equation}
\label{eq:G_l1_l2}
 \V{V}_{l}\V{V}_{l'}=\V{G}^{l}\V{G}^{l'}=\V{G}^{(l+l')},
\end{equation}
and since $\Omega$ is a finite group of length $L$, all its members (including $\V{G}$) must be $L^{\text{th}}$ root of unity, i.e. $\V{V}_l^L=\V{I}_M\:\:\forall\:l\in\{0,..,L-1\}$. Therefore, $\V{G}^{(l+l')}$ in (\ref{eq:G_l1_l2}) is the same as $\V{G}^{(l+l')\,\text{mod}\, L}$. Consequently, due to the cyclic property of the code, the index operator $\oplus$ in (\ref{eq:modulo_operator}) becomes
\begin{equation}
\label{eq:modulo_addition_cyclic}
 l''=(l+l')\,\text{mod}\,L.
\end{equation}
As a result, the transmitter does not even need a lookup table to calculate the resulting index of the next transmit member, it only performs the modulo addition operation defined in (\ref{eq:modulo_addition_cyclic}). This is analogous to the modulo addition of the angles of the phasors in DPSK as shown
\begin{eqnarray*}
 v_lv_{l'}=&\Big(e^{\frac{j2\pi}{L}}\Big)^l\Big(e^{\frac{j2\pi}{L}}\Big)^{l'}&=\Big(e^{\frac{j2\pi}{L}}\Big)^{(l+l')\,\text{mod}\,L}\\
\V{V}_l\V{V}_{l'}=&\V{G}^l\V{G}^{l'}&=\V{G}^{(l+l')\,\text{mod}\,L}
\end{eqnarray*}
Therefore the generator matrix $\V{G}$ is analogous to the basic DPSK phasor $e^{\frac{j2\pi}{L}}$. Hence, the generator matrix can be defined as 
\begin{equation}
\left. \V{G}=
 \begin{bmatrix}
e^{\frac{j2\pi}{L}u_1}& & \text{\huge{0}} \\
& \ddots &  \\
\:\text{\huge{0}}& &e^{\frac{j2\pi}{L}u_M}
\end{bmatrix}
\right.  \hspace{1cm}\begin{aligned} &u_{m}\in\{0,...,L-1\},\\
&m=1,...,M.\end{aligned}
\label{eq:generator_matrix}
\end{equation}
which is obviously an $L^{\text{th}}$ root of unity. In the special case when $\textstyle M=1$, the matrix collapses to one exponential which is the phasor base of the DPSK alphabet. $\V{G}$ is the building block used to construct the whole codebook using (\ref{eq:v_l_G_l}). Therefore the design of the codebook is based on the design of the generator matrix $\V{G}$ which is solely based on the design of the exponents $u_1,...,u_M$. Such exponents can be combined in one vector $\V{u}$ as
\begin{equation}
 \V{u}=[u_1,...,u_M].
\end{equation}

The question now is what is the optimal choice of vector $\V{u}$ that optimizes the error performance of the DUSTM scheme? To answer this question, we need to refer back to the design criteria defined in Section \ref{s:Design_Criteria} which was derived for a unitary STC scheme. Let the high SNR range be the operating range of interest, in this case one needs to check whether or not the full diversity condition is satisfied and how much coding gain can be achieved. According to the discussion on (\ref{eq:CG_full_diversity}), the distance matrix $\V{D}_{ll'}=\V{V}_{l'}-\V{V}_l$ needs to be investigated. When the codebook $\Omega$ is a group, it suffices (without loss of generality) to consider the distance matrix between any two different codewords, for example between $\V{V}_{l'}\eq \V{V}_0\eq \V{I}_M$ and $\V{V}_l\eq \V{G}^{l}$ for any $l\in\{1,...,L-1\}$. Therefore,
\begin{equation}
 \V{D}_{l0}=\V{I}_M-\V{G}^{l}= \begin{bmatrix}
1-e^{\frac{j2\pi l}{L}u_1} & & \text{\huge{0}} \\
& \ddots &  \\
\:\text{\huge{0}}& & 1-e^{\frac{j2\pi l}{L}u_M}
\end{bmatrix}, \:\:\:\forall\, l\in\{1,...,L-1\}.
\end{equation}

Since $l>0$, all diagonal elements in $\V{D}_{l0}$ are non-zero. Therefore matrix $\V{D}_{l0}$ is of full rank. Due to the group structure, this is also valid for all distance matrices $\V{D}_{ll'}\:\forall\:l\neq l'\in\{0,..,L-1\}$. Therefore the DUSTM with the cyclic group structure achieves full diversity. Hence, the design criterion of interest is maximizing the coding gain. For this we need to calculate the eigenvalues of the squared distance matrix $\V{D}_{l0}\hr\V{D}_{l0}$ which is the same as the square of the singular values of the distance matrix $\V{D}_{l0}$ using theorem \ref{theo:EV_SV2}. Namely,
\begin{eqnarray}
\label{eq:eig_val_Dl0}
 \lambda_m(\V{D}_{l0}\hr\V{D}_{l0})&=&\sigma_m^2(\V{D}_{l0})\nonumber\\
&=&\sigma_m^2(\V{I}_M-\V{G}^{l})\nonumber\\
&=& (1-e^{\frac{j2\pi l}{L}u_m})^*(1-e^{\frac{j2\pi l}{L}u_m})\nonumber\\
&=& 1-e^{\frac{j2\pi l}{L}u_m}-e^{-\frac{j2\pi l}{L}u_m} +1\nonumber\\
&=& 2-2\cos(\frac{2\pi lu_m}{L})\nonumber\\
&=& 4\sin^2(\frac{\pi lu_m}{L})
\end{eqnarray}
Therefore the coding gain in (\ref{eq:CG_full_diversity}) reduces to
\begin{eqnarray}
 \text{CG}&=&\mini\limits_{l=1,...,L-1}\:\,\Big(\prod\limits_{m=1}^{M}{\lambda_{l0m}}\Big)^{\frac{1}{M}}\nonumber\\
&=& \mini\limits_{l=1,...,L-1}\:\,\Big(\prod\limits_{m=1}^{M}{4\sin^2(\frac{\pi l}{L}u_m)}\Big)^{\frac{1}{M}}
\end{eqnarray}
To maximize the coding gain, it is required to search for the optimal vector $\V{u}$ that satisfies
\begin{equation}
\label{eq:design_criterion_DUSTM}
 \V{u}_{\text{opt}}=\argmax\limits_{\V{u}=[u_1,...,u_M]}\:\,\mini\limits_{l=1,...,L-1}\Big(\prod\limits_{m=1}^{M}{\sin^2(\frac{\pi l}{L}u_m)}\Big)^{\frac{1}{M}}
\end{equation}
\renewcommand{\arraystretch}{1} 
\begin{table}
\begin{center}
\caption{Optimal $\V{u}$ vectors for the DUSTM scheme with diagonal signals. \\See \cite[Table I]{Hochwald2000}}
\label{tab:opt_u}
  \begin{tabular}{|m{1cm}|m{1cm}|m{2cm}|m{5cm}|}
    \hline    
      $M$	 &  $\eta$ & $L=2^{\eta M}$ & $\V{u}=[u_1,...,u_M]$ \tn[5pt] \hline 
	1 & 1 & 2& [1] \tn \hline
	2 & 1 & 4& [1 1] \tn \hline
	3 & 1 & 8& [1 1 3] \tn \hline
	4 & 1 & 16& [1 3 5 7] \tn \hline
	5 & 1 & 32& [1 5 7 9 11] \tn \hline \hline
	1 & 2 & 4&[1] \tn \hline
	2 & 2 & 16&[1 7] \tn \hline
	3 & 2 & 64&[1 11 27] \tn \hline
	4 & 2 & 256&[1 25 97 107] \tn \hline
	5 & 2 & 1024&[1 157 283 415 487] \tn \hline
\end{tabular}	
\end{center}
\end{table}
Solving (\ref{eq:design_criterion_DUSTM}) analytically is cumbersome, therefore the authors in \cite{Hochwald2000} performed exhaustive computer searches trying to find optimal $\V{u}$ vectors for different number of transmit antennas $M$ and different spectral efficiencies $\eta$ in bits/channel use. The optimal $\V{u}$ vectors for $\eta=1,2$ and $M=1,...,5$ are published in \cite[Table I]{Hochwald2000} and shown here in Table \ref{tab:opt_u}. From the definition of spectral efficiency in (\ref{eq:spectral_efficiency_STM}), the codebook size is $L\eq 2^{\eta M}$, i.e. every information matrix carries $\eta M$ bits.
However, which bit sequence is to be assigned to which constellation matrix for achieving optimal performance is not intuitive to see. The following analysis uses the code properties to infer a good bit-to-matrix mapping. As argued in \cite{Hochwald2000}, restricting $u_1,...,u_M$ to be relatively prime to $L$ doesn't change the code properties. Since the $L$ values considered here are all even, therefore $u_1,...,u_M$ can be made all odd  (see Table \ref{tab:opt_u}). Now, consider matrix $\V{V}_{L/2}$, its $m^{\text{th}}$ element is
\begin{equation}
 e^{j\frac{2\pi l}{L}u_m}=e^{j\frac{2\pi L/2}{L}u_m}=e^{j\pi u_m}=-1, \:\forall\: m=1,...,M.
\end{equation}
Therefore $\V{V}_{L/2}\eq -\V{I}_M\eq -\V{V}_0$. In this case, using (\ref{eq:eig_val_Dl0}), the distance matrix $\V{D}_{\frac{L}{2}0}$ has squared singular values 
\begin{equation}
\sigma_m^2(\V{D}_{\frac{L}{2}0})=4\sin^2(\frac{\pi Lu_m}{2L})=4\sin^2(\frac{\pi}{2} u_m)=4,
\end{equation}
which is the maximum possible. This shows that matrices $\V{V}_0$ and $\V{V}_{L/2}$ are maximally separated. Due to the group structure, this is also the case for all matrix pairs $\V{V}_l$ and $\V{V}_{l+L/2}$ $\forall\:l=0,...,L-1$. Therefore matrices separated by $L/2$ should be assigned bit sequences with maximum hamming distance, i.e. complementary bit assignment. An example for such bit-to-matrix assignment in the case of a codebook of size $L\eq 8$ is shown in Table \ref{tab:bit_map_L_8}.
\begin{table}[!ht]
\begin{center}
\caption{The bit-to-matrix assignment for a codebook of size $L\eq 8$}
\label{tab:bit_map_L_8}
\begin{tabular}{|m{1cm}|m{1cm}|m{2cm}|m{2cm}|}
    \hline 
$l$ & $ l+\frac{L}{2}$ &$\V{V}_{l}$ & $\V{V}_{l+\frac{L}{2}}$ \tn \hline
0 &4 &000 & 111 \tn \hline
1 &5 &001 & 110 \tn \hline
2 &6 &010 &101 \tn \hline
3 &7 &011 &100 \tn \hline
 \end{tabular}	
\end{center}
\end{table}
\vspace{-0.5cm}
\subsection{DUSTM Performance Analysis}
\label{ss:DUSTM_simulations}
Simulations have been made over a piece-wise constant Rayleigh flat fading channel with the transmission scheme summarized in Figure \ref{fig:DUSTM_transmission}. The optimal $ \V{u}$ vectors in table \ref{tab:opt_u} are used to construct the generator matrix $\V{G}$. The receiver uses the decision metric of unitary transmission in (\ref{eq:ML_Unitary_STC2}).\\

\begin{figure}[!ht]
 \centering
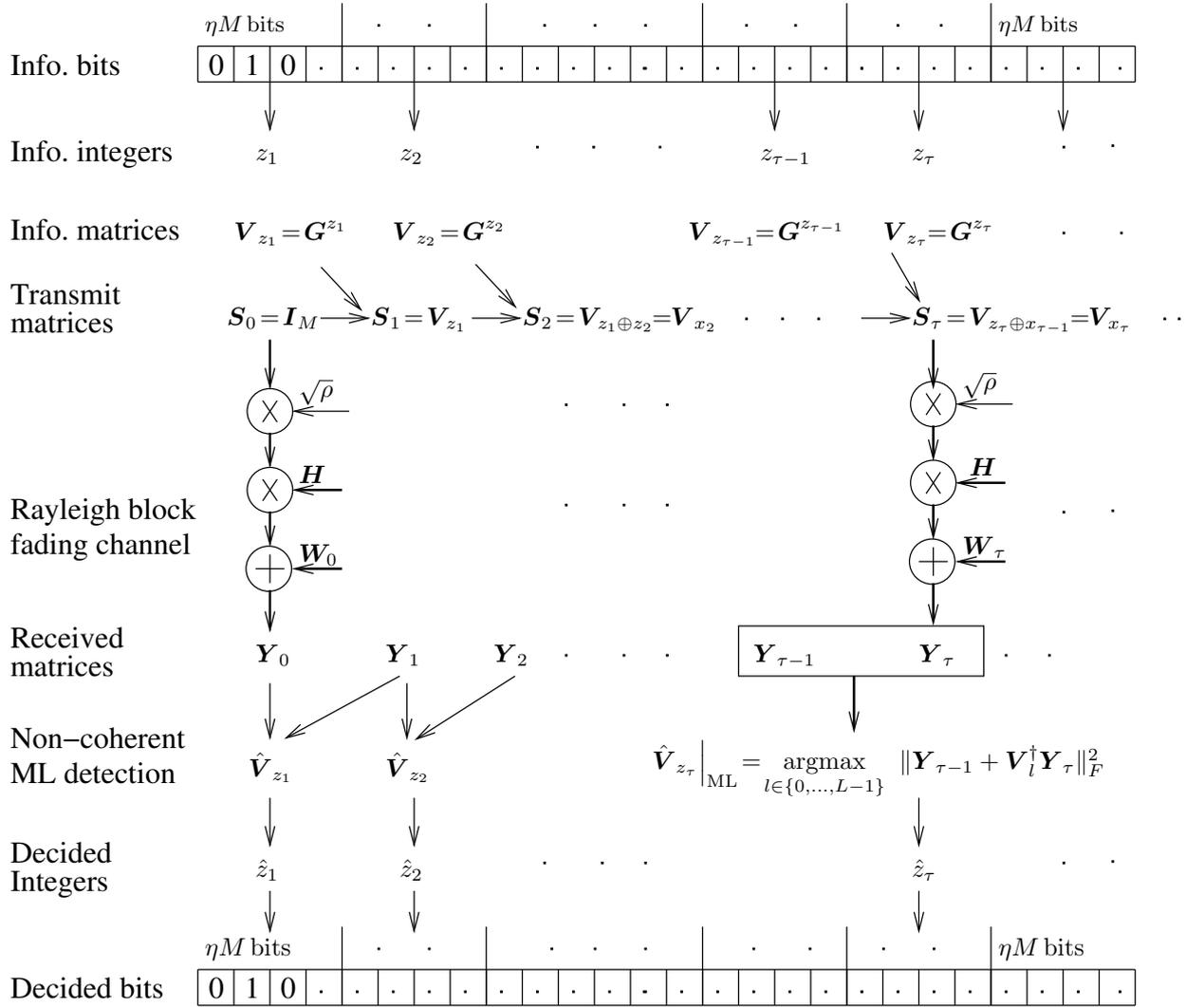
\caption{Transmission scheme of DUSTM with group cyclic code.}
\label{fig:DUSTM_transmission}
\end{figure}

The resulting BER curves for $M\eq 1,...,5\,$ transmit antennas and at $\eta=1,2$ bits/s/Hz are shown in Figure \ref{fig:DUSTM_BER_curves}.  We see that increasing the number of transmit antennas improves the performance only in the high SNR range, whereas in the low SNR range the performance is degraded with increasing $M$. This behaviour is expected due to the fact that the scheme considered is designed to enhance the error performance in the high SNR range. For example, if 2 bits/channel use are to be transmitted, then using the DUSTM diagonal scheme defined here is only meaningful for SNR $>$ 15 dB. \\

It is also clear that the slope of the curves increases with increasing $M$ resulting in diversity gain. However, the \emph{increase} in diversity gain itself decreases with increasing $M$. For example, going from $M\eq  1$ to $M\eq  2$, we gain a lot, but then the gain keeps decreasing until it does not make sense to increase $M$ anymore. This effect is shown when comparing the $M=4$ and $M=5$ curves at rate $\eta=1$, there is only a marginal insignificant improvement showing that the gain saturates at $M=4$ transmit antennas in case of transmitting 1 bit/channel use.\\
  
\begin{figure}[htp]
\centering
\subfloat[at $\eta\eq 1$ bit/channel use]{\label{subfig:DUSTM_R_1}\includegraphics[width=0.7\textwidth]{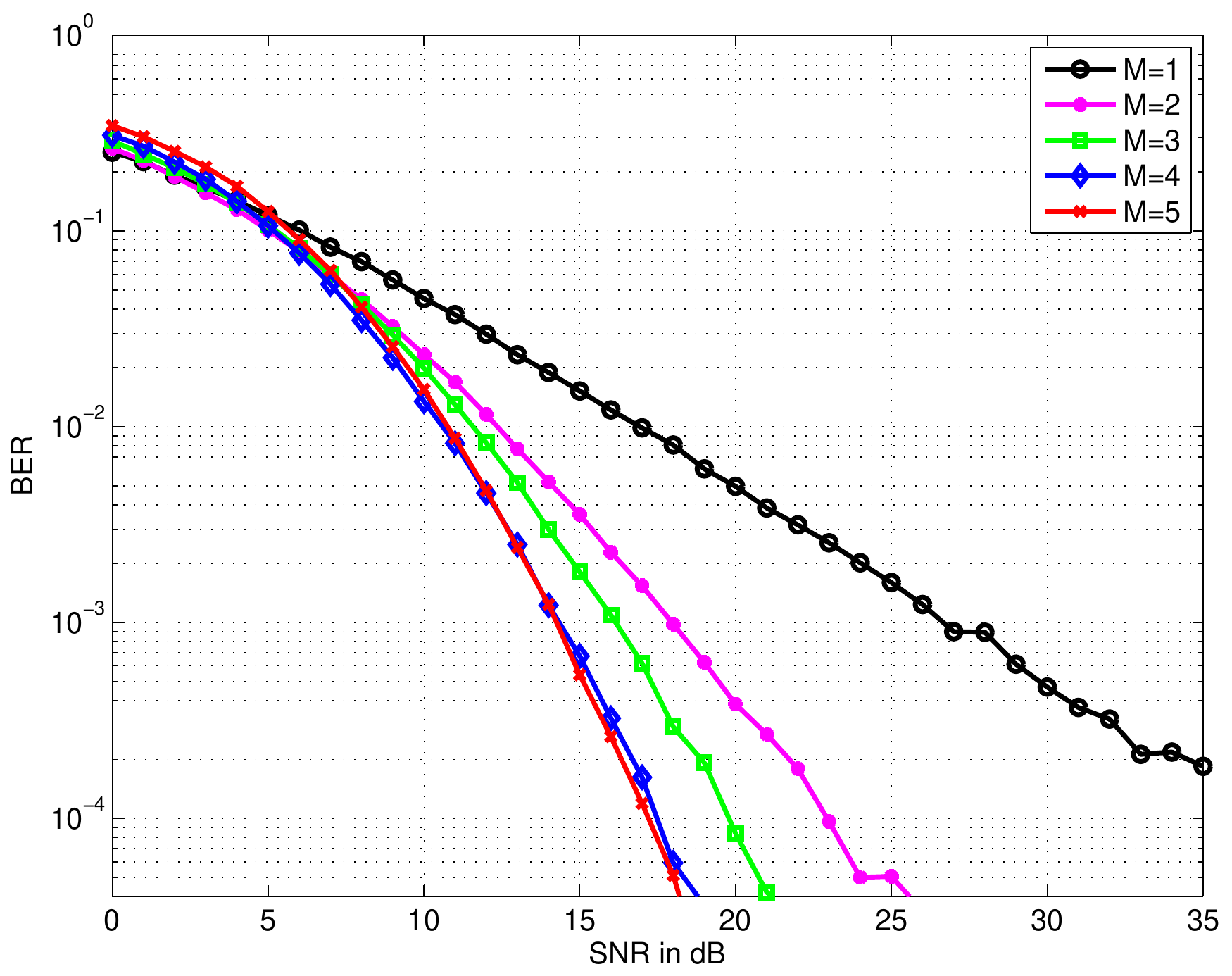}}\\
\subfloat[at $\eta\eq 2$ bits/channel use]{\label{subfig:DUSTM_R_2}\includegraphics[width=0.7\textwidth]{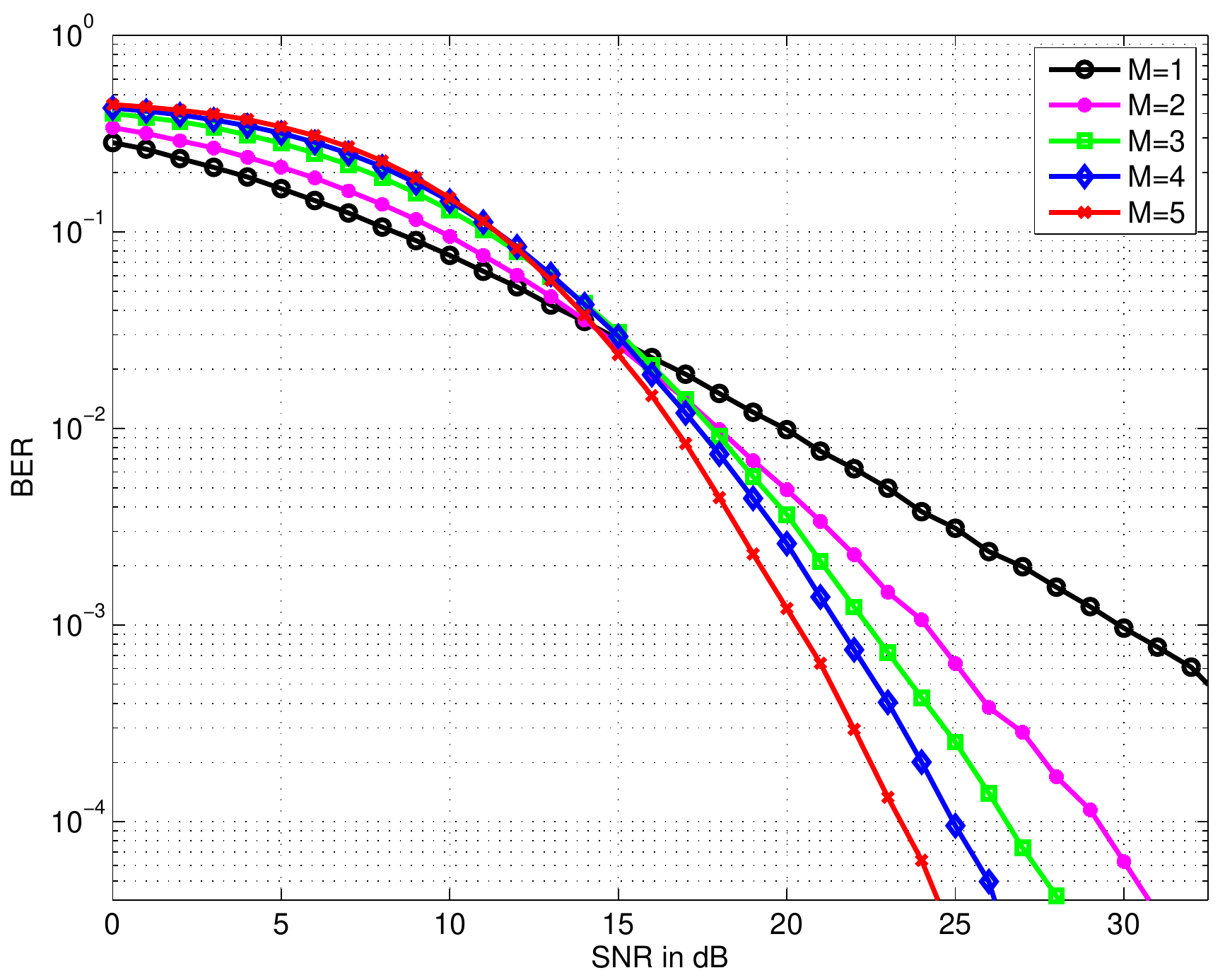}}
\caption{BER vs SNR ($\rho$) performance of the DUSTM scheme with diagonal cyclic design for $M$=1 to 5 transmit antennas and $N$=1 receive antenna}
 \label{fig:DUSTM_BER_curves}
\end{figure}

In conclusion, DUSTM is one possible way of realizing transmit diversity for non-coherent systems. It basically extends the differential single-antenna modulation scheme DPSK. Forcing the constellation to form a group dispenses the need of matrix multiplication for differential encoding,  which simplifies the transmitter's role. Although the design can be theoretically extended to any $\eta$ and $M$, it requires exhaustive computer searches to find optimal $\V{u}$ vectors for a certain $\eta$ and $M$, and therefore practically speaking the code is not easily extendable. A major disadvantage of the DUSTM scheme is that the receiver's complexity increases exponentially with $\eta$ and $M$. This is because the search space of the metric in (\ref{eq:ML_Unitary_STC2}) is the codebook size $L=2^{\eta M}$. \\

Furthermore, imposing the transmit matrices to have a diagonal structure has some practical perspective which is worth considering.  By definition of the transmit power normalization defined in (\ref{eq:power_st}), the total power transmitted by all transmit antennas is constant at any time slot. For diagonal signals, this means that  only one antenna delivers this total power every time slot. This makes the amplifiers connected to the transmit antennas deliver $M$ times the power they would otherwise deliver if all antennas were simultaneously operating. This in turn requires the amplifiers to have a large linear range making them more expensive to realize.  To avoid such a practical constraint, one may force a non-diagonal constellation using the transform $\V{V}_l\longrightarrow \V{U}\V{V}_l\V{U}\hr$ with any unitary matrix $\V{U}$. This will have the effect of distributing the transmit power over all $M$ antennas, leading to the use of cheap amplifiers without affecting the error performance.\\

\section{Orthogonal Space-Time Block Codes}
\label{s:OSTBC}
The previous section has shown one possible scheme for achieving transmit diversity by using unitary matrices that belong to a finite group code. Since the scheme maps bits directly to matrices (see the STM architecture in Figure \ref{subfig:DSTM}), the receiver needs to check all possible matrices belonging to the codebook to decide on the information bits. This results in an exponential increase in complexity with the spectral efficiency and the number of transmit antennas, making the scheme unfavourable for practical systems. This motivates the idea of converting bits first to symbols and then to matrices (see the STBC architecture in Figure \ref{subfig:DSTBC}) aiming at performing symbol based decoding rather than matrix based decoding. \\

The first space-time block coding scheme that adopted a bit-to-symbol-to-matrix mapping was proposed by Alamouti in \cite{Alamouti98}. 
The scheme was designed for two transmit antennas, and proved  to provide the same diversity order as that achieved by the MRC $1\!\times\!2$ SIMO system, i.e. a diversity order of 2. The scheme is remarkable in that it requires only \emph{linear} processing at the receiver, which is of substantially less complexity compared to that required by the STM scheme presented in Section \ref{s:DSTM}. Owing to its implementation simplicity and good performance, Alamouti's scheme has been adopted in 3G mobile technology standards like CDMA2000 and W-CDMA as well as in WIMAX technology which allows wireless broadband Internet access. In fact, Alamouti's two transmit antenna scheme is considered as a special case of a class of space-time codes known as Orthogonal Space-Time Block Codes (OSTBCs) which is defined in \cite{OSTBC_Tarokh99} for arbitrary number of transmit antennas.\\

This section starts by introducing OSTBCs through Alamouti's scheme and then shows the extension of OSTBCs to more than two transmit antennas. The section includes the case when information symbols are drawn from an equal-energy constellation (PSK symbols) resulting in the so-called unitary OSTBCs, as well as the case of using non-constant envelope constellations (like QAM alphabet) resulting in non-unitary OSTBCs. In each case, the differential non-coherent receiver decision metric is derived and the error performance is analyzed.
\subsection{Unitary OSTBCs}
\label{ss:USTBC}
The scheme proposed by Alamouti considers a MISO communication system with two transmit antennas and one receive antenna in an environment modeled by flat Rayleigh fading channel. The scheme was first designed with the assumption of perfect channel knowledge at the receiver and therefore the transmit symbols are the same as the information symbols and the detection is done coherently. After introducing OSTBCs in the coherent domain, we will show the possibility of using OSTBCs in non-coherent systems.
\subsubsection{Code Construction}
\label{sss:Unitary_code_construction}
Alamouti's scheme can be basically described as follows: Two information symbols $x_1$ and $x_2$ are buffered by the transmitter and transmitted in two time slots in the following manner.
In the first time slot, symbol $x_1$ is transmitted over the first antenna and simultaneously symbol $x_2$ is transmitted over the second antenna. In the second time slot, signal $-x_2^*$ is transmitted by the first antenna and $x_1^*$ is transmitted by the second antenna. In this case the transmit matrix $\V{S}$ described in Figure \ref{fig:STC_Tx_matrix} will have the form;
\begin{equation}
\label{eq:Alamouti}
\V{S}=\frac{1}{\sqrt{2}}
 \begin{bmatrix}
  x_1 & x_2\\
-x_2^*& x_1^*
 \end{bmatrix}.
\end{equation}\\

The transmit streams over the two antennas are mutually orthogonal since the columns of $\V{S}$ are orthogonal. Furthermore, Alamouti considered the case when symbols $x_1$ and $x_2$ are drawn from a constant envelope constellation (PSK) making the power of all antenna streams (columns of $\V{S}$) constant. Assuming a unit energy constellation (i.e. $|x|^2\eq1\, \forall\,x\in$ PSK alphabet), matrix $\V{S}$ is made unitary by including the factor $\nicefrac{1}{\sqrt{2}}$  making $\V{S}\hr\V{S}\eq\frac{|x_1|^2+|x_2|^2}{2}\V{I}_{2}\eq\V{I}_{2}$, hence the term unitary transmission. This subsection includes only the case of unitary signal constellation, namely using PSK symbol alphabet.\\

Since two independent symbols are transmitted over two time slots, Alamouti's scheme achieves full rate (i.e. in (\ref{eq:code_rate}) $R\eq 1$). It has been proved by Tarokh et al. in \cite{OSTBC_Tarokh99} that Alamouti's scheme is unique in that it is the only square OSTBC that achieves full rate using arbitrary complex signal constellation. They also proved that using real signal constellation such as PAM (Pulse Amplitude Modulation), full-rate OSTBCs with any number of transmit antennas exist. Real signal transmission is however not of much interest since in practical systems, the bandwidth needs to be best utilized by making best use of the complex space. For this reason, a study in \cite{OSTBC_Rate_upper_bounds} provides upper bounds on the achievable code rate for OSTBCs with more than two transmit antennas using arbitrary \emph{complex} signal constellation. The authors proved that the code rate of complex OSTBCs for three or more transmit antennas is upper-bounded by \nicefrac{3}{4}. Furthermore, they conjectured a tighter upper bound for the code rate of an OSTBC with $M$ transmit antennas to be
\begin{equation}
\label{eq:code_rate_conjecture}
 R\leq\frac{\left\lceil\frac{M}{2}\right\rceil+1}{2\vspace{0.2mm}\left\lceil\frac{M}{2}\right\rceil}.
\end{equation}\\

In the literature, several codes for up to $M\eq 5$ transmit antennas achieve the upper bound in (\ref{eq:code_rate_conjecture}) with equality \cite{OSTBC_Rate_upper_bounds}. For example, in the four transmit antenna system, an example of an OSTBC which achieves a code rate of $\nicefrac{3}{4}$ (the upper bound in (\ref{eq:code_rate_conjecture}) at $M\eq 4$) is a code proposed by Tirkkonen and Hottinen in \cite[eq. B.27]{Tirkkonen_Hottinen_2003} whose code matrix has the form
\begin{equation}
\label{eq:T_H}
\V{S}=\frac{1}{\sqrt{3}}
 \begin{bmatrix}
  x_1 & x_2 & x_3 & 0\\
-x_2^*& x_1^* & 0& -x_3\\
-x_3^* & 0& x_1^* & x_2\\
0&x_3^*&-x_2^* &x_1
 \end{bmatrix}.
\end{equation}
where $x_1$, $x_2$, and $x_3$ are independent complex information symbols, assumed here to be drawn from a unit energy PSK alphabet. This code matrix will be shortly referred to as T-H code. As shown, only three symbols are transmitted in four time slots resulting in a code rate of $\nicefrac{3}{4}$. Clearly matrix $\V{S}$ is unitary since its rows and also its columns are orthonormal.\\

Having shown examples for OSTBCs in two and four transmit antenna systems, this subsection continues by describing OSTBCs with an arbitrary number of transmit antennas $M$. Furthermore, STBCs are in general easily extendable to additionally employ multiple receive antennas. For example, Alamouti has shown in \cite{Alamouti98} that the extension from a $2\!\times\!1$ system to a $2\!\times\!2$ system is straightforward. For the sake of generality, the systems considered here are $M\!\times\! N$ MIMO systems.\\

Although OSTBCs were first proposed in coherent systems, they are also applicable for non-coherent systems. This will be proved in this subsection, where non-coherent detection is realized through the use of differential encoding. In \cite{Tarokh_DOSTBC_2000}, Tarokh and Jafarkhani presented a differential detection scheme for  Alamouti's code. For precise use of terms, (\ref{eq:Alamouti}) and (\ref{eq:T_H}) will in general be referred to as code matrices, which are transmit matrices (denoted by $\V{S}$) for coherent systems, but information matrices (denoted by $\V{V}$) for differential non-coherent systems. Since non-coherent systems are of interest as far as this thesis is concerned, in the following, the code matrices will be given the notation $\V{V}$. Furthermore, in the context of differential encoding, we restrict ourselves only to the number of transmit antennas where a square code matrix exist (i.e. $T\eq M$).\\

In general, an STBC encodes $K$ complex symbols $x_1,...,x_K$ by including linear combinations of $\pm x_1,...,\pm x_K$ and their conjugates $\pm x_1^*,...,\pm x_K^*$ in the code matrix. This can be expressed in a compact form as
\begin{equation}
\label{eq:dispersive_form}
 \V{V}=
\frac{1}{\sqrt{p}}\sum\limits_{i=1}^{K}{\V{A}_ix_i+\V{B}_ix_i^*},
\end{equation}
where $\V{A}_i$ and $\V{B}_i\,\forall i\in\{1,...,K\}$ are known as dispersion matrices since they disperse the symbols over the transmit antennas.  $p$ is a normalization factor used to make $\V{V}\hr\V{V}\eq\V{V}\V{V}\hr\eq\V{I}_M$. This also satisfies the power constraint in (\ref{eq:power_st}). As has been shown, $p\eq2$ in Alamouti's code in (\ref{eq:Alamouti}) and $p\eq3$ in T-H code matrix in (\ref{eq:T_H}).\\

Using the dispersive form in (\ref{eq:dispersive_form}), Alamouti's code matrix can be written as
\begin{equation*}
 \V{V}=\frac{1}{\sqrt{2}}\bigg(x_1\begin{bmatrix}1&0\\0&0\end{bmatrix}+x_1^*\begin{bmatrix}0&0\\0&1\end{bmatrix}+x_2\begin{bmatrix}0&1\\0&0\end{bmatrix}+x_2^*\begin{bmatrix}0&0\\-1&0\end{bmatrix}\bigg),
\end{equation*}
i.e. the dispersion matrices are
\begin{equation}
 \V{A}_1=\begin{bmatrix}1&0\\0&0\end{bmatrix},\,
\V{B}_1=\begin{bmatrix}0&0\\0&1\end{bmatrix},\,
\V{A}_2=\begin{bmatrix}0&1\\0&0\end{bmatrix},\,
\V{B}_2=\begin{bmatrix}0&0\\-1&0\end{bmatrix}.
\end{equation}
For the T-H four-antenna code matrix defined in (\ref{eq:T_H}), the dispersion matrices are
\begin{eqnarray}
\V{A}_1=\begin{bmatrix}1&0&0&0\\0&0&0&0\\0&0&0&0\\0&0&0&1\end{bmatrix},\,
\V{B}_1\hspace{-2mm}&=\hspace{-2mm}&\begin{bmatrix}0&0&0&0\\0&1&0&0\\0&0&1&0\\0&0&0&0\end{bmatrix},\,
\V{A}_2=\begin{bmatrix}0&1&0&0\\0&0&0&0\\0&0&0&1\\0&0&0&0\end{bmatrix},\,
\V{B}_2=\begin{bmatrix}0&0&0&0\\-1&0&0&0\\0&0&0&0\\0&0&-1&0\end{bmatrix},\nonumber\\
\V{A}_3\hspace{-2mm}&=\hspace{-2mm}&\begin{bmatrix}0&0&1&0\\0&0&0&-1\\0&0&0&0\\0&0&0&0\end{bmatrix},
\V{B}_3=\begin{bmatrix}0&0&0&0\\0&0&0&0\\-1&0&0&0\\0&1&0&0\end{bmatrix}.
\end{eqnarray}\\

Another alternative dispersive form for code matrices is defined based on the real and imaginary parts of the symbols instead of the symbols and their complex conjugate. The dispersion matrices of the alternative form are denoted as $\V{U}_i$ and $\V{Q}_i$ and matrix $\V{V}$ can be written as
\begin{equation}
\label{eq:dispersive_form2}
 \V{V}=
\frac{1}{\sqrt{p}}\sum\limits_{i=1}^{K}{\V{U}_ix_i^R+j\V{Q}_ix_i^I}
\end{equation}
where $x_i^R$ and $x_i^I$ are the real and the imaginary parts of symbol $x_i$, respectively. While the dispersion formations in (\ref{eq:dispersive_form}) and (\ref{eq:dispersive_form2}) are general for any STBC, the dispersion matrices for \emph{orthogonal} STBCs exhibit several special properties. The following properties hold for the dispersion matrices $\V{U}_i$ and $\V{Q}_i$ of any OSTBC \cite{Ganesan_Stoica_2001}.
\begin{equation}
\begin{aligned}
\label{eq:OSTBC_disp_matrices_properties}
\text{(i)}\,\,&\V{U}_i\hr\V{U}_i=\V{I}_M, \: \V{Q}_i\hr\V{Q}_i=\V{I}_M & 1\leq i\leq K\\
\text{(ii)}\,\,&\V{U}_i\hr\V{U}_d=-\V{U}_d\hr\V{U}_i, \:\: \V{Q}_i\hr\V{Q}_d=-\V{Q}_d\hr\V{Q}_i &  \:\:1\leq i\neq d\leq K\\
\text{(iii)}\,\,&\V{U}_i\hr\V{Q}_d=\V{Q}_d\hr\V{U}_i & 1\leq i,d\leq  K
\end{aligned}
\end{equation}
The above properties can be used to show that any OSTBC achieves full diversity. For a detailed proof of this, refer to Section \ref{s:full_diversity_OSTBCs}.\\

If all possible combinations of the $K$ symbols from their respective alphabets are substituted in the code matrix, we result in all possible code matrices $\V{V}_l \,\,\,\forall\, l\in\{0,...,L-1\}$ that construct a codebook $\Omega$. Based on the code construction, the total number of code matrices $L$ is $\prod\limits_{i=1}^{K}{q_i}$, where $q_i$ is the alphabet size of $\mathcal{A}_i$ from which symbol $x_i$ is drawn. This implies that the $K$ symbols in general belong to different alphabets. Since the information matrices $\V{V}_l$ are unitary $\forall\, l$, the differential encoding equation can be defined as in (\ref{eq:Diff_encoding_unitary2}) in Section \ref{s:DSTM}. Namely, if at block index $\tau-1$, $\V{S}_{\tau-1}$ is transmitted, then at the next block index $\tau$, the transmit matrix $\V{S}_\tau$ 
is constructed as
\begin{equation}
\label{eq:Diff_encoding_unitary_OSTBC}
\V{S}_\tau=\V{V}_{z_\tau}\V{S}_{\tau-1},
\end{equation}
where $z_\tau$ is the information integer corresponding to the $\sum\limits_{i=1}^{K}{\log_2q_i}$ information bits buffered by the transmitter at block index $\tau$. Based on this, the code matrix with index $z_\tau$ is drawn from the codebook $\Omega$ to form the information matrix  $\V{V}_{z_\tau}$. In STBCs, this is equivalent to saying that the $K$ generated symbols $x_1,...,x_K$ populate the code matrix forming the information matrix $\V{V}_{z_\tau}$. Since all information matrices $\V{V}_l$ are unitary matrices $\forall \,l\in\{0,...,L-1\}$, then by initializing the transmission with a unitary matrix $\V{S}_0$, it follows directly that all transmit matrices $\V{S}_{\tau}$ are unitary $\forall\,\tau$. Unlike the STM based on finite group codes defined in Section \ref{s:DSTM}, the transmit matrices in STBCs in general do \emph{not} belong to a finite group. \\
\subsubsection{ML Differential Decoder}
\label{sss:Unitary_ML_metric}
In the following we derive the ML receiver metric for non-coherent $M\!\times\!N$ MIMO systems that use unitary OSTBCs and experience quasi-static Rayleigh flat fading channel. In Section \ref{s:ML_STC_metric}, the ML decision metric for the special case of unitary transmission in the same channel conditions has been derived in (\ref{eq:ML_unitary2}) to be
\begin{equation*}
 \hat{\V{V}}_{z_\tau}\Big|_{\text{ML}}=\argmax\limits_{\V{V}_l\in\Omega}\:\,\Re\{\tr(\V{V}_l\V{Y}_{\tau-1}\V{Y}_{\tau}\hr)\}
\end{equation*}
where the decision is based on two consecutive received matrices $\V{Y}_{\tau-1}$ and $\V{Y}_{\tau}$. This metric was used in the STM scheme, where all possible candidate matrices $\V{V}_l$ in the codebook $\Omega$ are examined to decide on the most probable information matrix. Now in STBCs, the symbol-to-matrix mapping using the dispersive construction in (\ref{eq:dispersive_form}) can be used to modify the decision metric to
\begin{equation}
 \hat{\V{V}}_{z_\tau}\Big|_{\text{ML}}=\argmax\limits_{x_i\in\mathcal{A}_i}\:\,\Re\{\tr\Big((\frac{1}{\sqrt{p}}\sum\limits_{i=1}^{K}{\V{A}_ix_i+\V{B}_ix_i^*})\V{Y}_{\tau-1}\V{Y}_{\tau}\hr\Big)\},
\end{equation}
where $x_i$ is a candidate in the alphabet $\mathcal{A}_i$ from which the $i^{\text{th}}$ symbol in the code matrix is drawn. Since $\Re\{.\}$, $\tr\{.\}$, and $\sum\{.\}$ are linear operators, they are interchangeable and the metric can be further simplified to
\begin{equation}
\label{eq:ML_unitary_OSTBC}
 \hat{\V{V}}_{z_\tau}\Big|_{\text{ML}}=\argmax\limits_{x_i\in\mathcal{A}_i}\:\,\sum\limits_{i=1}^{K}{\Re\{\tr(\V{A}_i\V{Y}_{\tau-1}\V{Y}_{\tau}\hr x_i +\V{B}_i\V{Y}_{\tau-1}\V{Y}_{\tau}\hr x_i^*)\}}
\end{equation}
Taking the complex conjugate of any of the two terms in (\ref{eq:ML_unitary_OSTBC}) will not change the metric since at the end, only the real part is considered. Hence the metric can be rewritten as
\begin{eqnarray}
\label{eq:ML_unitary_OSTBC2}
 \hat{\V{V}}_{z_\tau}\Big|_{\text{ML}}&=&\argmax\limits_{x_i\in\mathcal{A}_i}\:\,\sum\limits_{i=1}^{K}{\Re\{\tr(\V{A}_i\V{Y}_{\tau-1}\V{Y}_{\tau}\hr\, x_i +\V{B}_i^*\V{Y}_{\tau-1}^*\V{Y}_{\tau}^T\, x_i)\}}\nonumber\\
&=&\argmax\limits_{x_i\in\mathcal{A}_i}\:\,\sum\limits_{i=1}^{K}{\Re\{\underbrace{\tr(\V{A}_i\V{Y}_{\tau-1}\V{Y}_{\tau}\hr +\V{B}_i^*\V{Y}_{\tau-1}^*\V{Y}_{\tau}^T)}_{\tilde{x}_i}x_i\}}\nonumber\\
&=&\argmax\limits_{x_i\in\mathcal{A}_i}\:\,\sum\limits_{i=1}^{K}{\Re\{\tilde{x}_ix_i\}}
\end{eqnarray}
The metric in (\ref{eq:ML_unitary_OSTBC2}) is therefore splittable among the $K$ symbols $x_1,...,x_K$. And the final differential non-coherent ML decision metric for unitary OSTBCs is
\begin{equation}
\label{eq:Fast_ML_Unitary_OSTBC}
 \boxed{\hat{x}_i\Big|_{\text{ML}}=\argmax\limits_{x_i\in\mathcal{A}_i}\:\,\Re\{\tilde{x}_ix_i\}.}
\end{equation}
This is the most remarkable advantage of unitary OSTBCs, where the ML decision on the different data symbols can be decoupled resulting in a significant reduction in complexity. Namely, the search space for the $K$ symbols is $\sum_{i=1}^{K}{q_i}$ instead of $\prod_{i=1}^{K}{q_i}$. If all symbols are drawn from the same alphabet of size $q$, then the search space complexity is $qK$ rather than $q^K$, i.e. the complexity increases only linearly with $K$ rather than exponentially. A decoder like in (\ref{eq:Fast_ML_Unitary_OSTBC}) decides on one complex symbol at a time, and therefore it is said to perform Single Complex Symbol Decoding (SCSD).  Since the $K$ symbols are independent, they can be decoded in parallel resulting in further reduction in the decoding delay.  Owing to such complexity and delay reduction, the metric in (\ref{eq:Fast_ML_Unitary_OSTBC}) is sometimes termed as fast ML decoding.\\

Figure \ref{fig:DSTBC_SCSD} describes the transmission of a differential STBC with $M$ transmit antennas and $N$ receive antennas. The block diagram is a more detailed description of the STBC architecture shown in Figure \ref{subfig:DSTBC}, with the receiver performing SCSD. All $h_{mn}\:\forall\,m=1,...,M$ and $n=1,...,N$ are independent quasi-static flat Rayleigh fading channels. The ST coder block populates the code matrix $\V{V}_{z_\tau}$ with the $K$ symbols $x_1,...,x_K$, and therefore $v_{tm},\,t=\{1,...,T\},$ and $\,m=\{1,...,M\}$ depend on the code matrix used. In unitary transmission, $\mathcal{A}_i$ for $i=1,...,K$ are PSK alphabets.\\
\renewcommand{\arraystretch}{0.8} 
\setlength{\arraycolsep}{1.5pt}
\begin{figure}
 \centering
\hspace{-1.3cm}
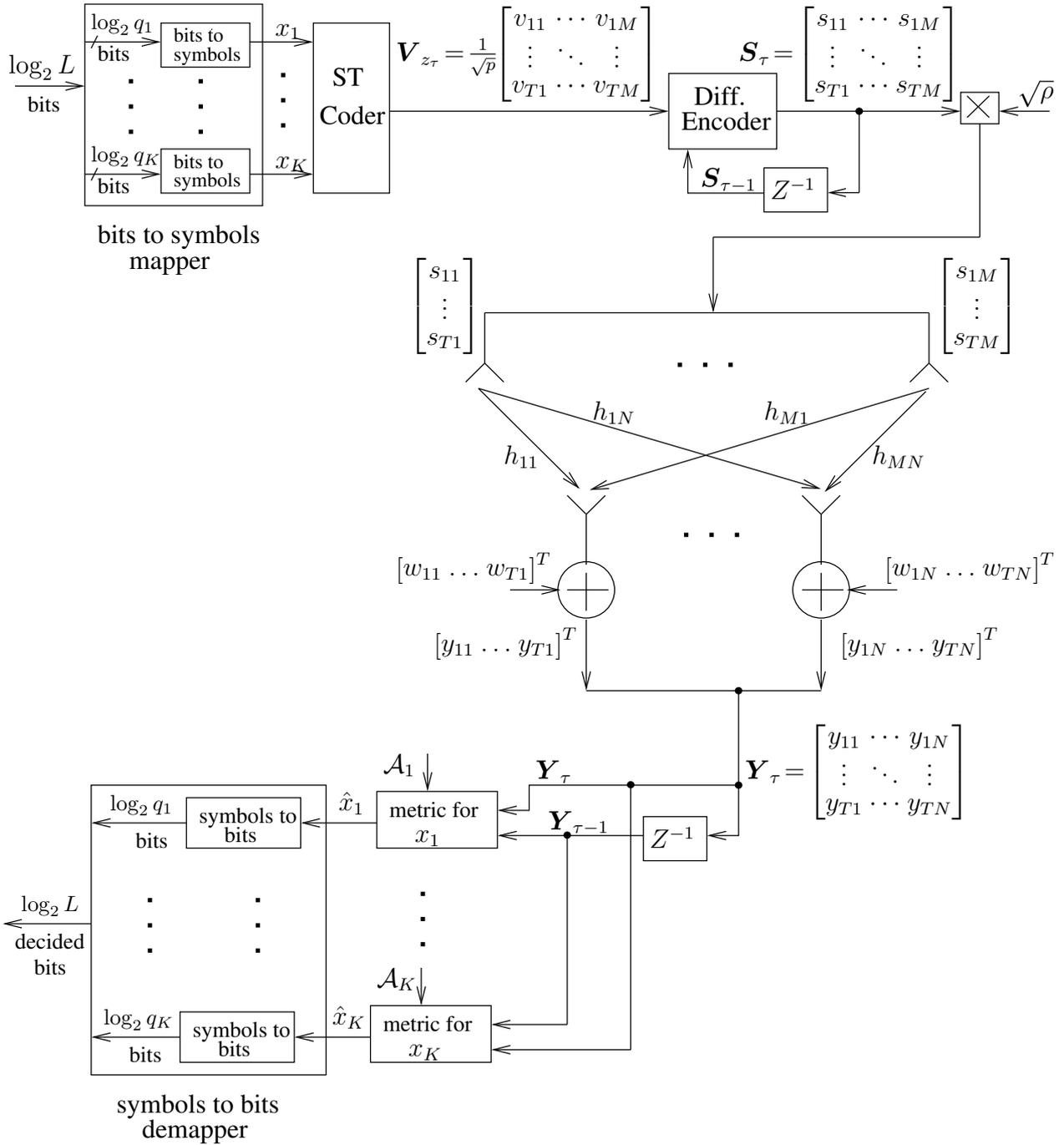
\caption{A block diagram showing the differential transmission of an STBC scheme over an $M\!\times\!N$ MIMO system with a non-coherent receiver that performs SCSD.}
\label{fig:DSTBC_SCSD}
\end{figure}
\renewcommand{\arraystretch}{1.2} 
\setlength{\arraycolsep}{5pt}

\subsubsection{Performance Analysis}
\label{sss:Unitary_performance_analysis}
Since all transmit diversity schemes covered in this thesis are easily extendable to include multiple receive antennas, the receiver metrics are always derived for a MIMO system. However, for the sake of performance comparison between the different schemes, it is enough to use only one receive antenna in system simulations. Monte Carlo simulations have been carried out over a Rayleigh block fading channel using the differential encoding in (\ref{eq:Diff_encoding_unitary_OSTBC}) and the non-coherent receiver metric in (\ref{eq:Fast_ML_Unitary_OSTBC}).\\

Figure \ref{fig:M_2_unitary_2_bits_s_Hz} compares the BER performance of differential non-coherent $2\!\times\!1$ systems using  Alamouti's scheme with QPSK alphabet for all symbols, and using the cyclic group DUSTM scheme at a transmission rate of 2 bits/s/Hz. Clearly Alamouti's scheme outperforms the cyclic DUSTM scheme by about $\unit[3]{dB}$. Also included is the error performance of the $1\!\times\!2$ SIMO non-coherent system described in section \ref{s:Rx-diversity}. All curves are parallel in the high SNR range since they have the same diversity order of 2. Recall that both OSTBCs and DUSTM achieve full diversity, i.e DO$\eq MN\eq2$ and the SIMO system achieves a diversity order of $N=2$.\\

\begin{figure}[!htp]
\centering
\includegraphics[scale=0.4]{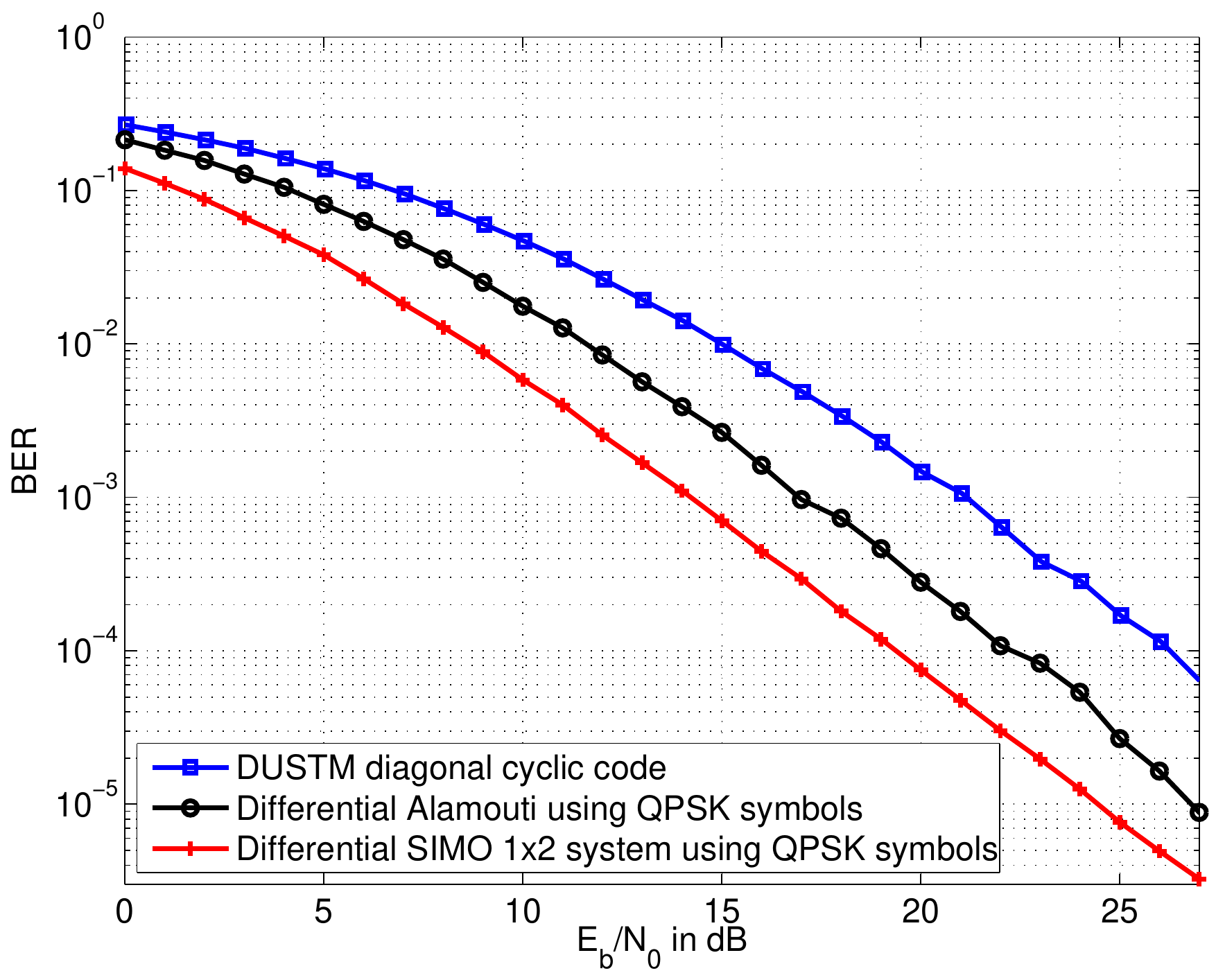} 
\caption{Comparing the 2$\times$1 DUSTM cyclic group scheme with differential Alamouti's scheme that uses PSK symbols at 2 bits/s/Hz. Additionally included is the error performance of the 1$\times$2 differential SIMO scheme.}
  \label{fig:M_2_unitary_2_bits_s_Hz}
\end{figure}

Observe that Alamouti's scheme suffers a $\unit[3]{dB}$ performance loss compared to the $1\!\times\!2$ SIMO system. As explained in \cite[section V.A]{Alamouti98}, this effect is due to the fact that, for the MISO system to transmit the same total power as that transmitted by the SIMO system, the total power in the MISO case is divided over the $M$ transmit antennas reducing the energy allocated to each symbol by a factor of $M$. This translates to a performance penalty of $\unit[10\log_{10}M]{dB}$, i.e. $\unit[3]{dB}$ in the two transmit antenna case. Therefore, the $1\!\times\!M$ SIMO error rate curve can be used to set an upper bound on the best achievable performance of an $M\!\times\!1$ transmit diversity system that achieves the same spectral efficiency. In other words, the best an $M\!\times\!1$  transmit diversity system can do is to be worse than the corresponding $1\!\times\!M$ receive diversity system by $\unit[10\log_{10}M]{dB}$.\\

For a transmission rate of $\unit[4]{bits/s/Hz}$, Alamouti's scheme is used with 16-PSK symbols and compared to a cyclic group code that achieves the same transmission rate using the generator matrix $\V{G}\eq\diag\{e^{\frac{j2\pi}{256}},\,e^{\frac{j2\pi75}{256}}\}$ \cite{Hassibi_Hochwald_2001}. Alamouti's scheme achieves an SNR advantage of about $\unit[5]{dB}$ as shown in Figure \ref{fig:M_2_unitary_4_bits_s_Hz}.\\
\begin{figure}[!htp]
\centering
\includegraphics[scale=0.4]{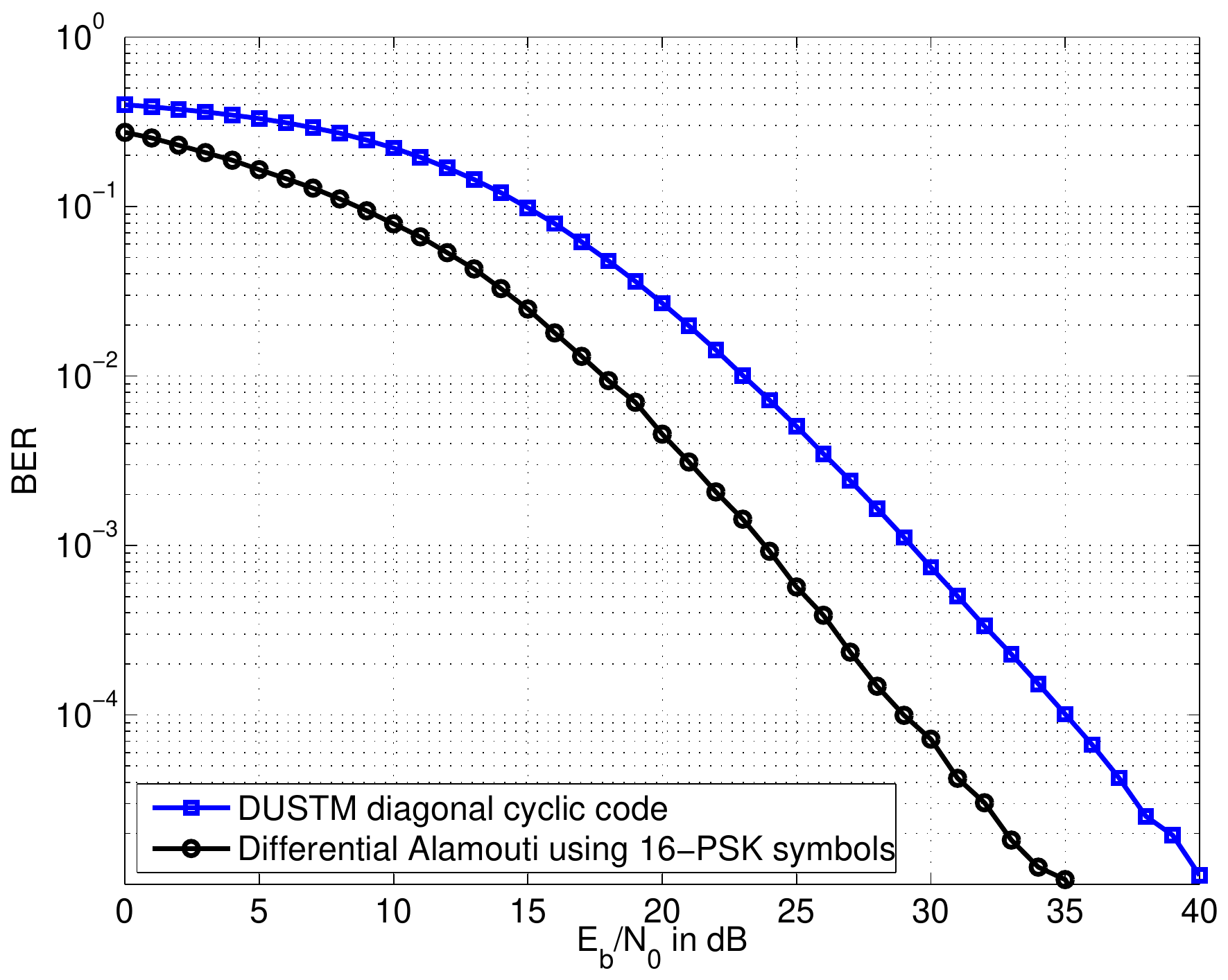} 
\caption{Comparing the 2$\times$1 DUSTM cyclic group scheme with differential Alamouti's scheme that uses PSK symbols at 4 bits/s/Hz.}
  \label{fig:M_2_unitary_4_bits_s_Hz}
\end{figure}

\begin{figure}[!htp]
 \centering
  \subfloat[$\eta=2$ bits/s/Hz]{\label{subfig:M_4_unitary_2_bits_s_Hz}\includegraphics[width=0.5\textwidth]{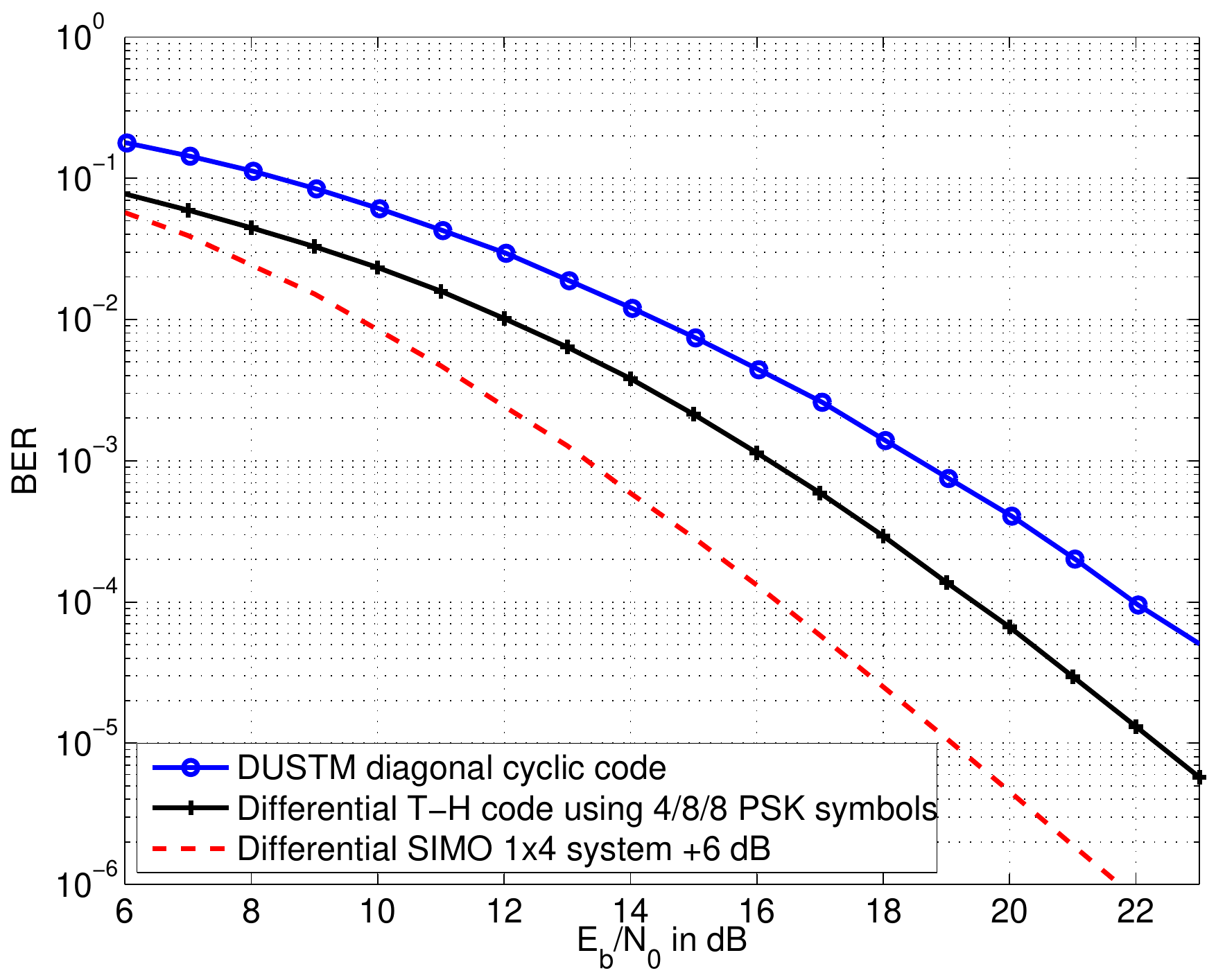}}\\
  \subfloat[$\eta=3$ bits/s/Hz]{\label{subfig:M_4_unitary_3_bits_s_Hz}\includegraphics[width=0.49\textwidth]{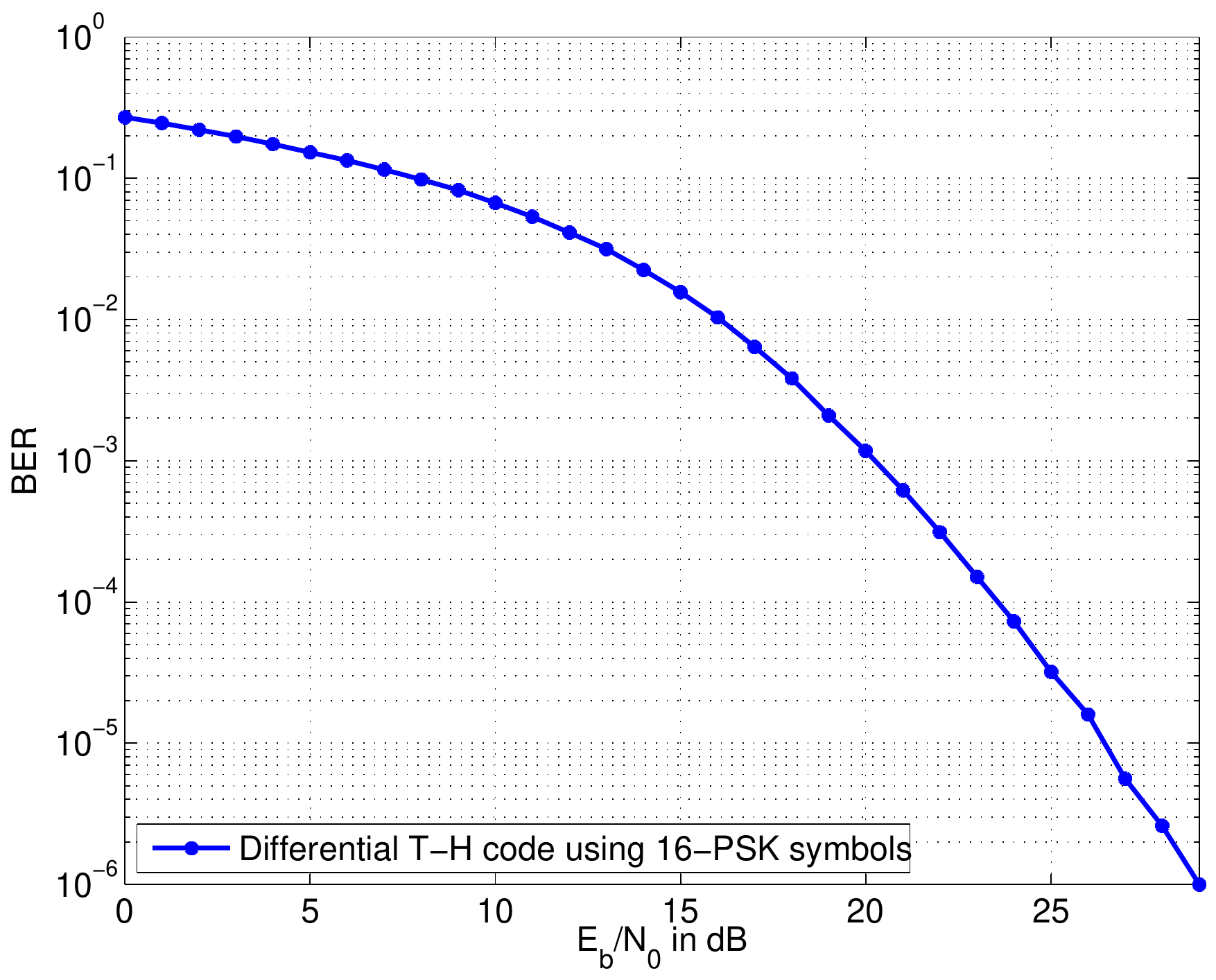}}
  \subfloat[$\eta=4$ bits/s/Hz]{\label{subfig:M_4_unitary_4_bits_s_Hz}\includegraphics[width=0.49\textwidth]{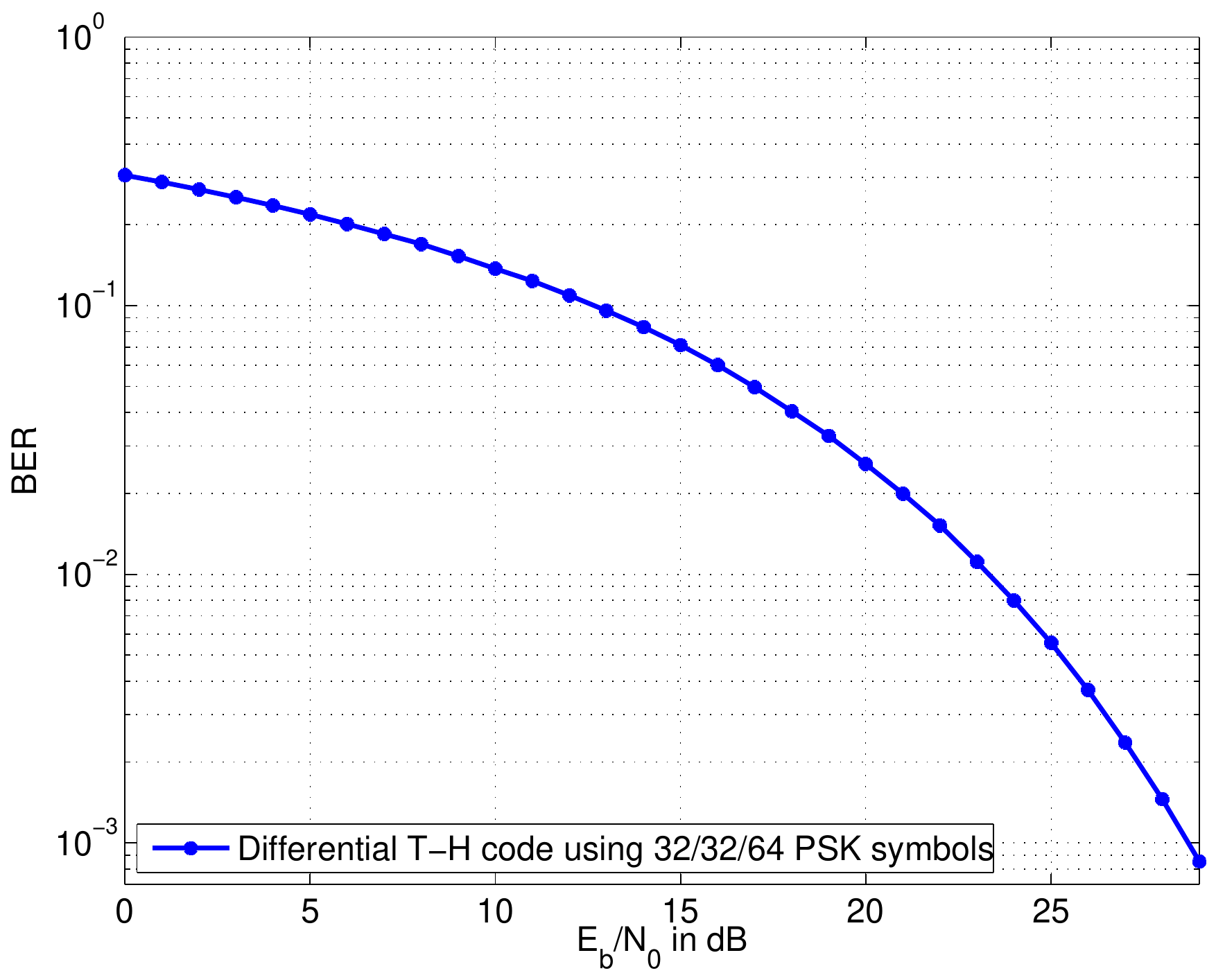}}
  \caption{BER curves for a differential $4\!\times\!1$ system that uses unitary OSTBC T-H code matrix in (\ref{eq:T_H}) with PSK symbols}
  \label{fig:T_H_BER}
\end{figure}

For four transmit antenna systems, the T-H code matrix in (\ref{eq:T_H}) is used. Since the code rate is $\nicefrac{3}{4}$, in order to achieve a transmission rate of $\unit[2]{bits/s/Hz}$ or $\unit[4]{bits/s/Hz}$, one needs to use unequal alphabet size. To achieve a spectral efficiency of $\eta$, one needs to satisfy
\begin{equation}
\eta=\frac{n_1+n_2+n_3}{4}
\end{equation}
where $n_i$ is the number of bits carried by symbol $x_i$ in (\ref{eq:T_H}). In general, the highest order alphabet is the one that dominates the performance. Therefore $n_1,\,n_2,$ and $n_3$ should be made as close as possible. For example to achieve $\unit[2]{bits/s/Hz}$, the combination $n_1\eq2,\,n_2\eq3,\,n_3\eq3$ is chosen, and is denoted as $4/8/8$ PSK alphabet. Similarly for the $\unit[4]{bits/s/Hz}$, $32/32/64$ PSK alphabet is used. For a transmission rate of $\unit[3]{bits/s/Hz}$, 16 PSK alphabet can be used for all symbols. The BER curves for the different transmission rates using the T-H code in $4\!\times\!1$ differential non-coherent systems are shown in Figure \ref{fig:T_H_BER}. For $\unit[2]{bits/s/Hz}$, the cyclic DUSTM code achieves about $\unit[2.25]{dB}$ performance loss compared to the $4/8/8$ PSK OSTBC. The BER curve of the SIMO $1\!\times\!4$ system at $\unit[2]{bits/s/Hz}$ is included after shifting it to the right by $10\log_{10}(4)=\unit[6]{dB}$ to show the best performance a $4\!\times\!1$ system can achieve. Again all curves are parallel since they all have the same diversity order of 4.


\subsection{Non-unitary OSTBCs}
\label{ss:Non_UOSTBCs}
Restricting the use of OSTBCs to unitary transmission using PSK alphabets will lead to performance degradation as more bits are to be transmitted per channel use. Just as in the SISO case, the higher the alphabet order of PSK, the more constellation points are placed on the same circle and therefore the higher the probability of decoding error. Allowing information bits to be carried over the amplitude in addition to the phase of the transmit signals will result in a more efficient utilization of the complex space. This will potentially lead to performance enhancement. Tao and Cheng proposed in \cite{Tao_2001} a differential encoding scheme for OSTBCs using QAM alphabets. This subsection shows their approach for differential encoding of non-unitary OSTBCs. The ML non-coherent metric is derived, together with another sub-optimum non-coherent metric that has a significantly reduced complexity. Finally, the error performance will be compared to unitary OSTBCs.\\
\subsubsection{Code Construction}
\label{sss:Non_Unitary_code_construction}
Consider the case when the information symbols are drawn from a non-unitary constellation. Since for any OSTBC, all $K$ symbols appear on all columns of the code matrix (like in (\ref{eq:Alamouti}) and (\ref{eq:T_H})) \footnote{In fact, the existence of every information symbol on all  columns of the information matrix is a necessary but insufficient condition for achieving full diversity. Since OSTBCs achieve full diversity, this condition is satisfied.}, all columns carry the same power and the code matrix $\V{V}_{z_\tau}$ at the $\tau^{\text{th}}$ block index satisfies
\begin{equation}
\label{eq:v_hr_V_non_unitary}
 \V{V}_{z_\tau}\hr\V{V}_{z_\tau}=\frac{\sum\limits_{i=1}^{K}|x_i|^2}{p}\V{I}_{M}\definedas a_\tau^2\,\V{I}_M.
\end{equation}
Here symbols $x_i$ are the symbols drawn at time index $\tau$, and $a_\tau$ is defined as the amplitude of matrix $\V{V}_{z_\tau}$. In this case, matrix $\V{V}_{z_\tau}$ is considered to be \emph{scaled} unitary, since all columns are mutually orthogonal, and all of them have the same power ($a_\tau^2$), but this power in general varies with $\tau$.\\

In the context of differential encoding, one must ensure that the power of the transmit matrices does not increase or decrease indefinitely. For this to be satisfied, some power normalization should be considered. The following differential encoding equation shows one way to satisfy such power constraint
\begin{equation}
\label{eq:Diff_encoding_non_unitary}
 \V{S}_\tau=\frac{\V{V}_{z_\tau}\V{S}_{\tau-1}}{a_{\tau-1}}.
\end{equation}
Compared to the differential encoding of unitary OSTBCs in (\ref{eq:Diff_encoding_unitary_OSTBC}), in non-unitary OSTBCs we need to divide by the amplitude of the previously transmitted matrix. This division removes the effect of the power of the previously transmitted block so that only the power of $\V{V}_{z_\tau}$ contributes to the power of $\V{S}_\tau$. If the transmission is initiated by an identity matrix, then by writing down the first few transmit blocks, we get
\begin{eqnarray*}
\V{S}_{0}=&\V{I}_M&\implies\V{S}_0\hr\V{S}_0=\V{I}_M\:\implies a_0=1;\\
\V{S}_1 =&\frac{\V{V}_{z_1}\V{S}_0}{a_0}&\implies\V{S}_1\hr\V{S}_1=\V{V}_{z_1}\hr\V{V}_{z_1}=a_1^2\V{I}_M\\
\V{S}_2 =&\frac{\V{V}_{z_2}\V{S}_1}{a_1}&\implies\V{S}_2\hr\V{S}_2=\frac{\V{S}_1\hr\V{V}_{z_2}\hr}{a_1}\frac{\V{V}_{z_2}\V{S}_1}{a_1}=\frac{a_2^2}{a_1^2}\V{S}_1\hr\V{S}_1=\frac{a_2^2}{a_1^2}\,a_1^2\V{I}_M=a_2^2\V{I}_M=\V{V}_{z_2}\hr\V{V}_{z_2}.
\end{eqnarray*}
Therefore at any time index $\tau$, the power of the transmit matrix $\V{S}_\tau$ is the same as the power of the information matrix $\V{V}_{z_\tau}$, namely
\begin{equation}
 \V{S}_\tau\hr\V{S}_\tau=\V{V}_{z_\tau}\hr\V{V}_{z_\tau}=a_\tau^2\V{I}_M,\:\:\:\:\forall\:\tau.
\label{eq:S_hr_S_V_hr_V}
\end{equation}\\

In the transmission equation $\V{Y}_{\tau}=\sqrt{\rho}\V{S}_{\tau}\V{H}+\V{W}_\tau$, in order to ensure that $\rho$ is the average SNR at each receive antenna per time slot, we need to satisfy the energy constraint $\E[\sum\limits_{m=1}^{M}{|s_{tm}|^2}]=1\:\:\forall \:t$ defined in (\ref{eq:power_st}). This constraint can be stated as: The average power of every \emph{row} in $\V{S}_\tau$ is required to be 1. Alamouti's code matrix in (\ref{eq:Alamouti}) and the T-H code matrix in (\ref{eq:T_H}) satisfy $\V{V}_{z_\tau}\hr\V{V}_{z_\tau}\eq\V{V}_{z_\tau}\V{V}_{z_\tau}\hr\eq a_{\tau}^2\V{I}_M$, and using (\ref{eq:S_hr_S_V_hr_V}), also $\V{S}_{\tau}\hr\V{S}_{\tau}\eq\V{S}_{\tau}\V{S}_{\tau}\hr\eq a_{\tau}^2\V{I}_M$. In other words, the power of all columns of $\V{S}_{\tau}$ is the same as the power of all rows of $\V{S}$ is $a_{\tau}^2$. Therefore satisfying the power constraint in (\ref{eq:power_st}) is equivalent to satisfying 
$\E[\V{V}_{z_\tau}\hr\V{V}_{z_\tau}]=\V{I}_M\:\forall\,\tau$. This is equivalent to satisfying $\E[a_{\tau}^2]\eq1\:\forall\,\tau$. Assuming all alphabets to be of unit \emph{average} energy (i.e. $\E[|x_i|^2]=1\:\forall\,i=1,...,K$) and using (\ref{eq:v_hr_V_non_unitary}), the energy constraint reduces to
\begin{equation}
 \E[a_{\tau}^2]\req 1 \implies \frac{\E\Big[\sum\limits_{i=1}^{K}{|x_i|^2}\Big]}{p}=\frac{K\E[|x_i|^2]}{p}\req 1 \implies p\req K.
\label{eq:p_eq_K}
\end{equation}

\subsubsection{ML Differential Decoder}
\label{sss:Non_Unitary_ML_metric}
To derive the non-coherent ML decision metric for non-unitary OSTBCs, we start from the general form of the non-coherent ML metric for any STC scheme in (\ref{eq:ML_STC_general2}), namely
\begin{equation}
\label{eq:ML_STC_general2_repeated}
\hat{\V{V}}_{z_\tau}\Big|_{\text{ML}}=\argmax\limits_{l=\{0,...,L-1\}}\:\,\tr\{\rho\V{\bar{Y}}_\tau\hr\V{\bar{S}}^{(l)}(\V{I}_{M}+ \rho\V{\bar{S}}\sp{(l)\dagger} \V{\bar{S}}^{(l)})^{-1}\V{\bar{S}}\sp{(l)\dagger}\V{\bar{Y}}_\tau\}-N\ln\{\det(\V{I}_M+ \rho\V{\bar{S}}\sp{(l)\dagger} \V{\bar{S}}^{(l)})\},
\end{equation}
where $\rho$ is the SNR at each receive antenna, $\V{\bar{Y}}_\tau=\begin{bmatrix}\V{Y}_{\tau-1}\\ \V{Y}_{\tau}\end{bmatrix}$ and $\V{\bar{S}}^{(l)}$ is the $l^{\text{th}}$ possible transmit two-block matrix. Based on the differential encoding defined in (\ref{eq:Diff_encoding_non_unitary}), $\V{\bar{S}}^{(l)}$ can be written as
\renewcommand{\arraystretch}{1.4} 
\begin{equation*}
\V{\bar{S}}^{(l)}= \begin{bmatrix}\V{S}_{\tau-1}\\\V{S}_\tau^{(l)}\end{bmatrix}=\begin{bmatrix}\V{S}_{\tau-1}\\\frac{\V{V}_{l}\V{S}_{\tau-1}}{a_{\tau-1}}\end{bmatrix}=\begin{bmatrix}a_{\tau-1}\V{I}_M\\\V{V}_{l}\end{bmatrix}\frac{\V{S}_{\tau-1}}{a_{\tau-1}}.
\end{equation*}
\renewcommand{\arraystretch}{1.2}
Recall from the discussion on (\ref{eq:P_Y_S_general_useful}) that multiplying $\V{\bar{S}}^{(l)}$ by a unitary matrix from the right does not change the non-coherent ML receiver metric. Therefore, the transmit two-block matrices $\begin{bmatrix}a_{\tau-1}\V{I}_M\\\V{V}_{l}\end{bmatrix}\frac{\V{S}_{\tau-1}}{a_{\tau-1}}$ and $\textstyle\begin{bmatrix}a_{\tau-1}\V{I}_M\\\V{V}_{l}\end{bmatrix}$ are indistinguishable to the receiver, since $\frac{\V{S}_{\tau-1}}{a_{\tau-1}}$ is a unitary matrix $\forall\,\tau$.
As a result, the transmit candidate matrix can be written in a canonical form as
\begin{equation}
\label{eq:S_l_non_unitary_canonical}
 \V{\bar{S}}^{(l)}\equiv\begin{bmatrix}a_{\tau-1}\V{I}_M\\\V{V}_{l}\end{bmatrix}.
\end{equation}
Using (\ref{eq:S_l_non_unitary_canonical}), the argument of the inverse and the determinant terms in (\ref{eq:ML_STC_general2_repeated}) becomes
\begin{eqnarray}
\label{eq:det_inv_argument}
\V{I}_M+ \rho\V{\bar{S}}\sp{(l)\dagger} \V{\bar{S}}^{(l)}&=& \V{I}_M+ \rho \begin{bmatrix}a_{\tau-1}\V{I}_M &\V{V}_{l}\hr\end{bmatrix}\begin{bmatrix}a_{\tau-1}\V{I}_M\\\V{V}_{l}\end{bmatrix}\nonumber\\
&=& \V{I}_M+ \rho(a_{\tau-1}^2+a_{\tau}^2(l))\V{I}_M\nonumber\\
&=&\Big(1+\rho(a_{\tau-1}^2+a_{\tau}^2(l))\Big)\V{I}_M,
\end{eqnarray}
where we have used $(l)$ in $a_{\tau}^2(l)\V{I}_M$ to emphasize that it is dependent on the candidate code matrix $\V{V}_l$, i.e. $\V{V}_l\hr\V{V}_l\eq a_{\tau}^2(l)$.
The determinant of (\ref{eq:det_inv_argument}) is 
\begin{equation*}
 \det\Big((1+\rho(a_{\tau-1}^2+a_{\tau}^2(l)))\V{I}_M\Big)=\Big(1+\rho(a_{\tau-1}^2+a_{\tau}^2(l))\Big)^M,
\end{equation*}
and the inverse is
\begin{equation*}
 \Big((1+\rho(a_{\tau-1}^2+a_{\tau}^2(l)))\V{I}_M\Big)^{-1}=\Big(1+\rho(a_{\tau-1}^2+a_{\tau}^2(l))\Big)^{-1}\V{I}_M.
\end{equation*}
Therefore, the ML metric in (\ref{eq:ML_STC_general2_repeated}) reduces to
\begin{equation}
 \hat{\V{V}}_{z_\tau}\Big|_{\text{ML}}=\argmax\limits_{\V{V}_l\in\Omega}\:\,\frac{\rho}{1+\rho(a_{\tau-1}^2+a_{\tau}^2(l))}\tr\{\V{\bar{Y}}_\tau\hr\V{\bar{S}}^{(l)}\V{\bar{S}}\sp{(l)\dagger}\V{\bar{Y}}_\tau\}-MN\ln\{1+\rho(a_{\tau-1}^2+a_{\tau}^2(l))\},
\end{equation}
where,
\begin{eqnarray}
 \tr\{\V{\bar{Y}}_\tau\hr\V{\bar{S}}^{(l)}\V{\bar{S}}\sp{(l)\dagger}\V{\bar{Y}}_\tau\}&=&\|\V{\bar{S}}\sp{(l)\dagger}\V{\bar{Y}}_\tau\|_F^2\nonumber\\
&=&\|\begin{bmatrix}a_{\tau-1}\V{I}_M & \V{V}_l\hr\end{bmatrix}\begin{bmatrix}\V{Y}_{\tau-1}\\\V{Y}_{\tau}\end{bmatrix}\|_F^2\nonumber\\
&=&\|a_{\tau-1}\V{Y}_{\tau-1}+\V{V}_l\hr\V{Y}_\tau\|_F^2,
\end{eqnarray}
and the non-coherent ML metric for differential non-unitary OSTBCs reduces finally to
\begin{equation}
\label{eq:ML_metric_non_unitary}
 \boxed{\hat{\V{V}}_{z_\tau}\Big|_{\text{ML}}=\argmax\limits_{\V{V}_l\in\Omega}\:\,\frac{\rho}{1+\rho(a_{\tau-1}^2+a_{\tau}^2(l))}\|a_{\tau-1}\V{Y}_{\tau-1}+\V{V}_l\hr\V{Y}_\tau\|_F^2-MN\ln\{1+\rho(a_{\tau-1}^2+a_{\tau}^2(l))\}.}
\end{equation}\\

As shown in (\ref{eq:ML_metric_non_unitary}), the ML differential decision metric for non-unitary OSTBCs requires symbols $x_1,...,x_K$ embedded in $\V{V}_l$ to be jointly detected. Losing the symbol decoupling advantage leads to an exponential increase in the receiver complexity with the number of encoded symbols. To alleviate this problem, one needs to consider some sub-optimum decision technique that maintains the symbol decoupling advantage and achieves an acceptable performance. For this purpose, the authors in \cite{Tao_2001} used another metric which they named a near-optimal metric as it provides a performance very close to the ML performance.
\subsubsection{Near-Optimal Differential Decoder}
Based on the Rayleigh block fading channel model, two consecutive received matrices at block indices $\tau-1$ and $\tau$ can be written as
\begin{eqnarray}
\label{eq:Y_t_1}
\V{Y}_{\tau-1}&=&\sqrt{\rho}\V{S}_{\tau-1}\V{H}+\V{W}_{\tau-1}\\
\V{Y}_{\tau}&=&\sqrt{\rho}\V{S}_{\tau}\V{H}+\V{W}_{\tau}.
\label{eq:Y_t}
\end{eqnarray}
Using the differential encoding equation in (\ref{eq:Diff_encoding_non_unitary}), (\ref{eq:Y_t}) can be rewritten as
\begin{equation}
 \V{Y}_{\tau}=\sqrt{\rho}\frac{\V{V}_{z_\tau}\V{S}_{\tau-1}}{a_{\tau-1}}\V{H}+\V{W}_{\tau}.
\label{eq:Y_t2}
\end{equation}
Multiplying (\ref{eq:Y_t_1}) from the left by $\frac{\V{V}_{z_\tau}}{a_{\tau-1}}$, then subtracting the resulting equation from (\ref{eq:Y_t2}), we get
\begin{equation}
 \V{Y}_{\tau}=a_{\tau-1}^{-1}\V{V}_{z_\tau}\V{Y}_{\tau-1}+\underbrace{\V{W}_{\tau}-a_{\tau-1}^{-1}\V{V}_{z_\tau}\V{W}_{\tau-1}}_{\definedas \sqrt{1+a_{\tau-1}^{-2}a_{\tau}^2}\V{W}_{\tau}'}
\label{eq:fund_receiver_eqn}
\end{equation}
where $\V{W}_\tau'$ is an equivalent noise matrix whose elements are i.i.d. complex Gaussian random variables with zero mean and variance $\nicefrac{1}{2}$ per dimension. (\ref{eq:fund_receiver_eqn}) looks as if $\V{V}_{z_\tau}$ was transmitted over a known channel matrix $a_{\tau-1}^{-1}\V{Y}_{\tau-1}$, and then corrupted by an AWGN noise with variance $1+a_{\tau-1}^{-2}a_{\tau}^2$. Since $a_{\tau-1}$ and $a_{\tau}$ are on average the same, it follows that the noise variance is on average 2, i.e. twice as much as the variance of the actual channel noise. This doubling of noise translates to to the well-known $\unit[3]{dB}$ performance degradation of non-coherent systems compared to coherent ones. \\

If we neglect the dependency of the noise variance on the transmitted signals, we can transform our problem to a virtual coherent system which decides on the candidate matrix that suffers the least the noise variance. This arrives us to the following near-optimal differential decoder
\begin{eqnarray}
\hat{\V{V}}_{z_\tau}&=&\argmin\limits_{\V{V}_l\in\Omega}\:\,\|\V{Y}_{\tau}-a_{\tau-1}^{-1}\V{V}_{l}\V{Y}_{\tau-1}\|_F^2\label{eq:sub_opt_Diff_Decoder}\\
&=&\argmin\limits_{\V{V}_l\in\Omega}\:\,\tr\{(\V{Y}_{\tau}-a_{\tau-1}^{-1}\V{V}_{l}\V{Y}_{\tau-1})(\V{Y}_{\tau}\hr-a_{\tau-1}^{-1}\V{Y}_{\tau-1}\hr\V{V}_l\hr)\}\nonumber\\
&=&\argmin\limits_{\V{V}_l\in\Omega}\:\,\tr\{\underbrace{\V{Y}_\tau\V{Y}_\tau\hr}_{\text{indep. of }l}-a_{\tau-1}^{-1}\V{V}_{l}\V{Y}_{\tau-1}\V{Y}_{\tau}\hr-a_{\tau-1}^{-1}\V{Y}_\tau\V{Y}_{\tau-1}\hr\V{V}_l\hr+a_{\tau-1}^{-2}\V{V}_{l}\V{Y}_{\tau-1}\V{Y}_{\tau-1}\hr\V{V}_{l}\hr\}\nonumber\\
&=&\argmax\limits_{\V{V}_l\in\Omega}\:\,a_{\tau-1}^{-1}\tr\{\V{V}_{l}\V{Y}_{\tau-1}\V{Y}_{\tau}\hr+\V{Y}_\tau\V{Y}_{\tau-1}\hr\V{V}_l\hr\}-a_{\tau-1}^{-2}\tr\{\underbrace{\V{V}_l\hr\V{V}_{l}}_{a_{\tau}^2(l)\V{I}_M}\V{Y}_{\tau-1}\V{Y}_{\tau-1}\hr\}\nonumber\\
&=&\argmax\limits_{\V{V}_l\in\Omega}\:\,\frac{2}{a_{\tau-1}}\Re\{\tr(\V{V}_l\V{Y}_{\tau-1}\V{Y}_\tau\hr)\}-\frac{a_{\tau}^2(l)}{a_{\tau-1}^2}\underbrace{\tr\{\V{Y}_{\tau-1}\V{Y}_{\tau-1}\hr\}}_{\tilde{y}}\nonumber
\end{eqnarray}
Using the dispersive form of $\V{V}_l$ in (\ref{eq:dispersive_form}), the near-optimal metric becomes
\begin{eqnarray*}
 \hat{\V{V}}_{z_\tau}&=&\argmax\limits_{x_i\in\mathcal{A}_i}\:\,\frac{2}{a_{\tau-1}}\Re\{\tr\Big((\frac{1}{\sqrt{K}}\sum\limits_{i=1}^{K}{\V{A}_ix_i+\V{B}_ix_i^*})\V{Y}_{\tau-1}\V{Y}_\tau\hr\Big)\}-\frac{a_{\tau}^2(l)}{a_{\tau-1}^2}\tilde{y}\\
&=&\argmax\limits_{x_i\in\mathcal{A}_i}\:\,\frac{2}{\sqrt{K}}\sum\limits_{i=1}^{K}{\Re\{\tr\Big(\V{A}_i\V{Y}_{\tau-1}\V{Y}_\tau\hr\, x_i+\V{B}_i\V{Y}_{\tau-1}\V{Y}_{\tau}\hr\,x_i^*\Big)\}}-\frac{a_{\tau}^2(l)}{a_{\tau-1}}\tilde{y}\\
&=&\argmax\limits_{x_i\in\mathcal{A}_i}\:\,\sum\limits_{i=1}^{K}{\frac{2}{\sqrt{K}}{\Re\{\underbrace{\tr\Big(\V{A}_i\V{Y}_{\tau-1}\V{Y}_\tau\hr+\V{B}_i^*\V{Y}_{\tau-1}^*\V{Y}_{\tau}^T\Big)}_{\tilde{x}_i}x_i\}}}-\frac{\tilde{y}}{a_{\tau-1}}\frac{\sum\limits_{i=1}^{K}{|x_i|^2}}{K}
\end{eqnarray*}
\begin{eqnarray*}
\hat{\V{V}}_{z_\tau}&=&\sum\limits_{i=1}^{K}{\argmax\limits_{x_i\in\mathcal{A}_i}\:\,\Re\{\tilde{x}_ix_i\}-\frac{\tilde{y}}{2a_{\tau-1}\sqrt{K}}|x_i|^2}
\end{eqnarray*}
Therefore, the near-optimal decoder decouples the data symbols $x_1,...,x_K$ resulting in linear complexity, and the decision metric for symbol $x_i$ finally becomes
\begin{equation}
\label{eq:x_hat_non_unitary}
 \boxed{\hat{x}_i=\argmax\limits_{x_i\in\mathcal{A}_i}\:\,\Re\{\tilde{x}_ix_i\}-\frac{\tilde{y}}{2a_{\tau-1}\sqrt{K}}|x_i|^2}
\end{equation}\\

In \cite{Tao_2001}, the authors showed that the performance of the near-optimal decoder in (\ref{eq:x_hat_non_unitary}) is only marginally worse --by about $\unit[0.3]{dB}$-- than  that of the optimal decoder in (\ref{eq:ML_metric_non_unitary}). Note also that the amplitude $a_{\tau-1}$ of the previously transmitted matrix appears in both the optimal and the near-optimal metrics. This indicates the possibility of error propagation, because if $\V{V}_{z_{\tau-1}}$ is wrongly decided, $a_{\tau-1}$ will also be wrong. However, the authors in \cite{Tao_2001} showed that the error rate curves resulting from the assumption of perfect knowledge of $a_{\tau-1}$ are almost the same as when no knowledge of $a_{\tau-1}$ is assumed, which indicates no error propagation. This was concluded for the optimal metric, and we assume the same behaviour for the near-optimal decoder.\\

Interestingly, if the symbols were drawn from a constant-energy alphabet, the metric in (\ref{eq:x_hat_non_unitary}) will be the same as the non-coherent \emph{Maximum Likelihood} metric for unitary OSTBCs in (\ref{eq:Fast_ML_Unitary_OSTBC}). In other words, the non-coherent ML metric for unitary OSTBCs can be also reached by which decides on the candidate matrix that suffers the least the noise variance in the virtual coherent system in (\ref{eq:fund_receiver_eqn}) with $a_\tau\eq 1\,\forall\,\tau$. \\

Owing to its significant reduction in decoding complexity, the near-optimal metric is the one we consider by default when differentially decoding non-unitary OSTBCs. The sub-optimal decoder achieves SCSD and the transmission can also be described by Figure \ref{fig:DSTBC_SCSD}, with $\mathcal{A}_i$ being the non-unitary alphabet from which symbol $x_i$ is drawn. In general, any non-unitary constellation can be used. Here, we study the use of QAM alphabets. In the case of \emph{rectangular} QAM, the real and the imaginary components of the symbols can be independently encoded, since rectangular QAM can be viewed as two independent PAM constellations. We may utilize this advantage by further splitting the decision on the real and the imaginary components of the symbols, reducing the decoder metric of (\ref{eq:x_hat_non_unitary}) to
\begin{equation}
\label{eq:decoder_rect_QAM}
\begin{aligned}
\hat{x}_i^R&=&\argmax\limits_{x_i^R\in\mathcal{A}_{\text{PAM}_i}}\:\,\tilde{x}_i^Rx_i^R-\frac{\tilde{y}}{2a_{\tau-1}\sqrt{K}}|x_i^R|^2\\
\hat{x}_i^I&=&\argmax\limits_{x_i^I\in\mathcal{A}_{\text{PAM}_i}}\:\,-\tilde{x}_i^Ix_i^I-\frac{\tilde{y}}{2a_{\tau-1}\sqrt{K}}|x_i^I|^2,
\end{aligned}
\end{equation}
where $\tilde{x}_i=\tilde{x}_i^R+j\tilde{x}_i^I$ and $x_i=x_i^R+jx_i^I$. Using the above metric for rectangular QAM constellations reduces the decoding complexity further. This is because the search space of one dimensional PAM alphabet is square root of the corresponding two dimensional QAM. For example in 16-QAM alphabet, the search space is reduced from $16$ candidates to $2\times\sqrt{16}\eq 8$ candidates, and in 64-QAM, the search space is reduced from $64$ to $2\times\sqrt{64}\eq 16$ candidates. A decoder like in (\ref{eq:decoder_rect_QAM}) that decides on one-dimensional component of each symbol independently will be referred to as performing Single Real Symbol Decoding (SRSD).\\
\subsubsection{Performance Analysis}
Simulations have been done for $2\!\times\!1$ and $4\!\times\!1$ systems using Alamouti's code matrix in (\ref{eq:Alamouti}) and the T-H code matrix in (\ref{eq:T_H}), respectively. The channel assumed is again a flat Rayleigh block fading channel with independent coefficients. The systems use the differential encoding equation in (\ref{eq:Diff_encoding_non_unitary}) and the near-optimal differential decoder in (\ref{eq:x_hat_non_unitary}), except for rectangular QAM alphabets where metric (\ref{eq:decoder_rect_QAM}) is used. \\

For QAM alphabets whose size has integer square roots, rectangular constellation is used (like 16-QAM and 64-QAM). Otherwise circular constellation is used (like 8-QAM and 32-QAM). Rectangular constellation can be easily defined, however for circular constellation we have the advantage of more freedom in placing the constellation points. To make use of this advantage, optimization for the 8-QAM constellation is carried out. We started with the constellation shown on the left of Figure \ref{fig:Opt_8_QAM_OSTBC_constellation} by defining two 4-PSK circles with different radii. We define two optimization parameters, namely $\theta$ which is the relative angle of rotation between the two circles and $a$ which is the ratio between the circles radii. Then error rate simulations were carried out at $E_b/N_0$ of $\unit[14]{dB}$ (a point in the high SNR range) for different $a$ and $\theta$. Figure \ref{fig:Opt_8_QAM_OSTBC_SER} shows that optimally the inner circle should be rotated by $45^0$ relative to the outer circle for all values of $a$ considered. For $a\eq1.6$, the SER is lowest. The optimized constellation is shown on the right of Figure \ref{fig:Opt_8_QAM_OSTBC_constellation}.\\

\begin{figure}[htp]
\vspace{-1.9cm}
\centering
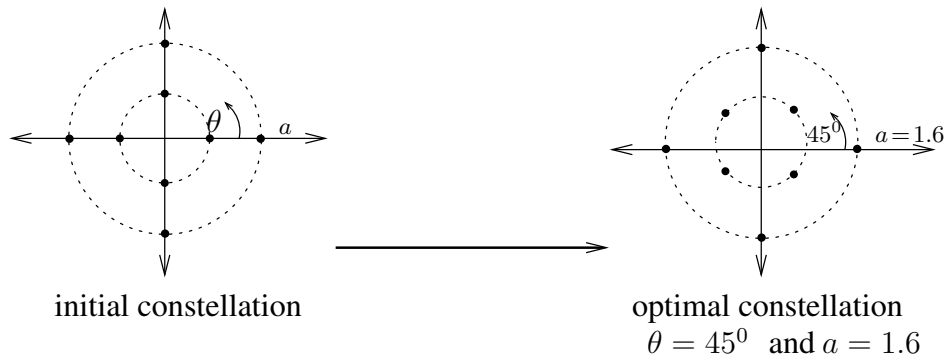
\caption{ Optimization for use of 8-QAM with OSTBCs}
  \label{fig:Opt_8_QAM_OSTBC_constellation}
\end{figure}

\begin{figure}[htp]
\centering
\includegraphics[scale=0.3]{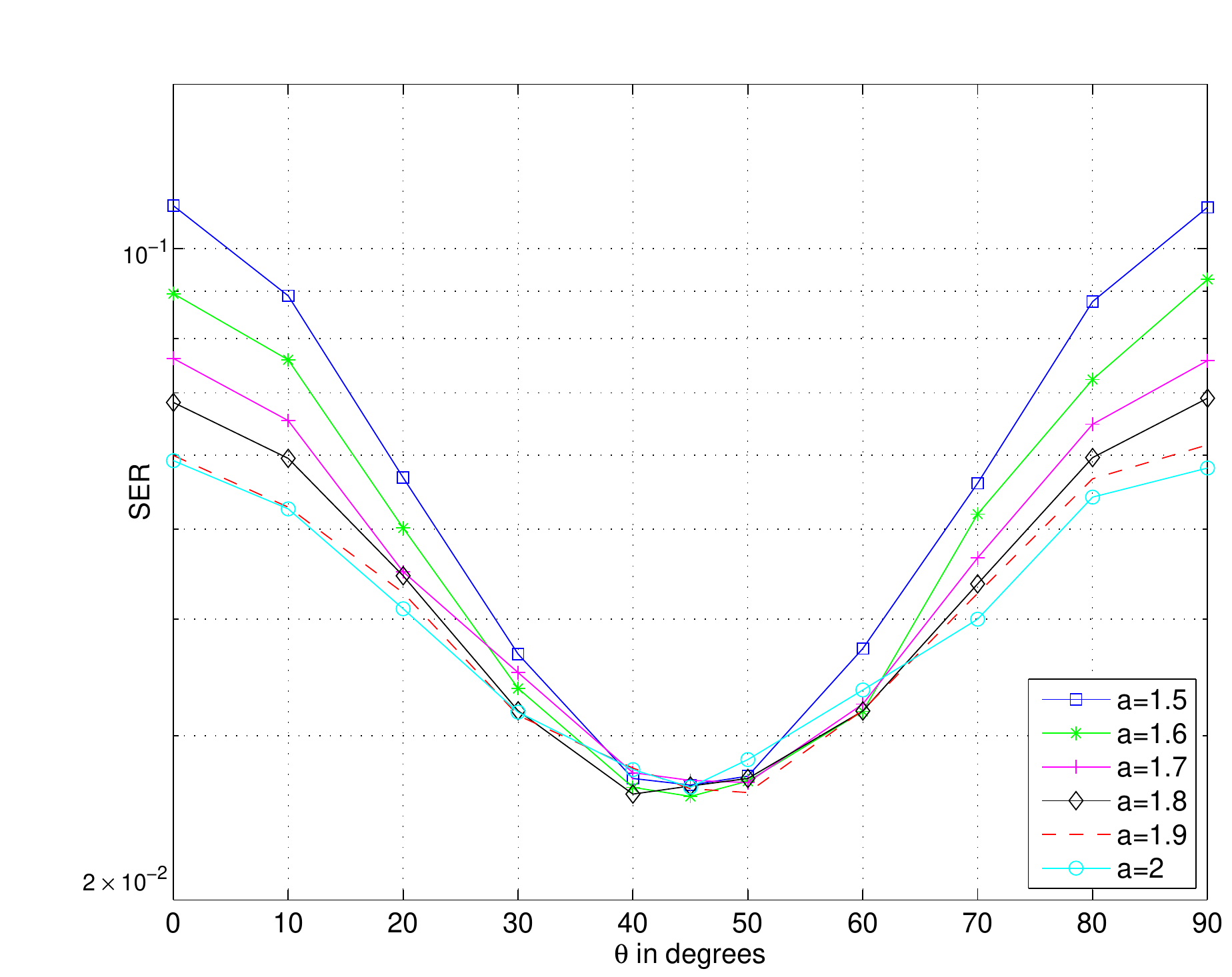}
\caption{SER vs angle of rotation for an 8-QAM constellation used with OSTBCs.}
  \label{fig:Opt_8_QAM_OSTBC_SER}
\end{figure}

Figure \ref{fig:M_2_non_unitary_4_bits_s_Hz} compares the BER curves of Alamouti's code using PSK vs. QAM symbols. At a transmission rate of $\unit[3]{bits/s/Hz}$, there is almost no difference in performance of the 8-QAM compared to 8-PSK. At a transmission rate of $\unit[4]{bits/s/Hz}$, the use of 16-QAM results in an SNR advantage of about $\unit[2]{dB}$ compared to 16-PSK. Moreover, the search space for 16-QAM is 8 candidates per complex symbol using the metric in (\ref{eq:decoder_rect_QAM}), whereas for 16-PSK all 16 candidates must be tested. For $4\!\times\!1$ systems, Figure \ref{fig:M_4_non_unitary} shows the BER curves at different transmission rates. At $\unit[2]{bits/s/Hz}$, the alphabet size combination 4/8/8 is chosen. As shown, using 4/8/8 QAM results in only $\unit[0.5]{dB}$ gain compared to 4/8/8 PSK. Also shown is the BER curve of the $1\!\times\!4$ non-coherent SIMO system after shifting it $\unit[6]{dB}$ to the right to account for the power division loss of the $4\!\times\!1$ MISO system. Clearly, using OSTBC with QAM alphabet is still far worse than the reference SIMO curve. For higher order constellation, the performance of PSK alphabet deteriorates significantly compared to QAM as shown for the transmission rates of $\unit[3]{bits/s/Hz}$ and $\unit[4]{bits/s/Hz}$.\\

Note that the curves of the OSTBC using QAM are parallel to those that use PSK constellation. This indicates that also non-unitary non-coherent OSTBCs achieve full diversity. This however was not explicitly proved in the literature. In coherent systems, the full diversity condition  based on the rank criteria was proved to apply for any STBC. Whereas in the differential non-coherent case, this was analytically proved true only for unitary transmission as derived in section \ref{s:Design_Criteria}. Based on our results, we conjecture that the rank criteria is still the diversity criteria for non-coherent non-unitary OSTBCs.\\

In conclusion, this chapter covered two classes of STCs, namely USTM and OSTBCs, in the differential domain. The USTM first looked appealing, since the transmitter's role is simplified. However, due to the direct mapping of bits to matrices, the receiver complexity increases exponentially with the number of transmit antennas and the spectral efficiency. On the other hand, OSTBCs map bits first to symbols and then to orthogonal code matrices. If the symbols belong to PSK alphabets, the code matrices are unitary and the ML metric performs SCSD. If they belong to QAM alphabets, the code matrices are scaled unitary and a near-optimal metric is used to achieve SCSD. DUSTM is inferior to all OSTBCs in terms of complexity and error performance for all transmission rates. The advantage of using QAM relative to PSK alphabets in OSTBCs is significant only for alphabet size larger than 8. \\

Although OSTBCs with QAM alphabets showed the best performance of the aforementioned schemes, still their performance is worse than the reference (shifted) SIMO curves. This indicates the possibility of existence of other schemes that make better use of the available resources to improve the reliability of transmission and increase the data rate. One such scheme is the so-called Quasi Orthogonal Space-Time Block Codes (QOSTBCs), which relaxes the condition of orthogonality of the code matrices to enhance the code rate. It is the role of the next chapter to motivate the use of QOSTBCs, and show how much advantage these codes may achieve.

\begin{figure}[htp]
\centering
\subfloat[$\eta=3$ bits/s/Hz]{\label{subfig:M_2_non_unitary_3_bits_s_Hz}\includegraphics[width=0.45\textwidth]{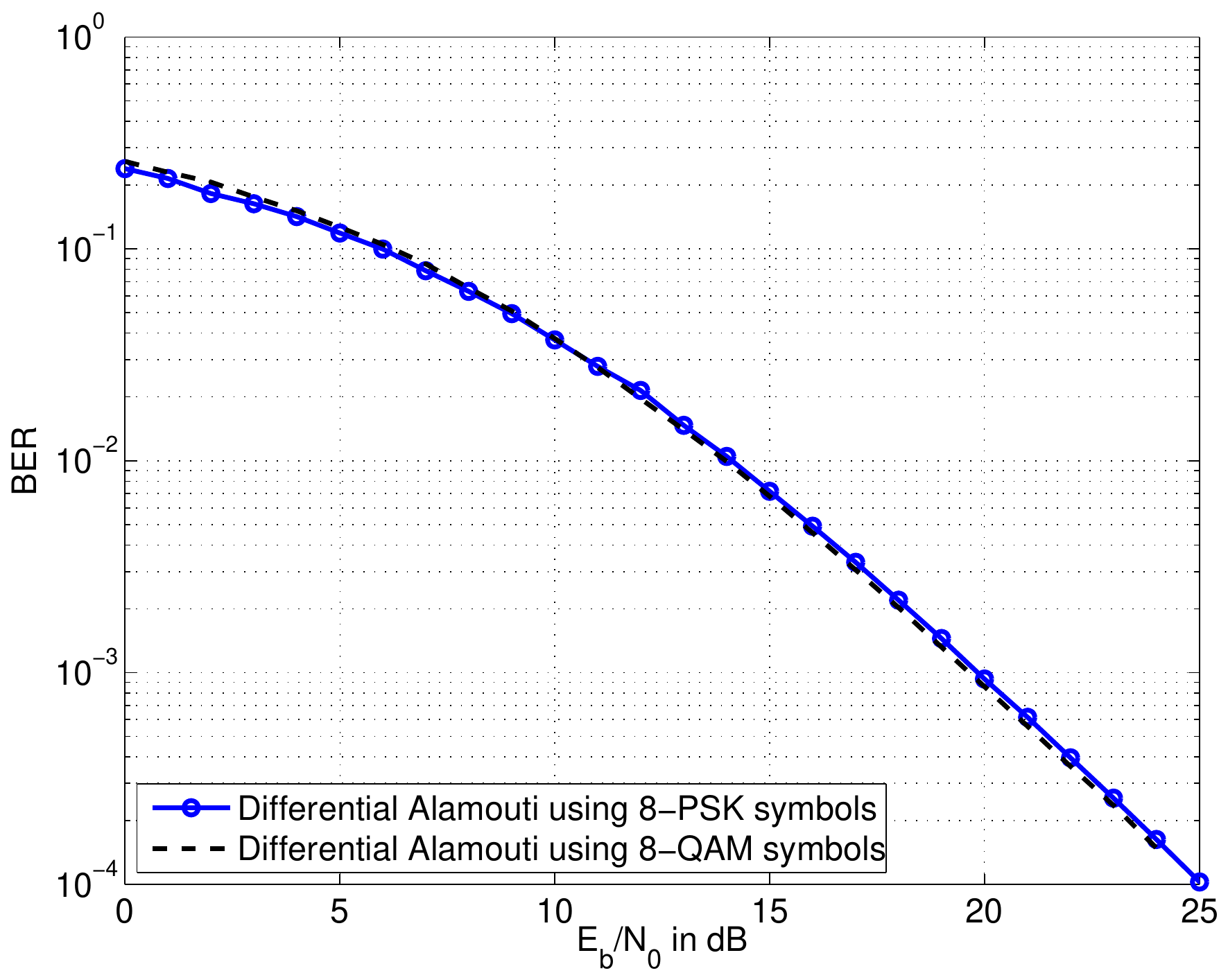}}
\subfloat[$\eta=4$ bits/s/Hz]{\label{subfig:M_2_non_unitary_4_bits_s_Hz}\includegraphics[width=0.45\textwidth]{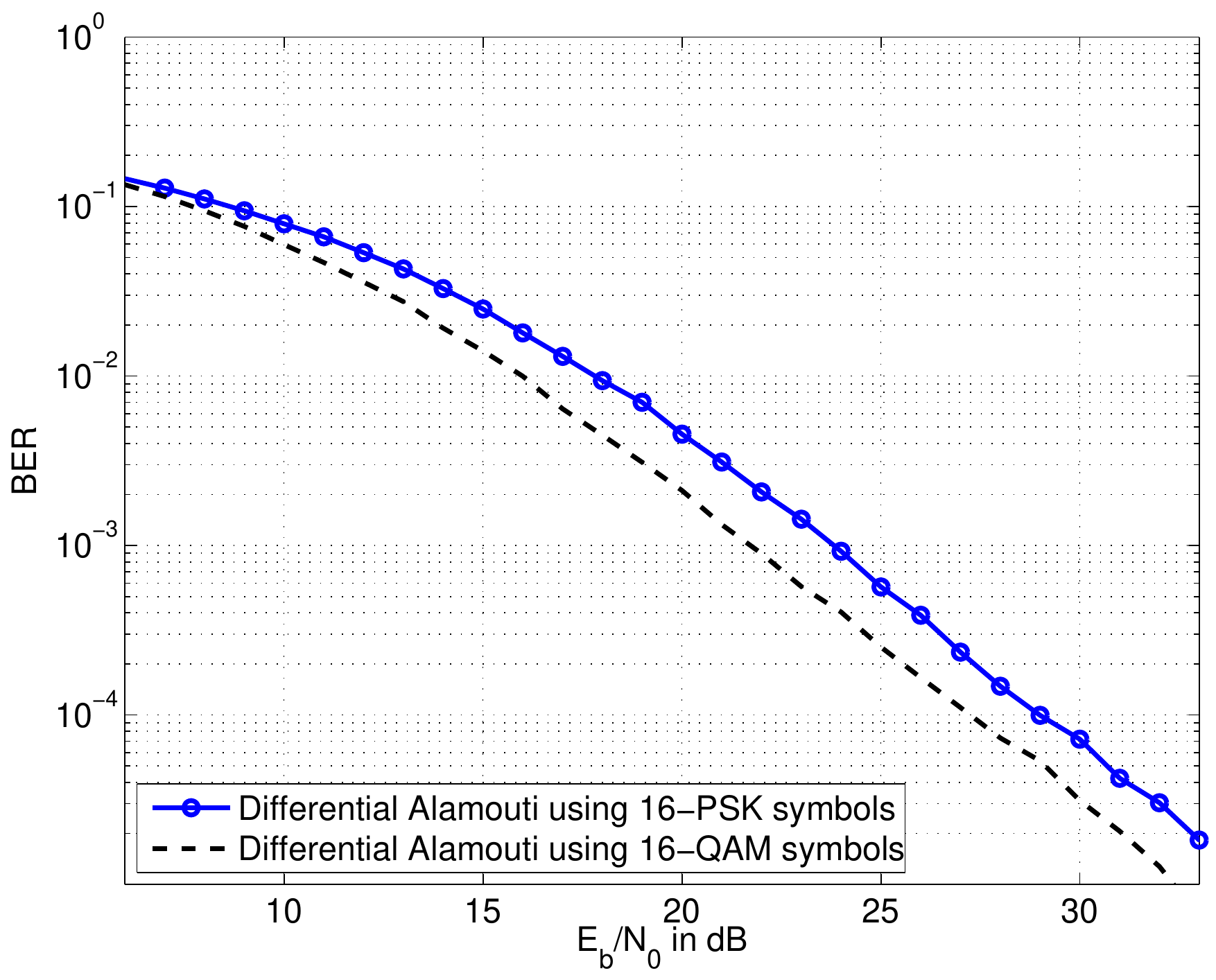}}
\caption{BER curves for $2\!\times\!1$ differential OSTBC systems that use Alamouti's code matrix in (\ref{eq:Alamouti}) using ML decoder in case of PSK symbols and near-optimal decoder in case of QAM symbols.}
  \label{fig:M_2_non_unitary_4_bits_s_Hz}
\end{figure}

\begin{figure}[htp]
 \centering
  \subfloat[$\eta=2$ bits/s/Hz]{\label{subfig:M_4_non_unitary_2_bits_s_Hz}\includegraphics[width=0.43\textwidth]{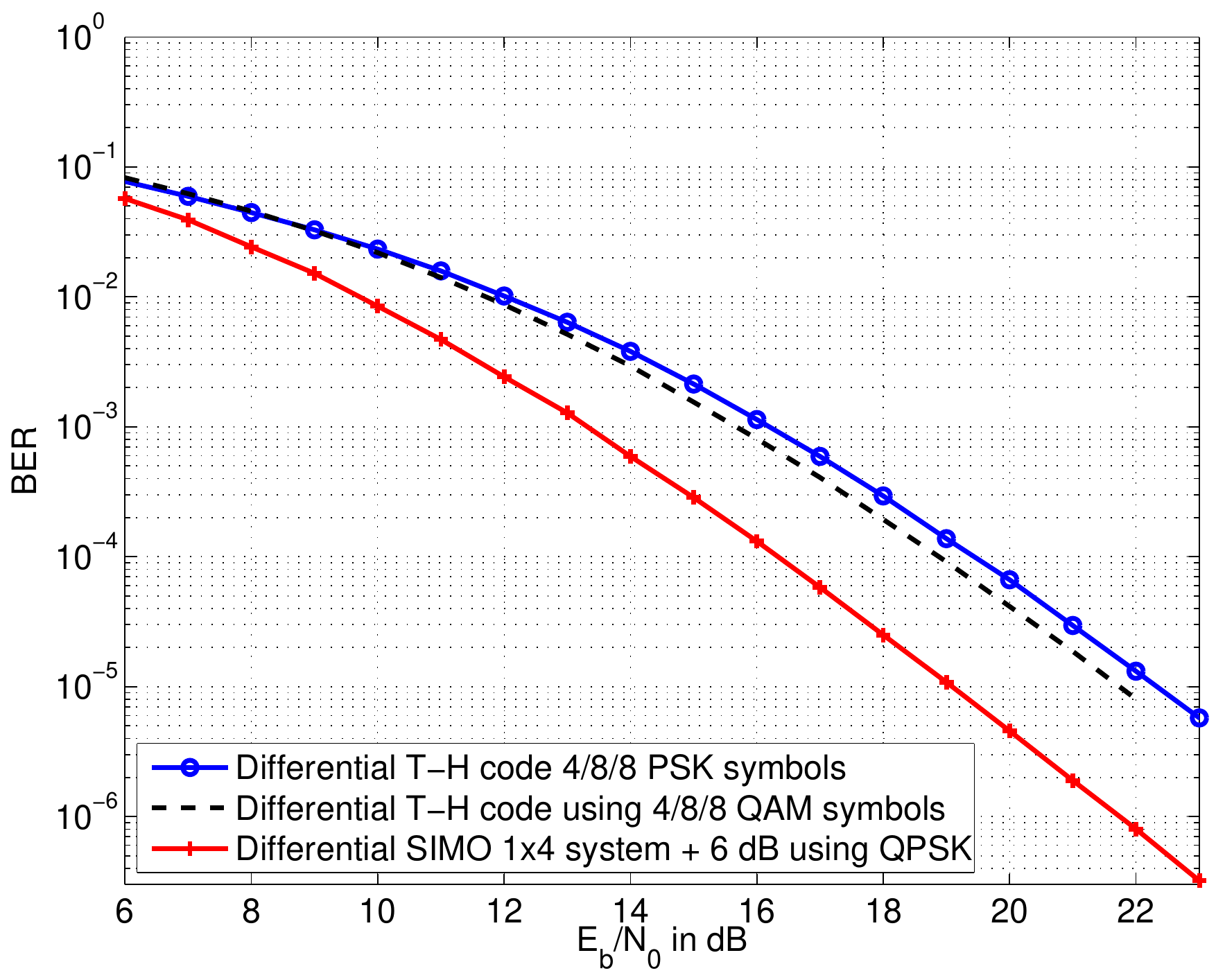}}\\
  \subfloat[$\eta=3$ bits/s/Hz]{\label{subfig:M_4_non_unitary_3_bits_s_Hz}\includegraphics[width=0.43\textwidth]{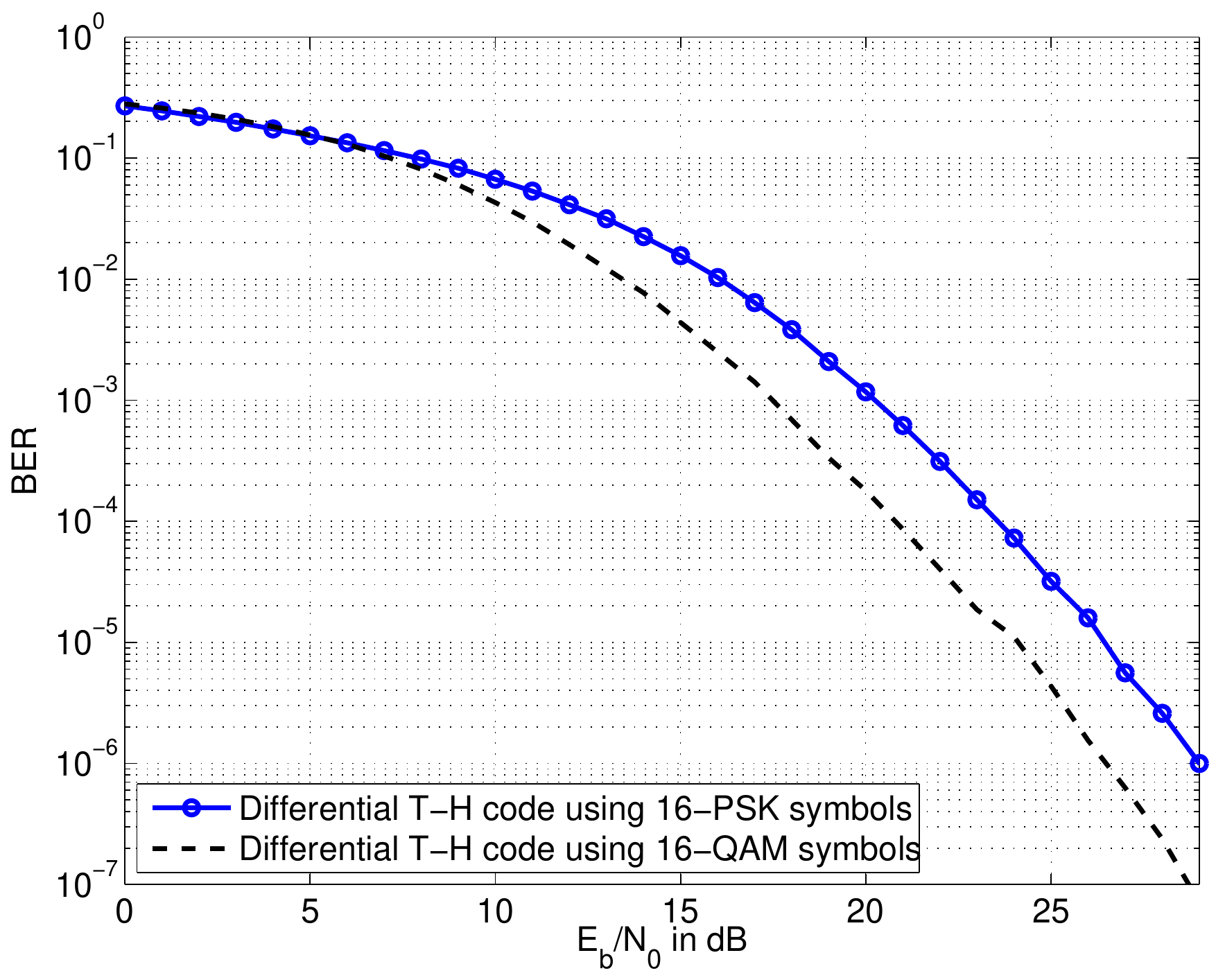}}
  \subfloat[$\eta=4$ bits/s/Hz]{\label{subfig:M_4_non_unitary_4_bits_s_Hz}\includegraphics[width=0.43\textwidth]{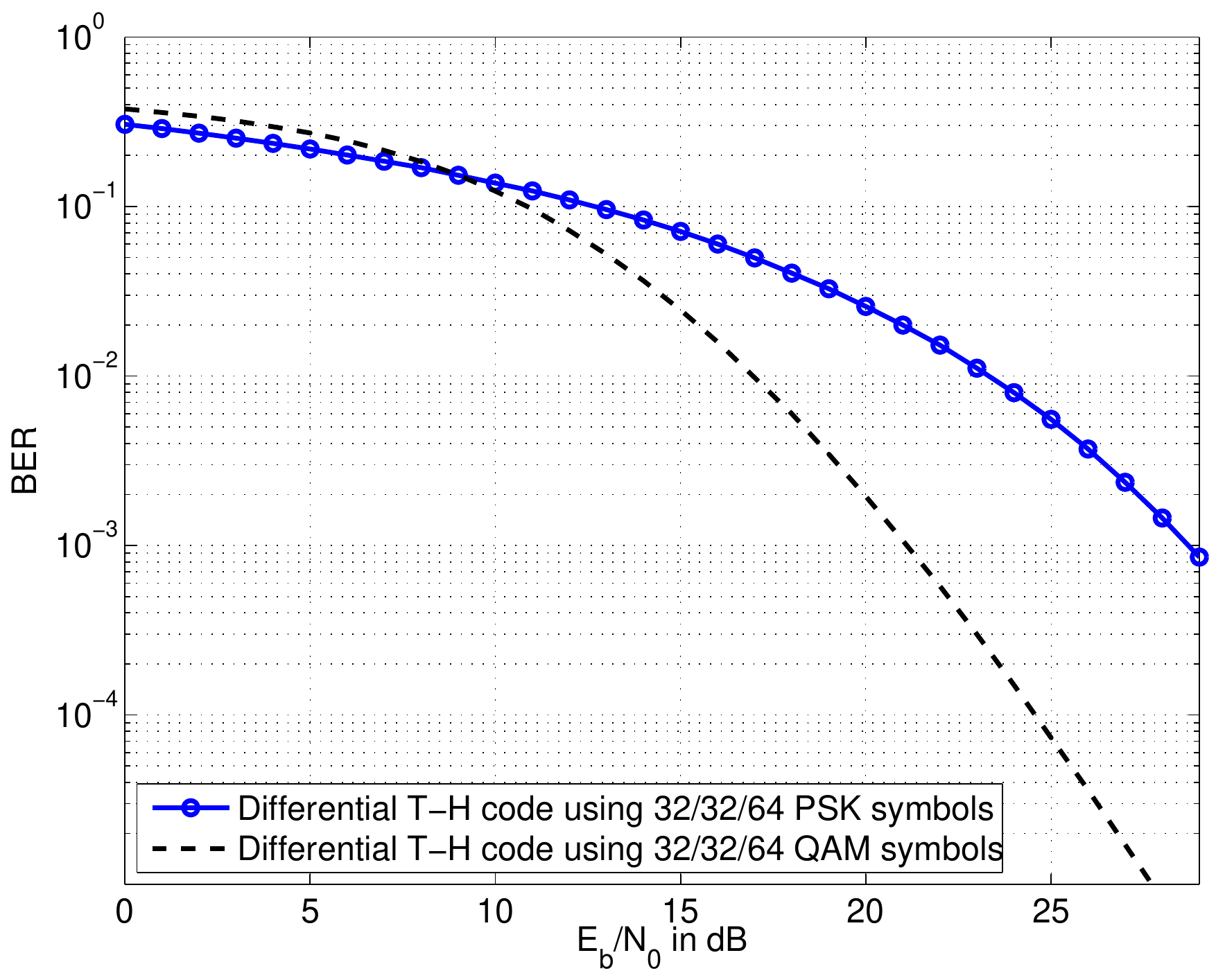}}
  \caption{BER curves for $4\!\times\!1$ differential OSTBC systems that use T-H code matrix in (\ref{eq:T_H}) using ML decoder in case of PSK symbols and near-optimal decoder in case of QAM symbols.}
  \label{fig:M_4_non_unitary}
\end{figure}


\abbrev{WIMAX}{Worldwide Interoperability for Microwave Access}
\abbrev{CDMA}{Code Division Multiple Access}
\abbrev{W-CDMA}{Wideband Code Division Multiple Access}
\abbrev{SCSD}{Single Complex Symbol Decoding}
\abbrev{SRSD}{Single Real Symbol Decoding}
\abbrev{T-H}{Tirkkonen and Hottinen}

\abbrev{DUSTM}{Differential Unitary Space-Time Modulation}
\abbrev{PAM}{Pulse Amplitude Modulation}

\chapter{Quasi Orthogonal Space-Time Block Codes}
\label{chap:QOSTBCs}
In the previous chapter, the class of orthogonal STBCs has shown to be an attractive transmit diversity approach due to its low decoding complexity and its ability to achieve full transmit diversity  for any number of transmitting antennas and any complex signal constellation. However, providing full code rate with complex OSTBCs is not possible for more than two transmit antennas. One way to achieve higher rates with STBCs is to relax the orthogonality condition of OSTBCs arriving at a new class of STCs known as Quasi-Orthogonal Space-Time Block Codes (QOSTBC). This new class was first presented by Jafarkhani in \cite{Jafarkhani_QOSTBC_2001} and independently by  Tirkkonen et al. in \cite{T_H_QOSTBC_2000}.\\

Unlike in OSTBCs, in QOSTBCs not all columns of the code matrix are orthogonal. The main idea of QOSTBCs lies in dividing the columns of the code matrix into different groups. Columns of different groups are mutually orthogonal, whereas columns within same group are not. This construction allows a higher code rate than that achieved by OSTBCs for the same number of transmit antennas. For most QOSTBCs presented in the literature, the improvement in the code rate comes at the expense of a higher decoding complexity. For four transmit antennas, the QOSTBCs proposed in \cite{Jafarkhani_QOSTBC_2001} and \cite{T_H_QOSTBC_2000} decode two complex symbols jointly, the so-called pair-wise complex symbol decoding. Additionally, these codes do not achieve full diversity. Later on, studies like in \cite{Xia_QOSTBCs_2002} and \cite{Papadias_QOSTBC_2003} showed the possibility of achieving full diversity for QOSTBCs by rotating the constellation. These codes however still require the joint detection of two complex symbols.\\

In 2004, Yuen et al. proposed in \cite{MDC_QOSTBC_2004} a construction of QOSTBCs that is based on interleaving the real and imaginary components of the symbols allowing the decoder to decouple all symbols achieving SCSD. For this reason, such codes were named Minimum Decoding Complexity QOSTBCs (MDC-QOSTBCs). With these codes, full diversity is also achievable with appropriate constellation rotation. In \cite{Yuen_MDC_QOSTBC_2005}, the same authors have shown a systematic approach for designing MDC-QOSTBCs by extending OSTBCs. They have also shown that this code achieves full rate for four transmit antennas and rate $\nicefrac{3}{4}$ for eight transmit antennas. Compared to other QOSTBCs proposed in the literature, MDC-QOSTBCs result in a marginal performance loss.\\

Owing to their reduced decoding complexity and good performance, we consider MDC-QOSTBCs to achieve a good rate-performance-complexity trade off. Consequently, this chapter will focus on this type of QOSTBCs. In Section \ref{s:MDC_QOSTBCs}, we explain the algebraic construction upon which MDC-QOSTBCs are built, as well as the full diversity condition for these codes. Section \ref{s:Orthogonalized_MDC_QOSTBCs} shows an approach used by the same authors in \cite{Yuen_DQOSTBC_2006} for differentially encoding MDC-QOSTBCs. This approach \emph{orthogonalize} the code matrix by imposing some conditions on the constellation. Such conditions will be shown to achieve degraded performance for high transmission rates.\\

Differential encoding --without orthogonalization-- for any QOSTBC remained unclear until Zhu and Jafarkhani published in \cite{Zhu_Jafarkhani_Diff_QOSTBCs_2005} one possible way out for using any QOSTBC in the differential domain. This differential transmission methodology will be explained in Section \ref{s:DQOSTBC}. The used QOSTBC in \cite{Zhu_Jafarkhani_Diff_QOSTBCs_2005} however requires pair-wise complex symbol decoding. Motivated by the good features of MDC-QOSTBCs, we desire to extend them to the differential framework using the differential encoding approach in \cite{Zhu_Jafarkhani_Diff_QOSTBCs_2005}. To the best of our knowledge, MDC-QOSTBCs have not been used in differential non-coherent systems for arbitrary complex signal constellation. In Section \ref{ss:QOSTBC_full_diversity}, we propose the differential version of MDC-QOSTBC for four transmit antennas that achieves full rate, full diversity and requires SCSD. Then in Section \ref{ss:QOSTBC_half_diversity}, a half-diversity differential MDC-QOSTBC that achieves SRSD with rectangular QAM constellation is proposed.
\section{Minimum Decoding Complexity QOSTBCs}
\label{s:MDC_QOSTBCs}
This section covers the construction and the properties of MDC-QOSTBCs proposed in \cite{MDC_QOSTBC_2004,Yuen_MDC_QOSTBC_2005,Yuen_MDC_QOSTBC_2005_V2} to serve all subsequent sections  of this chapter.
In \cite{MDC_QOSTBC_2004}, Yuen et al. have shown that the algebraic construction of an MDC-QOSTBC with even number of transmit antennas is based on extending an OSTBC that has half the number of transmit antennas and sends half the symbols in half the time slots. Namely, an MDC-QOSTBC that encodes $K$ symbols over $M$ transmit antennas in $T$ time slots is constructed from an OSTBC that encodes $\nicefrac{K}{2}$ symbols over $\nicefrac{M}{2}$ transmit antennas in $\nicefrac{T}{2}$ time slots. If the dispersion matrices of the $(\nicefrac{M}{2},\nicefrac{K}{2},\nicefrac{T}{2})$ OSTBC are denoted as $\underline{\V{U}}_i$ and $\underline{\V{Q}}_i$, then the dispersion matrices of the $(M,K,T)$ MDC-QOSTBC can be constructed by the following four mapping rules

\begin{equation}
\left.
\begin{aligned}
 \text{(i)}& \,\V{U}_i=\begin{bmatrix}\underline{\V{U}}_i&\V{0}\\\V{0}&\underline{\V{U}}_i\end{bmatrix}\,\,\,\,\,\,\,\,\,\,\:\,\,&\text{(ii)}\,\:\:\V{Q}_i=\begin{bmatrix}\V{0}&j\underline{\V{U}}_i\\j\underline{\V{U}}_i&\V{0}\end{bmatrix}\\
 \text{(iii)}&\, \V{U}_{i+\frac{K}{2}}=\begin{bmatrix}j\underline{\V{Q}}_i&\V{0}\\\V{0}&j\underline{\V{Q}}_i\end{bmatrix}\,\,\,\,\,\,&\text{(iv)}\,\V{Q}_{i+\frac{K}{2}}=\begin{bmatrix}\V{0}&\underline{\V{Q}}_i\\\underline{\V{Q}}_i&\V{0}\end{bmatrix}
\end{aligned}
\right. \hspace{1cm}1\leq i\leq\frac{K}{2}
\label{eq:MDC_construction}
\end{equation}

Using the properties of OSTBCs in (\ref{eq:OSTBC_disp_matrices_properties}) together with the above construction, the set of dispersion matrices $\V{U}_i$ and $\V{Q}_i$ of MDC-QOSTBCs can be shown to satisfy
\begin{equation}
\label{eq:MDC_QOSTBC_constraints}
\left.
\begin{aligned}
\text{(i)}\,\,&\V{U}_i\hr\V{U}_d&=&-\V{U}_d\hr\V{U}_i,\\
\text{(ii)}\,\,&\V{Q}_i\hr\V{Q}_d&=&-\V{Q}_d\hr\V{Q}_i\\
\text{(iii)}\,\,&\V{U}_i\hr\V{Q}_d&=&\,\V{Q}_d\hr\V{U}_i 
\end{aligned}
\right. \hspace{1cm}1\leq i\neq d\leq K
\end{equation}
Although the above MDC-QOSTBCs' properties look like the properties of OSTBCs in (\ref{eq:OSTBC_disp_matrices_properties}, ii-iii), (\ref{eq:MDC_QOSTBC_constraints}, iii) holds only for $i\neq d$. \\

In fact, for any QOSTBC to be able to decouple all symbols and achieve SCSD, its dispersion matrices must comply with the properties in (\ref{eq:MDC_QOSTBC_constraints}) \cite{Yuen_book}. Therefore, these properties are referred to as MDC-QO constraints. \\

The code rate of $(\nicefrac{M}{2},\nicefrac{K}{2},\nicefrac{T}{2})$ OSTBC is $\frac{K/2}{T/2}\eq\frac{K}{T}$ and that of $(M,K,T)$ MDC-QOSTBC is also $\frac{K}{T}$. In conclusion, an MDC-QOSTBC achieves the same code rate as the half-size OSTBC used to construct it. Therefore, an MDC-QOSTBC can achieve full rate for four transmit antennas and rate $\nicefrac{3}{4}$ for eight transmit antennas.\\

Based on the construction rules and the properties of an MDC-QOSTBC, its code matrix satisfies (see the proof in (\ref{eq:V_hr_V_MDC}))
\begin{equation}
\label{eq:V_hr_V_MDC_QOSTBC}
\V{V}\hr\V{V}=\frac{\alpha}{K}\V{I}_M+\frac{\beta}{K}\begin{bmatrix}\V{0}&\V{I}_{\frac{M}{2}}\\\V{I}_{\frac{M}{2}}&\V{0} \end{bmatrix},
\end{equation}
where
\begin{equation} 
 \label{eq:alpha_beta}
\alpha\eq\sum\limits_{i=1}^{K}{|x_i|^2},\,\,\beta\eq2\sum\limits_{i=1}^{\frac{K}{2}}{-x_i^R x_i^I+x_{i+\frac{K}{2}}^R x^I_{i+\frac{K}{2}}}.
\end{equation}
To study the diversity achieved by these codes, one needs to get an expression for the minimum determinant of the distance matrices $\V{D}_{ll'}, 0\leq l\neq l'\leq L-1$. Here, we omit the subscript $ll'$ for simplicity. Such an expression for the considered MDC-QOSTBCs was proved in (\ref{eq:det_min_MDC_QOSTBC}) to be
\begin{equation}
\label{eq:diversity_condition}
\det(\V{D}\hr\V{D})\Big|_{\text{min}}=\frac{1}{K}[(\Delta x^R)^2-(\Delta x^I)^2]^M
\end{equation}
where $\Delta x^R\text{ and }\Delta x^I$ represent the difference in the real and the imaginary components between two constellation points. This indicates that full diversity is achieved only when the absolute difference between the real parts of any two points in the constellation is not the same as the absolute difference between their imaginary parts.\\

Alamouti's code matrix in (\ref{eq:Alamouti}), whose $\underline{\V{U}}_i$ and $\underline{\V{Q}}_i$ dispersion matrices are

\begin{equation}
\underline{\V{U}}_1=\begin{bmatrix}1&0\\0&1\end{bmatrix},\,
\underline{\V{Q}}_1=\begin{bmatrix}1&0\\0&-1\end{bmatrix},\,
\underline{\V{U}}_2=\begin{bmatrix}0&1\\-1&0\end{bmatrix},\,
\underline{\V{Q}}_2=\begin{bmatrix}0&1\\1&0\end{bmatrix},
\end{equation}
can be extended using (\ref{eq:MDC_construction}) to construct the dispersion matrices of the four-transmit antenna MDC-QOSTBC as
\begin{equation*}
\begin{aligned}
\V{U}_1&=\left[\begin{array}{cc:cc}1&0&0&0\\0&1&0&0\\\hdashline0&0&1&0\\0&0&0&1\end{array}\right],\,
\V{Q}_1=\left[\begin{array}{cc:cc}0&0&j&0\\0&0&0&j\\\hdashline j&0&0&0\\0&j&0&0\end{array}\right],\,
\V{U}_2=\left[\begin{array}{cc:cc}0&1&0&0\\-1&0&0&0\\\hdashline0&0&0&1\\0&0&-1&0\end{array}\right],\,
\V{Q}_2=\left[\begin{array}{cc:cc}0&0&0&j\\0&0&-j&0\\\hdashline0&j&0&0\\-j&0&0&0\end{array}\right],\\
\V{U}_3&=\left[\begin{array}{cc:cc}j&0&0&0\\0&-j&0&0\\\hdashline0&0&j&0\\0&0&0&-j\end{array}\right],
\V{Q}_3=\left[\begin{array}{cc:cc}0&0&1&0\\0&0&0&-1\\\hdashline1&0&0&0\\0&-1&0&0\end{array}\right],
\V{U}_4=\left[\begin{array}{cc:cc}0&j&0&0\\j&0&0&0\\\hdashline0&0&0&j\\0&0&j&0\end{array}\right],
\V{Q}_4=\left[\begin{array}{cc:cc}0&0&0&1\\0&0&1&0\\\hdashline0&1&0&0\\1&0&0&0\end{array}\right],
\end{aligned}
\end{equation*}
and using the disperive form in (\ref{eq:dispersive_form2}), the full-rate MDC-QOSTBC code matrix for four transmit antennas becomes 

\begin{equation} 
\V{V}=\frac{1}{\sqrt{4}}
\left[
\begin{array}{cc:cc}
 x_1^R+jx_3^R & x_2^R+jx_4^R & -x_1^I+jx_3^I & -x_2^I+jx_4^I\\
-x_2^R+jx_4^R & x_1^R-jx_3^R & x_2^I+jx_4^I & -x_1^I-jx_3^I\\\hdashline
-x_1^I+jx_3^I & -x_2^I+jx_4^I & x_1^R+jx_3^R & x_2^R+jx_4^R \\
x_2^I+jx_4^I & -x_1^I-jx_3^I & -x_2^R+jx_4^R & x_1^R-jx_3^R \\ 
\end{array}
\right].
\label{eq:four_MDC_QOSTBC}
\end{equation}
Similarly extending the $(4,3,4)$ T-H code in (\ref{eq:T_H}), the code matrix of the $(M\eq8,K\eq6,T\eq8)$ rate $\nicefrac{3}{4}$ eight-transmit antenna MDC-QOSTBC becomes
\setlength{\arraycolsep}{2.5pt}
\begin{equation}
\hspace{-1cm}
\V{V}=\frac{1}{\sqrt{6}}
 \left[
\begin{array}{cccc:cccc}
  x_1^R+jx_4^R & x_2^R+jx_5^R & x_3^R+jx_6^R &0 &-x_1^I+jx_4^I & -x_2^I+jx_5^I & -x_3^I+jx_6^I & 0 \\
  -x_2^R+jx_5^R & x_1^R-jx_4^R & 0 & -x_3^R-jx_6^R & x_2^I+jx_5^I & -x_1^I-jx_4^I & 0 & x_3^I-jx_6^I\\
  -x_3^R+jx_6^R & 0 & x_1^R-jx_4^R & x_2^R+jx_5^R & x_3^I+jx_6^I & 0 & -x_1^I-jx_4^I & -x_2^I+jx_5^I \\
  0 & x_3^R-jx_6^R & -x_2^R+jx_5^R & x_1^R+jx_4^R & 0 & -x_3^I+jx_6^I & x_2^I+jx_5^I & -x_1^I+x_4^I\\\hdashline
  -x_1^I+jx_4^I & -x_2^I+jx_5^I & -x_3^I+jx_6^I & 0 & x_1^R+jx_4^R & x_2^R+jx_5^R & x_3^R+jx_6^R &0\\
  x_2^I+jx_5^I & -x_1^I-jx_4^I & 0 & x_3^I-jx_6^I & -x_2^R+jx_5^R & x_1^R-jx_4^R & 0 & -x_3^R-jx_6^R\\
  x_3^I+jx_6^I & 0 & -x_1^I-jx_4^I & -x_2^I+jx_5^I & -x_3^R+jx_6^R & 0 & x_1^R-jx_4^R & x_2^R+jx_5^R\\
  0 & -x_3^I+jx_6^I & x_2^I+jx_5^I & -x_1^I+x_4^I & 0 & x_3^R-jx_6^R & -x_2^R+jx_5^R & x_1^R+jx_4^R
\end{array}
\right].
\label{eq:eight_MDC_QOSTBC}
\end{equation}
\setlength{\arraycolsep}{5pt}
Interesting and useful insights can be observed from the MDC-QOSTBC code matrices as those shown in (\ref{eq:four_MDC_QOSTBC}) and (\ref{eq:eight_MDC_QOSTBC}). Clearly the construction defined in (\ref{eq:MDC_construction}) results in interleaving the real and imaginary components of the information symbols. Such interleaving is one possible way for achieving the single complex symbol decodability. Another observation is that the resulting code matrices have the so-called "ABBA" structure. In such a structure, the matrix is divided into four blocks, where each two diagonally opposite blocks are the same.  In fact, several QOSTBCs proposed in the literature posses such a structure. Moreover, each block matrix has the same form as the half-size OSTBC used to construct it. Namely, the block matrices in (\ref{eq:four_MDC_QOSTBC}) have an Alamouti structure shown in (\ref{eq:Alamouti}), and those in (\ref{eq:eight_MDC_QOSTBC}) have a T-H structure shown in (\ref{eq:T_H}).
\section{Orthogonalized Differential MDC-QOSTBCs}
\label{s:Orthogonalized_MDC_QOSTBCs}
The MDC-QOSTBC scheme described in Section \ref{s:MDC_QOSTBCs} was applied in \cite{Yuen_MDC_QOSTBC_2005,Yuen_MDC_QOSTBC_2005_V2} for coherent systems. However, several practical perspectives like fast channel variation and/or the requirement of low-complexity receivers may demand avoiding the use of channel estimation required by coherent systems. With OSTBCs, non-coherent detection was made possible through the differential encoding equations defined in (\ref{eq:Diff_encoding_unitary_OSTBC}) and (\ref{eq:Diff_encoding_non_unitary}). When considering QOSTBCs, it is not really easy to see how these codes can be encoded differentially. In particular, it is not easy to think of a differential encoding equation for a QOSTBC that maintains the properties of the code, and ensures a constant average power per transmit block. For this reason the authors of the MDC-QOSTBC scheme thought of \emph{orthogonalizing} the code as a way out to use it in the differential non-coherent domain. This is basically achieved by imposing some constraints on the signal constellation to make the code orthogonal, i.e. the resulting code is only \emph{conditionally} orthogonal and will be referred to as Orthogonalized MDC-QOSTBC (OMDC-QOSTBC).  This section will describe such an approach which was proposed in \cite{Yuen_DQOSTBC_2006} and will compare its error performance with the \emph{unconditionally} orthogonal STBCs that use arbitrary signal constellations.\\

It is worth mentioning that the full diversity criterion was proved to be the determinant (or rank) criterion for any STBC in coherent systems, but only for \emph{unitary} STBCs in non-coherent systems. Nevertheless, we use --without proof-- the same criterion for non-unitary and for quasi-orthogonal STBCs in the investigated non-coherent systems. The error rate curves will be shown to agree with our assumption that the rank criterion is still the full diversity criterion for non-unitary and QOSTBCs in non-coherent systems.

\subsection{Constellation Design}
For an $\!M\times\!M$ code matrix $\V{V}$ to be orthogonal, it must satisfy
\begin{equation*}
 \V{V}\hr\V{V}=a^2\V{I}_M,
\end{equation*}
for some scalar $a$. From (\ref{eq:V_hr_V_MDC_QOSTBC}), MDC-QOSTBCs can be made orthogonal by forcing $\beta$ in (\ref{eq:alpha_beta}) to be zero, i.e. requiring
\begin{equation}
 x_i^Rx_i^I\req x_{i+\frac{K}{2}}^Rx_{i+\frac{K}{2}}^I,\:\:\forall \,\,1\leq i\leq \frac{K}{2}.
\end{equation}
The easiest way for this to be satisfied is to ensure that all constellation points satisfy  
\begin{equation}
\label{eq:MDC_OSTBC_Condition1}
x^Rx^I\eq\nu, 
\end{equation}
where $x=x^R+jx^I$ is a point in the constellation being designed and $\nu$ is a constant. For a positive $\nu$, (\ref{eq:MDC_OSTBC_Condition1}) represents a hyperbola shown in Figure \ref{subfig:hyperbola}. Note that a negative $\nu$ can also be chosen resulting in an equivalent constellation design with a hyperbola in the other two quadrants of the complex space.\\

To satisfy the energy constraint in (\ref{eq:power_st}) or equivalently to satisfy $\E[\V{V}\hr\V{V}]\eq\V{I}_M$,  a normalization factor is used in the code matrices. In our work, such a factor assumes a unit average symbol energy. Therefore, the constellation points chosen from the loci of Figure \ref{subfig:hyperbola} must additionally satisfy
\begin{equation}
 \E[|x|^2]=1,
\label{eq:MDC_OSTBC_Condition2}
\end{equation}
which represents concentric circles of unit average power. The solutions to conditions (\ref{eq:MDC_OSTBC_Condition1}) and (\ref{eq:MDC_OSTBC_Condition2}) are the intersection points of the hyperbola with the concentric circles as shown in Figure \ref{subfig:hyperbola_circles} for the case of only two circles. The intersections are defined  by points $A_i$, $B_i$, $C_i$ and $D_i$, where $i$ is the circle number. The real and imaginary parts of these points are indicated by the  superscripts $R$ and $I$, respectively. For the same circle, the following is satisfied
\begin{equation}
 (A_i^R)^2+(A_i^I)^2=(B_i^R)^2+(B_i^I)^2
\label{eq:one_circle}
\end{equation}
Since $\nu=A_i^RA_i^I=B_i^RB_i^I$, then by subtracting $2\nu$ from both sides of (\ref{eq:one_circle}), we get
\begin{eqnarray}
 (A_i^R-A_i^I)^2&=&(B_i^R-B_i^I)^2\nonumber\\
 A_i^R-B_i^R&=&A_i^I-B_i^I\nonumber\\
\Delta^R&=&\Delta^I
\end{eqnarray}
which results in a zero determinant for the distance matrix in (\ref{eq:diversity_condition}). Therefore, points $A_i$ \& $B_i$ and similarly points $C_i$ \& $D_i$ of the same circle should not both exist in the constellation for the system to achieve full diversity. Additionally, for maximal constellation points separation, points $A$ and $C$ are chosen on one circle, then points $B$ and $D$ on the next circle, and so on resulting in the constellation shown in Figure \ref{subfig:Constellation_full_diversity} when two circles are considered. Since every circle contains two constellation points, the number of circles is half the alphabet size i.e. $\nicefrac{q}{2}$, and the condition in (\ref{eq:MDC_OSTBC_Condition2}) can be rewritten as
\begin{equation}
 \sum\limits_{i=1}^{\frac{q}{2}}{r_i^2}\req\frac{q}{2}
\label{eq:MDC_OSTBC_Condition2_ri}
\end{equation}

Next step is to maximize the coding gain by maximizing the minimum determinant in (\ref{eq:diversity_condition}). It is required to find the optimal choice of $\nu$ and of the circles radii $r_i$. The choice of $\nu$ governs the choice of the two angles $\theta_1$ and $\theta_2$ shown in Figure \ref{subfig:Constellation_full_diversity}. A clear proof in \cite{Yuen_DQOSTBC_2006} has shown that optimally $|\theta_1-\theta_2|=90^o$, or equivalently the optimal choice of $\nu$ is zero, making the constellation points lying on the x- or y-axis. Then the optimal choice of the circles radii $r_i,\,\,i\in\{1,...,\nicefrac{q}{2}\}$ was calculated. In \cite{Yuen_DQOSTBC_2006}, this optimization was done analytically for the 4-point constellation case, but only numerically for the 8-point constellation case since analytical optimization becomes more tedious in this case. The resulting optimal 4-point and 8-point constellations are shown in Figure \ref{fig:opt_constellation_OMDC_QOSTBC}.\\

These constellations can be used with the full-rate four-transmit antenna MDC-QOSTBC in (\ref{eq:four_MDC_QOSTBC}) or the $\frac{3}{4}$ rate eight-transmit antenna MDC-QOSTBC in (\ref{eq:eight_MDC_QOSTBC}). Thus for the four-transmit antenna case, the new code with the designed constellation is a full-rate full-diversity orthogonal STBC. This may sound as if this code violates the Hurwitz-Radon theorem, which states that for more than two transmit antennas, an STBC that achieves full rate, full diversity and orthogonality for all possible complex constellations can not exist. However, the theorem has not mentioned whether it is possible for such codes to exist for some \emph{specific} constellations (like the one used here). Therefore, the code used here with the designed constellation does not violate the Hurwitz-Radon theorem.

\begin{figure}
\vspace{-0.8cm}
 \centering
  \subfloat[hyperbola showing the loci of constellation points satisfying (\ref{eq:MDC_OSTBC_Condition1}).]{\label{subfig:hyperbola}\includegraphics[width=0.5\textwidth]{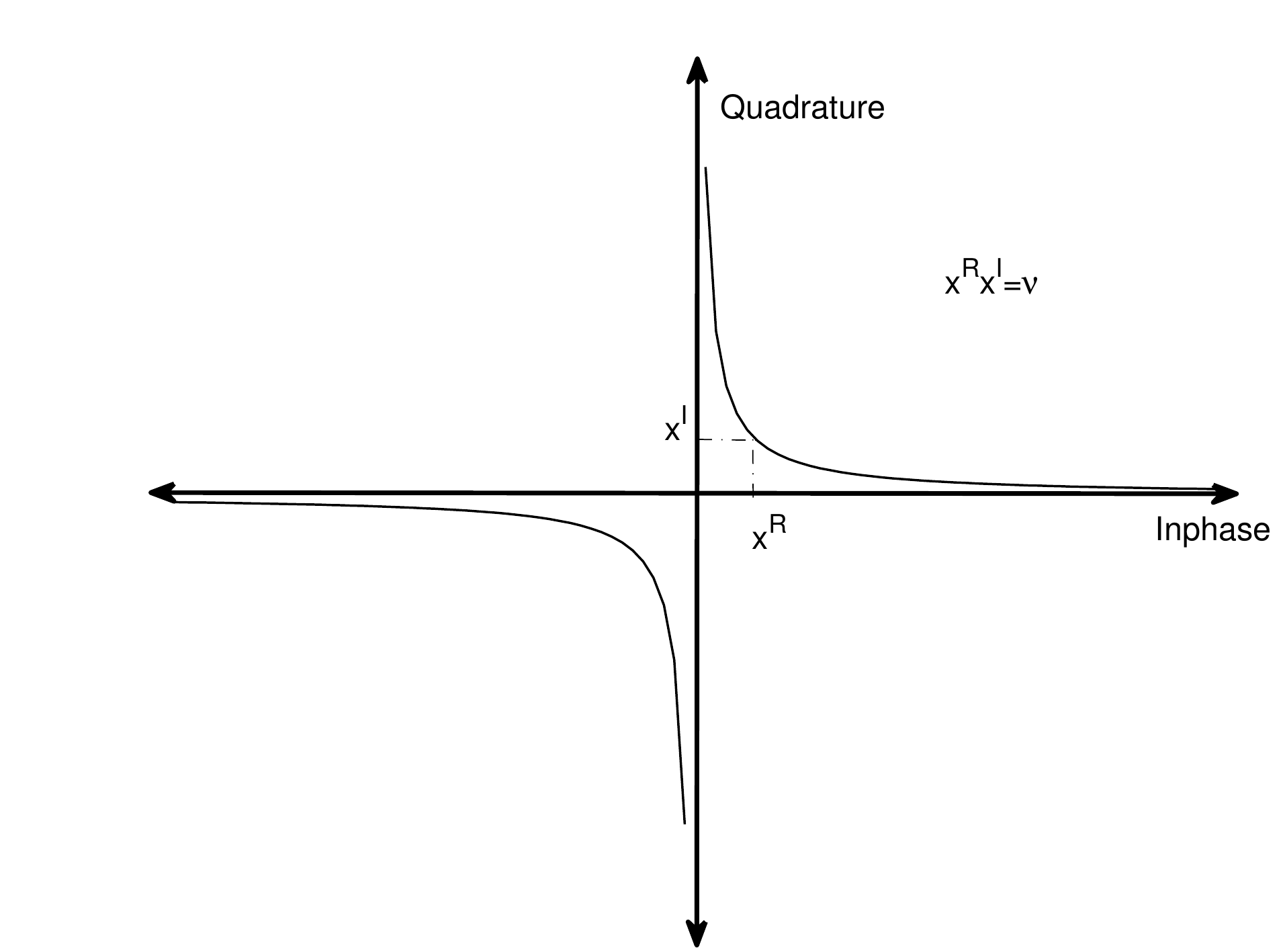}}
  \subfloat[Loci of constellation points satisfying (\ref{eq:MDC_OSTBC_Condition1}) and (\ref{eq:MDC_OSTBC_Condition2}).]{\label{subfig:hyperbola_circles}\includegraphics[width=0.49\textwidth]{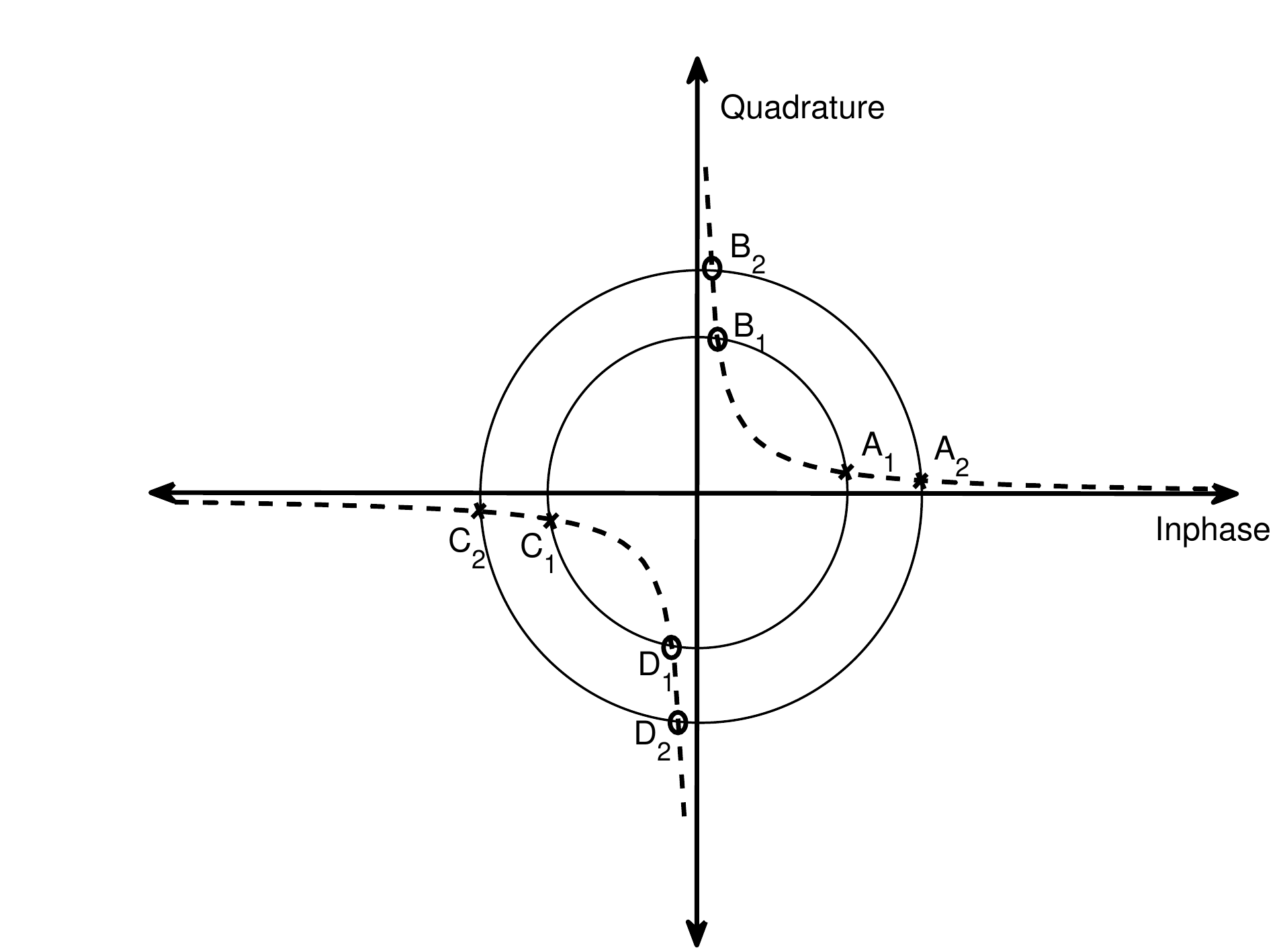}}\\\vspace{-0.2cm}
  \subfloat[A constellation that achieves full diversity for an MDC-QOSTBC]{\label{subfig:Constellation_full_diversity}\includegraphics[width=0.49\textwidth]{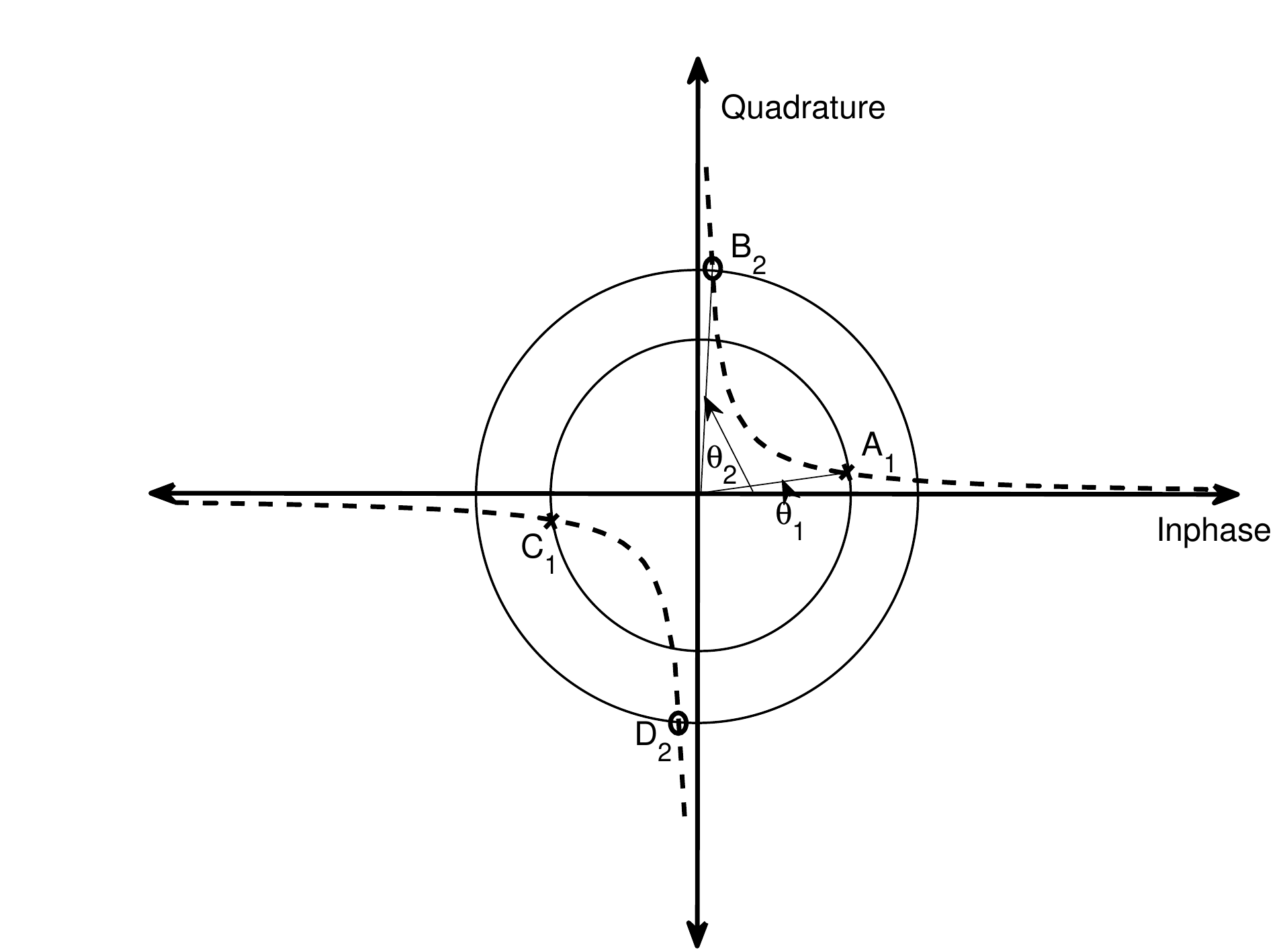}}
  \caption{Steps toward constellation design of OMDC-QOSTBC}
  \label{fig:constellation_loci}
\end{figure}

\begin{figure}
\vspace{-0.2cm}
 \centering
  \subfloat[Optimal 4-point constellation]{\label{subfig:4_point_constellation}\includegraphics[width=0.35\textwidth]{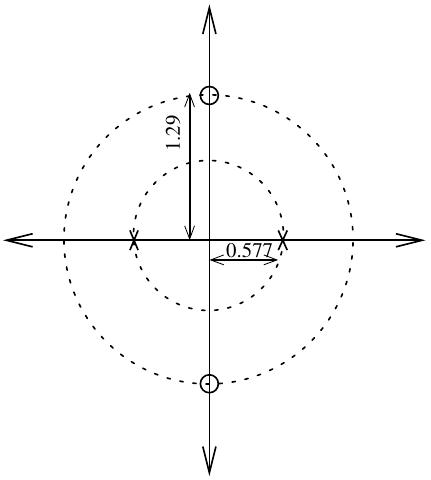}} \hspace{1cm}
\subfloat[Optimal 8-point constellation]{\label{subfig:8_point_constellation}\includegraphics[width=0.35\textwidth]{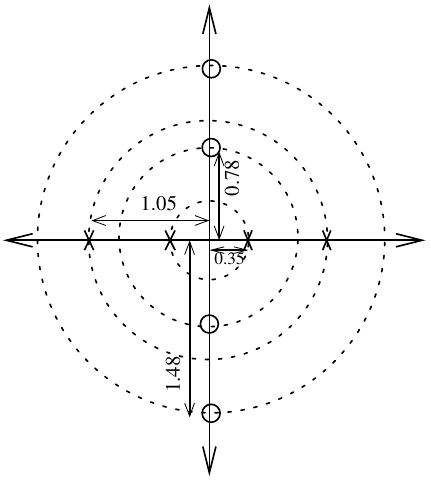}}
\caption{Optimal 4-point and 8-point constellations for the OMDC-QOSTBC}
\label{fig:opt_constellation_OMDC_QOSTBC}
\end{figure}
\subsection{Performance Analysis}
The MDC-QOSTBC with the constellation designed in the previous subsection is an orthogonal code with unequal energy for the constellation points. Thus the differential encoding in (\ref{eq:Diff_encoding_non_unitary}) and the non-coherent differential decoding in (\ref{eq:x_hat_non_unitary}) of non-unitary OSTBCs are applicable here. Simulations are performed for the four-transmit antenna code matrix in (\ref{eq:four_MDC_QOSTBC}) and the eight-transmit antenna code matrix in (\ref{eq:eight_MDC_QOSTBC}) using the 4-point and 8-point constellations shown in Figure \ref{fig:opt_constellation_OMDC_QOSTBC}. The channel used is a quasi-static flat Rayleigh fading channel.\\

For the four-transmit antenna case at $\unit[2]{bits/s/Hz}$, we compare in Figure \ref{fig:M_4_MDC_QOSTBC_vs_OSTBC_2bits} both the SER and the BER of the 4/8/8 PSK rate $\nicefrac{3}{4}$ OSTBC, the 4/8/8 QAM rate $\nicefrac{3}{4}$ OSTBC, the all 16-QAM rate $\nicefrac{1}{2}$ OSTBC and the full-rate OMDC-QOSTBC that uses the 4-point constellation shown in Figure \ref{subfig:4_point_constellation}. The rate $\nicefrac{1}{2}$ OSTBC uses a $4\!\times\!4$ code matrix that simply appends four Alamouti blocks, whereas the rate $\nicefrac{3}{4}$ OSTBC uses the T-H code matrix in (\ref{eq:T_H}).

\begin{figure}[htp]
 \centering
  \subfloat[SER]{\label{subfig:M_4_MDC_QOSTBC_vs_OSTBC_SER_2bits}\includegraphics[width=0.45\textwidth]{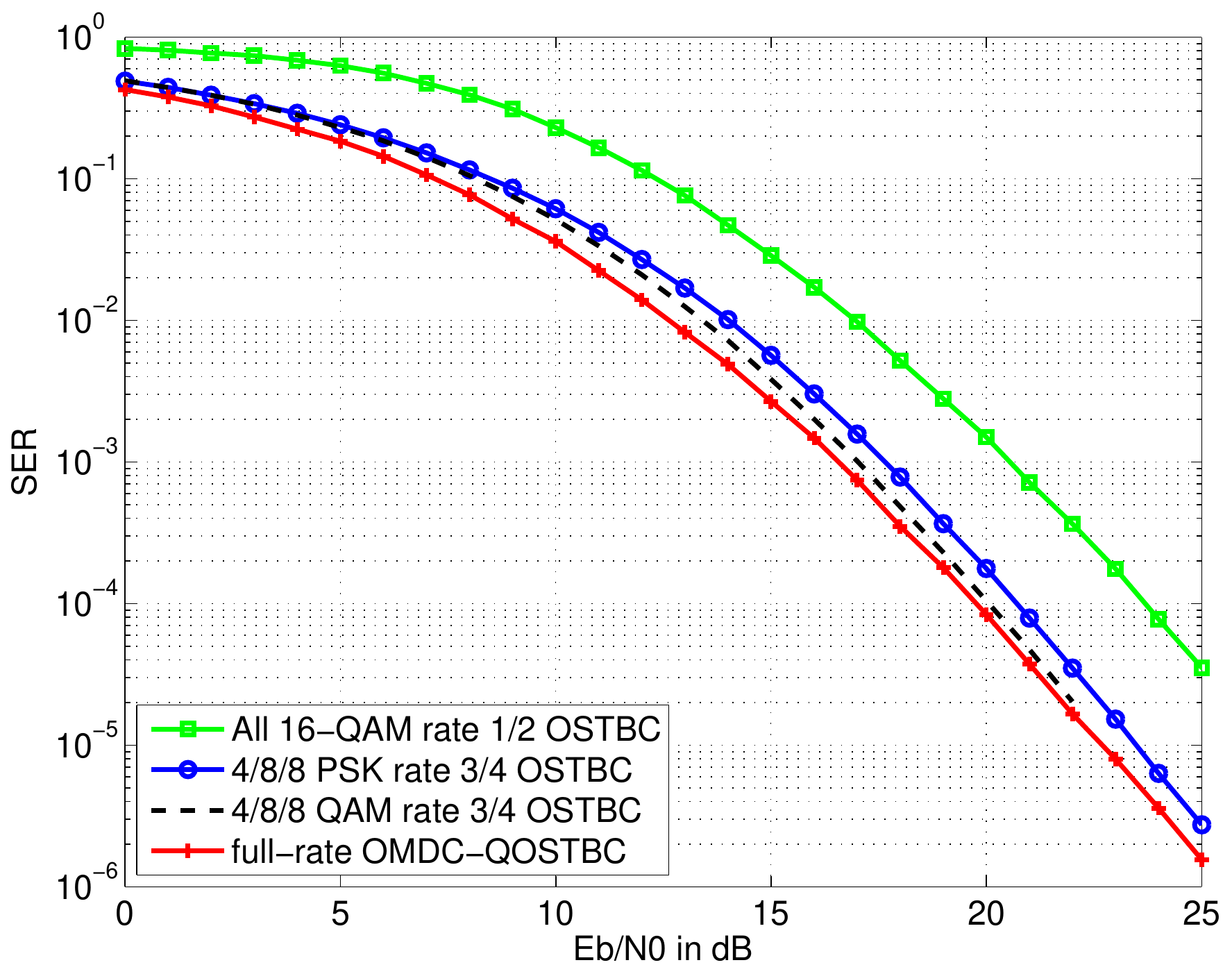}}
\subfloat[BER]{\label{subfig:M_4_MDC_QOSTBC_vs_OSTBC_BER_2bits}\includegraphics[width=0.45\textwidth]{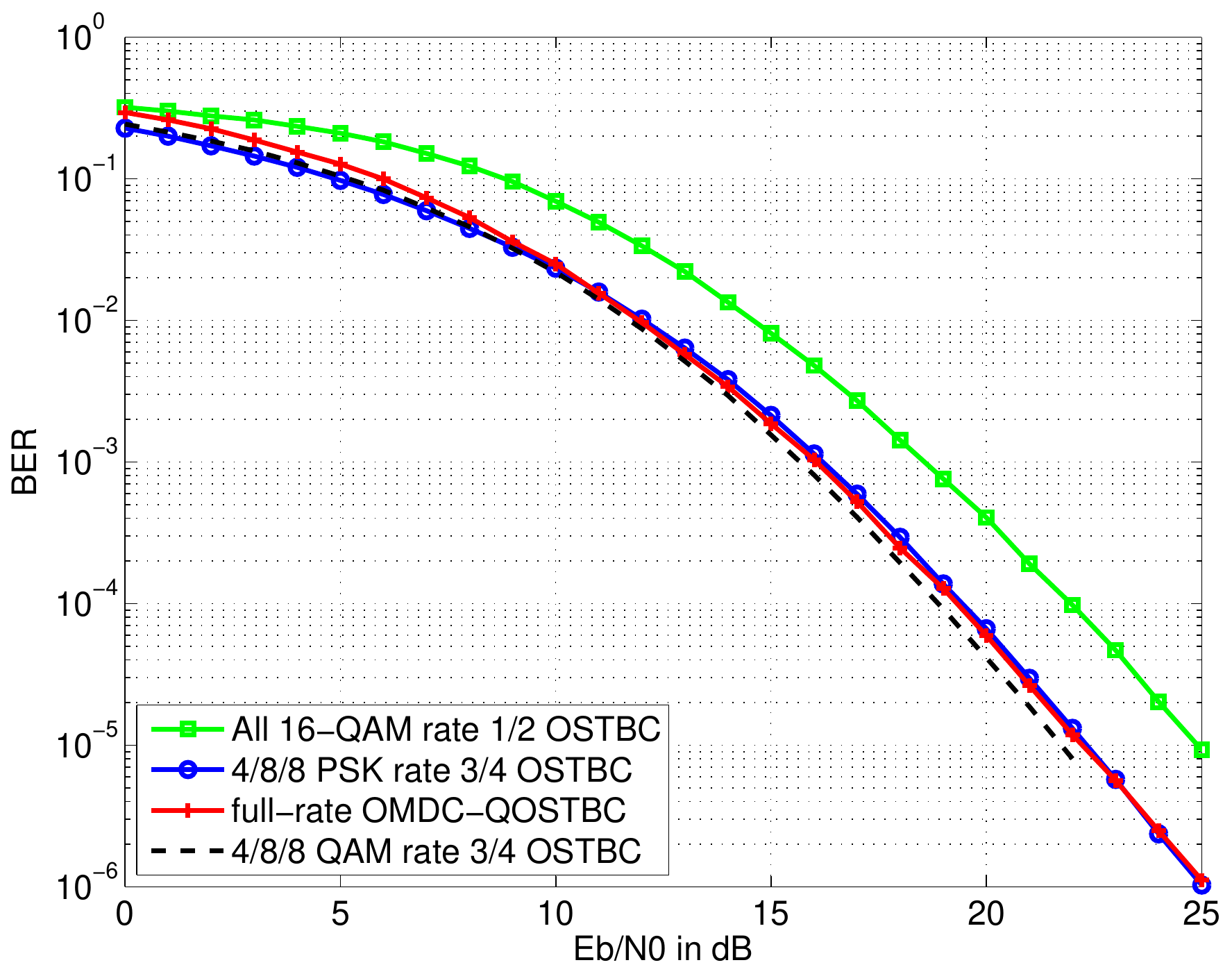}}
\caption{Comparing BER and SER of differential OMDC-QOSTBC with differential OSTBCs  for a $4\!\times\!1$ system at $\unit[2]{bits/s/Hz}$}
 \label{fig:M_4_MDC_QOSTBC_vs_OSTBC_2bits}
\end{figure}

The authors of the OMDC-QOSTBC scheme compared its performance only with the rate $\nicefrac{1}{2}$ OSTBC code that uses 16-QAM alphabet for all symbols. Clearly the performance of their proposed scheme is significantly better. However, it is also possible to use an OSTBC with a higher code rate (like the rate $\nicefrac{3}{4}$ T-H code) that requires a lower alphabet order and therefore potentially results in a better performance. Such a code however requires different alphabet sizes for the different symbols to achieve a spectral efficiency of $\unit[2]{bits/s/Hz}$. Here the alphabet size combination 4/8/8 is chosen. The 8-QAM used is the optimized one shown in Figure \ref{fig:Opt_8_QAM_OSTBC_constellation}. \\

Consider first the SER curves shown in Figure \ref{subfig:M_4_MDC_QOSTBC_vs_OSTBC_SER_2bits}, the OMDC-QOSTBC has the best performance of all other codes and clearly achieves full diversity. It achieves about $\unit[1]{dB}$ gain compared to the 4/8/8 PSK OSTBC, and only a slight improvement compared to the 4/8/8 QAM OSTBC, but it posses the advantage that all symbols are drawn from the same alphabet which simplifies the role of the Tx and the Rx. \\

The bit-to-symbol mapping used for all schemes is the gray coding. Nevertheless, the BER and the SER curves indicate that if a scheme shows a better SER performance compared to some other scheme, it will not necessarily show a better BER performance. For example, the SER performance of the OMDC-QOSTBC shows a $\unit[1]{dB}$ gain compared to that of the 4/8/8 PSK OSTBC, and yet they both show almost the same BER performance. One interpretation to such a behaviour is that the decoder of the OMDC-QOSTBC makes more correct symbol decisions than that of 4/8/8 PSK OSTBC, but a wrong decision in OMDC-QOSTBC results in more bit errors than a wrong decision in 4/8/8 PSK OSTBC. The number of bit errors that result from a symbol error is solely based on the bit-to-symbol mapping used as well as the constellation. The constellation of the OMDC-QOSTBC showed to be inferior to the standard PSK and QAM constellations which have less bit changes in the zone of the neighboring symbols for every symbol.\\

From the perspective of the end-to-end communication link, the BER is the actual measure of the link performance. From the BER curves, the 4/8/8 QAM OSTBC is marginally better in performance compared to the OMDC-QOSTBC. However, the search space of the OMDC-QOSTBC is $4\!\times\!4\eq16$ candidates and that of the 4/8/8 QAM code is $4\!+\!8\!+\!8\eq20$ candidates. Moreover, the OMDC-QOSTBC uses the same constellation for all symbols making the encoding and decoding simpler.  In conclusion, the OMDC-QOSTBC is so-far achieving the best performance-complexity trade off for a spectral efficiency of $\unit[2]{bits/s/Hz}$ in a four-transmit antenna system.\\

At a spectral efficiency of $\unit[3]{bits/s/Hz}$, Figure \ref{fig:M_4_MDC_QOSTBC_vs_OSTBC_3bits} shows the error performance of the OSTBC that uses 16-PSK symbols, the OSTBC that uses 16-QAM symbols and the OMDC-QOSTBC that uses the 8-point constellation in Figure \ref{subfig:8_point_constellation}. Through the SER curves, we see that the performance of the OMDC-QOSTBC is a little worse than that of the 16-QAM OSTBC indicating that the OMDC-QOSTBC scheme deteriorates for high transmission rates. This is because the constellation limits only two points on every circle --to achieve full diversity-- which does not utilize the complex space in the best way. When comparing the BER performance, the 16-QAM OSTBC has about $\unit[1]{dB}$ SNR advantage compared to the OMDC-QOSTBC. Furthermore, the decoding complexity of the 16-QAM OSTBC is less than that of the OMDC-QOSTBC. This is because the 16-QAM OSTBC can perform SRSD using the metric in (\ref{eq:decoder_rect_QAM}) which requires $4\!\times\!2\!\times\!3\eq24$ \footnote{calculated as 4 real alphabet size $\times$ 2 dimensions $\times$ 3 symbols per information block. When comparing the complexity of the different schemes, we compare the search space required for decoding one information block.} test candidates, whereas the OMDC-QOSTBC performs SCSD and requires $8\!\times\!4\eq32$ candidates. Thus, the 16-QAM OSTBC is superior in both complexity and performance to all the so-far used four-antenna systems at a spectral efficiency of $\unit[3]{bits/s/Hz}$.\\

\begin{figure}[htp]
 \centering
  \subfloat[SER]{\label{subfig:M_4_MDC_QOSTBC_vs_OSTBC_SER_3bits}\includegraphics[width=0.49\textwidth]{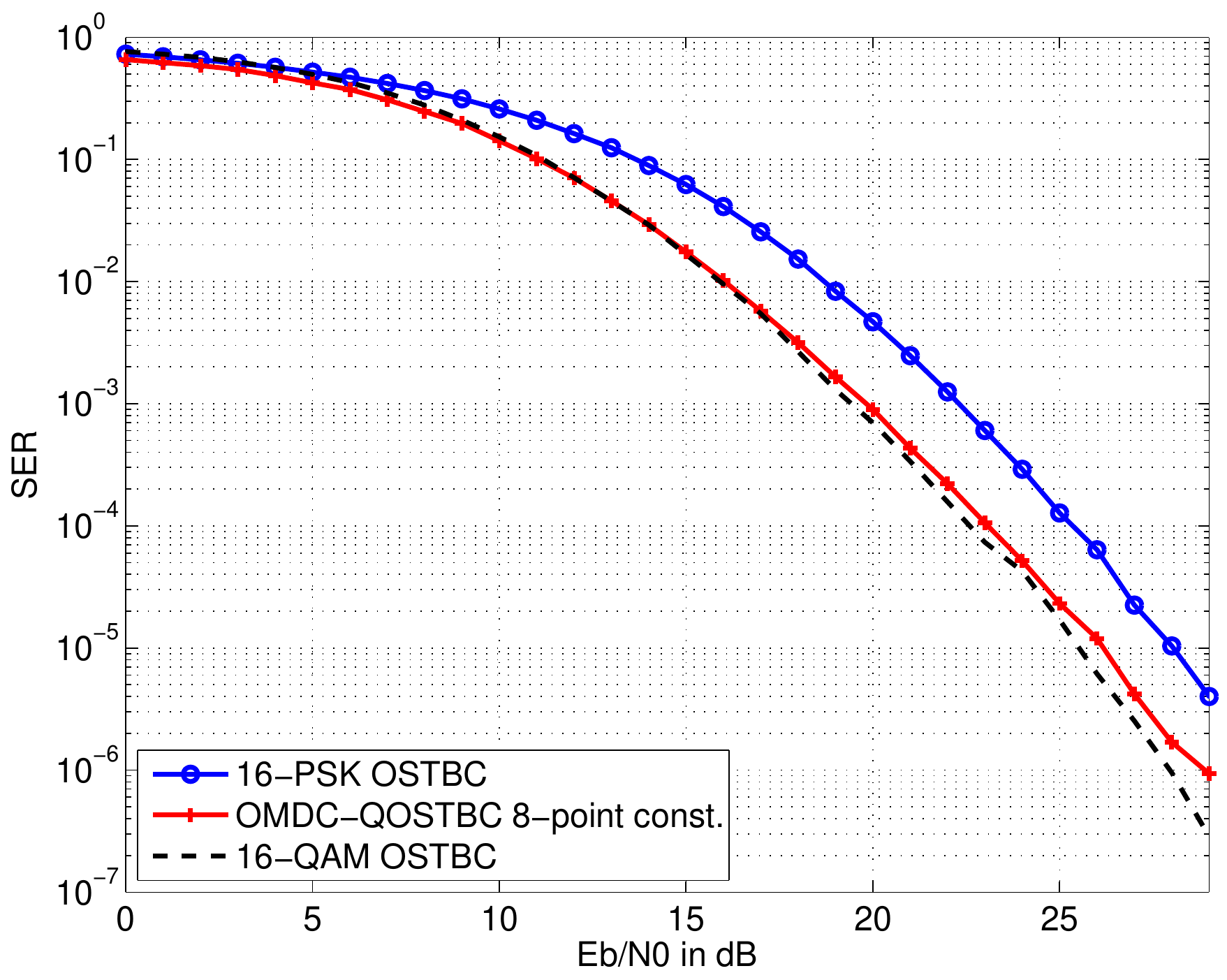}}
\subfloat[BER]{\label{subfig:M_4_MDC_QOSTBC_vs_OSTBC_BER_3bits}\includegraphics[width=0.49\textwidth]{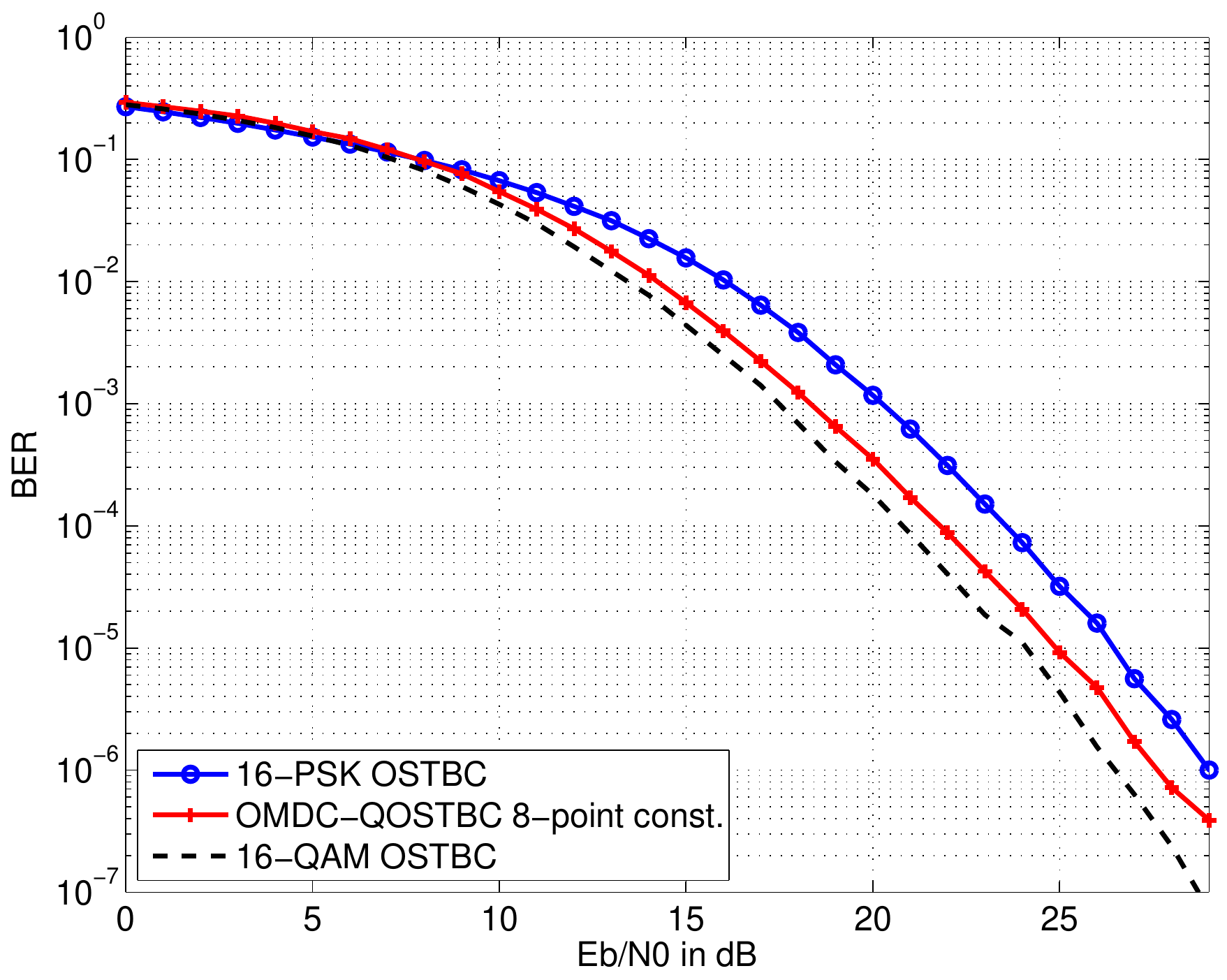}}
\caption{Comparing BER and SER of OMDC-QOSTBC with differential OSTBCs for a $4\!\times\!1$ system at $\unit[3]{bits/s/Hz}$}
 \label{fig:M_4_MDC_QOSTBC_vs_OSTBC_3bits}
\end{figure}
No previous simulations have been done for an eight-antenna system, so the performance of only the OMDC-QOSTBC is shown in Figure \ref{fig:M_8_0MDC_QOSTBC}. Using the 4-point constellation, a spectral efficiency of $\frac{3}{4}\times2\eq\unit[1.5]{bits/s/Hz}$ is achieved, and using the 8-point constellation, a spectral efficiency of $\frac{3}{4}\times3\eq\unit[2.25]{bits/s/Hz}$ is achieved.
\begin{figure}[!htp]
 \centering
  \includegraphics[width=0.45\textwidth]{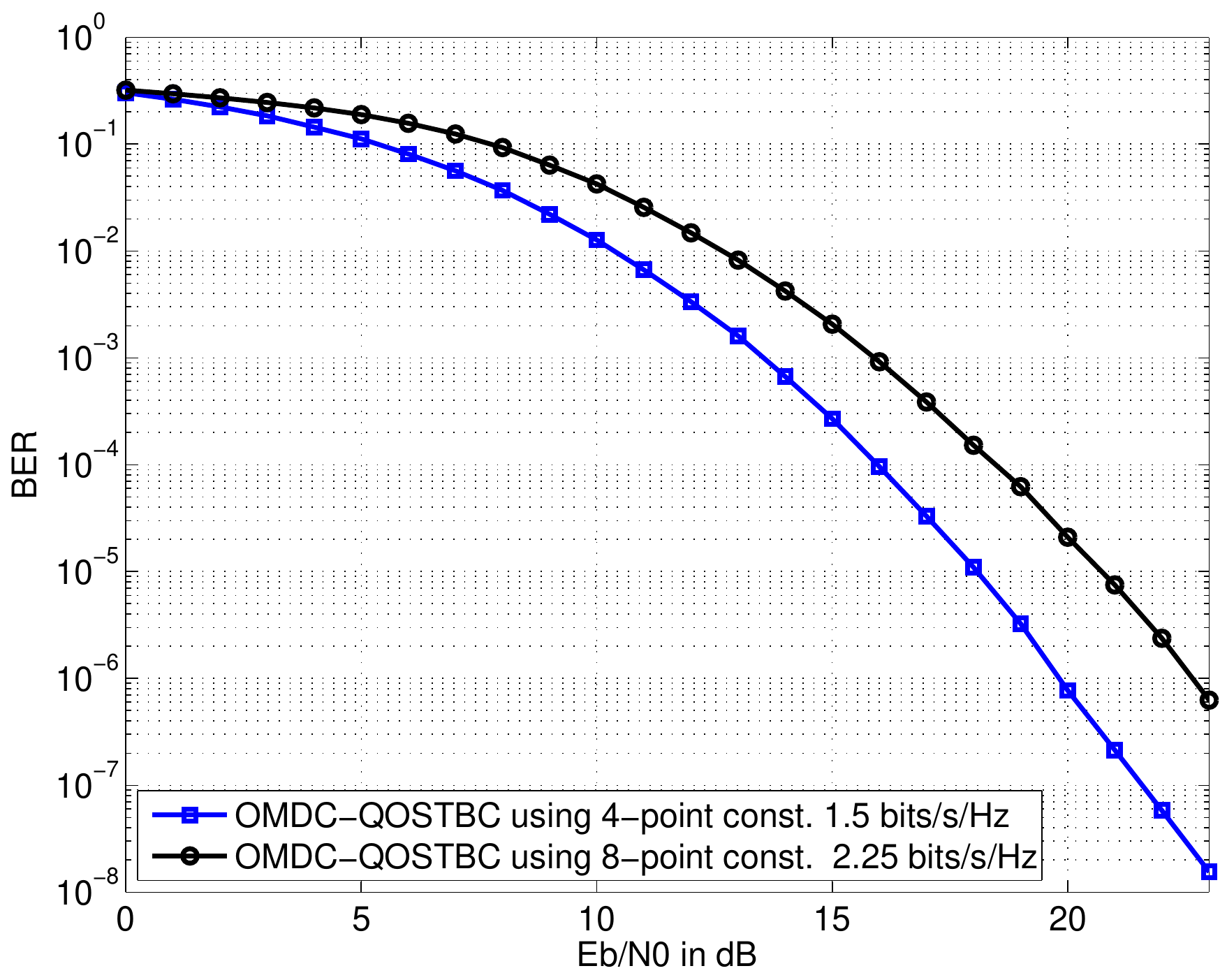} 
\caption{BER performance for an $8\!\times\!1$ differential system using OMDC-QOSTBC scheme}
\label{fig:M_8_0MDC_QOSTBC}
\end{figure}\\

In conclusion, this section has shown one possible way for using a QOSTBC code in the differential domain. The scheme is based on imposing some conditions on the constellation to make the code orthogonal. Here additionally more conditions have been imposed for the code to achieve full diversity. It has also been observed that although the scheme may show a better SER, this might not be the case for the BER. The interpretation for this is that the constellation used was designed to achieve orthogonality and full diversity, but not to achieve a distribution that reduces the bit changes within the zone containing the neighbouring symbols for every symbol. Such a condition that is achieved by default in the standard constellations like PSK and QAM.

\section{Differential QOSTBCs}
\label{s:DQOSTBC}
As discussed in the previous section, realizing differential encoding for QOSTBCs by forcing the constellation to achieve orthogonality results in performance degradation for high transmission rates. In this section we show a differential encoding approach for QOSTBCs proposed in \cite{Zhu_Jafarkhani_Diff_QOSTBCs_2005}. The section first describes such an approach for a general QOSTBC whose code matrix has an "ABBA" structure and then shows the special case of using MDC-QOSTBCs described in Section \ref{s:MDC_QOSTBCs} and how these codes reduce the decoding complexity. Two differential QOSTBCs are used, one achieves full diversity and performs single complex symbol decoding, and the other achieves half diversity and performs single real symbol decoding.
\subsection{Full-Diversity Differential MDC-QOSTBCs}
\label{ss:QOSTBC_full_diversity}
An $M\!\times\!N$ MIMO system in a Rayleigh block fading environment is described by the transmission equation
\begin{equation}
\label{eq:original_system}
\V{Y}_\tau\eq\sqrt{\rho}\,\V{S}_\tau\V{H}+\V{W}_\tau,
\end{equation}
where at block index $\tau$, matrix $\V{S}_\tau\in\mathbb{C}^{M\times M}$ is transmitted over the channel matrix $\V{H}\in\mathbb{C}^{M\times N}$ and corrupted by AWGN noise matrix $\V{W}_\tau\in\mathbb{C}^{M\times N}$ resulting in the received matrix $\V{Y}_\tau\in\mathbb{C}^{M\times N}$. The entries of $\V{H}$ and $\V{W}_\tau$ are i.i.d. complex Gaussian random variables with zero mean and unit variance. In the case when the transmit matrix has an "ABBA" structure as with MDC-QOSTBCs, (\ref{eq:original_system}) can be written as
\begin{equation}
\label{eq:Tx_matrix_2_halves}
 \begin{bmatrix}\V{Y}_\tau^1 \\\V{Y}_\tau^2 \end{bmatrix}=\sqrt{\rho} \begin{bmatrix}\V{A} & \V{B} \\\V{B} & \V{A} \end{bmatrix}\begin{bmatrix} \V{H}^1 \\ \V{H}^2\end{bmatrix} + \begin{bmatrix} \V{W}_\tau^1 \\ \V{W}_\tau^2 \end{bmatrix},
\end{equation}
where $\V{A}$ and $\V{B}$ $\in\mathbb{C}^{\frac{M}{2}\times\frac{M}{2}}$ and have a form of an OSTBC code matrix for $\frac{M}{2}$ transmit antennas. $\V{Y}_\tau^\kappa,\,\V{H}^\kappa,$ and $\V{W}_\tau^\kappa$ for $\kappa\eq1,2$ are $\in\mathbb{C}^{\frac{M}{2}\times N}$ and represent one half of the original matrices in (\ref{eq:original_system}). Multiplying (\ref{eq:Tx_matrix_2_halves}) from the left side by $\begin{bmatrix} \V{I}_{\frac{M}{2}} & \V{I}_{\frac{M}{2}}\\\V{I}_{\frac{M}{2}}& -\V{I}_{\frac{M}{2}}\end{bmatrix}$, we reside to an equivalent system derived as 
\begin{eqnarray}
  \begin{bmatrix} \V{I}_{\frac{M}{2}} & \V{I}_{\frac{M}{2}}\\\V{I}_{\frac{M}{2}}& -\V{I}_{\frac{M}{2}}\end{bmatrix}\begin{bmatrix}\V{Y}_\tau^1 \\\V{Y}_\tau^2 \end{bmatrix}&=&\sqrt{\rho}\begin{bmatrix} \V{I}_{\frac{M}{2}} & \V{I}_{\frac{M}{2}}\\\V{I}_{\frac{M}{2}}& -\V{I}_{\frac{M}{2}}\end{bmatrix} \begin{bmatrix}\V{A} & \V{B} \\\V{B} & \V{A} \end{bmatrix}\begin{bmatrix} \V{H}^1 \\ \V{H}^2\end{bmatrix} + \begin{bmatrix} \V{I}_{\frac{M}{2}} & \V{I}_{\frac{M}{2}}\\\V{I}_{\frac{M}{2}}& -\V{I}_{\frac{M}{2}}\end{bmatrix}\begin{bmatrix} \V{W}_\tau^1 \\ \V{W}_\tau^2 \end{bmatrix}\nonumber\\
\begin{bmatrix}\V{Y}_\tau^1+\V{Y}_\tau^2 \\ \V{Y}_\tau^1-\V{Y}_\tau^2 \end{bmatrix}&=&\sqrt{\rho}\begin{bmatrix} \V{A}+\V{B} & \V{A}+\V{B} \\ \V{A}-\V{B} & \V{B}-\V{A}\end{bmatrix}\begin{bmatrix} \V{H}^1 \\ \V{H}^2\end{bmatrix}+ \begin{bmatrix} \V{W}_\tau^1+\V{W}_\tau^2 \\ \V{W}_\tau^1-\V{W}_\tau^2 \end{bmatrix}\nonumber\\
\begin{bmatrix}\V{Y}_\tau^1+\V{Y}_\tau^2 \\ \V{Y}_\tau^1-\V{Y}_\tau^2 \end{bmatrix}&=&\sqrt{\rho}\begin{bmatrix} \V{A}+\V{B} & \V{0} \\ \V{0} & \V{A}-\V{B}\end{bmatrix}\begin{bmatrix} \V{I}_{\frac{M}{2}} & \V{I}_{\frac{M}{2}}\\\V{I}_{\frac{M}{2}}& -\V{I}_{\frac{M}{2}}\end{bmatrix}\begin{bmatrix} \V{H}^1 \\ \V{H}^2\end{bmatrix}+ \begin{bmatrix} \V{W}_\tau^1+\V{W}_\tau^2 \\ \V{W}_\tau^1-\V{W}_\tau^2 \end{bmatrix}\nonumber\\
\begin{bmatrix}\V{Y}_\tau^1+\V{Y}_\tau^2 \\ \V{Y}_\tau^1-\V{Y}_\tau^2 \end{bmatrix}&=&\sqrt{\rho}\begin{bmatrix} \V{A}+\V{B} & \V{0} \\ \V{0} & \V{A}-\V{B}\end{bmatrix}\begin{bmatrix} \V{H}^1+\V{H}^2 \\ \V{H}^1-\V{H}^2\end{bmatrix}+ \begin{bmatrix} \V{W}_\tau^1+\V{W}_\tau^2 \\ \V{W}_\tau^1-\V{W}_\tau^2 \end{bmatrix}\nonumber\\
\begin{bmatrix}\V{Y}_\tau^{1\E} \\ \V{Y}_\tau^{2\E} \end{bmatrix}&=&\sqrt{\rho}\,\begin{bmatrix} \V{S}_\tau^{1\E} & \V{0} \\ \V{0} & \V{S}_\tau^{2\E}\end{bmatrix}\begin{bmatrix} \V{H}^{1\E} \\ \V{H}^{2\E}\end{bmatrix}+ \begin{bmatrix} \V{W}_\tau^{1\E}\\ \V{W}_\tau^{2\E} \end{bmatrix}
\label{eq:subsystems}\\
\V{Y}_\tau^{\E}&=&\sqrt{\rho}\V{S}_\tau^{\E}\V{H}^{\E}+\V{W}_\tau^{\E}.
\label{eq:Equivalent_system}
\end{eqnarray}
Equation (\ref{eq:Equivalent_system}) represents a system that is mathematically equivalent to the original system in (\ref{eq:original_system}). Such an equivalent system encompasses two subsystems described in (\ref{eq:subsystems}) through the superscripts $1\E$ and $2\E$, where $\E$ stands for equivalent. The matrices of the equivalent subsystems are thus constructed from linear combinations of the submatrices of the original system. Namely,
\begin{equation}
\left.
\begin{aligned}
  \V{Y}_\tau^{\kappa\E}=\V{Y}_\tau^1\pm\V{Y}_\tau^2,\hspace{2cm}  &\V{S}_\tau^{\kappa\E}=\V{A}\pm\V{B},\\
 \V{H}^{\kappa\E}=\V{H}^1\pm\V{H}^2,\hspace{2cm} &\V{W}_\tau^{\kappa\E}=\V{W}_\tau^1\pm\V{W}_\tau^2
\end{aligned}
\right.\begin{aligned}
       &+ \,\text{for }\kappa=1\\
       &-\, \text{for }\kappa=2.
      \end{aligned}
\label{eq:linear_comb}
\end{equation}

In the differential domain, the information matrix in the original quasi-orthogonal system is also having an "ABBA" structure as
\begin{equation}
 \V{V}=\begin{bmatrix} \V{V}_1 & \V{V}_2 \\ \V{V}_2 & \V{V}_1\end{bmatrix},
\label{eq:matrix_V}
\end{equation}
where $\V{V}_1$ and $\V{V}_2$ have the form of an OSTBC for an $\frac{M}{2}$-antenna system. For example, in the four-transmit antenna case, they have an Alamouti structure as 
\begin{equation*}
 \left.\V{V}_1=\begin{bmatrix}v_1 & v_2\\ -v_2^* & v_1^*\end{bmatrix}, \right. \V{V}_2=\begin{bmatrix}v_3 & v_4\\ -v_4^* & v_3^*\end{bmatrix},
\end{equation*}
where $v_i$, $i=1,...,4$ are the information symbols to be encoded. Just as in (\ref{eq:linear_comb}), the information matrices in the equivalent subsystems are constructed as linear combinations of the original submatrices $\V{V}_1$ and $\V{V}_2$, i.e.
\begin{equation*}
\V{V}_{z_\tau}^{1\E}=\frac{1}{\sqrt{p_1}}\left[\V{V}_1+\V{V}_2\right],\,\, \V{V}_{z_\tau}^{2\E}=\frac{1}{\sqrt{p_2}}\left[\V{V}_1-\V{V}_2\right].
\end{equation*}
Due to this construction of the information matrices, the addressed scheme will be referred to as combined MDC-QOSTBC. The normalization factors $p_1$ and $p_2$ are used to ensure a constant average transmit power, with the constant chosen such that the total transmit power per time slot is on average $1$, making $\rho$ the average SNR per time slot at each receive antenna\footnote{$\rho$ is the average received SNR in both the original system in (\ref{eq:original_system}) and the equivalent system in (\ref{eq:Equivalent_system})}. Since matrices $\V{V}_1$ and $\V{V}_2$ are orthogonal, $\V{V}_{z_\tau}^{1\E}$  and $\V{V}_{z_\tau}^{2\E}$ are also orthogonal. Thus, they can be differentially encoded just as done for the non-unitary OSTBCs in (\ref{eq:Diff_encoding_non_unitary}). Next, we describe how the differential encoding of an "ABBA" QOSTBC can be realized by expressing it in terms of two equivalent \emph{orthogonal} subsystems, where each subsystem can be differentially encoded.
\begin{center} 
\textbf{Equivalent system} 
\end{center}
\begin{equation*}
 \V{Y}_\tau^{\E}=\sqrt{\rho}\,\V{S}_\tau^{\E}\V{H}^{\E}+\V{W}_\tau^{\E}
\end{equation*}
\begin{figure*}[!ht]
\vspace{-0.85cm}
\centering
 \scalebox{0.6}{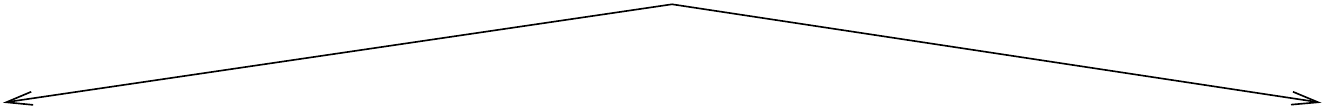}
\end{figure*}
\vspace{-0.85cm}
 \begin{table*}[htp]
\begin{center}
\begin{tabular}{p{9cm} p{9cm}}
\textbf{Equivalent subsystem 1} & \textbf{Equivalent subsystem 2}\tn[5pt]
 $\V{Y}_\tau^{1\E}=\sqrt{\rho}\,\V{S}_\tau^{1\E}\V{H}^{1\E}+\V{W}_\tau^{1\E}$ &  $\V{Y}_\tau^{2\E}=\sqrt{\rho}\,\V{S}_\tau^{2\E}\V{H}^{2\E}+\V{W}_\tau^{2\E}$\tn[3pt]
\underline{Differential Encoding:} & \underline{Differential Encoding:}\tn[3pt]
$\V{S}_\tau^{1\E}=\large\frac{\V{V}_{z_\tau}^{1\E}\V{S}_{\tau-1}^{1\E}}{a_{\tau-1}^{1\E}}$ & $\V{S}_\tau^{2\E}=\large\frac{\V{V}_{z_\tau}^{2\E}\V{S}_{\tau-1}^{2\E}}{a_{\tau-1}^{2\E}}$\tn[5pt]
$\V{S}_{\tau}^{1\E\hr}\V{S}_{\tau}^{1\E}=\V{V}_{z_\tau}^{1\E\hr}\V{V}_{z_\tau}^{1\E}=(a_{\tau}^{1\E})^2\V{I}_{\frac{M}{2}}$ & $\V{S}_{\tau}^{2\E\hr}\V{S}_{\tau}^{2\E}=\V{V}_{z_\tau}^{2\E\hr}\V{V}_{z_\tau}^{2\E}=(a_{\tau}^{2\E})^2\V{I}_{\frac{M}{2}}$ \tn
$\V{V}_{z_\tau}^{1\E}=\frac{1}{\sqrt{p_1}}\left[\V{V}_1+\V{V}_2\right]$ & $\V{V}_{z_\tau}^{2\E}=\frac{1}{\sqrt{p_2}}\left[\V{V}_1-\V{V}_2\right]$\tn[2pt]
$p_1$ is defined such that  & $p_2$ is defined such that \tn
$\E[\V{V}_{z_\tau}^{1\E\hr}\V{V}_{z_\tau}^{1\E}]=\V{I}_{\frac{M}{2}}$ $\forall\,\tau$ & $\E[\V{V}_{z_\tau}^{2\E\hr}\V{V}_{z_\tau}^{2\E}]=\V{I}_{\frac{M}{2}}$ $\forall\,\tau$\tn
i.e. $\E[(a_{\tau}^{1\E})^2]\req1$ & i.e. $\E[(a_{\tau}^{2\E})^2]\req1$\tn[5pt]
\multicolumn{2}{c}{\hspace{-3cm}\underline{For a four-transmit antenna system}}\tn[5pt]
$\V{V}_{z_\tau}^{1\E}=\frac{1}{\sqrt{p_1}}\begin{bmatrix} (v_1+v_3) & (v_2+v_4)\\ -(v_2+v_4)^* & (v_1+v_3)^*\end{bmatrix}$ & $\V{V}_{z_\tau}^{2\E}=\frac{1}{\sqrt{p_2}}\begin{bmatrix} (v_1-v_3) & (v_2-v_4)\\-(v_2-v_4)^* & (v_1-v_3)^*\end{bmatrix}$\tn
$\V{V}_{z_\tau}^{1\E\hr}\V{V}_{z_\tau}^{1\E}=\frac{|v_1+v_3|^2+|v_2+v_4|^2}{p_1}\V{I}_2$ & $\V{V}_{z_\tau}^{2\E\hr}\V{V}_{z_\tau}^{2\E}=\frac{|v_1-v_3|^2+|v_2-v_4|^2}{p_2}\V{I}_2$\tn[2pt]
\multicolumn{2}{c}{\hspace{-3cm}
\hspace{1cm}}\tn[3pt]
$(a_{\tau}^{1\E})^2=\frac{\alpha_\tau+\beta_\tau}{p_1} \implies p_1\req\E[\alpha_\tau+\beta_\tau]$ & $(a_{\tau}^{2\E})^2=\frac{\alpha_\tau-\beta_\tau}{p_2} \implies p_2\req\E[\alpha_\tau-\beta_\tau]$. \tn
 \end{tabular}	
\end{center}
\end{table*}\\

The described differential scheme is valid for any QOSTBC with any number of transmitting antennas, under the condition that the code matrix has an "ABBA" structure constructed from square orthogonal submatrices. In this section we give a general description for $M$ transmit antennas, but a detailed description only for 4 transmit antennas. In this case the equivalent information matrices $\V{V}_{z_\tau}^{1\E}$ and $\V{V}_{z_\tau}^{2\E}$ have an Alamouti structure as shown above.
\subsubsection{Sub-Optimal Differential Decoder}
	Next, we show a sub-optimal differential decoder as the one used for non-unitary OSTBCs in (\ref{eq:sub_opt_Diff_Decoder}). Since the differential encoding is done in the equivalent subsystems, the differential decoding is also performed for each equivalent subsystem and combined from both subsystems to get the final metric. To start with, consider the first equivalent subsystem and let $\V{V}_l^{1\E}$ be its $l^{\text{th}}$ candidate matrix out of $L$ possible information matrices, then  the sub-optimal metric of the first subsystem is 
	\begin{eqnarray}
	\xi_{1\E}&=&\mini\limits_{l=0,...,L-1}\:\,\|\V{Y}_\tau^{1\E}-\frac{\V{V}_l^{1\E}}{a_{\tau-1}^{1\E}}\V{Y}_{\tau-1}^{1\E}\|_F^2\nonumber\\
	&=&\mini\limits_{l=0,...,L-1}\:\,\frac{\tr(\V{Y}_{\tau-1}^{1\E\hr}\overbrace{\V{V}_l^{1\E\hr}\V{V}_l^{1\E}}^{(a_{\tau}^{1\E})^2}\V{Y}_{\tau-1}^{1\E})}{(a_{\tau-1}^{1\E})^2}-\frac{2\Re\{\tr(\V{Y}_{\tau}^{1\E\hr}\V{V}_l^{1\E}\V{Y}_{\tau-1}^{1\E})\}}{a_{\tau-1}^{1\E}}
	\label{eq:sub_opt_metric_general}
	\end{eqnarray}
Using the dispersive form $ \V{V}_{z_\tau}^{\kappa\E}=\frac{1}{\sqrt{p_\kappa}}\sum\limits_{i=1}^{K}{\V{U}_i^{\kappa\E}v_i^R+j\V{Q}^{\kappa\E}_iv_i^I},\,\, \kappa=1,2$, where $v_i,\,\,i=1,...,K$ are the $K$ information symbols to be encoded for transmission at block index $\tau$, the metric of the first subsystem can thus be written as 
\begin{eqnarray*} \xi_{1\E}\!\!\!&=\!\!\!&\mini\limits_{v_i\in\mathcal{A}_i}\:\,\underbrace{\frac{\tr(\V{Y}_{\tau-1}^{1\E\hr}\V{Y}_{\tau-1}^{1\E})}{2\sqrt{\alpha_{\tau-1}+\beta_{\tau-1}}}}_{\tilde{y}_1}.\sqrt{p_1}.\frac{\alpha_{\tau}+\beta_{\tau}}{p_1}-\frac{1}{\sqrt{p_1}}\Re\{\tr\big(\V{Y}_{\tau-1}^{1\E}\V{Y}_{\tau}^{1\E\hr}(\sum\limits_{i=1}^{K}{\V{U}_i^{1\E}v_i^R+j\V{Q}_i^{1\E}v_i^I})\big)\}\\
\!\!\!&=\!\!\!&\mini\limits_{v_i\in\mathcal{A}_i}\:\,\tilde{y}_1(\sum\limits_{i=1}^{4}{|v_i|^2}+2\Re\{v_1v_3^*+v_2v_4^*\})-\sum\limits_{i=1}^{4}{\Re\{\tr(\V{Y}_{\tau-1}^{1\E}\V{Y}_{\tau}^{1\E\hr}\V{U}_i^{1\E})v_i^R+\tr(j\V{Y}_{\tau-1}^{1\E}\V{Y}_{\tau}^{1\E\hr}\V{Q}_i^{1\E})v_i^I\}}
\end{eqnarray*}
Due to the term $\Re\{v_1v_3^*\}$, symbols $v_1$ and $v_3$ must be jointly decoded and similarly due to the term $\Re\{v_2v_4^*\}$, symbols $v_2$ and $v_4$ should as well be jointly decoded. This is also the case with the metric of the second subsystem. Thus the final metric --which is the same as the one used in \cite{Zhu_Jafarkhani_Diff_QOSTBCs_2005}-- achieves pair-wise complex symbol decoding. However, if we instead define symbols $v_i, \:i=1,...,4$ to carry the base information symbols $x_i, \:i=1,...,4$ after interleaving their real and imaginary components as is done in the MDC-QOSTBC code matrix, we get
\begin{equation}
v_1= x_1^R+jx_3^R,\:v_2=x_2^R+jx_4^R,\:v_3=-x_1^I+jx_3^I,\:v_4=-x_2^I+jx_4^I,
\end{equation}
and thus matrix $\V{V}$ in (\ref{eq:matrix_V}) is the same as the MDC-QOSTBC code matrix in (\ref{eq:four_MDC_QOSTBC}). This interleaving of components results in
\begin{eqnarray*}
\Re\{v_1v_3^*\}&=&\Re\{(x_1^R+jx_3^R)(-x_1^I-jx_3^I)\}=-x_1^Rx_1^I+x_3^Rx_3^I\\
\Re\{v_2v_4^*\}&=&\Re\{(x_2^R+jx_4^R)(-x_2^I-jx_4^I)\}=-x_2^Rx_2^I+x_4^Rx_4^I.
\end{eqnarray*}
Thus, interleaving the real and imaginary components of the information symbols allows the decoder to decide on each symbol $x_i,\,i=1,...,4$ independently. In fact, this is the rationale behind the reduced complexity achieved by MDC-QOSTBCs. In this case, $\alpha_\tau$ and $\beta_\tau$ in (\ref{eq:alpha_tau_beta_tau}) become
\begin{equation}
 \alpha_\tau=\sum\limits_{i=1}^{4}{|x_i|^2},\:\:\beta_\tau=2\sum\limits_{i=1}^{2}{-x_i^Rx_i^I+x_{i+2}^Rx_{i+2}^I},
\label{eq:alpha_beta_combined_M_4_MDC_DQOSTBCs}
\end{equation}
which are the same as $\alpha$ and $\beta$ in (\ref{eq:alpha_beta}) that were defined for a coherent MDC-QOSTBC. Now, by redefining the dispersion matrices $\V{U}_i^{\kappa\E}$ and $\V{Q}_i^{\kappa\E}$ to operate on symbols $x_i$ instead of $v_i$, the equivalent code matrix can be redefined as
\begin{equation*}
 \V{V}_{z_\tau}^{\kappa\E}=\frac{1}{\sqrt{p_\kappa}}\sum\limits_{i=1}^{K}{\V{U}_i^{\kappa\E}x_i^R+j\V{Q}^{\kappa\E}_ix_i^I},\hspace{0.5cm} \kappa=1,2.
\end{equation*}
And thus the metric of the first  subsystem reduces to
\begin{equation*}
\begin{aligned} \xi_{1\E}&=\mini\limits_{x_i\in\mathcal{A}_i}\:\,\tilde{y}_1\left(\sum\limits_{i=1}^{4}{|x_i|^2}+2\sum\limits_{i=1}^{2}{-x_i^Rx_i^I+x_{i+2}^Rx_{i+2}^I}\right)\\
& -\sum\limits_{i=1}^{4}{\underbrace{\Re\{\tr(\V{Y}_{\tau-1}^{1\E}\V{Y}_{\tau}^{1\E\hr}\V{U}_i^{1\E})\}}_{\tilde{x}_{i,1}}x_i^R+\Re\{\tr(j\V{Y}_{\tau-1}^{1\E}\V{Y}_{\tau}^{1\E\hr}\V{Q}_i^{1\E})\}x_i^I},
\end{aligned}
\end{equation*}
and similarly, the metric of the second subsystem is
\begin{equation*}
\begin{aligned}
 \xi_{2\E}&=\mini\limits_{x_i\in\mathcal{A}_i}\:\,\tilde{y}_2\left(\sum\limits_{i=1}^{4}{|x_i|^2}-2\sum\limits_{i=1}^{2}{-x_i^Rx_i^I+x_{i+2}^Rx_{i+2}^I}\right)\\
& -\sum\limits_{i=1}^{4}{\underbrace{\Re\{\tr(\V{Y}_{\tau-1}^{2\E}\V{Y}_{\tau}^{2\E\hr}\V{U}_i^{2\E})\}}_{\tilde{x}_{i,2}}x_i^R+\Re\{\tr(j\V{Y}_{\tau-1}^{2\E}\V{Y}_{\tau}^{2\E\hr}\V{Q}_i^{2\E})\}x_i^I},
\end{aligned}
\end{equation*}
where $\tilde{y}_2\eq\frac{\tr(\V{Y}_{\tau-1}^{2\E\hr}\V{Y}_{\tau-1}^{2\E})}{2\sqrt{\alpha_{\tau-1}-\beta_{\tau-1}}}$. By writing down the code matrices of the equivalent subsystems, one may observe the following relations for the dispersion matrices
\begin{equation*}
\begin{aligned}
&j\V{Q}_i^{1\E}=-\V{U}_i^{1\E}\definedas-\V{U}_i, \hspace{0.3cm} &i=1,2 \hspace{2cm} &j\V{Q}_i^{2\E}=\V{U}_i^{2\E}=\V{U}_i, \hspace{0.3cm}& i=1,2\\
&j\V{Q}_i^{1\E}=\V{U}_i^{1\E}=\V{U}_i, &i=3,4 \hspace{2cm} &j\V{Q}_i^{2\E}=-\V{U}_i^{2\E}=-\V{U}_i, & i=3,4.
\end{aligned}
\end{equation*}
Consequently, the metric from the first subsystem becomes
\begin{equation*}
 \xi_{i,1\E}=
\begin{cases}\mini\limits_{x_i\in\mathcal{A}_i}\:\,\tilde{y}_1(|x_i|^2-2x_i^Rx_i^I)-\tilde{x}_{i,1}(x_i^R-x_i^I) &i=1,2\\
 \mini\limits_{x_i\in\mathcal{A}_i}\:\,\tilde{y}_1(|x_i|^2+2x_i^Rx_i^I)-\tilde{x}_{i,1}(x_i^R+x_i^I) & i=3,4
\end{cases}
\end{equation*}
and the metric from the second subsystem becomes
\begin{equation*}
\xi_{i,2\E}=
\begin{cases}
 \mini\limits_{x_i\in\mathcal{A}_i}\:\,\tilde{y}_2(|x_i|^2+2x_i^Rx_i^I)-\tilde{x}_{i,2}(x_i^R+x_i^I) &i=1,2\\
 \mini\limits_{x_i\in\mathcal{A}_i}\:\,\tilde{y}_2(|x_i|^2-2x_i^Rx_i^I)-\tilde{x}_{i,2}(x_i^R-x_i^I) &i=3,4.
\end{cases}
\end{equation*}
Combining the metrics of both subsystems, we get the final decision metric as
\begin{equation}
\label{eq:MDC_QOSTBC_decoder}
\boxed{
\hat{x}_i=
 \begin{cases}
\argmin\limits_{x_i\in\mathcal{A}_i}\:\,(\tilde{y}_1+\tilde{y}_2)|x_i|^2-2(\tilde{y}_1-\tilde{y}_2)x_i^Rx_i^I-x_i^R(\tilde{x}_{i,1}+\tilde{x}_{i,2})+x_i^I(\tilde{x}_{i,1}-\tilde{x}_{i,2}) &i=1,2\\
\argmin\limits_{x_i\in\mathcal{A}_i}\:\,(\tilde{y}_1+\tilde{y}_2)|x_i|^2+2(\tilde{y}_1-\tilde{y}_2)x_i^Rx_i^I-x_i^R(\tilde{x}_{i,1}+\tilde{x}_{i,2})-x_i^I(\tilde{x}_{i,1}-\tilde{x}_{i,2})&i=3,4\\
 \end{cases}
}
\end{equation}
This is the differential non-coherent version of the MDC-QOSTBC decoder, whose coherent version is shown in \cite[eq.(8)]{Yuen_MDC_QOSTBC_2005}. While \cite[eq.(8)]{Yuen_MDC_QOSTBC_2005} is an ML decoder, (\ref{eq:MDC_QOSTBC_decoder}) is only a sub-optimal decoder as the ML differential decoder for non-unitary ST codes loses the advantage of symbol decoupling as has been shown in (\ref{eq:ML_metric_non_unitary}).
Figure \ref{fig:DMDC_QOSTBC_system} shows a detailed description of the differential encoding and decoding of MDC-QOSTBCs for a $4\!\times\!1$ system.
\setlength{\arraycolsep}{1.5pt}
\begin{figure}[htp]
 \centering
\scalebox{0.9}{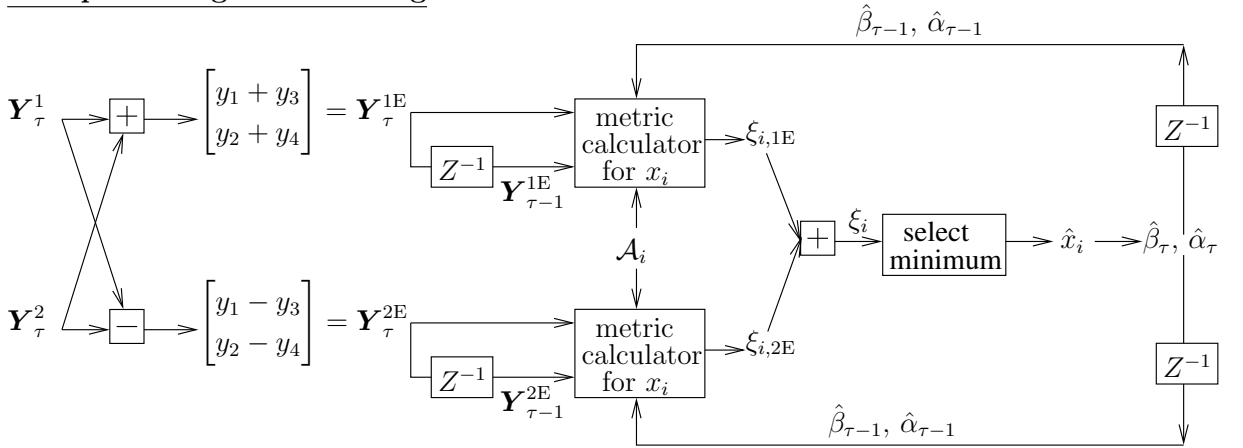}
\caption{Differential encoding and decoding of a $4\!\times\!1$ MDC-QOSTBC.}
\label{fig:DMDC_QOSTBC_system}
\end{figure}
\setlength{\arraycolsep}{5pt} 

\subsubsection{Diversity order and constellation design}
As has been noted before, the diversity order of an STC MISO system was proved for coherent transmission to be the minimum rank of all distance matrices. For non-coherent systems however, the same criterion was only proved to hold for unitary STCs as shown in Section \ref{s:Design_Criteria}. No explicit proof for non-coherent systems shows that the rank criterion is still the diversity measure for non-unitary and for non-orthogonal STBCs. Nevertheless, we conjecture that at least for non-coherent systems that use non-unitary QOSTBCs with an "ABBA" construction, still the minimum rank of the distance matrices governs the diversity order of the system. Simulation results will show to agree with this conjecture.\\

To study the diversity order of the differential QOSTBC described in this section, one first needs to get a form for the \emph{actual} information matrix which can be derived from the information matrices of the equivalent subsystems. That is, if the actual transmit matrix at block index $\tau$ is $\V{S}_\tau$, and that at block index $\tau-1$ is $\V{S}_{\tau-1}$, then the actual information matrix is the one that satisfies
\begin{equation}
 \V{S}_{\tau}=\V{V}'_{z_\tau}\V{S}_{\tau-1}.
\end{equation}\\

In Section \ref{s:Diversity_proof_DQOSTBC}, the actual information matrix $\V{V}'_{z_\tau}$ was proved to have an "ABBA" structure with entries derived in (\ref{eq:actual_info_symbols_full_diversity}). Knowing the form of the actual information matrix, one can get the minimum determinant among all distance matrices $\Delta \V{V}'_{z_\tau}$. It is proved in (\ref{eq:coding_gain_combined_MDC_QOSTBC}) that the minimum determinant is
\begin{equation} 
\det(\Delta \V{V}')\Big|_{\text{min}}\propto \mini((\Delta x^R)^2-(\Delta x^I)^2)^2,
\label{eq:coding_gain_combined_MDC_QOSTBC2}
\end{equation}
where $\Delta x=\Delta x^R+j\Delta x^I$ is the difference between any two points in the constellation. (\ref{eq:coding_gain_combined_MDC_QOSTBC2}) is having the same form as the minimum determinant for coherent MDC-QOSTBC in (\ref{eq:diversity_condition}). Thus the full diversity condition for both coherent and differential MDC-QOSTBCs is the same. Namely, full diversity is achieved if the absolute difference between the real parts of any two constellation points is not the same as the absolute difference between their imaginary parts, i.e.
\begin{equation}
|\Delta x^R|\neq|\Delta x^I|.
\label{eq:full_diversity_condition_MDC_QOSTBC2}
\end{equation} 
When a constellation like rectangular QAM is used for the information symbols $x_i$, constellation rotation is needed for the above full diversity condition to be satisfied. There exist many angles of rotation that satisfy (\ref{eq:full_diversity_condition_MDC_QOSTBC2}) making the differential MDC-QOSTBC achieve full diversity. However, only  some of these angles maximize the minimum determinant (the coding gain) in (\ref{eq:coding_gain_combined_MDC_QOSTBC2}). Since the minimum determinant of both coherent and non-coherent MDC-QOSTBCs have the same form, then the rotational angle that maximizes the coding gain is the same for both. In \cite{Yuen_MDC_QOSTBC_2005_V2}, it was proved that the optimal angle of rotation for a rectangular QAM is $\frac{1}{2}\tan^{-1}(\frac{1}{2})=13.28^o$. So this can be applied on the 4, 16 and 64- QAM constellations. To visualize the rotated constellation, Figure \ref{fig:opt_16_QAM_MDC_QOSTBC} shows the 16-QAM alphabet rotated with the optimal angle $13.28^o$. Also shown is the bit-to-symbol gray mapping.\\

\begin{figure}
 \centering
\includegraphics[width=0.45\textwidth,height=0.2\textheight]{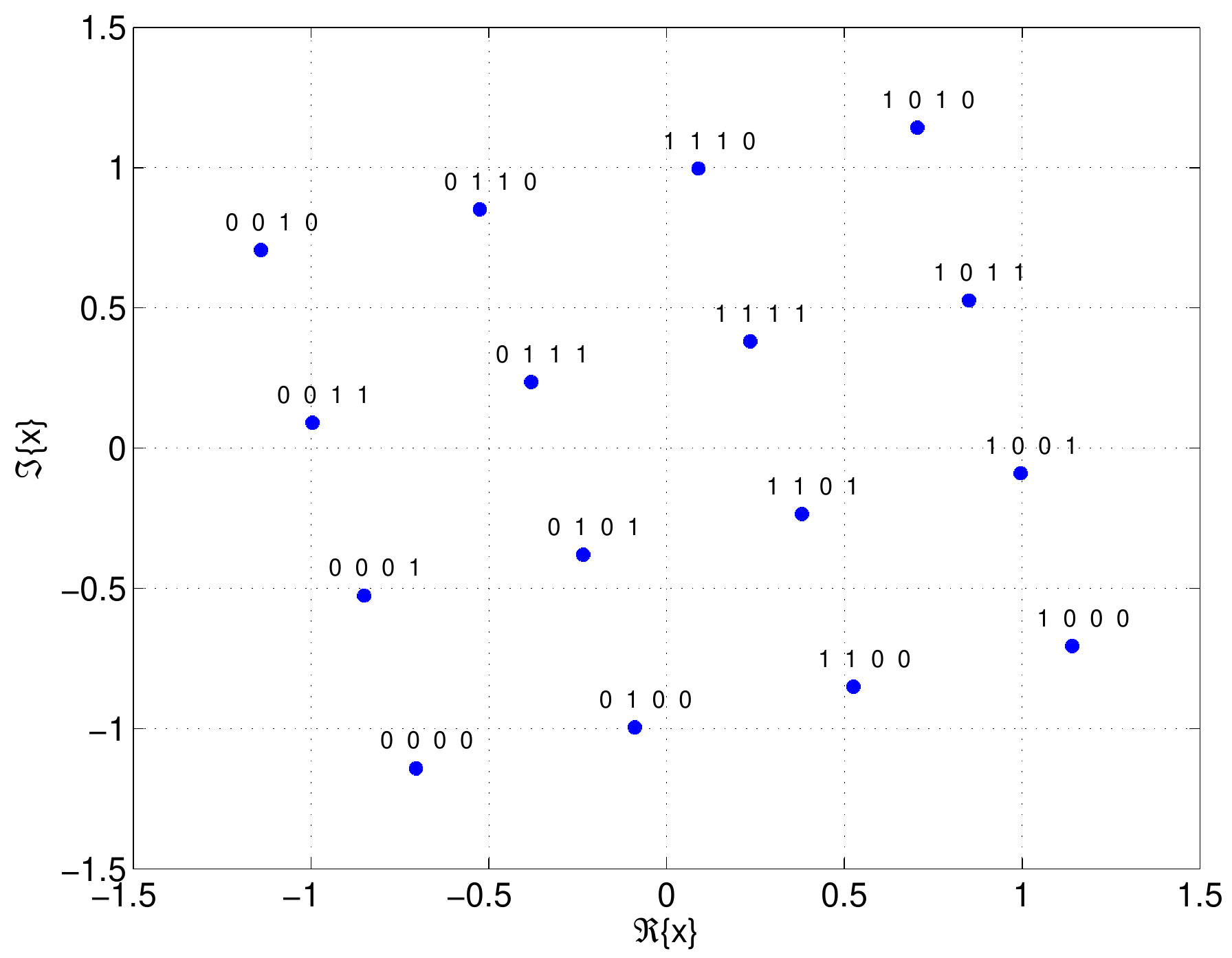}
\caption{Optimal rotated 16-QAM constellation for use with full-diversity coherent or differential MDC-QOSTBC.} 
\label{fig:opt_16_QAM_MDC_QOSTBC}
\end{figure}

For 8-QAM, a circular constellation is used and we searched for the optimal angle of rotation by applying the following procedure. Starting from the circular 8-QAM constellation shown in Figure \ref{subfig:8_QAM_MDC_QOSTBC_initial}. We consider three optimization parameters, namely $\theta_1$ and $\theta_2$ which are the angles of rotation for the inner and the outer circle, respectively and the third parameter $r$ is the ratio between the amplitude of both circles. The search span for $r$ is from 1.1 to 2.3 and that of $\theta_1$ and $\theta_2$ is from $0^o$ to $90^o$. For every ($r,\theta_1,\theta_2$) combination, the minimum determinant in (\ref{eq:coding_gain_combined_MDC_QOSTBC2}) is calculated. For every amplitude ratio $r$, we get the optimal combination ($\theta_1,\theta_2$) which has the maximum coding gain. For example at $r=1.37$, Figure \ref{subfig:opt_min_det_theta1_theta2} shows the minimum determinant as a function of both $\theta_1$ and $\theta_2$. As shown $\theta_{1,\text{opt}}$ and $\theta_{2,\text{opt}}$ are $12.73^o$ and $58.18^o$ or $77.27^o$ and $31.82^o$.\\

Figure \ref{subfig:max_min_det_thetas_opt} shows the maximum coding gain achieved as a function of $r$ with the optimal rotational angles for every $r$. As shown, several different values of $r$   result in close values of coding gain. Thus, error rate simulations have been carried out spanning $r$ from $1.3$ to $2.3$ with the optimal angles of rotation for every $r$. Figure \ref{subfig:SER_vs_r} shows the SER vs the amplitude ratio $r$ at $E_b/N_0=\unit[20]{dB}$. As shown the optimal choice of $r$ is in the range from $1.3$ to $1.4$.  We chose point $r\eq1.37$ which shows the absolute minimum SER. The corresponding optimal angles of rotation are shown in Figure \ref{subfig:opt_min_det_theta1_theta2} and chosen to be $\theta_1\eq12.73^o$ and $\theta_2\eq 58.18^o$. Figure \ref{subfig:opt_8_QAM_MDC_QOSTBC} shows the optimal 8-QAM constellation that achieves full diversity and maximum coding gain for the considered differential MDC-QOSTBC. Interestingly, the resulting optimal constellation can be viewed as a rectangular 8-QAM constellation with the same optimal angle of rotation $13.28^o$ derived in \cite{Yuen_MDC_QOSTBC_2005_V2} for a rectangular QAM constellation. \\

\begin{figure}
 \centering
\subfloat[initial 8-QAM constellation.]{\label{subfig:8_QAM_MDC_QOSTBC_initial}\scalebox{1}{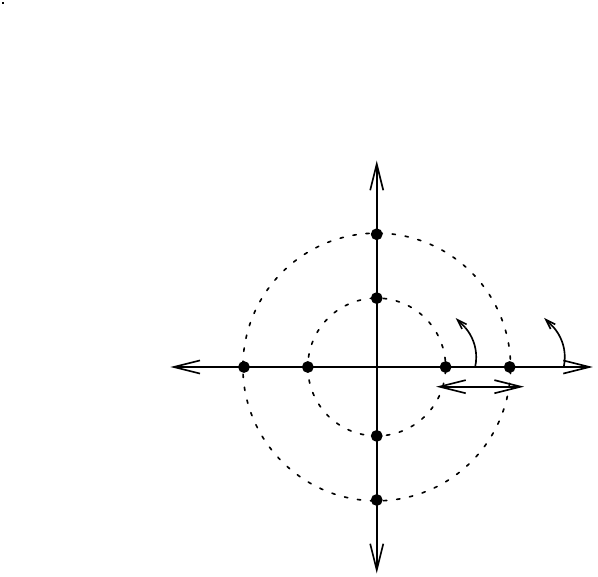\hspace{0.5cm}}}
  \subfloat[MDC-QOSTBC coding gain as a function of $\theta_1$ and $\theta_2$ at $r=1.37$ for an 8-QAM constellation.]{\label{subfig:opt_min_det_theta1_theta2}\includegraphics[width=0.6\textwidth]{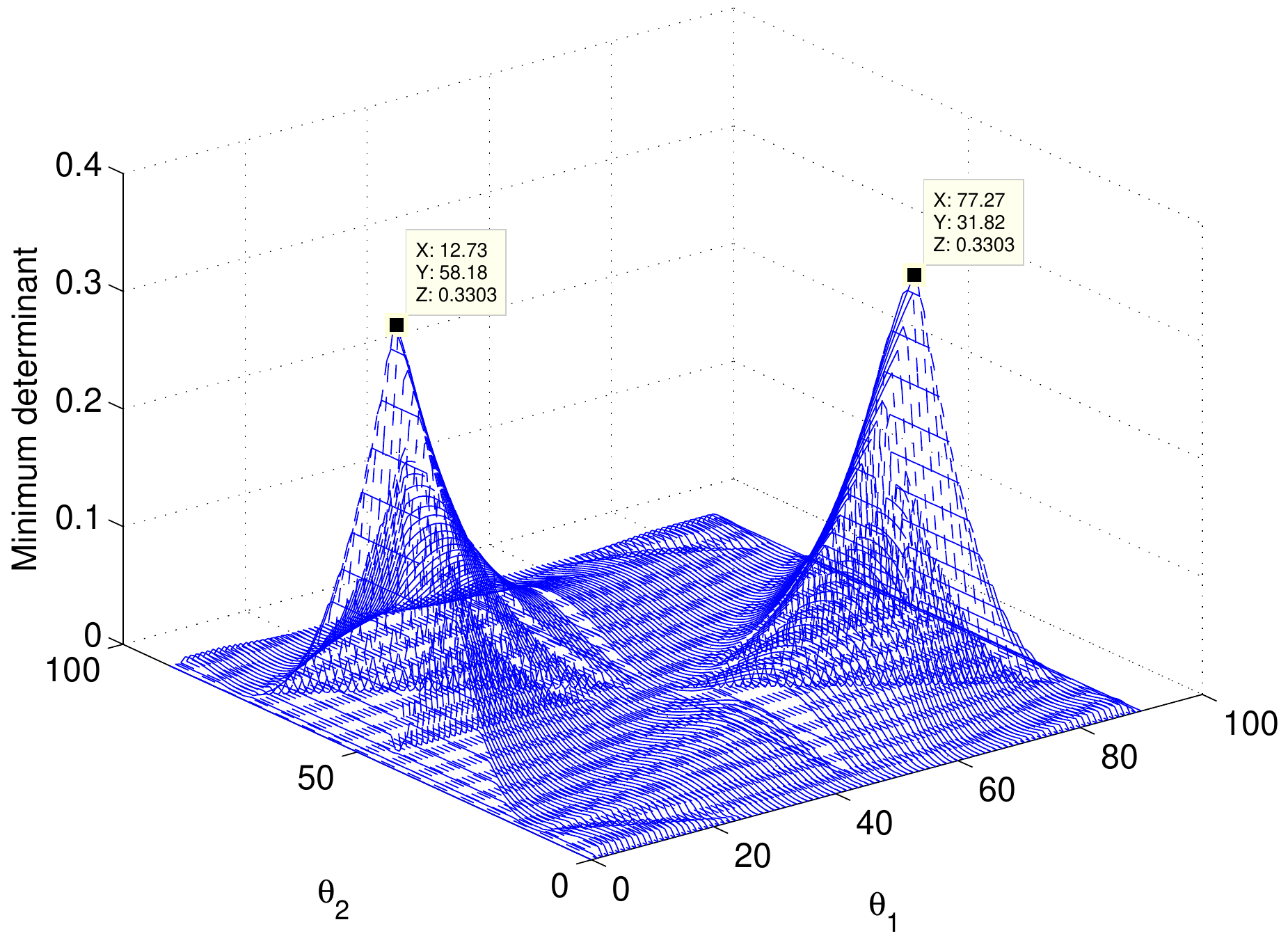}}\\
  \subfloat[MDC-QOSTBC coding gain at the optimal rotation angles vs. the amplitude ratio $r$ for an 8-QAM constellation.]{\label{subfig:max_min_det_thetas_opt}\includegraphics[width=0.49\textwidth]{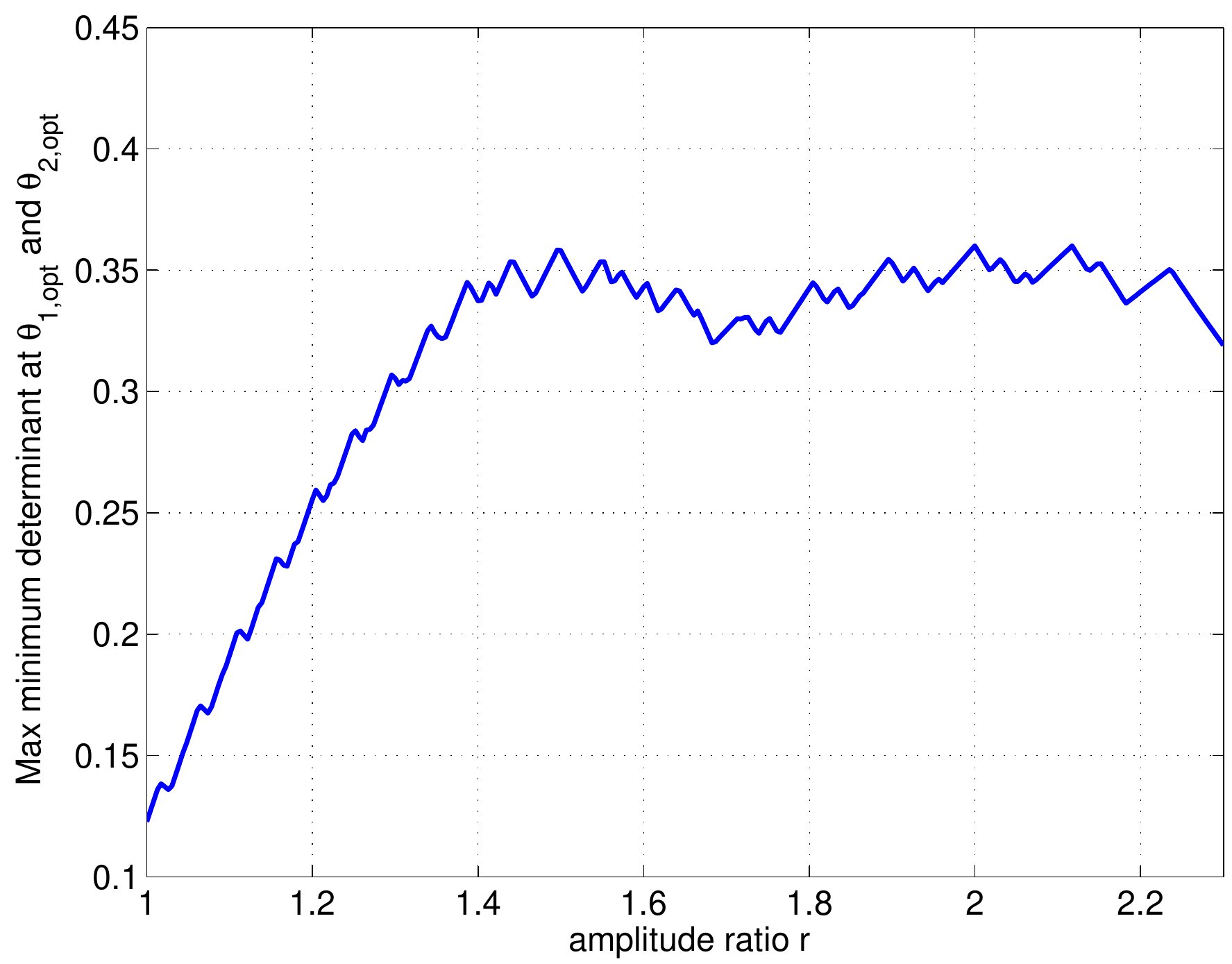}} \vspace{0.2cm}
  \subfloat[SER vs the amplitude ratio $r$ at optimal rotation angles for $E_b/N_0=20\,$dB]{\label{subfig:SER_vs_r}\includegraphics[width=0.49\textwidth]{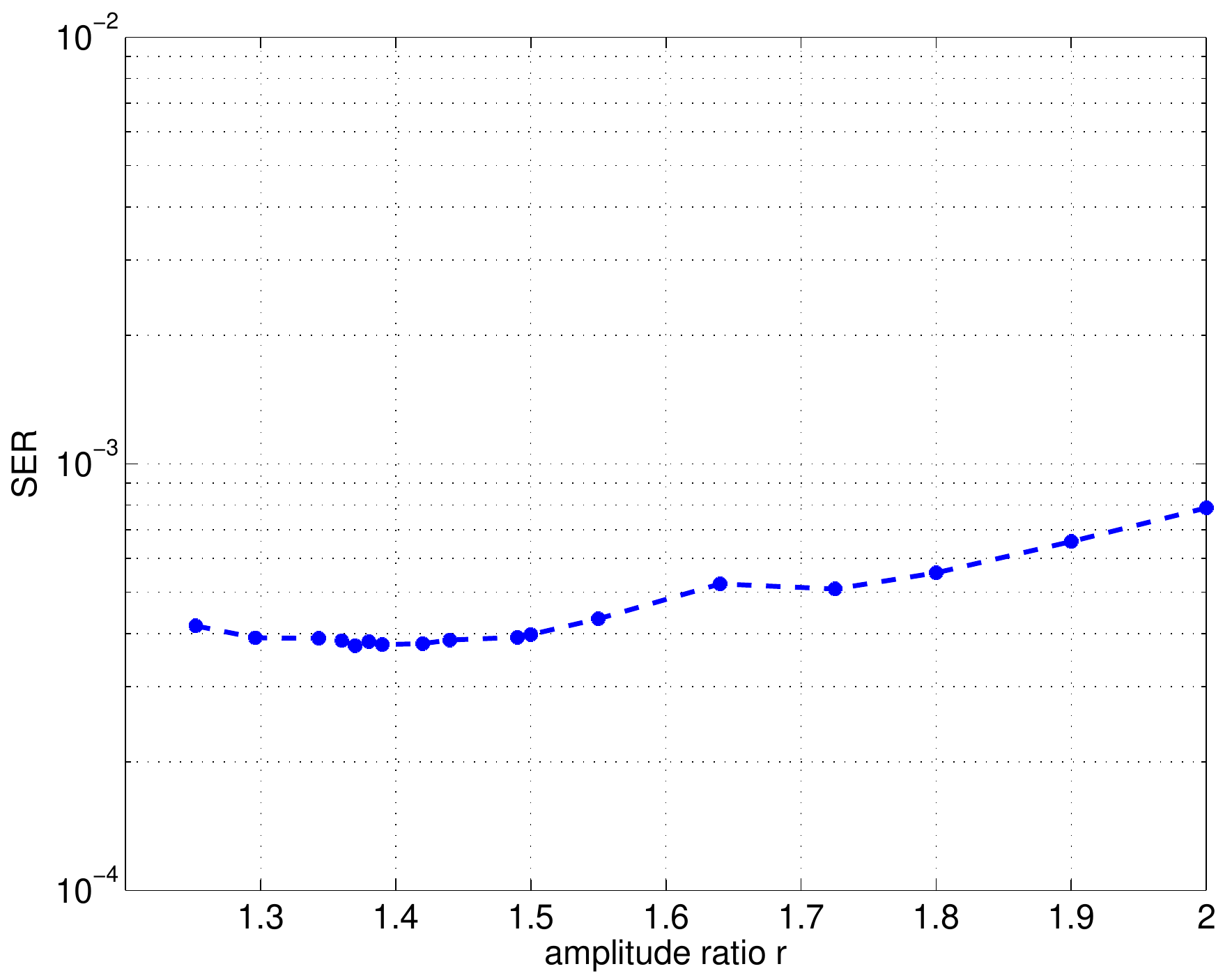}}\\
\subfloat[Optimal 8-QAM constellation at $r=1.37$, $\theta_1=12.73^o$ and $\theta_2=58.18^o$]{\label{subfig:opt_8_QAM_MDC_QOSTBC}\includegraphics[width=0.49\textwidth]{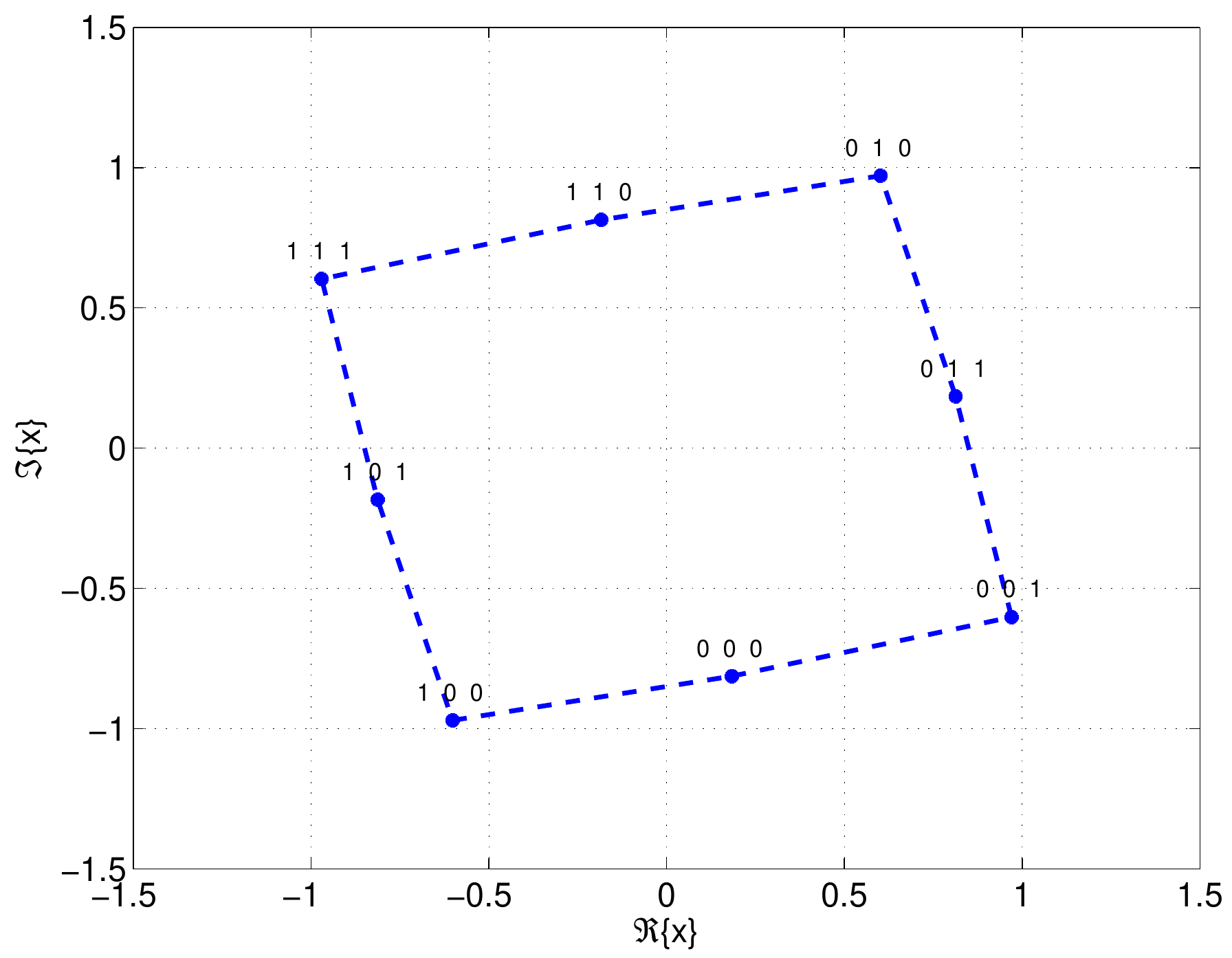}}
  \caption{8-QAM constellation design for MDC-QOSTBC.}
  \label{fig:8_QAM_MDC_QOSTBC}
\end{figure}

Finally, we evaluate the normalization constants $p_1$ and $p_2$. It has been shown above that
\begin{equation*}
 p_1\req\E[\alpha_\tau+\beta_\tau], \hspace{1cm} p_2\req\E[\alpha_\tau-\beta_\tau]
\end{equation*}
where $\alpha_\tau$ and $\beta_\tau$ are shown in (\ref{eq:alpha_beta_combined_M_4_MDC_DQOSTBCs}). Since all symbols $x_i$ $\forall$ $i=1,...,4$ are drawn from the same symmetric constellation, the average of $\beta_\tau$ vanishes, i.e.
\begin{equation*}
 \E[\beta_\tau]=\E\left[2\sum\limits_{i=1}^{2}{-x_i^Rx_i^I+x_{i+2}^Rx_{i+2}^I}\right]=0.
\end{equation*}
Therefore,
\begin{equation*}
 p_1=p_2=p=\E[\alpha_\tau]=\E\left[\sum\limits_{i=1}^{4}{|x_i|^2}\right]=4\E[|x|^2]=4.
\end{equation*}
The error rate performance of the addressed scheme will be shown in Section \ref{ss:DQOSTBC_performance_analysis}.\\

\subsection{Half Diversity Differential MDC-QOSTBCs}
\label{ss:QOSTBC_half_diversity}
As an attempt to reduce the receiver complexity further, this section shows an approach that achieves SRSD with rectangular QAM constellation while attaining only half diversity. The main difference between the currently addressed scheme compared to the scheme of the previous section lies in the construction of the information submatrices. In the previous section, linear combinations between the original submatrices $\V{V}_1$ and $\V{V}_2$ is made to get submatrices $\V{V}_{z_\tau}^{1\E}$ and $\V{V}_{z_\tau}^{2\E}$ of the equivalent subsystems. In this section, we do not use this linear combinations for the information matrices, thus we name the scheme un-combined MDC-QOSTBC. The information matrices are
 \begin{equation}
\begin{aligned}
\V{V}_{z_\tau}^{1\E}=&\frac{1}{\sqrt{p_1}}\V{V}_1=\frac{1}{\sqrt{p_1}}\begin{bmatrix}v_1 & v_2 \\-v_2^* & v_1^*\end{bmatrix}=\frac{1}{\sqrt{p_1}}\begin{bmatrix}(x_1^R+jx_3^R) & (x_2^R+jx_4^R)\\-(x_2^R+jx_4^R)^* & (x_1^R+jx_3^R)^*\end{bmatrix}\\
\V{V}_{z_\tau}^{2\E}=&\frac{1}{\sqrt{p_2}}\V{V}_2=\frac{1}{\sqrt{p_2}}\begin{bmatrix}v_3 & v_4 \\-v_4^* &v_3^*\end{bmatrix}=\frac{1}{\sqrt{p_2}}\begin{bmatrix}(-x_1^I+jx_3^I) & (-x_2^I+jx_4^I)\\-(-x_2^I+jx_4^I)^*  &(-x_1^I+jx_3^I)^*\end{bmatrix},
\end{aligned}
\label{eq:submatrices_no_comb}
\end{equation}
where $x_i$ for $i=1,...,4$ are the information bearing symbols to be encoded at block time index $\tau$.  Since matrices $\V{V}_{z_\tau}^{1\E}$ and $\V{V}_{z_\tau}^{2\E}$ are still orthogonal matrices, differential encoding is done exactly as shown in the previous section. Unlike in the previous section, the resulting system of the un-combined scheme is equivalent to a different system other than the one we started with. The resulting information matrices show that the first subsystem sends only the real parts of the 4 information symbols, whereas the second subsystem sends their imaginary parts.\\

From the submatrices in (\ref{eq:submatrices_no_comb}), the following holds;
\begin{equation}
 \left. 
\begin{aligned}
&\V{V}_{z_\tau}^{1\E\hr}\V{V}_{z_\tau}^{1\E}=\frac{1}{p_1}(|v_1|^2+|v_2|^2)\V{I}_{2}\\
&=\frac{1}{p_1}(|x_1^R+jx_3^R|^2+|x_2^R+jx_4^R|^2)\V{I}_2\\
&=\frac{1}{p_1}\sum\limits_{i=1}^{4}{(x_i^R)^2}\V{I}_2 \definedas (a_{\tau}^{1\E})^2\V{I}_2
\end{aligned}
\right. 
\hspace{1cm}
\begin{aligned}
&\V{V}_{z_\tau}^{2\E\hr}\V{V}_{z_\tau}^{2\E}=\frac{1}{p_2}(|v_3|^2+|v_4|^2)\V{I}_{2}\\
&=\frac{1}{p_2}(|-x_1^I+jx_3^I|^2+|-x_2^I+jx_4^I|^2)\V{I}_2\\
&=\frac{1}{p_2}\sum\limits_{i=1}^{4}{(x_i^I)^2}\V{I}_2 \definedas (a_{\tau}^{2\E})^2\V{I}_2.
\end{aligned}
\end{equation}
and the power normalization factors $p_1$ and $p_2$ are defined such that
\begin{equation}
\left.
\begin{aligned}
 &\E\left[(a_{\tau}^{1\E})^2\right]\req 1\\
&\implies p_1\req\E\left[\sum\limits_{i=1}^{4}{(x_i^R)^2}\right]\\
&=4\times\E[(x^R)^2]=4\times\frac{1}{2}=2
\end{aligned}
\right.
\hspace{1cm}
\begin{aligned}
 &\E\left[(a_{\tau}^{2\E})^2\right]\req 1\\
&\implies p_2\req\E\left[\sum\limits_{i=1}^{4}{(x_i^I)^2}\right]\\
&=4\times\E[(x^I)^2]=4\times\frac{1}{2}=2,
\end{aligned}
\end{equation}
as the average power of a constellation point is 1 which is equally distributed among the real and the imaginary components of the symbols making $p_1\eq p_2\eq  p\eq2$. Using (\ref{eq:submatrices_no_comb}), the information submatrices can be written in a dispersive form as $ \V{V}_{z_\tau}^{1\E}=\frac{1}{\sqrt{p}}\sum\limits_{i=1}^{4}{\V{U}_i^{1\E}}x_i^R,\,\, \V{V}_{z_\tau}^{2\E}=\frac{1}{\sqrt{p}}\sum\limits_{i=1}^{4}{j\V{Q}_i^{2\E}}x_i^I.$
\subsubsection{Sub-optimal Differential Decoder}
The same sub-optimal decoding approach used in the scheme of the previous section will be used here. Starting from (\ref{eq:sub_opt_metric_general}), the metric for the first subsystem can be derived as
\begin{eqnarray*}
 \xi_{1\E}&=&\mini\limits_{l=0,...,L-1}\:\,\frac{\tr(\V{Y}_{\tau-1}^{1\E\hr}\overbrace{\V{V}_l^{1\E\hr}\V{V}_l^{1\E}}^{(a_{\tau}^{1\E})^2}\V{Y}_{\tau-1}^{1\E})}{a_{\tau-1}^{1\E}}-2\Re\{\tr(\V{Y}_{\tau}^{1\E\hr}\V{V}_l^{1\E}\V{Y}_{\tau-1}^{1\E})\}\\
&=&\mini\limits_{x_i^R\in\mathcal{A}_i^R}\:\,\underbrace{\frac{\tr(\V{Y}_{\tau-1}^{1\E\hr}\V{Y}_{\tau-1}^{1\E})}{2\sqrt{\sum\limits_{i=1}^{4}{(x_{i,\tau-1}^R)^2}}}}_{\tilde{y}_1}.\sqrt{p}.\frac{\sum\limits_{i=1}^{4}{(x_{i}^R)^2}}{p}-\frac{1}{\sqrt{p}}\sum\limits_{i=1}^{4}{\underbrace{\Re\{\tr\big(\V{Y}_{\tau-1}^{1\E}\V{Y}_{\tau}^{1\E\hr}\V{U}_i^{1\E}\big)\}}_{\tilde{x}_i^R}x_i^R}\\
&=&\mini\limits_{x_i^R\in\mathcal{A}_i^R}\:\,\tilde{y}_1\sum\limits_{i=1}^{4}{(x_{i}^R)^2}-\sum\limits_{i=1}^{4}{\tilde{x}_i^Rx_i^R},
\end{eqnarray*}
thus the metric of the first subsystem decides on the real components of the symbols independently and similarly the metric of the second subsystem decides on the imaginary components of the symbols. The final metric is
\begin{equation}
 \boxed{
\begin{aligned}
\hat{x}_i^R=&\argmin\limits_{x_i^R\in\mathcal{A}_i^R}\:\,\tilde{y}_1(x_{i}^R)^2-\tilde{x}_i^Rx_i^R\\
\hat{x}_i^I=&\argmin\limits_{x_i^I\in\mathcal{A}_i^I}\:\,\tilde{y}_2(x_{i}^I)^2-\tilde{x}_i^Ix_i^I,
\end{aligned}}
\label{eq:un_combined_MDC_QOSTBC_decoder}
\end{equation}
where $\tilde{y}_2=\large\frac{\tr(\V{Y}_{\tau-1}^{2\E\hr}\V{Y}_{\tau-1}^{2\E})}{2\sqrt{\sum\limits_{i=1}^{4}{(x_{i,\tau-1}^I)^2}}}$, $\tilde{x}_i^I=\Re\{\tr\big(j\V{Y}_{\tau-1}^{2\E}\V{Y}_{\tau}^{2\E\hr}\V{Q}_i^{2\E}\big)\}$ and $\mathcal{A}_i^R$ \& $\mathcal{A}_i^I$ are respectively the alphabets from which $x_i^R$ and $x_i^I$ are drawn. For the information symbols at block index $\tau-1$, we use the notation $x_{i,\tau-1}$, whereas for the information symbols at block index $\tau$, the notation $x_i$ is used for consistency and simplicity.\\

To make use of the lower decoding complexity achieved by the un-combined differential MDC-QOSTBC, the constellation used should encode the real and imaginary components of the symbols independently. So we only use rectangular QAM constellation for this purpose. It was proved in section \ref{s:Diversity_proof_DQOSTBC} that the addressed un-combined differential MDC-QOSTBC with rectangular QAM achieves half diversity. Thus the price paid by the lower decoding complexity of this scheme is the reduced diversity order. The transmission scheme of the un-combined MDC-QOSTBC follows the same description in Figure \ref{fig:DMDC_QOSTBC_system} except that $\V{V}_{z_\tau}^{1\E}\eq\frac{1}{\sqrt{p_1}}\V{V}_1$ and $\V{V}_{z_\tau}^{2\E}\eq\frac{1}{\sqrt{p_2}}\V{V}_2$ and that the metric of each subsystem decides on one dimension of the information symbols, i.e. no metric combination between both subsystems is performed.\\

Due to the reduced diversity order, the slope of the error rate curves are expected to be lower than that of the full-diversity codes. This indicates a better performance in the low SNR range up to some SNR value. Hence, the un-combined MDC-QOSTBC with rectangular QAM constellation is of interest if the complexity is of more concern than the error performance, or if the operating SNR point is in the range below the intercept point with the full-diversity curve. Another case in which this scheme might be of interest is when a channel coding block is included in the system. In this case, good performance is needed in the low SNR range.
\subsection{Performance Analysis}
\label{ss:DQOSTBC_performance_analysis}

In this section, we show the error performance of both the full-diversity combined and the half-diversity un-combined differential MDC-QOSTBCs and compare them to all previously investigated schemes at spectral efficiencies 2, 3, 4 and 6 bits/s/Hz. Figures \ref{fig:M_4_at_2_bits_s_Hz}-\ref{fig:M_4_at_6_bits_s_Hz} show the BER curves and Tables \ref{tab:M_4_at_2_bits_s_Hz}-\ref{tab:M_4_at_6_bits_s_Hz} summarize the characteristics associated with every technique in terms of diversity order, code rate, error performance and complexity. For every spectral efficiency the schemes are ordered in the tables based on the error performance from the best to the worst. The scheme that achieves the best BER performance-complexity trade off is highlighted, except for rate $\unit[6]{bits/s/Hz}$ where the difference in complexity and performance is significant, so the operating SNR governs the decision on the most suitable scheme. To show the slope of a diversity order 2 curve, we included the Alamouti two transmit antenna curves that achieve a diversity order of 2. The complexity is measured by the search space (number of test candidates) for the decision of one information block. In case of SCSD, the search space is $L\eq$ alphabet size $\times$ number of symbols per information block and in case the symbols are drawn from different alphabets, the search space is the sum of the alphabet size of the $K$ information symbols. In SRSD decoders, the search space is measured as one-dimensional alphabet size $\times$ 2 $\times$ number of information symbols $K$.\\

In conclusion, this chapter covered the use of QOSTBCs in differential non-coherent systems. The used code for all investigated schemes is the MDC-QOSTBC. The main advantage this class of QOSTBC offers is the reduced decoding complexity from pair-wise to single complex symbol decoding. The rationale behind such complexity reduction is the interleaving of the real and imaginary components of the information symbols. In order to use MDC-QOSTBCs in differential systems, two main approaches were investigated. The first approach imposes limitations on the constellation to orthogonalize the code matrix and proved to have poor performance for high spectral efficiencies. The second approach treats the system as two orthogonal subsystems and differentially encode each subsystem. By taking linear combinations between the information matrices of both subsystems, one gets an equivalent system that achieves SCSD and can achieve full diversity with appropriate constellation rotation. Differential full-diversity combined MDC-QOSTBCs proved to achieve significant performance improvement compared to differential orthogonal STBCs. For a further reduction of the complexity, one may skip the linear combination step of the information submatrices to allow independent decision on the real and imaginary components of the symbols with the penalty of a reduced diversity order of two.

\begin{figure}[!ht]
\centering
\vspace{-0.1cm}
\includegraphics[width=0.55\textwidth]{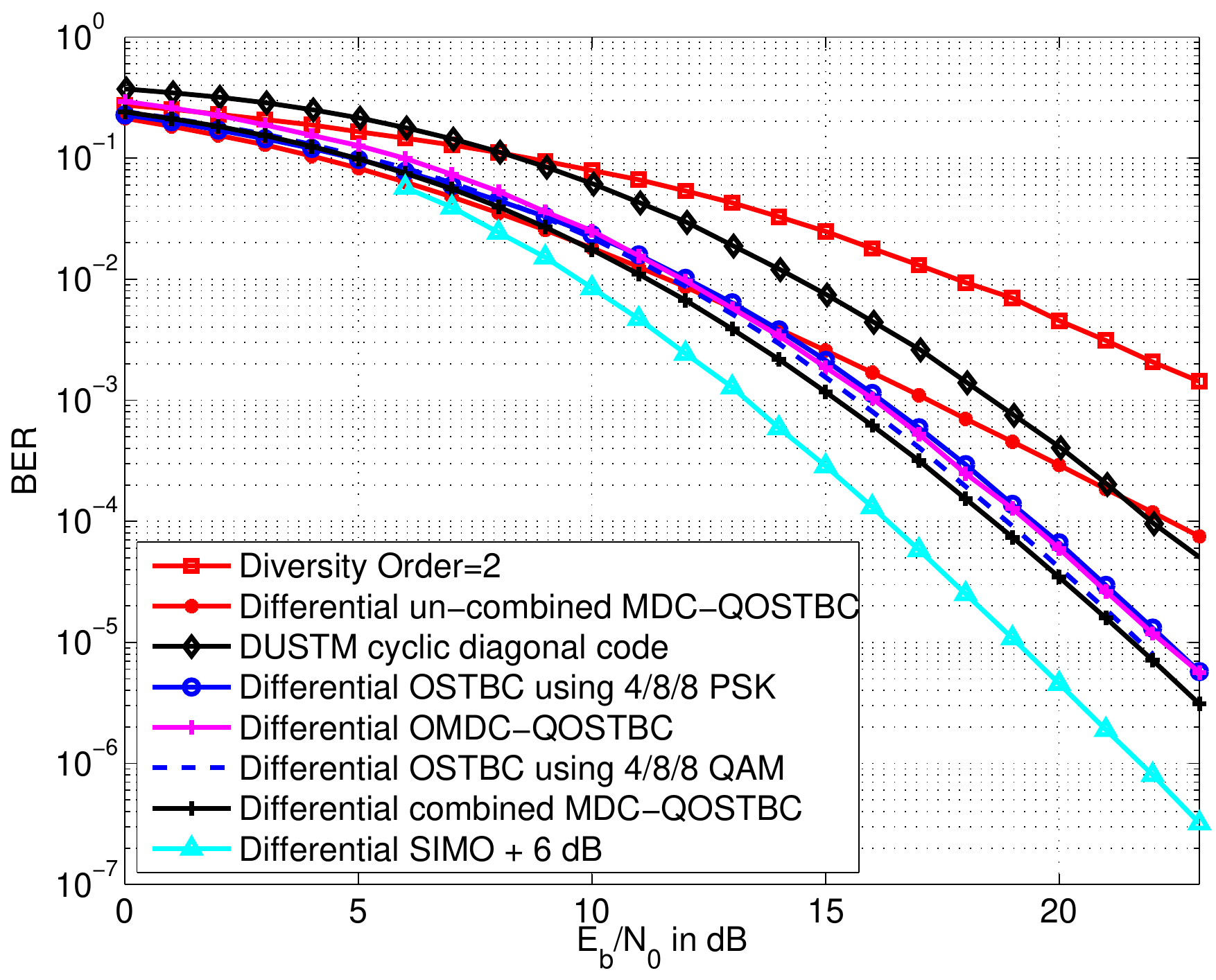}
 \caption{BER comparison for differential STC schemes in a $4\!\times\!1$ system at $\unit[2]{bits/s/Hz}$.}
\label{fig:M_4_at_2_bits_s_Hz}
\end{figure}

 \begin{table}[!ht]
\begin{center}
\caption{Comparison of differential STC schemes in a $4\!\times\!1$ system at $\unit[2]{bits/s/Hz}$.}
\label{tab:M_4_at_2_bits_s_Hz}
\begin{tabular}{|p{2cm}|p{2.5cm}|p{2cm}|p{6.5cm}|} \hline 
\textbf{Scheme} & \textbf{Constellation} & \textbf{Search space complexity} & \textbf{Comments} \tn\hline
\shadeRow Combined MDC-QOSTBC & 4-PSK rotated by $\theta=13.28^o$ & $4\!\times\!4=16$ & achieves full diversity and full rate.\tn\hline
OSTBC using $\text{T-H}$ code matrix & 4/8/8 QAM & $L=4+8+8=20$ & achieves full diversity and rate $\frac{3}{4}$. BER is slightly worse than combined MDC-QOSTBC. The scheme requires constellation change every 2 time slots at Tx and Rx.\tn\hline
OMDC-QOSTBC & Optimized 4-point constellation in Figure \ref{subfig:4_point_constellation} & $L=4\!\times\!4=16$ & achieves full rate and full-diversity. BER performance is worse than combined MDC-QOSTBC by about $\unit[0.8]{dB}$. \tn \hline
OSTBC using $\text{T-H}$ code matrix & 4/8/8 PSK & $L=4+8+8=20$ & full-diversity and rate $\frac{3}{4}$. BER is almost the same as OMDC-QOSTBC. The scheme requires constellation change every 2 time slots at Tx and Rx.\tn\hline
DUSTM cyclic diagonal code & bits mapped directly to matrices that belong to a group codebook &  $L\eq2^{\eta M}\eq2^{2\times 4}\eq256$ & achieves full diversity. BER is worse than combined MDC-QOSTBC by more than $\unit[3]{dB}$ and the complexity is significantly higher.\tn\hline
Un-combined MDC-QOSTBC & 4-PSK &  $L=2\times 2 \times 4=16$ & achieves half diversity and same complexity as combined MDC-QOSTBC, so the only advantage it offers compared to combined scheme is the simultaneously decision on the real and imaginary components of the symbols leading to a faster decoding.\tn\hline
 \end{tabular}	
\end{center}
\end{table}
\begin{figure}[!ht]
\centering
\includegraphics[width=0.6\textwidth]{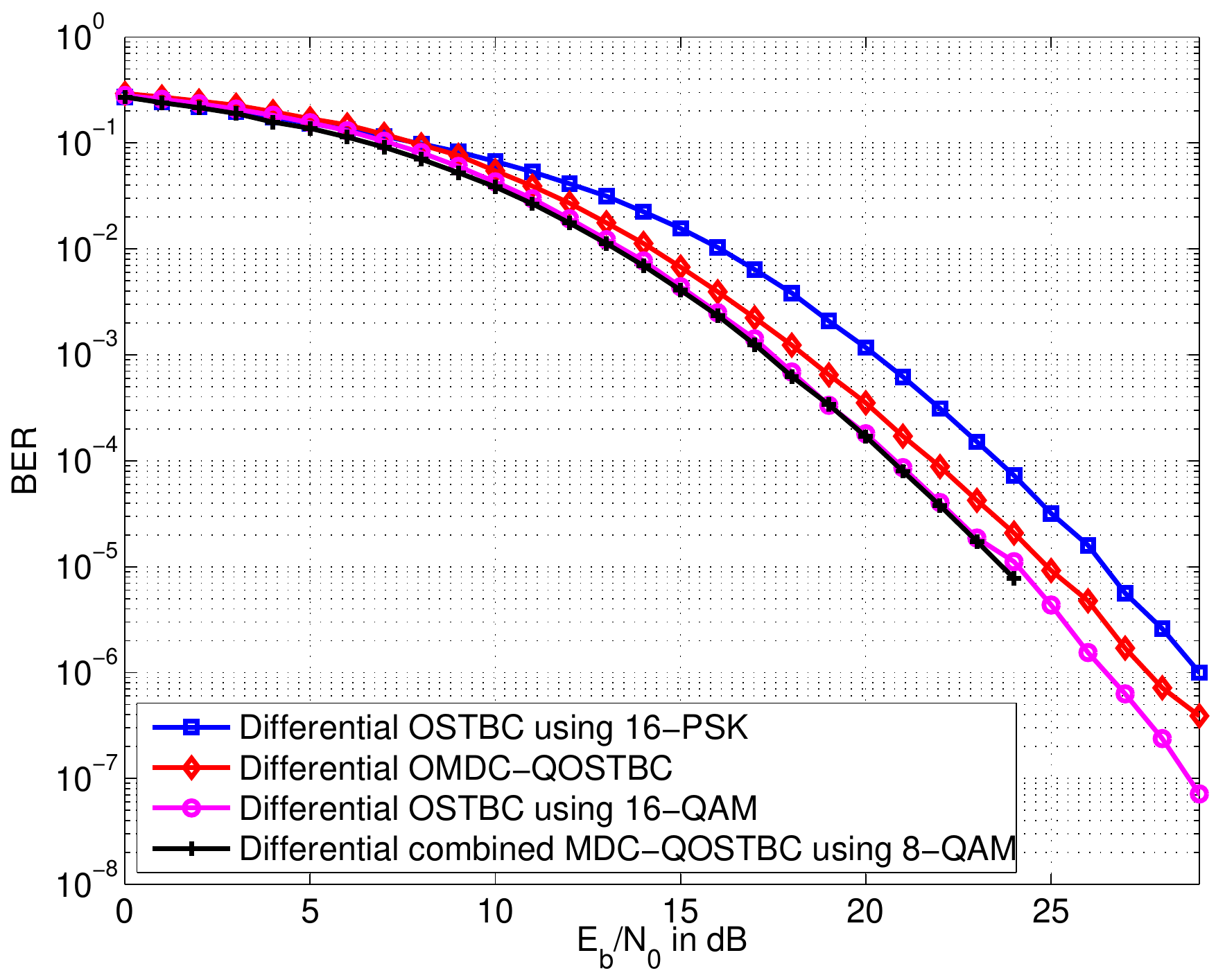}
 \caption{BER comparison for differential STC schemes in a $4\!\times\!1$ system at $\unit[3]{bits/s/Hz}$.}
\label{fig:M_4_at_3_bits_s_Hz}
\end{figure}
\begin{table}[!ht]
\begin{center}
\caption{Comparison of differential STC schemes in a $4\!\times\!1$ system at $\unit[3]{bits/s/Hz}$.}
\label{tab:M_4_at_3_bits_s_Hz}
\begin{tabular}{|p{2cm}|p{2.5cm}|p{2cm}|p{6.5cm}|}
    \hline 
\textbf{Scheme} & \textbf{Constellation} & \textbf{Search space complexity} & \textbf{Comments} \tn\hline
Combined MDC-QOSTBC & Optimized 8-QAM shown in Figure \ref{subfig:opt_8_QAM_MDC_QOSTBC}. & $L=8\times4=32$ & achieves full diversity  and full-rate. Placed first since its SER is $\unit[1]{dB}$ better than OSTBC with 16-QAM but both have same BER. \tn \hline
\shadeRow OSTBC using $\text{T-H}$ code matrix & 16-QAM &  $L=4\times2\times3=24$ & achieves full diversity and rate $\frac{3}{4}$, same BER as combined MDC-QOSTBC but with lower decoding complexity.\tn\hline
OMDC-QOSTBC & Optimized 8-point constellation in Figure \ref{subfig:8_point_constellation} & $L=8\times4=32$ & BER is worse than combined MDC-QOSTBC or OSTBC with 16-QAM by about $\unit[1]{dB}$. \tn\hline
OSTBC using $\text{T-H}$ code matrix & 16-PSK &  $L=16\times3=48$ & achieves full diversity and rate $\frac{3}{4}$. BER performance is worse than combined-MDC QOSTBC  or OSTBC with 16-QAM by about $\unit[2.5]{dB}$. \tn\hline
 \end{tabular}	
\end{center}
\end{table}

\begin{figure}[!ht]
\centering
\includegraphics[width=0.6\textwidth]{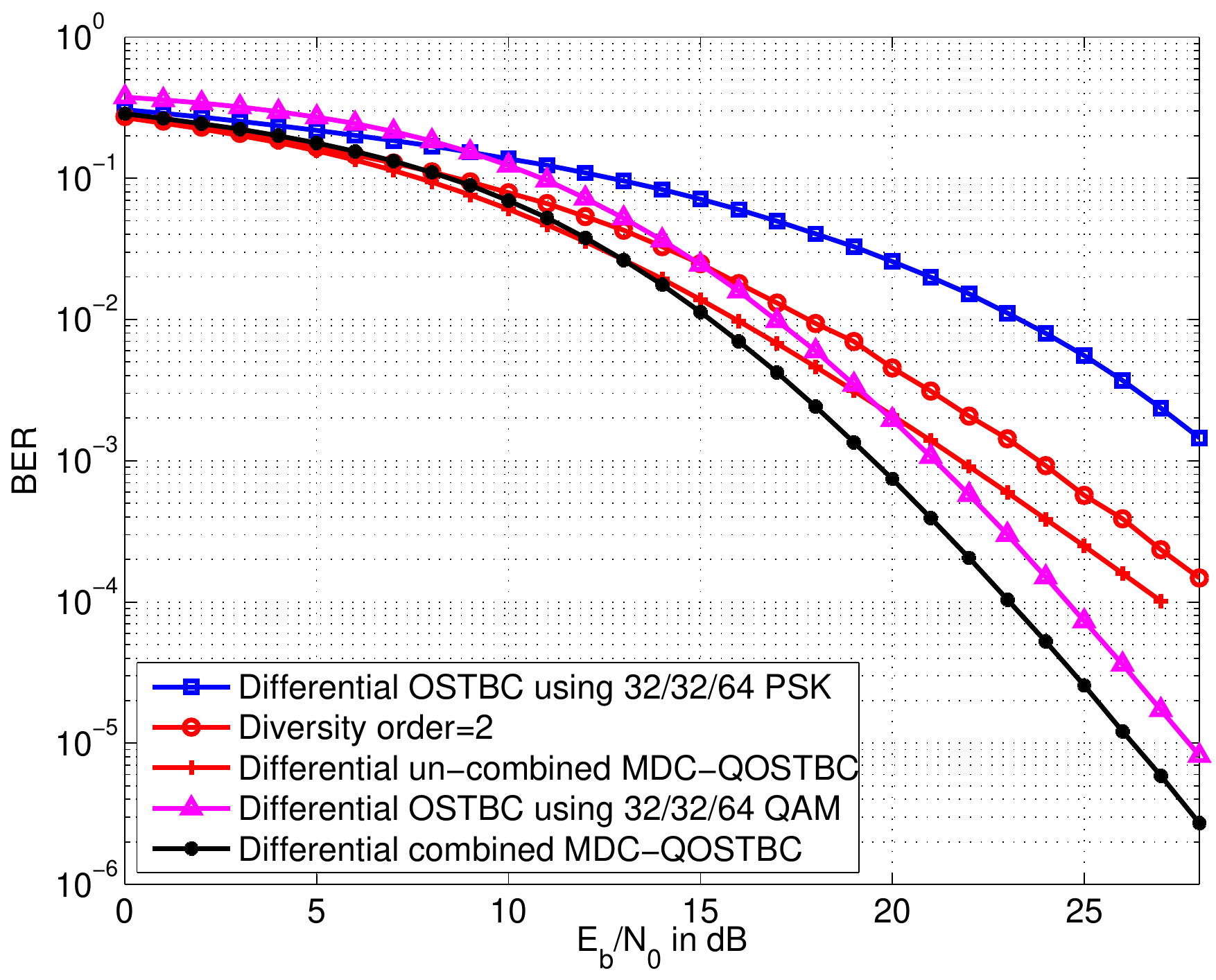}
 \caption{BER comparison for differential STC schemes in a $4\!\times\!1$ system at $\unit[4]{bits/s/Hz}$.}
\label{fig:M_4_at_4_bits_s_Hz}
\end{figure}
\begin{table}[htp]
\begin{center}
\caption{Comparison of differential STC schemes in a $4\!\times\!1$ system at $\unit[4]{bits/s/Hz}$.}
\label{tab:M_4_at_4_bits_s_Hz}
\vspace{-0.5cm}
\begin{tabular}{|p{2cm}|p{2.5cm}|p{2cm}|p{8cm}|}
    \hline 
\textbf{Scheme} & \textbf{Constellation} & \textbf{Search space complexity} & \textbf{Comments} \tn \hline
\shadeRow Combined MDC-QOSTBC & 16-QAM rotated by $\theta=13.28^o$ shown in Figure \ref{fig:opt_16_QAM_MDC_QOSTBC}. & $L=16\times4=64$ & achieves full diversity and full-rate. \tn\hline
OSTBC using $\text{T-H}$ code matrix & 32/32/64 QAM & $L=32+32+8\times2=80$ & achieves full diversity and rate $\frac{3}{4}$. BER is worse than combined MDC-QOSTBC by about $\unit[1.5]{dB}$. The scheme requires constellation change every 2 time slots at Tx and Rx.\tn\hline
Un-combined MDC-QOSTBC & 16-QAM & $L=4\times2\times4=32$ & achieves half diversity and full-rate. BER worse than Combined MDC-QOSTBC starting $E_b/N_0=\unit[14]{dB}$. The scheme sacrifices diversity for lower decoding complexity.\tn\hline
OSTBC using $\text{T-H}$ code matrix & 32/32/64 PSK & $L=32+32+64=128$ & achieves full diversity and rate $\frac{3}{4}$. BER is worse than combined MDC-QOSTBC by about $\unit[9]{dB}$ at BER=$10^{-3}$ and requires twice as much decoding search space. Additionally, the scheme requires constellation change every 2 time slots at Tx and Rx.\tn\hline
 \end{tabular}	
\end{center}
\end{table}
\begin{figure}[!ht]
\centering
\includegraphics[width=0.7\textwidth]{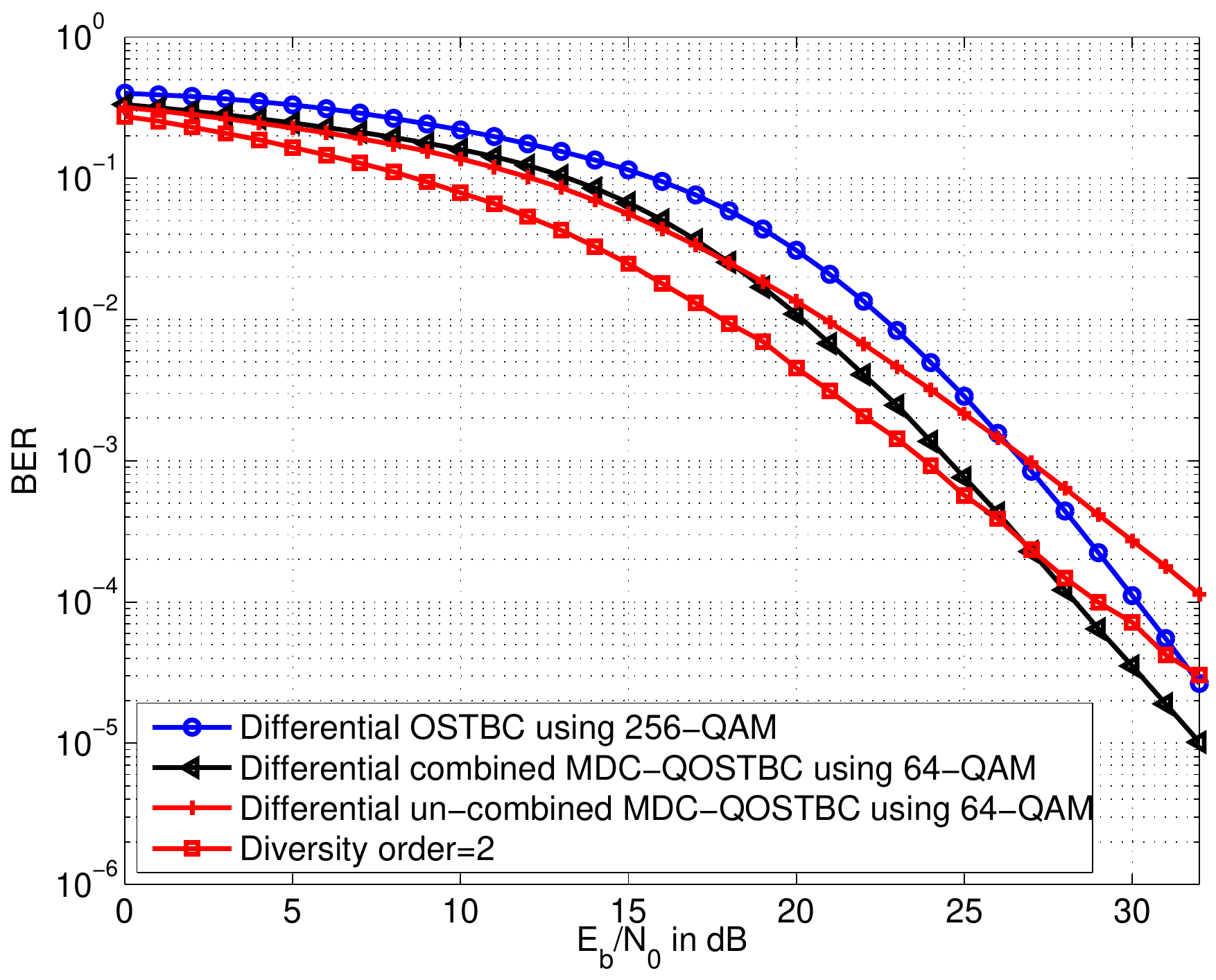}
 \caption{BER comparison for differential STC schemes in a $4\!\times\!1$ system at $\unit[6]{bits/s/Hz}$.}
\label{fig:M_4_at_6_bits_s_Hz}
\end{figure}

\begin{table}[htp]
\begin{center}
\caption{Comparison of differential STC schemes in a $4\!\times\!1$ system at $\unit[6]{bits/s/Hz}$.}
\label{tab:M_4_at_6_bits_s_Hz}
\begin{tabular}{|p{2cm}|p{2.5cm}|p{2cm}|p{7cm}|}
    \hline 
\textbf{Scheme} & \textbf{Constellation} & \textbf{Search space complexity} & \textbf{Comments} \tn \hline
Combined MDC-QOSTBC &64-QAM rotated by $\theta=13.28^o$. & $L=64\times4=256$ & achieves full diversity and full-rate. \tn\hline
OSTBC using $\text{T-H}$ code matrix & 256-QAM & $L=16\times2\times3=96$ & achieves full diversity and rate $\frac{3}{4}$. BER is worse than combined MDC-QOSTBC by about $\unit[2]{dB}$.\tn\hline
Un-combined MDC-QOSTBC & 64-QAM & $L=8\times2\times4=64$ & achieves half diversity and full-rate. BER worse than combined MDC-QOSTBC starting $E_b/N_0=\unit[20]{dB}$. The scheme sacrifices diversity for lower decoding complexity.\tn\hline
 \end{tabular}
\end{center}
\end{table}

\abbrev{QOSTBC}{Quasi-Orthogonal Space-Time Block Code}
\abbrev{MDC-QOSTBC}{Minimum Decoding Complexity Quasi-Orthogonal Space-Time Block Code}
\abbrev{OMDC-QOSTBC}{Orthogonalized MDC-QOSTBC}

\chapter{Conclusion}
\label{chap:conclusion}

In this thesis, we investigated the theory and applications of differential modulation schemes in single carrier wireless systems. In SISO systems, the use of DAPSK proved to achieve significant performance enhancement over DPSK especially for high spectral efficiencies. Through MSDD, the error floor associated with fast varying channels is removed.\\

With no channel knowledge at the transmitter nor at the receiver, differential space-time coding techniques proved to  achieve transmit diversity leading to faster decay of error rate curves. Having proved the design criteria for unitary ST codes in non-coherent systems, the simulations of all investigated schemes proved to comply with this criteria. After visiting orthogonal STBCs, the error rate curves showed to be still inferior to the best achievable performance limit. To this end, we studied quasi-orthogonal STBCs which promise higher code rate and potential performance improvement. With the goals of preserving the low complexity of OSTBCs while decoding with no CSI at the receiver, we proposed the use of MDC-QOSTBCs in the differential non-coherent domain. The proposed code proved to achieve significant performance improvement compared to OSTBCs for high spectral efficiencies.\\

Although a plethora of research work has been done on the topic of space-time codes, there are still open problems that need further investigation. The design of codes that achieve a code rate more than one is an interesting area of research. These codes may achieve better transmission quality for high data rate applications. For next generation wireless communications, providing good quality of service for high speed mobile stations motivates extending good coherent STBCs to the differential non-coherent domain to dispense the need of CSI. \\

To bridge the performance gap between coherent and non-coherent STBCs and to remove the error floor incurred in fast fading channels, MSDD can be extended to MIMO systems. The brute force search of ML MSDD leads to an exponential increase in the decoding complexity with the MSDD window length. Reduced search space through sub-optimal MSDD schemes like sphere decoding was proved in the literature to achieve near-ML performance. Thus, we suggest for future work the investigation of MSDD using sphere decoding with the differential MDC-QOSTBCs proposed in this thesis. Studying the performance of the addressed schemes in multipath environments is as well encouraged. Lastly, the research covered in this thesis consider single user systems, so we suggest extending the addressed schemes to incorporate multi users.

\cleardoublepage
\phantomsection
\addcontentsline{toc}{chapter}{\listfigurename}
\listoffigures
\cleardoublepage
\phantomsection
\addcontentsline{toc}{chapter}{\listtablename}
\listoftables

\appendix
\chapter{Fundamentals of Linear Algebra}
\label{App:math_formulas}
This appendix aims at defining some of the fundamental topics in linear Algebra. It specifically explains the notion of trace, rank, determinant, eigenvalues, and singular values of a matrix. Moreover, some of the useful properties of such entities are included and the relation between them is also provided.
\section{Properties of the Trace Operator}
\label{s:trace_properties}
The trace of a square matrix $\V{A}$ is defined as the sum of the elements of its main diagonal. Let $\V{A}$ be a square matrix defined as $\V{A}_{M\!\times\! M}=\{a_{mn}\},\:m,n\in\{1,...,M\}$, then the trace of $\V{A}$ is
\begin{equation}
\tr(\V{A})=\sum\limits_{m=1}^{M}{a_{mm}}.
\end{equation}
Based on the above definition, the following properties trivially hold. Let $\V{A}$ and $\V{B}$ be square matrices of the same dimension, and let $a$ be a complex scalar, then
\begin{eqnarray*}
\tr(a)&=&a\\
\tr(\V{A}+\V{B})&=&\tr(\V{A})+\tr(\V{B})\\
\tr(\V{A}^T)&=&\tr(\V{A})\\
\tr(\V{A^*})&=&\tr(\V{A})^*
\end{eqnarray*}
From the above properties, the following relationships also hold;
\begin{eqnarray*}
\tr(\V{A}\hr)&=&\tr(\V{A}^T)^*=\tr(\V{A})^*\\
\tr(\V{A}+\V{A}\hr)&=&\tr(\V{A})+\tr(\V{A})^*=2\,\Re\{\tr(\V{A})\}\\
\end{eqnarray*}
Moreover, the trace operator is invariant under cyclic permutations, i.e.
\begin{equation}
\label{eq:tr_permutation}
\tr(\V{ABC})=\tr(\V{CAB})=\tr(\V{BCA})
\end{equation}
\section{Properties of the Determinant Operator}
\label{s:Determinant}
The determinant is an operator that operates on a square matrix and results in a scalar. This section summarizes some useful properties of the determinant function. Let $\V{A}$ and $\V{B}$ be any two square matrices then the following holds;
\begin{eqnarray*}
\det(\V{AB})&=&\det(\V{A})\det(\V{B})\\
\det(\V{A}^{-1})&=&(\det(\V{A}))^{-1}\\
\det(\V{A}^{T})&=&\det(\V{A})\\
\det(\V{A}\hr)&=&\det((\V{A}^T)^*)=(\det(\V{A}))^*
\label{eq:Det_properties}
\end{eqnarray*}
Furthermore the determinant of diagonal matrices is the product of the elements along the main diagonal. So if $\V{A}$ is an $M\!\times\!M$ diagonal matrix of diagonal entries $a_m$, $m\in\{1,...,M\}$, then
\begin{equation*}
\det(\V{A})=\prod\limits_{m=1}^{M}{a_m}
\end{equation*}
An important information that a determinant of a matrix provides, is whether or not the matrix has an inverse. Explicitly, a matrix is invertible if and only if its determinant is non-zero.\\

Another very useful theorem that has been extensively used in this thesis to evaluate certain determinants is known as Sylvester's determinant theorem and it states the following;
\begin{theorem}
Let $\V{A}$ be an $M\!\times\!N$ matrix and $\V{B}$ an $N\!\times\!M$ matrix, then
\begin{equation*}
\det(\V{I}_M+\V{AB})=\det(\V{I}_{N}+\V{BA})
\end{equation*}
\end{theorem}

\section{Rank of a Matrix}
\label{rank}
The rank of an $M\!\times\!N$ matrix $\V{A}$ is defined as the number of linearly independent rows or columns of $\V{A}$ which is at most the minimum of both dimensions $M$ and $N$, namely
\begin{equation}
\rank(\V{A})\leq\min(M,N).
\end{equation}
A matrix whose rank is the maximum achievable is said to have full rank, otherwise it is rank deficient. Full rank matrices have non-zero determinant and therefore are invertible, where as rank deficient matrices have zero determinant and therefore are non-invertible (singular).\\

For a diagonal matrix $\V{A}$, if there exist $r$ non-zero elements on the main diagonal, then there exist $r$ linearly independent vectors in $\V{A}$, therefore $r$ is the rank of $\V{A}$. Consequently, the rank of a diagonal matrix is the number of non-zero elements on its main diagonal.\\

Another useful property of the rank operator is that it is invariant under multiplication from any side by a full rank square matrix. Let $\V{U}$ and $\V{V}$ be full rank matrices of dimension $M\!\times\!M$ and $N\!\times\!N$, respectively, and let $\V{A}$ be an arbitrary $M\!\times\!N$ matrix, then the following holds
\begin{equation}
\label{eq:rank_unitary}
\rank(\V{UA})=\rank(\V{AV})=\rank(\V{A})
\end{equation}
\section{Eigenvalue Decomposition and Singular Value Decomposition}
\label{s:EVD_SVD}
This section provides the basic theory of eigenvalues and singular values as one of the important concepts in the field of linear Algebra. Two well-known matrix decompositions are defined, namely the eigenvalue decomposition and the singular value decomposition. The section also includes some of the useful properties of eigenvalues and singular values that have been used in this thesis.
\subsection{Eigenvalue Decomposition}
\label{ss:EVD}
When a square matrix $\V{A}$ acts on vector $\V{x}$, it may change its magnitude or its direction or both. If only the magnitude of $\V{x}$ is changed by a factor $\lambda$, this can be described as
\begin{equation}
\label{eq:Eigen_value_problem}
\V{A}\V{x}=\lambda\V{x}
\end{equation}
where $\lambda$ is in general positive or negative complex scalar, so the direction of $\V{x}$ might be reversed in case $\lambda$ is negative. The special vectors which keep their direction unchanged (or possibly reversed) after being acted upon by matrix $\V{A}$ are known as the eigenvectors of $\V{A}$, and the special factors of magnitude change of the corresponding eigenvectors are known as the eigenvalues of $\V{A}$. On the other hand, if $\V{A}$ acts on a non-eigenvector $\V{x}$, the output vector $\V{Ax}$ points in a direction other than that of $\V{x}$ or $-\V{x}$.\\

"Eigen" is a German prefix that means "own" or "Characteristic". This indicates that eigenvectors and eigenvalues are of characteristic and unique nature to a matrix. An $M\!\times\!M$ matrix $\V{A}$ can have at most $M$ non-zero eigenvalues, and it can be decomposed as 
\begin{equation}
\label{eq:EVD}
\V{A}=\V{X\Lambda X}^{-1}
\end{equation}
where $\V{X}$ is an $M\!\times\!M$ matrix whose columns are the eigenvectors of $\V{A}$, and $\V{\Lambda}$ is a diagonal matrix whose diagonal contains the corresponding eigenvalues of $\V{A}$, namely $\lambda_i,\, i\in\{1,...,M\}$. Such a decomposition is known as Eigen-Value Decomposition (EVD) and is defined only for square matrices. Nevertheless, not all square matrices can be eigen-decomposed. Only matrices whose all eigenvectors are linearly independent can be eigen-decomposed, because otherwise matrix $\V{X}$ will be rank deficient, and therefore will not be invertible. It is obvious to see that for diagonal matrices, $\V{X}$ is an identity matrix, and therefore the elements on the main diagonal of $\V{A}$ are themselves the eigenvalues.\\

Using the properties of the determinant operator defined in Section \ref{s:Determinant}. The following theorem holds;
\begin{theorem}
\label{theo:det_EV}
Let $\V{A}$ be an $M\!\times\!M$ matrix which is eigen-decomposable. Then the determinant of $\V{A}$ is the product of its eigenvalues.
\end{theorem}
\begin{proof}
\begin{equation}
 \begin{aligned}
  \V{A}&=\V{X\Lambda X}^{-1}\\
\det(\V{A})&=\det(\V{X\Lambda X}^{-1})=\det(\V{X})\det(\V{\Lambda})\det(\V{X}^{-1})\\
&=\det(\V{X})\det(\V{\Lambda})\frac{1}{\det(\V{X})}=\det(\V{\Lambda})=\prod\limits_{i=1}^{M}{\lambda_i}
 \end{aligned}
\label{eq:det_properties}
\end{equation}
\end{proof}
with $\lambda_i$ being the $i^{\text{th}}$ eigenvalue of $\V{A}$.\\

Using the properties of the trace operator defined in Section \ref{s:trace_properties}. The following theorem holds;
\begin{theorem}
Let $\V{A}$ be an $M\!\times\!M$ matrix which is eigen-decomposable. Then the trace of $\V{A}$ is the sum of its eigenvalues.
\end{theorem}
\begin{proof}
\begin{eqnarray*}
\V{A}&=&\V{X\Lambda X}^{-1}\\
\tr(\V{A})&=&\tr(\V{X\Lambda X}^{-1})=\tr(\V{\Lambda}\V{X}^{-1}\V{X})=\tr(\V{\Lambda})=\sum\limits_{i=1}^{M}{\lambda_i}
\end{eqnarray*}
\end{proof}
The next theorem relates the eigenvalues of a matrix to its rank. 
\begin{theorem}
\label{theo:rank_EV}
The rank of a matrix is the number of its non-zero eigenvalues.
\end{theorem}
\begin{proof}
Since both $\V{X}$ and $\V{X}^{-1}$ are full rank square matrices by the definition of EVD. Then it follows from (\ref{eq:rank_unitary}) that
\begin{equation}
\rank(\V{A})=\rank(\V{X\Lambda X}^{-1})=\rank(\V{\Lambda})
\end{equation}
Since the rank of $\V{\Lambda}$ is the number of non-zero elements (eigenvalues) on its main diagonal, therefore it follows that the rank of a matrix $\V{A}$ is the number of its non-zero eigenvalues.
\end{proof}
\subsection{Singular Value Decomposition}
\label{ss:SVD}
In this subsection, another useful matrix decomposition is defined. For this we need to first define the unitary property of matrices.
\begin{definition}
A square $M\!\times\!M$ matrix $\V{A}$ is said to be unitary if its inverse is the same as its conjugate transpose, namely
\begin{equation*}
\V{A}^{-1}=\V{A}\hr.
\end{equation*} Therefore,
\begin{equation*} 
\V{U}\hr\V{U}=\V{UU}\hr=\V{U}^{-1}\V{U}=\V{I}_{M}.
\end{equation*}
\end{definition}
Singular Value Decomposition (SVD) is a matrix factorization which applies not only for square matrices as EVD but also for rectangular matrices. Let $\V{A}$ be $\in\mathbb{C}^{M\!\times\!N}$, SVD is defined as
\begin{equation}
\V{A}=\V{U\Sigma V}\hr,
\label{eq:SVD}
\end{equation}
where $\V{U}\in\mathbb{C}^{M\!\times\!M}$, $\V{V}\in\mathbb{C}^{N\!\times\!N}$ and both are unitary matrices. The columns of $\V{U}$ are orthonormal basis vectors known as left-singular vectors, where as the columns of $\V{V}$ are  orthonormal basis vectors known as right-singular vectors. $\V{\Sigma}$ is an $M\!\times\!N$ diagonal matrix whose diagonal entries are known as the singular values of $\V{A}$, and they are non-negative real numbers. Matrix $\V{A}$ has at most $\min(M,N)$ singular values, which are denoted as $\sigma_i,\:i\in\{1,...,\min(M,N)\}$. Unlike EVD, any matrix can be singular-decomposed.\\

There exist some relationships between the singular values and the eigenvalues of matrices. Here only one relation will be shown in the next theorem. However, before describing such a relation, it is useful to first define the notion of \emph{similar} matrices as follows
\begin{definition}
Two $M\!\times\! M$ matrices $\V{A}$ and $\V{B}$ are said to be similar if
\begin{equation*}
\V{B}=\V{QA}\V{Q}^{-1}
\end{equation*}
for some invertible $M\!\times\!M$ matrix $\V{Q}$. Similar matrices have the same rank, determinant, trace and same eigenvalues.
\end{definition}
This can be easily verified using respectively the rank property in (\ref{eq:rank_unitary}), the first two determinant properties in (\ref{eq:det_properties}), the trace property in (\ref{eq:tr_permutation}), and the definition of the EVD, namely
\begin{eqnarray*}
\V{A}&=&\V{X\Lambda X}^{-1}\\
\V{B}&=&\V{QA}\V{Q}^{-1}=\underbrace{\V{QX}}_{\V{X}'}\V{\Lambda}\underbrace{\V{X}^{-1}\V{Q}^{-1}}_{\V{X}'^{-1}}.
\end{eqnarray*}
Now it is possible to define the desired relation between eigenvalues and singular values as shown in the following theorem.
\begin{theorem}
\label{theo:EV_SV2}
The eigenvalues of matrix $\V{A}\hr\V{A}$ or $\V{AA}\hr$ are the square of the singular values of matrix $\V{A}$
\end{theorem}
\begin{proof}
Let $\V{A}$ be an $M\!\times\!M$ matrix
\begin{eqnarray*}
\V{A}&=&\V{U\Sigma V}\hr\\
\V{A}\hr\V{A}&=&\V{V\Sigma}\hr\underbrace{\V{U}\hr\V{U}}_{\V{I}_M}\V{\Sigma V}\hr\\
&=&\V{V}(\V{\Sigma}\hr\V{\Sigma})\V{V}\hr\\
&=&\V{V}(\V{\Sigma}\hr\V{\Sigma})\V{V}^{-1}
\end{eqnarray*}
where $\V{V}\hr=\V{V}^{-1}$ since $\V{V}$ is unitary. In the last line we see that matrices $\V{A}\hr\V{A}$ and $\V{\Sigma}\hr\V{\Sigma}$ are \emph{similar} matrices, and therefore they have the same eigenvalues, namely
\begin{equation*}
\lambda_m(\V{A}\hr\V{A})=\lambda_m(\V{\Sigma}\hr\V{\Sigma})\:\:\:\forall\: m\in\{1,...,M\}
\end{equation*}
Furthermore, since $\V{\Sigma}\hr\V{\Sigma}$ is a diagonal matrix, therefore its eigenvalues are the elements on the main diagonal which are the square of the singular values of $\V{A}$. This concludes the theorem that the eigenvalues of $\V{A}\hr\V{A}$ are the same as the square of the singular values of $\V{A}$, i.e.
\begin{equation}
\label{eq:EV_SV2}
\lambda_m(\V{A}\hr\V{A})=\sigma_m^2(\V{A})\:\:\:\forall\: m\in\{1,...,M\}
\end{equation}
A similar proof holds for $\V{AA}\hr$.
\end{proof}

In the following theorem, the relation between the rank of a matrix and its singular values is shown.
\begin{theorem}
\label{theo:rank_SV}
The rank of matrix $\V{A}$ is the number of non-zero singular values of $\V{A}$.
\end{theorem}
\begin{proof}
Using the SVD of $\V{A}$,
\begin{eqnarray*}
\V{A}&=&\V{U\Sigma V}\hr\\
\rank(\V{A})&=&\rank(\V{U\Sigma V}\hr)
\end{eqnarray*}
Since both $\V{U}$ and $\V{V}\hr$ are unitary and therefore full rank square matrices, then it follows from (\ref{eq:rank_unitary}) that
\begin{equation*}
\rank(\V{A})=\rank(\V{\Sigma})
\end{equation*}
Since the rank of $\V{\Sigma}$ is the number of non-zero elements (singular values) on its main diagonal, therefore it follows that the rank of a matrix is the number of its non-zero singular values.
\end{proof}
Next is another theorem that uses the SVD to prove the rank of some matrix.
\begin{theorem}
\label{theo:rank_AA_A}
The rank of matrix $\V{A}\hr\V{A}$ is the same as the rank of matrix $\V{A}$, i.e.
\begin{equation}
\label{eq:rank_AA_A}
\rank(\V{A}\hr\V{A})=\rank(\V{A})
\end{equation}
\end{theorem}
\begin{proof}
if $\V{A}$ has the SVD as $\V{A}=\V{U\Sigma V}\hr$, then 
\begin{eqnarray*}
\V{A}&=&\V{U\Sigma V}\hr\\
\V{A}\hr\V{A}&=&\V{V\Sigma}\hr\underbrace{\V{U}\hr\V{U}}_{\V{I}_M}\V{\Sigma V}\hr\\
&=&\V{V}(\V{\Sigma}\hr\V{\Sigma})\V{V}\hr\\
&=&\V{V}(\V{\Sigma}\hr\V{\Sigma})\V{V}^{-1}
\end{eqnarray*}
where $\V{V}\hr=\V{V}^{-1}$ since $\V{V}$ is unitary. In the last line we see that matrices $\V{A}\hr\V{A}$ and $\V{\Sigma}\hr\V{\Sigma}$ are \emph{similar} matrices, and therefore they have the same rank. Now, the rank of $\V{\Sigma}\hr\V{\Sigma}$ is the number of non-zero square singular values of $\V{A}$ which is the same as the number of non-zero singular values of $\V{A}$ which is from theorem \ref{theo:rank_SV} the rank of $\V{A}$. Therefore the rank of $\V{A}\hr\V{A}$ equals the rank of $\V{A}$.
\end{proof}


\abbrev{EVD}{Eigen-Value Decomposition}
\abbrev{SVD}{Singular Value Decomposition}

\chapter{Diversity Proofs}
\label{App:Lengthy_Proofs}
\section{Full Diversity of OSTBCs}
\label{s:full_diversity_OSTBCs}
Consider an OSTBC defined by the dispersive form
\begin{equation}
\label{eq:dispersive_form2_again}
 \V{V}=
\frac{1}{\sqrt{p}}\sum\limits_{i=1}^{K}{\V{U}_ix_i^R+j\V{Q}_ix_i^I}
\end{equation}
Using the properties of dispersion matrices for OSTBCs
\begin{eqnarray*}
&\V{U}_i\hr\V{U}_i=\V{I}_M, \: \V{Q}_i\hr\V{Q}_i=\V{I}_M & 1\leq i\leq K\\
&\V{U}_i\hr\V{U}_d=-\V{U}_d\hr\V{U}_i, \:\: \V{Q}_i\hr\V{Q}_d=-\V{Q}_d\hr\V{Q}_i &  1\leq u\neq i\leq K\\
&\V{U}_i\hr\V{Q}_d=\V{Q}_d\hr\V{U}_i & 1\leq u,i\leq  K,
\end{eqnarray*}
any OSTBC can be proved to achieve full diversity. As defined in Section \ref{s:Design_Criteria}, a STC achieves full diversity if the distance matrix $\V{D}_{ll'}$ or equivalently the squared distance matrix $\V{D}_{ll'}\hr\V{D}_{ll'}$ between any two code matrices $\V{V}_{l}$ and $\V{V}_{l'}\:\forall\:l\neq l'\in\{0,...,L\}$ is of full rank. 
Using the dispersive form of the code matrix $\V{V}$ in (\ref{eq:dispersive_form2_again}), the distance matrix $\V{D}_{ll'}$ can be written as 
\begin{equation*}
 \V{D}_{ll'}=\frac{1}{\sqrt{p}}\sum\limits_{i=1}^{K}{\V{U}_i(x_{i,l}^R-x_{i,l'}^R)+j\V{Q}_i(x_{i,l}^I-x_{i,l'}^I)},
\end{equation*}
where code $\V{V}_l$ carries symbols $x_{i,l}=x_{i,l}^R+jx_{i,l}^I$, and $\V{V}_{l'}$ carries symbols $x_{i,l'}=x_{i,l'}^R+jx_{i,l'}^I$
and the squared distance matrix $\V{D}_{ll'}\hr\V{D}_{ll'}$ is
\begin{eqnarray}
  \V{D}_{ll'}\hr\V{D}_{ll'}&=&\frac{1}{p}\bigg[\sum\limits_{i=1}^{K}{\V{U}_i\hr\underbrace{(x_{i,l}^R-x_{i,l'}^R)}_{\Delta x_i^R}-j\V{Q}_i\hr\underbrace{(x_{i,l}^I-x_{i,l'}^I)}_{\Delta x_i^I} }\bigg]\bigg[\sum\limits_{d=1}^{K}{\V{U}_d\underbrace{(x_{d,l}^R-x_{d,l'}^R)}_{\Delta x_d^R}+j\V{Q}_d\underbrace{(x_{d,l}^I-x_{d,l'}^I)}_{\Delta x_d^I}}\bigg]\nonumber\\
&=&\frac{1}{p}\bigg[\sum\limits_{i=1}^{K}\underbrace{\V{U}_i\hr\Delta x_i^R-j\V{Q}_i\hr\Delta x_i^I}_{\text{first block}}\bigg]\bigg[\sum\limits_{d=1}^{K}\underbrace{\V{U}_d\Delta x_d^R+j\V{Q}_d\Delta x_d^I}_{\text{second block}}\bigg]
\label{eq:V_hr_V_disp}
\end{eqnarray}
The resulting terms of the above multiplication can be divided into same-index terms and different-index terms. Same index terms result when both blocks are at the same index, let it be $i$, and will have the form
\begin{eqnarray}
\label{eq:same_index_terms}
 &\frac{1}{p}\bigg[\V{U}_i\hr\V{U}_i(\Delta x_i^R)^2+\V{Q}_i\hr\V{Q}_i(\Delta x_i^I)^2 +\Delta x_i^R\Delta x_i^I\underbrace{(-j\V{Q}_i\hr\V{U}_i+j\V{U}_i\hr\V{Q}_i)}_{\text{cancel out}}\bigg]\nonumber\\
&=\frac{((\Delta x_i^R)^2+(\Delta x_i^I)^2)}{p}\V{I}_{M}.
\end{eqnarray}
For different-index terms, consider together the addition of the term resulting from the multiplication of the first block at index $i$ with the second block at index $d$, and the term resulting from the multiplication of the first block at index $d$ with the second block at index $i$, namely
\begin{eqnarray}
\label{eq:different_index_terms}
&\frac{1}{p}\bigg[(\V{U}_i\hr\Delta x_i^R-j\V{Q}_i\hr\Delta x_i^I)(\V{U}_d\Delta x_d^R+j\V{Q}_d\Delta x_d^I)+(\V{U}_d\hr\Delta x_d^R-j\V{Q}_d\hr\Delta x_d^I)(\V{U}_i\Delta x_i^R+j\V{Q}_i\Delta x_i^I)\bigg]\nonumber\\
&=\frac{1}{p}\bigg[\underbrace{(\V{U}_i\hr\V{U}_d+\V{U}_d\hr\V{U}_i)}_{\text{cancel out}}\Delta x_i^R\Delta x_d^R+\underbrace{(\V{Q}_i\hr\V{Q}_d+\V{Q}_d\hr\V{Q}_i)}_{\text{cancel out}}\Delta x_i^I\Delta x_d^I\nonumber\\
&+\Delta x_d^R\Delta x_i^I(\underbrace{-j\V{Q}_i\hr\V{U}_d+j\V{U}_d\hr\V{Q}_i}_{\text{cancel out}})+\Delta x_i^R\Delta x_d^I(\underbrace{j\V{U}_i\hr\V{Q}_d-j\V{Q}_d\hr\V{U}_i}_{\text{cancel out}})\bigg]
\end{eqnarray}
Therefore, all different-index terms cancel out, and matrix $\V{D}_{ll'}\hr\V{D}_{ll'}$ becomes
\begin{equation}
\label{eq:Dll_dash_OSTBC}
 \V{D}_{ll'}\hr\V{D}_{ll'}=\frac{\bigg(\sum\limits_{i=1}^{K}{((\Delta x_i^R)^2+(\Delta x_i^I)^2)}\bigg)}{p}\V{I}_{M}.
\end{equation}\\

Since code matrices $\V{V}_l$ and $\V{V}_{l'}$ are different, it follows that at lease one symbol differs making the summation in (\ref{eq:Dll_dash_OSTBC}) never zero. Matrix  $\V{D}_{ll'}\hr\V{D}_{ll'}$ is therefore a full rank matrix for any arbitrary complex constellation, which proves that any OSTBC unconditionally achieves full diversity.
\section{Properties of MDC-QOSTBCs}
\label{s:Properties_of_MDC-QOSTBCs}
In this section we derive some properties of the MDC-QOSTBC used in Chapter \ref{chap:QOSTBCs}. We first derive an expression for $\V{V}\hr\V{V}$, with $\V{V}$ being the code matrix of an $(M,K,T)$ MDC-QOSTBC. Then we conclude an expression for the minimum determinant of the distance matrix which governs the diversity order of the code. \\

Since both OSTBCs and any QOSTBC that satisfies the MDC-QO constraints in (\ref{eq:MDC_QOSTBC_constraints})-- will  be referred to as general MDC-QOSTBC -- share these constraints for unequal indices $i$ and $d$, we will refer to the previous section in deriving the form of $\V{V}\hr\V{V}$. Note that although in the previous section the derivation was done for the distance matrix $\V{D}_{ll'}$, the equations still have the same form if they were instead derived  for the code matrix $\V{V}$. The only difference is in replacing the difference symbols $\Delta x_i$ by $x_i$. \\

By grasping the same form of (\ref{eq:V_hr_V_disp}), we can write $\V{V}\hr\V{V}$ as 
\begin{equation}
\label{eq:V_hr_V_disp_QO}
  \V{V}\hr\V{V}=\frac{1}{p}\bigg[\sum\limits_{i=1}^{K}\underbrace{\V{U}_i\hr x_i^R-j\V{Q}_i\hr x_i^I}_{\text{first block}}\bigg]\bigg[\sum\limits_{d=1}^{K}\underbrace{\V{U}_d x_d^R+j\V{Q}_d x_d^I}_{\text{second block}}\bigg]\nonumber\\
\end{equation}
If we do the same analysis for same-index terms and different-index terms as done in the previous section, we conclude the following. For different-index terms, all OSTBC properties used in (\ref{eq:different_index_terms}) still hold for a general MDC-QOSTBC, therefore the different-index terms ($i\neq d$) cancel out. However, for the same-index terms, no specific properties are defined for a general MDC-QOSTBC and therefore, all we can say is that their code matrices satisfy the following
\begin{equation}
\label{eq:V_hr_V_MDC_general}
 \V{V}\hr\V{V}=\frac{1}{p}\sum\limits_{i=1}^{K}{\V{U}_i\hr\V{U}_i( x_i^R)^2+\V{Q}_i\hr\V{Q}_i( x_i^I)^2 + x_i^R x_i^I(-j\V{Q}_i\hr\V{U}_i+j\V{U}_i\hr\V{Q}_i)}.
\end{equation}

On the other hand, for the used MDC-QOSTBC based on the code construction in (\ref{eq:MDC_construction}), it can be verified that these codes possess the following two additional properties
\begin{equation}
\begin{aligned}
 &\text{(iv)}\,\,\V{U}_i\hr\V{U}_i=\V{I}_M,\,\V{Q}_i\hr\V{Q}_i=\V{I}_M \hspace{1cm}1\leq i\leq K\\
 &\text{(v)}\,\,\V{U}_i\hr\V{Q}_i=-\V{Q}_i\hr\V{U}_i=\V{Q}_{i+\frac{K}{2}}\hr\V{U}_{i+\frac{K}{2}}=-\V{U}_{i+\frac{K}{2}}\hr\V{Q}_{i+\frac{K}{2}}=j\begin{bmatrix}\V{0}&\V{I}_{\frac{M}{2}}\\\V{I}_{\frac{M}{2}}&\V{0} \end{bmatrix} \hspace{0.2cm}1\leq i\leq\frac{K}{2}
\end{aligned}.
\end{equation}
Using a similar proof like in (\ref{eq:p_eq_K}), it is easy to see that here also $p=K$ holds. Now using the additional two properties back in (\ref{eq:V_hr_V_MDC_general}), we get

\begin{eqnarray}
\label{eq:V_hr_V_MDC}
 \V{V}\hr\V{V}&=&\frac{1}{K}\sum\limits_{i=1}^{K}{|x_i|^2\V{I}_{M} + 2j\,x_i^R x_i^I\V{U}_i\hr\V{Q}_i}\nonumber\\
&=&\frac{1}{K}\big[\underbrace{\sum\limits_{i=1}^{K}{|x_i|^2}}_{\alpha}\V{I}_{M} + \underbrace{2\sum\limits_{i=1}^{\frac{K}{2}}{-x_i^R x_i^I+x_{i+\frac{K}{2}}^R x^I_{i+\frac{K}{2}}}}_{\beta}\begin{bmatrix}\V{0}&\V{I}_{\frac{M}{2}}\\\V{I}_{\frac{M}{2}}&\V{0} \end{bmatrix}\Big]\nonumber\\
&=&\frac{\alpha}{K}\V{I}_M+\frac{\beta}{K}\begin{bmatrix}\V{0}&\V{I}_{\frac{M}{2}}\\\V{I}_{\frac{M}{2}}&\V{0} \end{bmatrix}
\end{eqnarray}
where $\alpha\eq\sum\limits_{i=1}^{K}{|x_i|^2}$, and $\beta\eq2\sum\limits_{i=1}^{\frac{K}{2}}{-x_i^R x_i^I+x_{i+\frac{K}{2}}^R x^I_{i+\frac{K}{2}}}$. \\

Next, we derive the full diversity condition for MDC-QOSTBCs. To study diversity, one needs to evaluate the minimum determinant of the distance matrix $\V{D}$. Similar to the code matrix which satisfies (\ref{eq:V_hr_V_MDC}), the distance matrix satisfies
\begin{equation}
\label{eq:D_hr_D_MDC_QOSTBC}
\V{D}\hr\V{D}=\frac{\alpha'}{K}\V{I}_M+\frac{\beta'}{K}\begin{bmatrix}\V{0}&\V{I}_{\frac{M}{2}}\\\V{I}_{\frac{M}{2}}&\V{0} \end{bmatrix}
\end{equation}
where 
\begin{equation}
\label{eq:alpha_beta_dash}
\alpha'=\sum\limits_{i=1}^{K}{|\Delta x_i|^2}, \,\,\beta'\eq2\sum\limits_{i=1}^{\frac{K}{2}}{-\Delta x_i^R \Delta x_i^I+\Delta x_{i+\frac{K}{2}}^R \Delta x^I_{i+\frac{K}{2}}},
 \end{equation}
with $\Delta x_i=\Delta x_i^R+\Delta x_i^I$ is the difference between the symbols at index $i$ in the code matrix. The determinant of $\V{D}\hr\V{D}$ can be derived from (\ref{eq:D_hr_D_MDC_QOSTBC}) to be \cite{Yuen_MDC_QOSTBC_2005_V2}
\begin{equation}
 \det(\V{D}\hr\V{D})\Big|_{\text{min}}=\frac{1}{K}[(\alpha' + \beta')(\alpha' - \beta')]^{\frac{M}{2}}.
\end{equation}
The minimum determinant occurs when the distance matrix is as sparse as possible which is the case when only one symbol is different between the pair of code matrices considered. This results in dropping the summation and index $i$ in (\ref{eq:alpha_beta_dash}) making
\begin{equation*}
 \begin{aligned}
  \alpha'+\beta'=(\Delta x^R)^2+(\Delta x^I)^2-2\Delta x^R\Delta x^I=(\Delta x^R-\Delta x^I)^2\\
  \alpha'-\beta'=(\Delta x^R)^2+(\Delta x^I)^2+2\Delta x^R\Delta x^I=(\Delta x^R+\Delta x^I)^2\\
 \end{aligned}
\end{equation*}
reducing the minimum determinant to
\begin{equation}
\label{eq:det_min_MDC_QOSTBC}
 \det(\V{D}\hr\V{D})\Big|_{\text{min}}=\frac{1}{K}[(\Delta x^R)^2-(\Delta x^I)^2]^M
\end{equation}
In order to achieve full diversity, (\ref{eq:det_min_MDC_QOSTBC}) should be made non-zero. In words, for MDC-QOSTBCs to achieve full diversity, the absolute difference between the real parts between any two points in the constellation should not be the same as the absolute difference between their imaginary parts. To achieve optimal coding gain, (\ref{eq:det_min_MDC_QOSTBC}) should be maximized.

\section{Diversity Order of Differential MDC-QOSTBCs}
\label{s:Diversity_proof_DQOSTBC}
In this section, it is aimed to derive the diversity order of the differential MDC-QOSTBCs used in section \ref{s:DQOSTBC}.  We first derive the actual information matrix from the information matrices of the equivalent subsystems. Then we get an expression for the minimum determinant of the distance matrix which is the difference between two possible actual information matrices. This minimum determinant governs the diversity order of the system and the conditions needed for full diversity. The following derivation holds for both combined and un-combined differential QOSTBCs proposed in Section \ref{ss:QOSTBC_full_diversity} and  \ref{ss:QOSTBC_half_diversity} when four transmit antennas are used.\\

When constructing the matrices of the equivalent subsystems, the combined scheme of Section \ref{ss:QOSTBC_full_diversity} takes linear combinations of symbols $v_1,...,v_4$, while the un-combined scheme in Section \ref{ss:QOSTBC_half_diversity} does not. So to make the following derivation valid for both schemes, we define symbols $c_1,c_2$ as the entries of the information matrix $\V{V}_{z_\tau}^{1\E}$ in the first equivalent subsystem, and $c_3,c_4$ as the entries of the information matrix $\V{V}_{z_\tau}^{2\E}$ in the second equivalent subsystem, i.e.
\begin{equation}
 \V{V}_{z_\tau}^{1\E}=\frac{1}{\sqrt{p_1}}\begin{bmatrix}c_1 & c_2 \\ -c_2^* & c_1^*\end{bmatrix}, \V{V}_{z_\tau}^{2\E}=\frac{1}{\sqrt{p_2}}\begin{bmatrix}c_3 & c_4 \\ -c_4^* & c_3^*\end{bmatrix}.
\end{equation}\\

In the un-combined scheme $c_i=v_i$ $\forall\:i=1,...,4$, whereas in the combined one
\begin{equation}
 c_1=v_1+v_3, \,c_2=v_2+v_4,\, c_3=v_1-v_3,\,c_4=v_2-v_4.
\end{equation}
The differential encoding equations of the two subsystems can be expanded as
\begin{equation*}
\begin{aligned}
\V{S}_\tau^{1\E}&=\frac{\V{V}_{z_\tau}^{1\E}\V{S}_{\tau-1}^{1\E}}{a_{\tau-1}^{1\E}}\\
\begin{bmatrix}(s_{1,\tau}+s_{3,\tau}) & (s_{2,\tau}+s_{4,\tau})\\-(s_{2,\tau}+s_{4,\tau})^* & (s_{1,\tau}+s_{3,\tau})^*\end{bmatrix}&=\frac{1}{\sqrt{p_1}\,a_{\tau-1}^{1\E}}\begin{bmatrix}c_1 & c_2 \\ -c_2^* & c_1^*\end{bmatrix}\begin{bmatrix}(s_{1,\tau-1}+s_{3,\tau-1}) & (s_{2,\tau-1}+s_{4,\tau-1})\\-(s_{2,\tau-1}+s_{4,\tau-1})^* & (s_{1,\tau-1}+s_{3,\tau-1})^*\end{bmatrix}
\end{aligned}
\end{equation*}
 and
\begin{equation*}
\begin{aligned}
\V{S}_\tau^{2\E}&=\frac{\V{V}_{z_\tau}^{2\E}\V{S}_{\tau-1}^{2\E}}{a_{\tau-1}^{2\E}}\\
\begin{bmatrix}(s_{1,\tau}-s_{3,\tau}) & (s_{2,\tau}-s_{4,\tau})\\-(s_{2,\tau}-s_{4,\tau})^* & (s_{1,\tau}-s_{3,\tau})^*\end{bmatrix}&=\frac{1}{\sqrt{p_2}\,a_{\tau-1}^{1\E}}\begin{bmatrix}c_3 & c_4 \\ -c_4^* & c_3^*\end{bmatrix}\begin{bmatrix}(s_{1,\tau-1}-s_{3,\tau-1}) & (s_{2,\tau-1}-s_{4,\tau-1})\\-(s_{2,\tau-1}-s_{4,\tau-1})^* & (s_{1,\tau-1}-s_{3,\tau-1})^*\end{bmatrix},
\end{aligned}
\end{equation*}
where $s_{i,\tau}$ for $i=1,...,4$ are the signals contained in the actual transmit matrix $\V{S}_\tau$. To get $s_{i,\tau}$ in terms of $s_{i,\tau-1}$ for $i=1,...,4$, consider the equations of the first row of $\V{S}_\tau^{1\E}$ and $\V{S}_\tau^{2\E}$
\begin{eqnarray}
 \label{eq:s_s_1_1}s_{1,\tau}+s_{3,\tau}&=\frac{1}{\sqrt{p_1}\,a_{\tau-1}^{1\E}}\left(c_1(s_{1,\tau-1}+s_{3,\tau-1})-c_2(s_{2,\tau-1}+s_{4,\tau-1})^*\right)\\
\label{eq:s_s_1_2}s_{1,\tau}-s_{3,\tau}&=\frac{1}{\sqrt{p_2}\,a_{\tau-1}^{2\E}}\left(c_3(s_{1,\tau-1}-s_{3,\tau-1})-c_4(s_{2,\tau-1}-s_{4,\tau-1})^*\right)\\
\label{eq:s_s_1_3}s_{2,\tau}+s_{4,\tau}&=\frac{1}{\sqrt{p_1}\,a_{\tau-1}^{1\E}}\left(c_1(s_{2,\tau-1}+s_{4,\tau-1})+c_2(s_{1,\tau-1}+s_{3,\tau-1})^*\right)\\
\label{eq:s_s_1_4}s_{2,\tau}-s_{4,\tau}&=\frac{1}{\sqrt{p_2}\,a_{\tau-1}^{2\E}}\left(c_3(s_{2,\tau-1}-s_{4,\tau-1})+c_4(s_{1,\tau-1}-s_{3,\tau-1})^*\right)
\end{eqnarray}
It has been proved that $p_1=p_2=p$ for both schemes. Adding and subtracting (\ref{eq:s_s_1_1}) and (\ref{eq:s_s_1_2}) as well as (\ref{eq:s_s_1_3}) and (\ref{eq:s_s_1_4}), one can get $s_{i,\tau}$ in terms of $s_{i,\tau-1}$, $i=1,...,4$ as follows
\begin{equation}
\begin{aligned}
 s_{1,\tau}&=\Bigg[\overbrace{\frac{1}{2\sqrt{p}}\left(\frac{c_1}{a_{\tau-1}^{1\E}}+\frac{c_3}{a_{\tau-1}^{2\E}}\right)}^{v_1'}s_{1,\tau-1}-\overbrace{\frac{1}{2\sqrt{p}}\left(\frac{c_2}{a_{\tau-1}^{1\E}}+\frac{c_4}{a_{\tau-1}^{2\E}}\right)}^{v_2'}s_{2,\tau-1}^*\\
&+\overbrace{\frac{1}{2\sqrt{p}}\left(\frac{c_1}{a_{\tau-1}^{1\E}}-\frac{c_3}{a_{\tau-1}^{2\E}}\right)}^{v_3'}s_{3,\tau-1}-\overbrace{\frac{1}{2\sqrt{p}}\left(\frac{c_2}{a_{\tau-1}^{1\E}}-\frac{c_4}{a_{\tau-1}^{2\E}}\right)}^{v_4'}s_{4,\tau-1}^*\Bigg]
\end{aligned},
\label{eq:actual_info_syms_derivation}
\end{equation}
i.e.
\begin{equation*}
 s_{1,\tau}=v_1's_{1,\tau-1}-v_2's_{2,\tau-1}^*+v_3's_{3,\tau-1}-v_4's_{4,\tau-1}^*.
\end{equation*}
Similarly, the following relations can as well be concluded;
\begin{eqnarray*}
 s_{2,\tau}&=&v_1's_{2,\tau-1}+v_2's_{1,\tau-1}^*+v_3's_{4,\tau-1}+v_4's_{3,\tau-1}^*\\
s_{3,\tau}&=&v_1's_{3,\tau-1}-v_2's_{4,\tau-1}^*+v_3's_{1,\tau-1}-v_4's_{2,\tau-1}^*\\
s_{4,\tau}&=&v_1's_{4,\tau-1}+v_2's_{3,\tau-1}^*+v_3's_{2,\tau-1}+v_4's_{1,\tau-1}^*.
\end{eqnarray*}
Using the above relations, the following can be verified
\begin{eqnarray}
 \left[\begin{array}{cc:cc}s_{1,\tau} & s_{2,\tau} & s_{3,\tau} & s_{4,\tau}\\
        -s_{2,\tau}^* & s_{1,\tau}^* & -s_{4,\tau}^* & s_{3,\tau}^*\\ \hdashline
	s_{3,\tau} & s_{4,\tau} & s_{1,\tau} & s_{2,\tau}\\ 
	-s_{4,\tau}^* & s_{3,\tau}^* & -s_{2,\tau}^* & s_{1,\tau}^* 
\end{array}\right]&\eq&
 \left[\begin{array}{cc:cc}v_1' & v_2' & v_3' & v_4'\\
        -v_2'^* & v_1'^* & -v_4'^* & v_3'^*\\ \hdashline
	v_3' & v_4' & v_1' &v_2'\\ 
	-v_4'^* & v_3'^* & -v_2'^* & v_1'^* 
\end{array}\right]
 \left[\begin{array}{cc:cc}s_{1,\tau-1} & s_{2,\tau-1} & s_{3,\tau-1} & s_{4,\tau-1}\\
        -s_{2,\tau-1}^* & s_{1,\tau-1}^* & -s_{4,\tau-1}^* & s_{3,\tau-1}^*\\ \hdashline
	s_{3,\tau-1} & s_{4,\tau-1} & s_{1,\tau-1} & s_{2,\tau-1}\\ 
	-s_{4,\tau-1}^* & s_{3,\tau-1}^* & -s_{2,\tau-1}^* & s_{1,\tau-1}^* 
\end{array}\right]\nonumber\\
\V{S}_{\tau}&\eq&\V{V}'_{z_\tau}\V{S}_{\tau-1}\label{eq:actual_diff_transmission}
\end{eqnarray}
From the last relation, matrix $\V{V}'_{z_\tau}$ is the information matrix of the \emph{actual} differential QOSTBC system. As shown, all matrices in the actual system also have an "ABBA" structure. From (\ref{eq:actual_info_syms_derivation}), the actual information symbols are
\begin{equation}
\boxed{
 \begin{aligned}
  v_1'=&\frac{1}{2\sqrt{p}}\left(\frac{c_1}{a_{\tau-1}^{1\E}}+\frac{c_3}{a_{\tau-1}^{2\E}}\right),\hspace{1cm}
  v_2'=&\frac{1}{2\sqrt{p}}\left(\frac{c_2}{a_{\tau-1}^{1\E}}+\frac{c_4}{a_{\tau-1}^{2\E}}\right)\\
  v_3'=&\frac{1}{2\sqrt{p}}\left(\frac{c_1}{a_{\tau-1}^{1\E}}-\frac{c_3}{a_{\tau-1}^{2\E}}\right),\hspace{1cm}
  v_4'=&\frac{1}{2\sqrt{p}}\left(\frac{c_2}{a_{\tau-1}^{1\E}}-\frac{c_4}{a_{\tau-1}^{2\E}}\right).
 \end{aligned}
\label{eq:actual_info_syms}}
\end{equation}\\
\subsubsection{Diversity Order of the Differential Combined MDC-QOSTBC}
In the combined MDC-QOSTBC scheme in Section \ref{ss:QOSTBC_full_diversity}, symbols $c_i$ are
\begin{equation}
\begin{aligned}
 c_1&=v_1+v_3=(x_1^R-x_1^I)+j(x_3^R+x_3^I) \hspace{1cm} &c_2=&v_2+v_4=(x_2^R-x_2^I)+j(x_4^R+x_4^I)\\
c_3&=v_1-v_3=(x_1^R+x_1^I)+j(x_3^R-x_3^I) \hspace{1cm} &c_4=&v_2-v_4=(x_2^R+x_2^I)+j(x_4^R-x_4^I).
\end{aligned}
\end{equation}
Thus the actual information symbols in (\ref{eq:actual_info_syms}) become
\begin{equation}
 \begin{aligned}
  v_1'&=\frac{1}{2\sqrt{p}}\left(\frac{(x_1^R-x_1^I)+j(x_3^R+x_3^I)}{a_{\tau-1}^{1\E}}+\frac{(x_1^R+x_1^I)+j(x_3^R-x_3^I)}{a_{\tau-1}^{2\E}}\right)\\
  v_2'&=\frac{1}{2\sqrt{p}}\left(\frac{(x_2^R-x_2^I)+j(x_4^R+x_4^I)}{a_{\tau-1}^{1\E}}+\frac{(x_2^R+x_2^I)+j(x_4^R-x_4^I)}{a_{\tau-1}^{2\E}}\right)\\
 v_3'&=\frac{1}{2\sqrt{p}}\left(\frac{(x_1^R-x_1^I)+j(x_3^R+x_3^I)}{a_{\tau-1}^{1\E}}-\frac{(x_1^R+x_1^I)+j(x_3^R-x_3^I)}{a_{\tau-1}^{2\E}}\right)\\
 v_4'&=\frac{1}{2\sqrt{p}}\left(\frac{(x_2^R-x_2^I)+j(x_4^R+x_4^I)}{a_{\tau-1}^{1\E}}-\frac{(x_2^R+x_2^I)+j(x_4^R-x_4^I)}{a_{\tau-1}^{2\E}}\right).
 \end{aligned}
\label{eq:actual_info_symbols_full_diversity}
\end{equation}\\
We consider all information symbols $x_i$, $i=1,...,4$ to be drawn from the same constellation.
To get the minimum determinant of the distance matrix, one considers the least change between two information matrices.  First, for one information matrix, consider the worst case when $a_{\tau-1}^{1\E}\eq a_{\tau-1}^{2\E}\eq a_{\tau-1}$ which can happen in case $\beta_{\tau-1}$ in (\ref{eq:alpha_beta_combined_M_4_MDC_DQOSTBCs})
is zero. In this case the actual information symbols $v_i'$ reduce to
\begin{equation*}
 \begin{aligned}
  v_1'&=\frac{1}{\sqrt{p}a_{\tau-1}}(x_1^R+jx_3^R) \hspace{1cm} &v_2'&=\frac{1}{\sqrt{p}a_{\tau-1}}(x_2^R+jx_4^R)\\
  v_3'&=\frac{1}{\sqrt{p}a_{\tau-1}}(-x_1^I+jx_3^I) \hspace{1cm} &v_4'&=\frac{1}{\sqrt{p}a_{\tau-1}}(-x_2^I+jx_4^I).
 \end{aligned}
\end{equation*}
and the elements of the distance matrix are
\begin{equation*}
 \begin{aligned}
  \Delta v_1'&=\frac{1}{\sqrt{p}a_{\tau-1}}(\Delta x_1^R +j\Delta x_3^R) \hspace{1cm} &\Delta v_2'&=\frac{1}{\sqrt{p}a_{\tau-1}}(\Delta x_2^R +j\Delta x_4^R)\\
  \Delta v_3'&=\frac{1}{\sqrt{p}a_{\tau-1}}(-\Delta x_1^I +j\Delta x_3^I) \hspace{1cm} &\Delta v_4'&=\frac{1}{\sqrt{p}a_{\tau-1}}(-\Delta x_2^I +j\Delta x_4^I).
  \end{aligned}
\end{equation*}\\

In getting the minimum determinant, consider the change of only one information symbol, let it be $x_1$ without loss of generality. In this case, the elements of the distance matrix are
\begin{equation}
 \begin{aligned}
  \Delta v_1'&=\frac{\Delta x_1^R}{\sqrt{p}\,a_{\tau-1}}, \hspace{0.3cm} \Delta v_2'=0,\hspace{0.3cm} \Delta v_3'&=\frac{-\Delta x_1^I}{\sqrt{p}\,a_{\tau-1}},\hspace{0.3cm} \Delta v_4'=0.
  \end{aligned}
\label{eq:delta_v1}
\end{equation}\\

Similar to the information matrix $\V{V}'$ in (\ref{eq:actual_diff_transmission}), the distance matrix also has an "ABBA" structure with elements $\Delta v_i'$, $i=1,...,4$. When only $\Delta v_1'$ and $\Delta v_3'$ are non-zero, the distance matrix becomes
\begin{equation}
 \Delta \V{V}'=\begin{bmatrix}\begin{array}{cc:cc}
\Delta v_1' & 0 & \Delta v_3' & 0\\
0 & \Delta v_1'^* &  0 & \Delta v_3'^*\\ \hdashline
\Delta v_3' & 0 & \Delta v_1' & 0\\
0 & \Delta v_3'^* & 0 & \Delta v_1'^*
\end{array}\end{bmatrix},
\label{eq:distance_matrix_MDC_QOSTBC}
\end{equation}
whose determinant can be obtained through Gaussian elimination by reducing the matrix to an upper triangular form then multiplying the elements of the main diagonal resulting in
\begin{eqnarray}
 \det(\Delta \V{V}')\Big|_{\text{min}}&=&\mini (\Delta v_1'^2- \Delta v_3'^2)(\Delta v_1'^{*2}- \Delta v_3'^{*2})\\
&=&\mini (\Delta v_1'^2- \Delta v_3'^2)^2 \propto \mini((\Delta x^R)^2-(\Delta x^I)^2)^2,
\label{eq:coding_gain_combined_MDC_QOSTBC}
\end{eqnarray}
Thus the differential combined MDC-QOSTBC achieves full diversity if $|\Delta v_1'|\neq|\Delta v_3'|$, which from (\ref{eq:delta_v1}) is equivalent to 
\begin{equation}
 |\Delta x^R|\neq|\Delta x^I|,
\label{eq:full_diversity_condition_MDC_QOSTBC}
\end{equation}
for any two constellation points whose difference is $\Delta x=\Delta x^R+j\Delta x^I$.

\subsubsection{Diversity Order of the Differential Un-combined MDC-QOSTBC}
For the un-combined MDC-QOSTBC scheme of Section \ref{ss:QOSTBC_half_diversity}, symbols $c_i$ are
\begin{equation}
\begin{aligned}
 c_1&=v_1=x_1^R+jx_3^R \hspace{1cm} &c_2=&v_2=x_2^R+jx_4^R\\
 c_3&=v_3=-x_1^I+jx_3^I  &c_4=&v_4=-x_2^I+jx_4^I.
\end{aligned}
\end{equation}\\

Thus the actual information symbols in (\ref{eq:actual_info_syms}) become
\begin{equation}
 \begin{aligned}
  v_1'&=\frac{1}{2\sqrt{p}}\left(\frac{x_1^R+jx_3^R}{a_{\tau-1}^{1\E}}+\frac{-x_1^I+jx_3^I}{a_{\tau-1}^{2\E}}\right),\hspace{1cm}
  v_2'&=\frac{1}{2\sqrt{p}}\left(\frac{x_2^R+jx_4^R}{a_{\tau-1}^{1\E}}+\frac{-x_2^I+jx_4^I}{a_{\tau-1}^{2\E}}\right)\\
  v_3'&=\frac{1}{2\sqrt{p}}\left(\frac{x_1^R+jx_3^R}{a_{\tau-1}^{1\E}}-\frac{-x_1^I+jx_3^I}{a_{\tau-1}^{2\E}}\right),\hspace{1cm}
  v_4'&=\frac{1}{2\sqrt{p}}\left(\frac{x_2^R+jx_4^R}{a_{\tau-1}^{1\E}}-\frac{-x_2^I+jx_4^I}{a_{\tau-1}^{2\E}}\right),
 \end{aligned}
\label{eq:actual_info_syms_half_div}
\end{equation}
where $(a_{\tau}^{1\E})^2=\frac{1}{p}\sum\limits_{i=1}^{4}{(x_i^R)^2}$ and $(a_{\tau}^{2\E})^2=\frac{1}{p}\sum\limits_{i=1}^{4}{(x_i^I)^2}$. Since in the rectangular QAM used, $x_i^R$ and $x_i^I$ are drawn from the same one-dimensional alphabet, it can happen that $a_{\tau}^{1\E}\eq a_{\tau}^{2\E}$. so in (\ref{eq:actual_info_syms_half_div}), consider the worst case scenario w.r.t. diversity when $a_{\tau-1}^{1\E}\eq a_{\tau-1}^{2\E}\eq a_{\tau-1}$. In this case the actual information symbols reduce to
\begin{equation}
 \begin{aligned}
  v_1'&=\frac{1}{2\sqrt{p}\,a_{\tau-1}}\big((x_1^R-x_1^I)+j(x_3^R+x_3^I)\big),\hspace{1cm}
  v_2'&=\frac{1}{2\sqrt{p}\,a_{\tau-1}}\big((x_2^R-x_2^I)+j(x_4^R+x_4^I)\big),\\
  v_3'&=\frac{1}{2\sqrt{p}\,a_{\tau-1}}\big((x_1^R+x_1^I)+j(x_3^R-x_3^I)\big),\hspace{1cm}
  v_4'&=\frac{1}{2\sqrt{p}\,a_{\tau-1}}\big((x_2^R+x_2^I)+j(x_4^R-x_4^I)\big), 
 \end{aligned}
\label{eq:actual_info_syms_half_div2}
\end{equation}
In getting the minimum determinant, consider the change of only one information symbol, let it be $x_1$ without loss of generality. In this case, the elements of the distance matrix are
\begin{equation*}
 \begin{aligned}
  \Delta v_1'&=\frac{\Delta x_1^R-\Delta x_1^I}{\sqrt{p}\,a_{\tau-1}}, \hspace{0.3cm} \Delta v_2'=0,\hspace{0.3cm} \Delta v_3'&=\frac{\Delta x_1^R+\Delta x_1^I}{\sqrt{p}\,a_{\tau-1}},\hspace{0.3cm} \Delta v_4'=0.
  \end{aligned}
\end{equation*}
Similar to the previous section, the distance matrix has the same form as in (\ref{eq:distance_matrix_MDC_QOSTBC}) and its minimum determinant is 
\begin{equation}
 \det(\Delta \V{V}')\Big|_{\text{min}}=\mini\,(\Delta v_1'^2- \Delta v_3'^2)^2.
\label{eq:coding_gain_un_combined_MDC_QOSTBC}
\end{equation}
In the used rectangular QAM constellation, two constellation symbols can vary in only the real or the imaginary part. In this case $|\Delta v_1'|\eq|\Delta v_3'|$ resulting in a zero minimum determinant and the distance matrix will be
\begin{equation}
 \Delta \V{V}'=\begin{bmatrix}\begin{array}{cc:cc}
\Delta v_1' & 0 & \Delta v_1' & 0\\
0 & \Delta v_1' &  0 & \Delta v_1'\\ \hdashline
\Delta v_1' & 0 & \Delta v_1' & 0\\
0 & \Delta v_1' & 0 & \Delta v_1'
\end{array}\end{bmatrix},
\label{eq:distance_matrix_MDC_QOSTBC_rank2}
\end{equation}
which is of rank 2. Thus the diversity order achieved by the un-combined differential MDC-QOSTBC is 2, hence the name half-diversity.

\cleardoublepage
\phantomsection
\addcontentsline{toc}{chapter}{\bibname}
\bibliographystyle{ieeetr}
\bibliography{references} 

\end{document}